\documentclass[conference,romanappendices]{IEEEtran}
\pdfoutput=1

\newif\ifdraft
 \draftfalse

\newif\iffull
\fulltrue

\usepackage{mathtools}
\usepackage{stmaryrd}
\usepackage{amsthm}
\usepackage{float}
\usepackage{wrapfig}
\usepackage{framed}
\usepackage{xcolor}
\usepackage{xspace}
\usepackage{amssymb}
\usepackage[T1]{fontenc}
\usepackage[utf8]{inputenc}
\usepackage{lipsum}
\usepackage[normalem]{ulem}

\usepackage{url}
\usepackage{graphicx}
\usepackage{tikz}
\usepackage{proof}
\usepackage{placeins}
\usepackage{mathpartir}
\usepackage{hyperref}
\usepackage{enumerate}

\def\BibTeX{{\rm B\kern-.05em{\sc i\kern-.025em b}\kern-.08em
    T\kern-.1667em\lower.7ex\hbox{E}\kern-.125emX}}

\usetikzlibrary{decorations.pathreplacing}
\usetikzlibrary{intersections}

\input{macros.sty}

\newtheorem{definition}{Definition}
\newtheorem{proposition}{Proposition}
\newtheorem{theorem}{Theorem}
\newtheorem{lemma}{Lemma}
\newtheorem{corollary}{Corollary}

\newtheorem*{theorem*}{Theorem}
\newtheorem*{lemma*}{Lemma}
\newtheorem*{proposition*}{Proposition}

\theoremstyle{remark}
\newtheorem{remark}{Remark}
\newtheorem{example}{Example}

\renewcommand{\todo}[1]{}
\renewcommand{\toadd}[1]{}
\newcommand{\subparagraph}{\paragraph}

\begin{document}

\title{Deciding Indistinguishability\\
  {\Large
    Decidability of a Sound Set of Inference Rules for Computational Indistinguishability}}

\author{\IEEEauthorblockN{Adrien Koutsos}
  \IEEEauthorblockA{\textit{LSV, CNRS, ENS Paris-Saclay,
      Université Paris-Saclay} \\
    Cachan, France\\
    adrien.koutsos@lsv.fr}}

\maketitle

% \ifdraft
\thispagestyle{plain}
\pagestyle{plain}
% \fi

\begin{abstract}
  Computational indistinguishability is a key property in cryptography and verification of security protocols. Current tools for proving it rely on cryptographic game transformations.

  We follow Bana and Comon's approach~\cite{Bana:2014:CCS:2660267.2660276, DBLP:conf/post/BanaC12}, axiomatizing what an adversary cannot distinguish. We prove the decidability of a set of first-order axioms \ch{which are computationally sound, though incomplete, for protocols with a bounded number of sessions whose security is based on an $\textsc{ind-cca}_2$ encryption scheme.} Alternatively, our result can be viewed as the decidability of a family of cryptographic game transformations. Our proof relies on term rewriting and automated deduction techniques.
\end{abstract}


\section{Introduction}
\label{section:introduction}

Designing security protocols is notoriously hard. For example, the TLS protocol used to secure most of the Internet connections was successfully attacked several times at the protocol level, e.g. the \textsc{LogJam} attack~\cite{weakdh15} or the \textsc{TripleHandshake} attack~\cite{6956559}. This shows that, even for high visibility protocols, and years after their design, attacks are still found.

Using formal methods to prove a security property is the best way to get a strong confidence. However, there is a difficulty, which is not present in standard program verification: we need not only to specify formally the program and the security property, but also the attacker. Several attacker models have been considered in the literature.
% with different strengths and weaknesses.

A popular attacker model, the \emph{Dolev-Yao attacker}, grants the attacker the complete control of the network: he can intercept and re-route all messages. Besides, the adversary is allowed to modify messages using a fixed set of rules (e.g. given a cipher-text and its decryption key, he can retrieve the plain-text message). Formally, messages are terms in a term algebra and the rules are given through a set of rewrite rules. This model is very amenable to automatic verification of security properties. There are several automated tools, such as, e.g., \textsf{ProVerif}~\cite{proverif}, \textsf{Tamarin}~\cite{Meier:2013:TPS:2526861.2526920} and \textsf{Deepsec}~\cite{DBLP:conf/sp/ChevalKR18}.

Another attacker model, closer to a real world attacker, is the \emph{computational attacker} model. This adversary also controls the network, but this model does not restrict the attacker to a fixed set of operations: the adversary can perform any probabilistic polynomial time computation. More formally, messages are bit-strings, random numbers are sampled uniformly among bit-strings in $\{0,1\}^\eta$ (where $\eta$ is the \emph{security parameter}) and the attacker is any probabilistic polynomial-time Turing machine (PPTM). This model offers stronger guarantees than the Dolev-Yao model (DY model), but formal proofs are harder to complete and more error-prone. There exist several formal verification tools in this model: for example, \textsc{EasyCrypt}~\cite{DBLP:conf/crypto/BartheGHB11} which relies on pRHL, and \textsc{CryptoVerif}~\cite{DBLP:journals/tdsc/Blanchet08} which performs game transformations. As expected, such tools are less automatic than the verification tools in the DY model. Moreover, the failure to find a proof in such tools, either because the proof search failed or did not terminate, or because the user could not manually find a proof, does not give any indication on the actual security of the protocol.

There is an alternative approach, the Bana-Comon model (BC model), introduced in~\cite{DBLP:conf/post/BanaC12}. In this model, we  express the security of a protocol as the unsatisfiability of a set of formulas in first-order logic. The formulas contain the negation of the security property and \emph{axioms}, which reflect implementation assumptions, such as functional correctness and cryptographic hypotheses on the security primitives. This method has several advantages over pRHL and game transformations. First, it is simpler, as there is no security game and no probabilities, only a first-order formula. Then, carrying out a proof of unsatisfiability in this logic entails the security of the protocol in the computational model. Finally, the absence of such a proof implies the existence of a model of the formula, i.e. an attack, albeit not necessarily a computational one; nonetheless, we know that the security of the protocol \emph{cannot} be obtained without extra assumptions. \ch{Note that the Bana-Comon approach is only valid for protocols with a finite number of sessions (there is no unbounded replication). Since this is the model we use, we inherit this restriction.}

There is another input to security proofs that we did not discuss yet: the class of security properties considered. Roughly, there are two categories. \emph{Reachability} properties state that some bad state is unreachable. This includes, for example, authentication or (weak) secrecy. \emph{Indistinguishability} properties state that an adversary cannot distinguish between the executions of two protocols. This allows for more complex properties, such as strong secrecy and unlinkability.

\paragraph{Deciding Security}
When trying to prove a protocol, there are three possible outcomes: either we find a proof, which gives security guarantees corresponding to the attacker model; or we find an attack, meaning that the protocol is insecure; or the tool or the user (for interactive provers) could not carry out the proof and failed to find an attack. \ch{The latter case may happen for two different reasons. First, we could neither find a proof nor an attack because the proof method used is incomplete. In that case, we need either to make new assumptions and try again, or to use another proof technique. Second, the tool may not terminate on the protocol considered. This is problematic, as we do not know if we should continue waiting, and consume more resources and memory, or try another method.}
% In the last case, we have no idea whether the protocol is secure.

This can be avoided for decidable classes of protocols and properties.
% , ensuring that either an attack or a proof can be automatically found.
Of course, such classes depend on both the attacker model and the security properties considered. We give here a non-exhaustive survey of such results.
In the symbolic model, \cite{DBLP:journals/tocl/Comon-LundhCZ10} shows decidability of secrecy (a reachability property) for a bounded number of sessions. In~\cite{DBLP:conf/csfw/DOsualdoOT17}, the authors show the decidability of a secrecy property for \emph{depth-bounded} protocols, with an unbounded number of sessions, using Well-Structured Transition Systems~\cite{DBLP:journals/tcs/FinkelS01}. Chrétien et al~\cite{7243732} show the decidability of indistinguishability properties for a restricted class of protocols. E.g., they consider processes communicating on distinct channels and without \textsf{else} branches. The authors of~\cite{DBLP:journals/iandc/ChevalCD17} show the decidability of symbolic equivalence for a bounded number of sessions, but with conditional branching.

\ch{In the computational model, we are aware of only one direct result. In~\cite{DBLP:conf/cade/Comon-LundhCS13}, the authors show the decidability of the security of a formula in the BC model, for \emph{reachability properties}, for a bounded number of sessions. But there is an indirect way of getting decidability in the computational model, through a \emph{computational soundness} theorem (e.g.\ \cite{DBLP:journals/joc/AbadiR02}). A computational soundness theorem states that, for some given classes of protocols and properties, symbolic security implies computational security. These results usually make strong implementation assumptions (e.g. parsing assumptions, or the absence of dishonest keys), and require that the security primitives satisfy strong cryptographic hypothesis. By combining a decidability result in the symbolic model with a computational soundness theorem, which applies to the considered classes of protocols and properties (e.g.~\cite{DBLP:conf/ccs/BackesMU12} for reachability properties, or~\cite{DBLP:conf/post/BackesMR14} for indistinguishability properties), we obtain a decidability result in the computational model.}

\ch{We discuss further related works later, in Section~\ref{sec:related}.}
% We see that the state of the art lacks a decidability result for indistinguishability in the computational or BC model.

\paragraph{Contributions}
In this paper, we consider the BC model for indistinguishability properties~\cite{Bana:2014:CCS:2660267.2660276}. This is a first-order logic in which we design a set of axioms $\mathsf{Ax}$ which includes, in particular, axioms for the $\textsc{ind-cca}_2$ cryptographic assumption \cite{DBLP:conf/crypto/BellareDPR98}. Given a protocol and a security property, we can build, using a folding technique described in~\cite{Bana:2014:CCS:2660267.2660276}, a ground atomic formula $\psi$ expressing the security of the protocol. Showing the unsatisfiability of the conjunction of the axioms $\mathsf{Ax}$ and the negation of $\psi$ entails the security of the protocol in the computational model, \ch{assuming that the encryption scheme is $\textsc{ind-cca}_2$ secure.}

Our main result is the decidability of the problem:\\[0.3em]
\begin{tabular}{l}
  \textbf{Input:} A ground formula $\vec u \sim \vec v$.\\
  \textbf{Question:} Is $\textsf{Ax} \wedge \vec u \not \sim \vec v$ unsatisfiable?
\end{tabular}\\[0.5em]
\ch{That is, we show the decidability of a sound, though incomplete, axiomatization of computational indistinguishability.}

All the formulas in $\textsf{Ax}$ are Horn clauses, therefore to show the unsatisfiability of $\textsf{Ax} \wedge \vec u \not \sim \vec v$ we use resolution with a negative strategy: we see axioms in $\textsf{Ax}$ as inference rules and look for a derivation of the goal $\vec u \sim \vec v$. We prove the decidability of the corresponding satisfiability problem.
% Entscheidungsproblem.

The main difficulty lies in dealing with equalities (defined through a term rewriting system $R$). First we show the completeness of an ordered strategy by commuting rule applications. This allows us to have only one rewriting modulo $R$ at the beginning of the proof. We then bound the size of the terms after this rewriting as follows: we identify a class of proof cuts introducing arbitrary subterms; we give proof cut eliminations to remove them; and finally, we show that cut-free proofs are of bounded size w.r.t. the size of the conclusion.

\paragraph{Game Transformations}
Our result can be reinterpreted as the decidability of the problem of determining whether there exits a sequence of game transformations~\cite{cryptoeprint:2004:332, Bellare2006} that allows to prove the security of a protocol. Indeed, one can associate to every axiom in $\textsf{Ax}$ either a cryptographic assumption or a game transformation.

Each unitary axiom in $\textsf{Ax}$ (an atomic formula) corresponds to an instantiation of the $\textsc{ind-cca}_2$ game. For instance, in the simpler case of \textsc{ind-cpa} security of an encryption $\enc{\_}{\pk}{\nonce}$, no polynomial-time adversary can distinguish between two cipher-texts, even if it chooses the two corresponding plain-texts (here, $\nonce$ is the explicit encryption randomness). Initially, the public key $\pk$ is given to the adversary, who computes a pair of plain-texts $g(\pk)$: $g$ is interpreted as the adversary's computation. Then the two cipher-texts, corresponding to the encryptions of the first and second components of $g(\pk)$, should be indistinguishable. This yields the unitary axiom:
\[
  \enc{\pi_1(g(\pk))}{\pk}{\nonce} \sim \enc{\pi_2(g(\pk))}{\pk}{\nonce}
\]

Similarly, non-unitary axioms correspond to cryptographic game transformations. E.g., the function application axiom:
\[
  \vec u \sim \vec v \rightarrow f(\vec u) \sim f(\vec v)
\]
states that if no adversary can distinguish between the arguments of a function call, then no adversary can distinguish between the images. As for a cryptographic game transformation, the soundness of this axiom is shown by reduction. Given a winning adversary $\mathcal{A}$ against the conclusion $f(\vec u) \sim f(\vec v)$, we build a winning adversary $\mathcal{B}$ against $\vec u \sim \vec v$: the adversary $\mathcal{B}$, on input $\vec w$ (which was sampled from $\vec u$ or $\vec v$), computes $f(\vec w)$ and then gives the result to the distinguisher $\mathcal{A}$. The advantage of $\mathcal{B}$ against $\vec u \sim \vec v$ is then the advantage of $\mathcal{A}$ against $f(\vec u) \sim f(\vec v)$, which is (by hypothesis) non negligible.

By interpreting every axiom in $\textsf{Ax}$ as a cryptographic assumption or a game transformation, and the goal formula $\vec u \sim \vec v$ as the initial game, our result can be reformulated as showing the decidability of the following problem:\\[0.5em]
\textbf{Input:} An initial game $\vec u \sim \vec v$.\\
\textbf{Question:} \hspace{-0.55em}\begin{tabular}[t]{l}Is there a sequence of game transformations in $\textsf{Ax}$ \\showing that $\vec u \sim \vec v$ is secure?
\end{tabular}\\

From this point of view, our result guarantees a kind of sub-formula property for the intermediate games appearing in the game transformation proof. We may only consider intermediate games that are in a finite set computable from the original protocol: the other games are provably unnecessary detours.
% we can safely ignore games that have nothing to do with the protocol.
To our knowledge, our result is the first showing the decidability of a class of game transformations.

\paragraph{\ch{Scope and Limitations}}
\ch{To achieve decidability, we had to remove or restrict some axioms. The most important restriction is arguably that we do not include the transitivity axiom. The transitivity axiom states that to show that $\vec{u} \sim \vec{v}$, it is sufficient to find a $\vec{w}$ such that $\vec{u} \sim \vec{w}$ and $\vec{w} \sim \vec{v}$. Obviously, this axiom is problematic for decidability, as the vector of term $\vec{w}$ must be guessed, and may be arbitrarily large. Therefore, instead of directly including transitivity, we push it inside the $\cca$ axiom schema, by allowing instances of the $\cca$ axiom to deal simultaneously with multiple keys and interleaved encryptions. Of course, this is at the cost of a more complex axiom. We do not know if our problem remains decidable when we include the transitivity axiom.}

\paragraph{Applications}
The BC indistinguishability model has been used to analyse RFID protocols~\cite{DBLP:conf/csfw/ComonK17}, a key-wrapping API~\cite{DBLP:conf/csfw/ScerriS16} and an e-voting protocol~\cite{DBLP:conf/esorics/BanaCE18}. \ch{Ideally, we would like future case studies to be carried out automatically and machine checked. Because our procedure has a high complexity, it is unclear whether it can be used directly for this. Still, our procedure could be a building block in a tool doing an incomplete but faster heuristic exploration of the proof space.}

\textsc{CryptoVerif} and \textsc{EasyCrypt} are based on game transformations, directly in the former and through the pRHL logic in the latter. Therefore, our result could be used to bring automation to these tools. Of course, both tools allow for more rules. Still, we could identify which game transformations or rules correspond to our axioms, and apply our result to obtain decidability for this subset of game transformations.

\paragraph{Outline}
We introduce the logic and the axioms in Section~\ref{section:logic} and \ref{section:axioms}. We then state the main result in Section~\ref{section:difficulties}, and depict the difficulties of the proof. Finally we sketch the proof: in Section~\ref{section:sketch-proof} we show the rule commutations and some cut eliminations; in Section~\ref{section:decision} we show a normal form for proofs and its properties; and in Section~\ref{section:bounding-body} we give more cut eliminations and the decision procedure. We discuss in details the related works in Section~\ref{sec:related}.
\iffull
Most of the proofs are in appendix.
\else
For space reasons, most of the proofs are sketched or omitted. The full proofs can be found in the long version of this paper~\cite{long}.
\fi

%%% Local Variables:
%%% mode: latex
%%% TeX-master: "ms"
%%% End:

\section{The Logic}
\label{section:logic}

We recall here the logic introduced in~\cite{Bana:2014:CCS:2660267.2660276}. In this logic, terms represent messages of the protocol sent over the network, including the adversary's inputs, which are specified using additional function symbols. Formulas are built using the usual Boolean connectives and FO quantifiers, and a single predicate, $\sim$, which stands for indistinguishability. The semantics of the logic is the usual first-order semantics, though we are particularly interested in computational models, in which terms are interpreted as PPTMs, and $\sim$ is interpreted as computational indistinguishability.
 
This logic is then used as follows: given a protocol and a security property, we can build (automatically) a single formula $\vec u \sim \vec v$ expressing the security of the protocol. We specify, through a (recursive) set of axioms, what the adversary \emph{cannot} do. This yields a set of axioms $\mathsf{Ax}$. We show that $\mathsf{Ax} \wedge \vec u \not \sim \vec v$ is unsatisfiable, and that the axioms $\mathsf{Ax}$ are valid in the computational model. We deduce from this the security of the protocol in the computational model.

\subsection{Syntax}
\paragraph{Terms}
Terms are built upon a set of function symbols $\sig$, a countable set of names $\Nonce$ and a countable set of variables $\mathcal{X}$. This is a sorted logic with two sorts $\term,\bool$, with~$\bool \subseteq \term$.

The set $\sig$ of function symbols is composed of a countable set of adversarial function symbols $\mathcal{G}$ (representing the adversary computations), and the following function symbols% , which includes symbols for pairs and public key encryptions
: the pair $\pair{\_}{\_}$, projections $\pi_1,\pi_2$, public and private key generation $\pk(\_),\sk(\_)$, encryption with random seed $\enc{\_}{\_}{\_}$, decryption $\dec(\_,\_)$, $\symite$, $\true$, $\false$, zero $\zero(\_)$ and equality check $\eq{\_}{\_}$. We give their types below:
  \begin{alignat*}{4}
    &\pair{\_}{\_},\dec(\_,\_)&\;:\;& \term^2 \ra \term &\quad\;\;&
    \eq{\_}{\_} &\;:\;& \term^2 \rightarrow \bool\\
    &\pi_1,\pi_2,\zero,\pk,\sk&\;:\;&\term \ra \term &&
    \enc{\_}{\_}{\_}&\;:\;& \term^3 \ra \term\\
    &\symite &\;:\;& \bool \times \term^2 \ra \term&&
    \true,\false&\;:\;& \rightarrow \bool
  \end{alignat*}
  % \begin{mathpar}
  %   \pair{\_}{\_},\dec(\_,\_): \term^2 \ra \term
    
  %   \eq{\_}{\_} : \term^2 \rightarrow \bool

  %   \pi_1,\pi_2,\zero,\pk,\sk:\term \ra \term

  %   \enc{\_}{\_}{\_}: \term^3 \ra \term

  %   \true,\false: \rightarrow \bool

  %   \symite : \bool \times \term^2 \ra \term
  % \end{mathpar}
Moreover all the names in $\Nonce$ have sort $\term$, and each variable in $\mathcal{X}$ comes with a sort. We let $\ssig$ be $\sig$ without the $\symite$ function symbol, and for any subset $\mathcal{S}$ of the union of $\sig$, $\Nonce$ and $\mathcal{X}$, we let $\mathcal{T}(\mathcal{S})$ be the set of terms built upon $\mathcal{S}$.

\paragraph{Formulas}
For every integer $n$, we have one predicate symbol $\sim_{n}$ of arity $2n$, which represents equivalence between two vectors of terms of length $n$. Formulas are then obtained using the usual Boolean connectives and first-order quantifiers.

\paragraph{Semantics}
We use the classical first-order logic semantics: every sort is interpreted by some domain, and function symbols and predicates are interpreted as, resp., functions of the appropriate domains and relations on these domains.

We focus on a particular class of such models, the \emph{computational models}. We informally describe the properties of a computational model $\cmodel$ (a full description is given in~\cite{Bana:2014:CCS:2660267.2660276}):
\begin{itemize}
\item $\term$ is interpreted as the set of probabilistic polynomial time Turing machines equipped with a working tape and two random tapes $\rho_1,\rho_2$ (one for the protocol random values, the other for the adversary random samplings). Moreover its input is of length $\eta$ (the security parameter). $\bool$ is the restriction of $\term$ to machines that return $0$ or~$1$.

\item A name $\nonce \in \Nonce$ is interpreted as a machine that, on input of length $\eta$, extracts a word of length $\eta$ from the first random tape $\rho_1$. Furthermore we require that different names extract disjoint parts of $\rho_1$.

\item $\true$, $\false$, $\zero(\_)$, $\eq{\_}{\_}$, and $\symite$ are interpreted as expected. For instance, $\eq{\_}{\_}$ takes two machines $M_1$, $M_2$, and returns $M$ such that on input $w$ and random tapes $\rho_1,\rho_2$, $M$ returns 1 if $M_1(w,\rho_1,\rho_2)=M_2(w,\rho_1,\rho_2)$ and 0 otherwise. The function symbol $\zero$ is interpreted as the function that, on input of length $l$, returns the bit-string~$0^{l}$.

\item A function symbol $g \in \mathcal{G}$ with $n$ arguments is interpreted as a function $\sem{g}$ such that there is a polynomial-time Turing machine $M_g$ such that for every machines $(m_i)_{i \le n}$ in the interpretation domains, and for every inputs $w,\rho_1,\rho_2$:
  \[
    \sem{g}\big((m_i)_{i \le n})(w,\rho_1,\rho_2\big) =
    M_g\big((m_i(w,\rho_1,\rho_2))_{i \le n},\rho_2)
  \]
  Observe that $M_g$ cannot access directly the tape $\rho_1$.
\item Protocol function symbols are interpreted as deterministic polynomial-time Turing machine. Their interpretations will be restricted using \emph{implementation axioms} later.
\item The interpretation of function symbols is lifted to terms: given an assignment $\sigma$ of the variables of a term $t$ to elements of the appropriate domains, we write $\sem{t}_{\eta,\rho_1,\rho_2}^\sigma$ the interpretation of the term with respect to $\eta,\rho_1,\rho_2$. $\sigma$ is omitted when empty. We also omit the other parameters when they are irrelevant.
\item The predicate $\sim_n$ is interpreted as \emph{computational indistinguishability} $\approx$, defined by
  \(
    m_1,\ldots,m_n \approx m'_1,\ldots,m'_n
  \)
  iff for every PPTM $\mathcal{A}$ with random tape $\rho_2$:
\begin{alignat*}{2}
  \big|\,&\prob(\rho_1,\rho_2:  \mathcal{A}((m_i(1^\eta,\rho_1,\rho_2))_{1 \le i \le n},\rho_2)=1)& -\\
  &\prob(\rho_1,\rho_2: \mathcal{A}((m'_i(1^\eta,\rho_1,\rho_2))_{1 \le i \le n},\rho_2)=1) \,\big|&
\end{alignat*}
is negligible in $\eta$ (a function is negligible if it is asymptotically smaller than the inverse of any polynomial).

Moreover, for all ground terms $u,v$, we write $\cmodel \models u \sim v$ when $\sem{u}\approx\sem{v}$ in $\cmodel$.
\end{itemize}

\begin{example}
  \label{example:cs-intro-semantics}
  Let $\nonce_0,\nonce_1,\nonce \in \Nonce$ and $g \in \sig$ of arity zero. For every computational model~$\cmodel$:
  \[
    \cmodel \models \ite{g()}{\nonce_0}{\nonce_1} \sim \nonce
  \]
  Indeed, the term on the left represents the message obtained by letting the adversary choose a branch, and then sampling from $\nonce_0$ or $\nonce_1$ accordingly. This is semantically equivalent to directly performing a random sampling, as done on the right.
\end{example}

%%% Local Variables:
%%% mode: latex
%%% TeX-master: "ms"
%%% End:

\section{Axioms}
\label{section:axioms}
We present the axioms $\mathsf{Ax}$, which are of two kinds:
\begin{itemize}
\item \emph{structural axioms} represent properties that hold in every computational model. This includes axioms such as the symmetry of $\sim$, or properties of the $\symite$.
\item \emph{implementation axioms} reflect implementation assumptions, such as the functional correctness of the pair and projections (e.g. $\pi_1(\pair{u}{v}) = u$), or cryptographic assumptions on the security primitives (e.g. $\textsc{ind-cca}_2$).
\end{itemize}

All our axioms $\mathsf{Ax}$ are universally quantified Horn clauses. To show the unsatisfiability of $\mathsf{Ax} \wedge \vec u \not \sim \vec v$, we use resolution with a negative strategy (which is complete, see~\cite{DBLP:books/daglib/0070484}). As all axioms are Horn clauses, a proof by resolution with a negative strategy can be seen as a proof tree where each node is indexed by the axiom of $\mathsf{Ax}$ used at this resolution step. Hence, axioms will be given as inference rules (where variables are implicitly universally quantified).

\subsection{Equality and Structural Axioms}
Some notation conventions: we use $\vec u$ to denote a vector of terms; and we use an infix notation for $\sim$, writing $\vec u \sim \vec v$ when $\vec u$ and $\vec v$ are of the same length.

\ch{The equality and structural axioms we present here already appeared in the literature~\cite{Bana:2014:CCS:2660267.2660276,DBLP:conf/csfw/ComonK17,DBLP:journals/iacr/BanaC16}, sometimes with slightly different formulations.}

  % We argue for the expressiveness of our axioms later, in Section~\ref{example:proof-intro}, by providing examples of proofs.

\paragraph{Equality}
Computational indistinguishability is an equivalence relation (i.e. reflexive, symmetric and transitive).
% \footnote{Even though we omit transitivity, as we can avoid it in proofs. We do not know if our problem remains decidable when we include it.}
But we can observe that it is not a congruence. E.g. take a computational model $\cmodel$, we know that two names $\nonce$ and $\nonce'$ are indistinguishable (since they are interpreted as independent uniform random sampling in $\{0,1\}^\eta$), and $\nonce$ is indistinguishable from itself. Therefore:
\[
  \cmodel \models \nonce \sim \nonce'
  \quad\text{ and }\quad
  \cmodel \models \nonce \sim \nonce
\]
But there is a simple PPTM that can distinguish between $\pair{\nonce}{\nonce}$ and $\pair{\nonce'}{\nonce}$: simply test whether the two arguments are equal, if so return $1$ and otherwise return $0$. Then, with overwhelming probability, this machine will guess from which distribution its input was sampled from.

Even though $\sim$ is not a congruence, we can get a congruence from it: if $\eq{s}{t} \sim \true$ holds in all models then, using the semantics of $\eq{\_}{\_}$, in every computational model $\cmodel$, $\sem{s}$ and $\sem{t}$ are identical except for a negligible number of samplings. Hence we can replace any occurrence of $s$ by $t$ in a formula without changing its semantics with respect to computational indistinguishability.

We use this in our logic as follows: we let $s = t$ be a shorthand for $\eq{s}{t} \sim \true$, and we introduce a set of equalities $R$ (given in Fig.~\ref{fig:trs}) and its congruence closure $=_R$. We split $R$ in four sub-parts: $R_1$ contains the functional correctness assumptions on the pair and encryption; $R_2$ and $R_3$ contain, respectively, the homomorphism properties and simplification rules of the $\symite$; and $R_4$ allows to change the order in which conditional tests are performed.

We then introduce a recursive set of rules:
\begin{align*}
  \null\qquad\begin{array}[c]{c}
    \infer[R]{\vec u, s \sim \vec v}{\vec u, t \sim \vec v}
  \end{array}
  &&
  \left(
    s,t \text{ ground terms with }
    s =_R t
  \right)
\end{align*}
By orienting $R_1,R_2,R_3$ from left to right, and carefully choosing an orientation for the ground instances of $R_4$, we obtain a recursive term rewriting system $\ra_R$. We have the following theorem
\iffull
(proven in Appendix~\ref{app-section:trs}):
\else
(the proof is in the long version~\cite{long}):
\fi
\begin{theorem}
  The TRS $\ra_R$ is convergent on ground terms.
\end{theorem}

\begin{figure}[tb]
  \small
  \[\begin{array}{l}
      R_1
      \left\{\begin{array}{lcl}
          \pi_i ( \langle x_1,x_2\rangle) = x_i & & \eq{x}{x} = \true\\
          \dec(\{x\}_{\pk(y)}^z,\sk(y)) = x & & \\
        \end{array}\right.\\[0.7em]
      R_2
      \left\{\begin{array}{lrr}
          \lefteqn{f(\vec{u},\ite{b}{x}{y},\vec v) =}\\
          & \ite{b}{f(\vec{u},x,\vec v)}{f(\vec{u},y,\vec v)} & \qquad\quad(f \in \ssig) \\
          \lefteqn{\ite{(\ite{b}{a}{c})}{x}{y} =}\\
          & \ite{b}{(\ite{a}{x}{y})}{\lefteqn{(\ite{c}{x}{y})}}
        \end{array}\right.\\[1.7em]
      R_3
      \left\{\begin{array}{l}
          \ite{b}{x}{x} = x \\
          \ite{\true}{x}{y} = x \hfill
          \ite{\false}{x}{y} = y \\
          \ite{b}{(\ite{b}{x}{y})}{z} = \ite{b}{x}{z}\\
          \ite{b}{x}{(\ite{b}{y}{z})} = \ite{b}{x}{z}
        \end{array}\right.\\[2em]
      R_4
      \left\{\begin{array}{lrr}
          \lefteqn{\ite{b}{(\ite{a}{x}{y})}{z} =}\\
          & \ite{a}{(\ite{b}{x}{z})}{(\ite{b}{y}{z})} \\
          \lefteqn{\ite{b}{x}{(\ite{a}{y}{z})} =}\\
          & \ite{a}{(\ite{b}{x}{y})}{(\ite{b}{x}{z})}
        \end{array}\right.
    \end{array}\]
  \caption{\label{fig:trs} $R \;=\; R_1 \cup R_2 \cup R_3 \cup R_4$}
\end{figure}

\paragraph{Structural Axioms}
\begin{figure*}[tb]
  \small
  \begin{mathpar}
    \begin{gathered}[c]
      \infer[\perm]{
        u_1,\ldots,u_n \sim v_1,\ldots, v_n
      }{
        u_{\pi(1)},\ldots,u_{\pi(n)} \sim v_{\pi(1)},\ldots,v_{\pi(n)}
      }
    \end{gathered}
    
    \begin{gathered}[c]
      \infer[\restr]{\vec u \sim \vec v}{\vec u,t \sim \vec v,t'}
    \end{gathered}

    \text{for any $s =_R t$, }
    \begin{array}[c]{c}
      \infer[R]{\vec u, s \sim \vec v}{\vec u, t \sim \vec v}
    \end{array}

    \begin{gathered}[c]
      \infer[\fa]{
        f(\vec{u}_1),\vec{v}_1 \sim f(\vec{u}_2), \vec{v}_2
      }{
        \vec{u}_1,\vec{v}_1\sim \vec{u}_2, \vec{v}_2
      }
    \end{gathered}

    \begin{gathered}[c]
      \infer[\dup]{\vec u,t,t \sim \vec v,t',t'}{\vec u,t \sim \vec v,t'}
    \end{gathered}

    \begin{gathered}[c]
      \infer[\sym]{\vec u \sim \vec v}{\vec v \sim \vec u}
    \end{gathered}
    
    \text{for any $b,b' \in \mathcal{T}(\ssig,\Nonce)$,}
    \begin{array}[c]{c}
      \infer[\csm]{\vec w, (\ite{b}{u_i}{v_i})_i \sim \vec w',(\ite{b'}{u_i'}{v_i'})_i}
      {
        \vec w, b, (u_i)_i \sim \vec w', b', (u_i')_i \; & \;
        \vec w, b, (v_i)_i \sim \vec w', b', (v_i')_i
      }
    \end{array}
  \end{mathpar}
  \textbf{Conventions:} $\pi$ is a permutation of $\{1,\ldots, n\}$ and $f \in \nzsig$.
  \caption{\label{figure:axioms} The Axioms $\textsf{Struct-Ax}$.}
\end{figure*}

%%% Local Variables:
%%% mode: latex
%%% TeX-master: "ms"
%%% End:

We now give an informal description of the axioms given in Fig.~\ref{figure:axioms}. We describe in details the case study axiom $\cs$, which is the most complicated one. It states that in order to show that:
\[
  \ite{b}{u}{v} \sim \ite{b'}{u'}{v'}
\]
it is sufficient to show that the \textsf{then} branches and the \textsf{else} branches are indistinguishable, \emph{when giving to the adversary the value of the conditional} (i.e. $b$ on the left and $b'$ on the right). We can do better, by considering simultaneously several terms starting with the same conditional. We also allow some terms $\vec w$ and $\vec w'$ on the left and right to stay untouched:
\[
  \infer{\vec w, (\ite{b}{u_i}{v_i})_i \sim \vec w',(\ite{b'}{u_i'}{v_i'})_i}
  {
    \vec w, b, (u_i)_i \sim \vec w', b', (u_i')_i \quad & \quad
    \vec w, b, (v_i)_i \sim \vec w', b', (v_i')_i
  }
\]
This is the only axiom with more than one premise. Furthermore we assume that $b,b'$ do not contain conditionals. % This is a technical restriction which is used in the decidability proof, but might be unnecessary.

We quickly describe the other structural axioms: $\perm$ allows to change the terms order, using the same permutation on both sides of $\sim$; $\restr$ is a strengthening axiom;  $R$ allows to replace a term $s$ by any $R$-equal term $t$; the function application axiom $\fa$ states that to prove that two images are indistinguishable, it is sufficient to show that the arguments are indistinguishable (we restrict this axiom to the case where $f$ is in $\nzsig$); $\mathsf{Sym}$ states that indistinguishability is symmetrical; and $\dup$ states that giving twice the same value to an adversary is equivalent to giving it only once. All the above axioms are computationally valid.
\begin{proposition}
  \label{prop:sound-axiom}
  The axioms given in Fig.~\ref{figure:axioms} are valid in any computational model in which  the functional correctness assumptions $R_1$ on pairs and encryptions hold.
\end{proposition}

\begin{proof}
  The proof can be found in~\cite{Bana:2014:CCS:2660267.2660276}.
\end{proof}

\paragraph{\ch{Restrictions}}
\ch{As mentioned earlier, we restricted some axioms to achieve decidability. For example, the $\csm$ and $\fa$ axioms presented above are weaker than the corresponding axioms in~\cite{Bana:2014:CCS:2660267.2660276}: in the $\csm$ axiom, we forbid the terms $b$ and $b'$ from containing conditionals; and we do not allow $\fa$ applications on the $\zero$ function symbols. These are technical restrictions which are used in the proof, but might be unnecessary.}

\subsection{Cryptographic Assumptions}

We now show how cryptographic assumptions are translated into unitary axioms. In the computational model, the security of a cryptographic primitive is expressed through a game between a challenger and an attacker (which is a PPTM) that tries to break the primitive.

We present here the $\textsc{ind-cca}_2$ game (for Indistinguishability against Chosen Ciphertexts Attacks, see~\cite{DBLP:conf/crypto/BellareDPR98}). First, the challenger computes a public/private key pair $(\pk(\nonce)$, $\sk(\nonce))$ (using a nonce $\nonce$ of length $\eta$ uniformly sampled), and sends $\pk(\nonce)$ to the attacker. The adversary then has access to two oracles: i) a left-right oracle $\mathcal{O}^b_{\textsf{LR}}(\nonce)$ that takes two messages $m_0,m_1$ as input and returns $\{m_b\}_{\pk(\nonce)}^{\nonce_r}$, where $b$ is an internal bit uniformly sampled at the beginning by the challenger and $\nonce_r$ is a fresh nonce; ii) a decryption oracle $\mathcal{O}_{\textsf{dec}}(\nonce)$ that, given $m$, returns $\dec(m,\sk(\nonce))$ if $m$ was not the result of a previous $\mathcal{O}_{\textsf{LR}}$ oracle query, and length of $m$ zeros otherwise. Remark that the two oracles have a shared memory. For simplicity, we omit the length constraints of these oracles
\iffull
(we give them in Appendix~\ref{app-section:cca-axiom}).
\else
(they can be found the long version~\cite{long})
\fi
The advantage $\text{Adv}_{\mathcal{A}}^{\CCA}(\eta)$ of an adversary $\mathcal{A}$ against this game is the probability for $\mathcal{A}$ to guess the bit $b$:
\begin{multline*}
  \big|\;
  \Pr\big(\nonce :
  \mathcal{A}^{\mathcal{O}^1_{\textsf{LR}}(\nonce),\mathcal{O}_{\textsf{dec}}(\nonce)}
  \left(1^\eta\right) = 1 \big)\\
  -\;  \Pr\big(\nonce  :
  \mathcal{A}^{\mathcal{O}^0_{\textsf{LR}}(\nonce),\mathcal{O}_{\textsf{dec}}(\nonce)}
  \left(1^\eta\right) = 1 \big) 
  \;\big|
\end{multline*}
An encryption scheme is $\textsc{ind-cca}_2$ if the advantage $\text{Adv}_{\mathcal{A}}^{\CCA}(\eta)$ of any adversary $\mathcal{A}$ is negligible in $\eta$. The $\textsc{ind-cca}_1$ game is the restriction of this game where the adversary cannot call $\mathcal{O}_{\textsf{dec}}$ after having called $\mathcal{O}_{\textsf{LR}}$. An encryption scheme is $\textsc{ind-cca}_1$ if  $\text{Adv}_{\mathcal{A}}^{\CCAO}(\eta)$ is negligible for any adversary $\mathcal{A}$.

\paragraph{$\CCAO$ Axiom}
Before introducing the $\CCA$ axioms, we recall informally the $\CCAO$ axioms from~\cite{Bana:2014:CCS:2660267.2660276}. \ch{First, we define a syntactic property on secret keys used as a side-condition of the $\CCAO$ axioms:}
\begin{definition}
  \ch{For every ground term $t$, we say that a secret key $\sk(\nonce)$ appears only in \emph{decryption position} in $t$ if it appears only in subterms of $t$ of the form $\dec(\_,\sk(\nonce))$.}
\end{definition}
We now define the $\CCAO$ axioms:
\begin{definition}
  \label{def:cca1}
  $\CCAO$ is the computable set of unitary axioms:
  \[
    \vec w, t[\enc{u}{\pk(\nonce)}{\nonce_r}] \sim \vec w, t[\enc{v}{\pk(\nonce)}{\nonce_r}]
  \]
  where: $\nonce_r$ does not appear in $t,u,v, \vec w$; $\nonce$ appears only in $\pk(\nonce)$ or $\sk(\nonce)$ in $t,u,v, \vec w$; $\sk(\nonce)$ does not appear in $t, \vec w$; $\sk(\nonce)$ appears only in decryption position in $u,v$; and the terms $u$ and $v$ are always of the same length.
\end{definition}

\begin{proposition}
  $\CCAO$ is valid in every computational model where the encryption scheme interpretation is $\textsc{ind-cca}_1$.
\end{proposition}

\begin{proof}
  (\emph{sketch}) The proof is a reduction that, given a PPTM $\mathcal{A}$ that can distinguish between $\vec w, t[\enc{u}{\pk(\nonce)}{\nonce_r}]$ and $\vec w, t[\enc{v}{\pk(\nonce)}{\nonce_r}]$, builds a winning adversary against the $\textsc{ind-cca}_1$ game.

  We define the adversary. First, it computes $\sem{u}$ and $\sem{v}$, calling the decryption oracle if necessary. It then sends them to the challenger who answers $c$, which is either $\sem{\enc{u}{\pk(\nonce)}{\nonce_r}}$ or $\sem{\enc{v}{\pk(\nonce)}{\nonce_r}}$. Observe that we need the freshness hypothesis on $\nonce_r$ as it is drawn by the challenger and the adversary cannot sample it. Using $c$, the adversary computes $\sem{t[c]}$, which it can do because the secret key does not appear in $t$, and then returns the bit $\mathcal{A}(\sem{t[c]})$. The advantage of the adversary is exactly the advantage of $\mathcal{A}$, which we assumed non-negligible, hence the adversary wins the game.
\end{proof}

\begin{remark}
  \ch{In the $\CCAO$ axiom, we did not specify how we ensure that $u$ and $v$ are always of the same length. Since the length of a term depends on implementation details (e.g. how the pair $\pair{\_}{\_}$ implemented), we let the user supply implementation assumptions, but require that these assumptions satisfy some properties (this is necessary to get decidability). To simplify the presentation, we omit all length constraints
    \iffull
    for now. We describe them later, in Section~\ref{subsection:length}.
    \else
    in the rest of this paper. We explain how they are handled in the long version~\cite{long}.
    \fi}
\end{remark}
\paragraph{$\CCA$ Axiom}
To extend this axiom to the $\textsc{ind-cca}_2$ game, we need to deal with calls to the decryption oracle performed after some calls to the left-right oracle. For example, consider the case where one call $(u,v)$ was made. Let $\alpha \equiv \enc{u}{\pk(\nonce)}{\nonce_r}$ and $\alpha' \equiv \enc{v}{\pk(\nonce)}{\nonce_r}$ (where $\equiv$ denotes syntactic equality) be the result of the call on, respectively, the left and the right. A naive first try could be to state that decryptions are indistinguishable. That is, if we let $s \equiv t[\alpha]$ and $s' \equiv t[\alpha']$, then $\dec(s,\sk(\nonce)) \sim  \dec(s',\sk(\nonce))$. But this is not valid: for example, take $u \equiv 0, v \equiv 1$, $t \equiv g([])$ (where $[]$ is a hole variable). Then the adversary can, by interpreting $g$ as the identity function, obtain a term semantically equal to $0$ on the left and $1$ on the right. This allows him to distinguish between the left and right cases.

We prevent this by adding a guard checking that we are not decrypting $\alpha$ on the left (resp. $\alpha'$ on the right): if not, we return the decryption $\dec(s,\sk(\nonce))$ (resp. $\dec(s',\sk(\nonce))$)  asked for, otherwise we return a dummy message $\zero(\dec(s,\sk(\nonce)))$ (resp. $\zero(\dec(s',\sk(\nonce)))$).
\begin{definition}
  \label{def:cca2simp}
  $\CCAs$ is the (recursive) set of unitary axioms:
  \begin{alignat*}{2}
    &\vec w, t[\alpha],&&
    \begin{alignedat}[t]{2}
      \ite{\eq{s}{\alpha}}{&\zero(\dec(t[\alpha],\sk(\nonce)))\\}{&\dec(t[\alpha],\sk(\nonce))}
    \end{alignedat}\\
    \sim\;&
    \vec w, t[\alpha'],&&
    \begin{alignedat}[t]{2}
      \ite{\eq{s'}{\alpha'}}{&\zero(\dec(t[\alpha'],\sk(\nonce)))\\}{&\dec(t[\alpha'],\sk(\nonce))}
    \end{alignedat}
  \end{alignat*}
  under the side-conditions of Definition~\ref{def:cca1}.
\end{definition}
 This axiom is valid whenever the encryption is $\textsc{ind-cca}_2$.
\begin{proposition}
  $\CCAs$ is valid in every computational model where the encryption scheme interpretation is $\textsc{ind-cca}_2$.
\end{proposition}

This construction can be generalized to any number of calls to the left-right oracle, by adding a guard for each call, and to any number of keys.
\iffull
We refer the reader to Appendix~\ref{app-section:cca-axiom}, where we define formally the general $\CCA$ axioms.
\else
The general $\CCA$ axioms can be found in the long version~\cite{long}.
\fi
\footnote{\ch{Note that axioms for the $\textsc{ind-cca}_2$ cryptographic assumption have already appeared in the literature, in~\cite{DBLP:journals/iacr/BanaC16}. These axioms are only for a single call to the left-right oracle, and a single key. Our axiom schema is more general.}}
Still, a few comments:
we use extra syntactic side-conditions to remove superfluous guards; 
we allow for $\alpha$-renaming of names; 
we restrict $t$ to be without $\symite$ and $\zero$;
the axioms allow for an arbitrary number of public/private key pairs to be used simultaneously;
and finally, an instance of the axiom can contain any number of interleaved left-right and decryption oracles calls.

\begin{remark}
  \label{rem:cca-trans}
  The last point is what allows us to avoid transitivity in proofs. For example, consider four encryptions, two of them ($\alpha$ and $\gamma$) using the public key $\pk(\nonce)$, and the other two ($\beta$ and $\delta$) using the public key $\pk(\nonce')$:
  \[
    \alpha \equiv \enc{A}{\pk(\nonce)}{\nonce_0} \;\;\;
    \beta \equiv \enc{B}{\pk(\nonce')}{\nonce_1} \;\;\;
    \gamma \equiv \enc{C}{\pk(\nonce)}{\nonce_0} \;\;\;
    \delta \equiv \enc{D}{\pk(\nonce')}{\nonce_1}
  \]
  Then the following formula is a valid instance of the $\CCA$ axioms on, simultaneously, keys $\pk(\nonce)$ and $\pk(\nonce')$:
  \begin{equation*}
    \label{eq:trans-ex}
    \infer[\CCA(\pk(\nonce),\pk(\nonce'))]{
      \alpha,\beta \sim \gamma,\delta
    }{}
  \end{equation*}
  However, proving the above formula using $\CCA$ only on one key at a time, as in~\cite{Bana:2014:CCS:2660267.2660276}, requires transitivity:
  \[
    \infer{
      \alpha,\beta \sim \gamma,\delta
    }{
      \infer[\CCA(\pk(\nonce'))]{
        \alpha,\beta \sim \alpha,\delta
      }{}
      \quad & \quad
      \infer[\CCA(\pk(\nonce))]{
        \alpha,\delta \sim \gamma,\delta
      }{}
    }
  \]
\end{remark}

\subsection{Comments and Examples}
\label{example:proof-intro}

 Our set of axioms is not complete w.r.t. the computational interpretation semantics. Indeed, being so would mean axiomatizing exactly which distributions (computable in polynomial time) can be distinguished by PPTMs, which is unrealistic and would lead to undecidability. E.g., if we completely axiomatized $\textsc{ind-cca}_2$, then showing the satisfiability of our set of axioms would show the existence of $\textsc{ind-cca}_2$ functions, which is an open problem. 

 Still, our axioms are expressive enough to complete concrete proofs of security. We illustrate this with two simple examples: a proof of the formula in Example~\ref{example:cs-intro-semantics}, and a proof of the security of one round of the NSL protocol~\cite{DBLP:journals/ipl/Lowe95}. Of course, such proofs can be found automatically using our decision procedure.% , which, to our knowledge, is the first decidability result for computational indistinguishability.

\begin{example}
  \label{example:proof-base}
  We give a proof of the formula of Example~\ref{example:cs-intro-semantics}:
  \[
    \ite{g()}{\nonce_0}{\nonce_1} \sim \nonce
  \]
  First, we introduce a conditional $g()$ on the right to match the structure of the left side using $R$. Then, we split the proof in two using the $\cs$ axiom. We conclude using the reflexivity modulo $\alpha$-renaming axiom (this axiom is subsumed by $\CCA$, therefore we do not include it in $\textsc{Ax}$).
  \[
    \infer[R]{\ite{g()}{\nonce_0}{\nonce_1} \sim \nonce}
    {
      \infer[\cs]{\ite{g()}{\nonce_0}{\nonce_1} \sim \ite{g()}{\nonce}{\nonce}}
      {
        \infer[\textsc{Refl}]{g(),\nonce_0 \sim g(),\nonce}{}
        \quad & \quad
        \infer[\textsc{Refl}]{g(),\nonce_1 \sim g(),\nonce}{}
      }
    }
  \]
\end{example}

\paragraph{Proof of NSL}
\begin{wrapfigure}{r}{0.45\linewidth}
  % \begin{figure}[b]
  \begin{center}
    \begin{tikzpicture}[every node/.style={font=\small}]
      \tikzset{>=stealth}

      \draw[->] (3,1.6) -- ++(3,0) 
      node[midway,above]{$\lrenc{\pair{\nonce_{\agent{A}}}{\agent{A}}}{\pk_B}{\nonce_0}$};
      \draw[<-] (3,0.8) -- ++(3,0) 
      node[midway,above]{$\lrenc{\pair{\nonce_{\agent{A}}}{\pair{\nonce_{\agent{B}}}{\agent{B}}}}{\pk_{\agent{A}}}{\nonce_1}$};
      \draw[->] (3,0) -- ++(3,0) 
      node[midway,above]{$\lrenc{\nonce_{\agent{B}}}{\pk_{\agent{B}}}{\nonce_2}$};

      \draw (3,-0.2) -- ++(0,2.3) node[left,pos=0.95]{$\agent{A}$};
      \draw (6,-0.25) -- ++(0,2.3) node[right,pos=0.95]{$\agent{B}$};
    \end{tikzpicture}
  \end{center}
  \caption{\label{fig:nsl-inf}The NSL protocol.}
\end{wrapfigure}
We consider a simple setting with one initiator $\agent{A}$, one responder $\agent{B}$ and no key server. An execution of the NSL protocol is given in Fig.~\ref{fig:nsl-inf}.

We write this in the logic. First, we let $\pk_{\agent{A}} \equiv \pk(\nonce_{\agent{A}})$ and $\sk_{\agent{A}} \equiv \sk(\nonce_{\agent{A}})$ be the public/private key pair of agent $\agent{A}$ (we define similarly $(\pk_{\agent{B}}$, $\sk_{\agent{B}})$). Since $A$ does not wait for any input before sending its first message, we put it into the initial frame:
\[
  % \small
  \phi_0 \equiv
  \pk_{\agent{A}},\pk_{\agent{B}},
  \enc{\pair{\nonce_{\agent{A}}}{\agent{A}}}{\pk_B}{\nonce_0}
\]
Then, both agents wait for a message before sending a single reply. When receiving $\textbf{x}_{\agent{A}}$ (resp. $\textbf{x}_{\agent{B}}$), the answer of agent $\agent{A}$ (resp. $\agent{B}$) is expressed in the logic as follows:
{% \small
  \begin{alignat*}{2}
    t_{\agent{A}}[\textbf{x}_{\agent{A}}] &\;\;\equiv\;\;&&
    \begin{alignedat}[t]{2}
      &\itene
      {\eq
        {\pi_1(\dec(\textbf{x}_{\agent{A}},\sk_{\agent{A}}))}
        {\nonce_{\agent{A}}}&}
      {\\ &\itene
        {\eq
          {\pi_2(\pi_2(\dec(\textbf{x}_{\agent{A}},\sk_{\agent{A}})))}
          {\agent{B}}&}
        {\\ &\quad\enc
          {\pi_1(\pi_2(\dec(\textbf{x}_{\agent{A}},\sk_{\agent{A}})))}
          {\pk_{\agent{B}}}{\nonce_2}}}
    \end{alignedat} \\[0.3em]
    t_{\agent{B}}[\textbf{x}_{\agent{B}}]&\;\;\equiv\;\;&&
    \begin{alignedat}[t]{2}
      &\itene
      {\eq
        {\pi_2(\dec(\textbf{x}_{\agent{B}},\sk_{\agent{B}}))}
        {\agent{A}}&}
      {\\ &\quad\enc{
          \pair
          {\pi_1(\dec(\textbf{x}_{\agent{B}},\sk_{\agent{B}}))}
          {\pair{\nonce_{\agent{B}}}
            {\agent{B}}}}
        {\pk_{\agent{A}}}{\nonce_1}}
    \end{alignedat}
  \end{alignat*}
}
  During an execution of the protocol, the adversary has several choices. First, it decides whether to interact first with $\agent{A}$ or $\agent{B}$. We focus on the case where it first sends a message to $\agent{B}$ (the other case is similar). Then, it can honestly forward the messages or forge new ones. E.g., when sending the first message to $\agent{B}$, it can either forward $\agent{A}$'s message $\enc{\pair{\nonce_{\agent{A}}}{\agent{A}}}{\pk_B}{\nonce_0}$ or forge a new message. We are going to prove the security of the protocol in the following case (the other cases are similar):
\begin{itemize}
\item the first message, sent to $\agent{B}$, is honest. Therefore we take $\textbf{x}_{\agent{B}} \equiv \enc{\pair{\nonce_{\agent{A}}}{\agent{A}}}{\pk_B}{\nonce_0}$, and the answer from $\agent{B}$ is:
  \[
    t_{\agent{B}}[\textbf{x}_{\agent{B}}] =_R
    \lrenc
    {\pair{\nonce_{\agent{A}}}{\pair{\nonce_{\agent{B}}}{\agent{B}}}}
    {\pk_{\agent{A}}}{\nonce_1}
  \]
\item the second message, sent to $\agent{A}$, is forged. Therefore we take $\textbf{x}_{\agent{A}} \equiv g(\phi_1)$, where $\phi_1 \equiv \phi_0,t_{\agent{B}}[\textbf{x}_{\agent{B}}]$. As, a priori, nothing prevents $g(\phi_1)$ from being equal to $t_{\agent{B}}[\textbf{x}_{\agent{B}}]$, we use the conditional $\eq{g(\phi_1)}{t_{\agent{B}}[\textbf{x}_{\agent{B}}]}$ to ensure that this message is forged. The answer from $\agent{A}$ is then:
\begin{equation}
  \label{eq:nsl-example}
  s \equiv \ite{\eq{g(\phi_1)}{t_{\agent{B}}[\textbf{x}_{\agent{B}}]}}{0}{t_{\agent{A}}[g(\phi_1)]}
\end{equation}
\end{itemize}
We show the secrecy of the nonce $\nonce_{\agent{B}}$: we let $t'_{\agent{B}}[\textbf{x}_{\agent{B}}]$ (resp. $s'$) be the term $t_{\agent{B}}[\textbf{x}_{\agent{B}}]$ (resp. $s$) where we replaced all occurrences of $\nonce_{\agent{B}}$ by $0$. For example,
\(
  t'_{\agent{B}}[\textbf{x}_{\agent{B}}] =_R
  \lrenc{\pair{\nonce_{\agent{A}}}{\pair{0}{\agent{B}}}}{\pk_{\agent{A}}}{\nonce_1}
\).
This yields the following goal formula:
\begin{equation}
  \label{eq:goal-nsl-ex}
  \phi_0,t_{\agent{B}}[\textbf{x}_{\agent{B}}],s \sim
  \phi_0,t'_{\agent{B}}[\textbf{x}_{\agent{B}}],s'
\end{equation}

\begin{remark}
  The process of computing the formula from the protocol description can be done automatically, using a simple procedure similar to the folding procedure from \cite{Bana:2014:CCS:2660267.2660276}. The formula in \eqref{eq:goal-nsl-ex} has already been split between the honest and dishonest cases using the case study axiom $\cs$ (we omit the $\cs$ applications to keep the proof readable). For example, the term in \eqref{eq:nsl-example} is the ``$\textsf{else}$'' branch of a $\cs$ application on conditional $\eq{g(\phi_1)}{t_{\agent{B}}[\textbf{x}_{\agent{B}}]}$ (which does not contain nested conditionals, as required by the $\cs$ side-condition).
\end{remark}

We now proceed with the proof. We let $\delta$ be the guarded decryption that will be used in the $\cca$ axiom:
\begin{equation}
  \label{eq:nsl-ex-intro}
  \delta \,\equiv
  \begin{alignedat}[t]{2}
    \ite{\eq{g(\phi_1)}{t_{\agent{B}}[\textbf{x}_{\agent{B}}]}}{&\zero(\dec(g(\phi_1),\sk_{\agent{A}}))\\}{&\dec(g(\phi_1),\sk_{\agent{A}})}
  \end{alignedat}
\end{equation}
and $s_\delta$ be the term $s$ where all occurrences of $\dec(g(\phi_1),\sk_{\agent{A}})$ have been replaced by $\delta$. We have $s =_R s_\delta$. We also introduce shorthands for some subterms of $s_\delta$: we let $a_\delta$, $b_\delta$ and $e_\delta$ be the terms $\eq{\pi_1(\delta)}{\nonce_{\agent{A}}}$, $\eq{\pi_2(\pi_2(\delta)))}{\agent{B}}$ and $\enc{\pi_1(\pi_2(\delta))}{\pk_{\agent{B}}}{\nonce_2}$. We define $\delta', s'_{\delta'}, a'_{\delta'}, b'_{\delta'}$ and $e'_{\delta'}$ similarly.

We then rewrite $s$ and $s'$ into $s_{\delta}$ and $s'_{\delta'}$ using $R$. Then we apply $\fa$ several times, first to deconstruct $s_{\delta}$ and $s'_{\delta'}$, and then to deconstruct $a_{\delta}, b_{\delta}$ and $a'_{\delta'}, b'_{\delta'}$. Finally, we use $\dup$ to remove duplicates, and we apply $\cca$ simultaneously on key pairs $(\pk_{\agent{A}},\sk_{\agent{A}})$ and $(\pk_{\agent{B}},\sk_{\agent{B}})$ (we omit here the details of the syntactic side-conditions that have to be checked):
\[
  \infer[R]{
    \phi_0,t_{\agent{B}}[\textbf{x}_{\agent{B}}],s \sim
    \phi_0,t'_{\agent{B}}[\textbf{x}_{\agent{B}}],s'
  }{
    \infer[(\fa, \dup)^*]{
      \phi_0,t_{\agent{B}}[\textbf{x}_{\agent{B}}],s_\delta \sim
      \phi_0,t'_{\agent{B}}[\textbf{x}_{\agent{B}}],s'_{\delta'}
    }{
      \infer[(\fa, \dup)^*]{
        \phi_0,t_{\agent{B}}[\textbf{x}_{\agent{B}}],
        a_{\delta}, b_{\delta},e_{\delta} 
        \sim
        \phi_0,t'_{\agent{B}}[\textbf{x}_{\agent{B}}],
        a'_{\delta'}, b'_{\delta'},e'_{\delta'}
      }{
        \infer[\cca]{
          \phi_0,t_{\agent{B}}[\textbf{x}_{\agent{B}}],
          \nonce_{\agent{A}},\delta,e_{\delta} 
          \sim
          \phi_0,t'_{\agent{B}}[\textbf{x}_{\agent{B}}],
          \nonce_{\agent{A}},\delta',e'_{\delta'}
        }{
        }
      }
    }
  }
\]

%%% Local Variables:
%%% mode: latex
%%% TeX-master: "ms"
%%% End:

\section{Main Result and Difficulties}
\label{section:difficulties}

We let $\mathsf{Ax}$ be the conjunction of $\textsf{Struct-Ax}$ and $\CCA$. We now state the main result of this paper.
\begin{theorem*}[Main Result]
  The following problem is decidable:\\
  \textbf{Input:} A ground formula $\vec u \sim \vec v$.\\
  \textbf{Question:} Is $\textsf{Ax} \wedge \vec u \not \sim \vec v$ unsatisfiable?
\end{theorem*}
We give here an overview of the problems that have to be overcome in order to obtain the decidability result. Before starting, a few comments. We close all rules under permutations. The $\sym$ rule commutes with all the other rules, and the $\CCA$ unitary axioms are closed under $\sym$. Therefore we can remove $\perm$ and $\sym$ from the set of rules. Observe that $\cs, \fa, \dup$ and $\cca$ are all \emph{decreasing rules}, i.e. the premises are smaller than the conclusion. The only non-decreasing rules are $R$, which may rewrite a term into a larger one, and $\restr$, which we eliminate later. Therefore we now focus on $R$.

\paragraph{Necessary Introductions} 
As we saw in Example~\ref{example:proof-base}, it might be necessary to use $R$ in the ``wrong direction'', typically to introduce new conditionals. A priori, this yields an unbounded search space. Therefore our goal is to characterize in which situations we need to use $R$ in the ``wrong direction'', and with which instances. We identify two necessary reasons for introducing new conditionals.

First, to match the shape of the term on the other side, like $g()$ in Example~\ref{example:proof-base}. In this case, the introduced conditional is exactly the conditional that appeared on the other side of $\sim$. With more complex examples this may not be the case. Nonetheless, an introduced conditional is always bounded by the conditional it matches.

Second, we might introduce a guard in order to fit to the definition of safe decryptions in the $\CCA$ axioms, as in \eqref{eq:nsl-ex-intro}. Here also, the introduced guard will be of bounded size. Indeed, guards of $\dec(s,\sk)$ are of the form $\eq{s}{\alpha}$ where $\alpha$ is a subterm of $s$. Therefore, for a fixed $s$, there are a bounded number of them, and they are of bounded size.

\begin{example}[Cut Elimination]
\label{ex:cut-example}
These conditions are actually sufficient. We illustrate this on an example where the $\cs$ rule is applied on two conditionals that have just been introduced.
\begin{equation*} 
    \infer[R]{ s \sim t}
    {
      \infer[\cs]{\ite{a}{s}{s} \sim \ite{b}{t}{t}}
      {
        a,s \sim b,t
        \qquad & \qquad
        a,s \sim b,t
      }
    }
\end{equation*}
Here $a$ and $b$ can be of arbitrary size. Intuitively, this is not a problem since any proof of $a,s \sim b,t$ includes a proof of $s \sim t$. Formally, we have the following weakening lemma.
\end{example}
\begin{lemma}
  \label{lem:body-restr-elim}
  For every proof $P$ of a ground formula $\vec u,s \sim \vec v,t$, there exists a proof $P'$ of $\vec u \sim \vec v$ where $P'$ is no larger than~$P$.
\end{lemma}
\begin{proof}
  (\emph{sketch})
  \iffull
  The full proof is in Appendix~\ref{app-section:ordering}.
  \else
  The full proof is in the long version~\cite{long}.
  \fi
  We prove by induction on $P$ that the $\restr$ rule is admissible using $\textsf{Ax} \backslash \{\restr\}$. For this to work, we need the $\cca$ axioms to be closed under $\restr$.
  \iffull
  Note that this creates some problems, which are dealt with in Appendix~\ref{sub:restr-cca}.
  \else
  Note that this creates some problems, which are dealt with in~\cite{long}.
  \fi
\end{proof}

Using this lemma, we can deal with Example~\ref{ex:cut-example} by doing a proof cut elimination. More generally, by induction on the proof size, we can guarantee that no such proof cuts appear.

This is the strategy we are going to follow: look for proof cuts that introduce unbounded new terms, eliminate them, and show that after sufficiently many cut eliminations all the subterms appearing in the proof are bounded by the ($R$-normal form of the) conclusion.

But a proof may contain more complex behaviors than just the introduction of a conditional followed by a $\cs$ application. For example the conditional being matched could have been itself introduced earlier to match another conditional, which itself was introduced to match a third conditional etc.
  \newcommand{\tnum}[1]{\text{\textcolor{red}{\tiny #1}}}
  
\begin{example} 
  \label{example:complex-cut}
  We illustrate this on an example. When it is more convenient, we write terms containing only $\symite$ and other subterms (handled as constants) as binary trees; we also index some subterms with a number, which helps keeping track of them across rule applications.
  \begin{equation*}
    % \label{eq:cut-example2}
    \begin{array}[c]{c}
      \infer[R]{ \ite{a}{u}{v} \sim \ite{c}{s}{t}}
      {
        \infer[\fa^{(3)}]{
          \begin{array}[c]{c}
            \begin{tikz}
              \tikzstyle{level 1}=[level distance=1.8em, sibling distance=3em]
              \tikzstyle{level 2}=[level distance=1.8em, sibling distance=2em]
              \tikzstyle{level 3}=[level distance=1.8em, sibling distance=2em]
              \node[anchor=base] (a) at (0,0){$a_{\tnum{1}}$}
              child {
                node {$b_{\tnum{2}}$}
                child {node{$u_{\tnum{4}}$}}
                child { 
                  node{$b_{\tnum{3}}$}
                  child {node{$w_{\tnum{5}}$}}
                  child {node{$u_{\tnum{6}}$}}
                }
              }
              child {node{$v_{\tnum{7}}$}};
            \end{tikz}
          \end{array}
          \quad\sim\quad
          \begin{array}[c]{c}
            \begin{tikz}
              \tikzstyle{level 1}=[level distance=1.8em, sibling distance=3em]
              \tikzstyle{level 2}=[level distance=1.8em, sibling distance=2em]
              \tikzstyle{level 3}=[level distance=1.8em, sibling distance=2em]
              \node[anchor=base] (a) at (0,0){$d_{\tnum{1}}$}
              child {
                node {$c_{\tnum{2}}$}
                child {node{$s_{\tnum{4}}$}}
                child { 
                  node{$d_{\tnum{3}}$}
                  child {node{$t_{\tnum{5}}$}}
                  child {node{$r_{\tnum{6}}$}}
                }
              }
              child {
                node{$p_{\tnum{7}}$}
              };
            \end{tikz}
          \end{array}
        }
        {
          a_{\tnum{1}},b_{\tnum{2}},b_{\tnum{3}},u_{\tnum{4}},w_{\tnum{5}},u_{\tnum{6}},v_{\tnum{7}}
          \sim
          d_{\tnum{1}},c_{\tnum{2}},d_{\tnum{3}},s_{\tnum{4}},t_{\tnum{5}},r_{\tnum{6}},p_{\tnum{7}}
        }
      }
    \end{array}
  \end{equation*}
  where $p_{\tnum{7}} \equiv \ite{c}{s}{t}$. Here the conditionals $b,d$ and the terms $w,r$ are, a priori, arbitrary. Therefore we would like to bound them or remove them through a cut elimination. The cut elimination technique used in Example~\ref{ex:cut-example} does not apply here because we cannot extract a proof of $a \sim c$. 

  But we can extract a proof of $b_{\tnum{2}},b_{\tnum{3}} \sim c_{\tnum{2}},d_{\tnum{3}}$. Using Proposition~\ref{prop:sound-axiom}, this means that in every appropriate computational model, $\sem{b,b} \approx \sem{c,d}$. It means that no adversary can distinguish between getting twice the same value sampled from $\sem{b}$ and getting a pair of values sampled from $\sem{c,d}$. In particular, this means that $\sem{c}_{\eta,\rho} = \sem{d}_{\eta,\rho}$, except for a negligible number of random tapes $\rho$.
\end{example}

\paragraph{A First Key Lemma}
A natural question is to ask whether this semantic equality $\sem{c} = \sem{d}$ implies a syntactic equality. While this is not the case in general, there are fragments of our logic in which this holds. We annotate the rules $\fa$ by the function symbol involved, and we let $\fas = \{\fa_f \mid f \in \ssig\}$.
\begin{definition}
  Let $\Sigma$ be the set of axiom names, seen as an alphabet. For all $\mathcal{L} \subseteq \Sigma^*$, we let $\mathfrak{F}(\mathcal{L})$ be the fragment of our logic defined by: a formula $\phi$ is in the fragment iff there exists a proof $P$ such that $P \vdash \phi$ and, for every branch $\rho$ of $P$, the word $w$ obtained by collecting the axiom names along $\rho$ (starting from the root) is in~$\mathcal{L}$.
\end{definition}
\begin{lemma}
  \label{lem:cond-equiv-body}
  For all $b,b',b''$, if $b,b \sim b',b''$ is in the fragment $\mathfrak{F}(\fas^*\cdot\dup^*\cdot\cca)$ then $b' \equiv b''$.
\end{lemma}
\begin{proof}
  The proof relies on the shape of the $\cca$ axioms, and can be found in
  \iffull
  Appendix~\ref{app-section:proof-form}.
  \else
  the long version~\cite{long}.
  \fi
\end{proof}

Using this lemma, we can deal with Example~\ref{example:complex-cut} if $a_{\tnum{1}},b_{\tnum{2}},b_{\tnum{3}} \sim d_{\tnum{1}},c_{\tnum{2}},d_{\tnum{3}}$ lies in the fragment $\mathfrak{F}(\fas^*\cdot\dup^*\cdot\cca)$. Using a first time the lemma on $b_{\tnum{2}},b_{\tnum{3}} \sim c_{\tnum{2}},d_{\tnum{3}}$ we obtain $c \equiv d$, and using again the lemma on $a_{\tnum{1}},b_{\tnum{2}} \sim d_{\tnum{1}},c_{\tnum{2}}$ (since $d \equiv c$) we deduce $a \equiv b$. Hence the cut elimination introduced before applies.

\paragraph{Proof Sketch} We now state the sketch of the proof:
\begin{itemize}
\item \textbf{Commutations:} first we show that we can assume that rules are applied in some given order. We prove this by showing some commutation results and adding new rules.
\item \textbf{Proof Cut Eliminations:} % we define a formal notion of introduced conditionals, depending only on the $R$-normal form of the conclusion.
  % We show that we can assume a precise form for the terms, strongly connected to the shape of the proof.
  through proof cut eliminations, we guarantee that every conditional appearing in the proof is \abounded. Intuitively a conditional is \abounded if it is a subterm of the conclusion or if it guards a decryption appearing in an \abounded term.
\item \textbf{Decision Procedure:} we give a procedure that, given a goal formula $t \sim t'$, computes the set of \abounded terms for this formula. We show that this procedure computes a finite set, and deduce that the proof search is finite. This yields an effective algorithm to decide our problem.
\end{itemize}

%%% Local Variables:
%%% mode: latex
%%% TeX-master: "ms"
%%% End:

\section{Commutations and Cut Eliminations}
\label{section:sketch-proof}
In this section we show, through rule commutations, that we can restrict ourselves to proofs using rules in some given order. Then, we show how this restricts the shapes of the terms.

\subsection{Rule Commutations}
Everything in this subsection applies to any set $\textsf{U}$ of unitary axioms closed under $\restr$. We specialize to $\CCA$ later.

We start by showing a set of rule commutations of the form $w \Rightarrow w'$, where $w$ and $w'$ are words over the set of rule names. An entry $w \Rightarrow w'$ means that a derivation in $w$ can be rewritten into a derivation  in $w'$, with the same conclusion and premises. Here are the basic commutations we use:
{\[
  \begin{array}{cc}
    \begin{array}{|ccc|}
      \hline
      \vphantom{\Big|}\dup\cdot R &\;\Rightarrow\;& R\cdot\dup
      \\[0.2em]
      \dup\cdot\fa&\;\Rightarrow\;& \fa^*\cdot\dup
      \\[0.2em]
      \dup\cdot \cs&\;\Rightarrow\;& \cs\cdot\dup
      \\\hline
    \end{array}
    &
    \begin{array}{|ccc|}
      \hline
      \vphantom{\Big|}\fa\cdot R &\;\Rightarrow\;& R\cdot\fa
      \\[0.2em]
      \fa\cdot \cs&\;\Rightarrow\;& R\cdot \cs\cdot\fa\\
      \hline
    \end{array}
  \end{array}
\]}
\begin{lemma}
  \label{lem:rule-commute-body}
  All the above rule commutations are correct.
\end{lemma}
\begin{proof}
  We show only $\fa\cdot R \Rightarrow R\cdot\fa$ (the full proof is in
  \iffull
  Appendix~\ref{app-section:ordering}):
  \else
  the long version~\cite{long}):
  \fi
  {\begin{gather*}
      \begin{array}[c]{c}
        \infer[\fa]{\vec u,f(\vec v) \sim \pvec{u}', f(\pvec{v}')}
        {
          \infer[R]{ \vec u,\vec v \sim \pvec{u}', \pvec{v}'}
          { \vec u_1,\vec v_1 \sim \pvec{u}_1', \pvec{v}_1'}
        }
      \end{array}
      \Rightarrow
      \begin{array}[c]{c}
        \infer[R]{\vec u,f(\vec v) \sim \pvec{u}', f(\pvec{v}')}
        {
          \infer[\fa]{\vec u_1,f(\vec v_1) \sim \pvec{u}_1', f(\pvec{v}_1')}
          { \vec u_1,\vec v_1 \sim \pvec{u}_1', \pvec{v}_1'}
        }
      \end{array}
      \\[-1em]
    \qedhere
  \end{gather*}}
\end{proof}
Using these rules, we obtain a first restriction.
\begin{lemma}
  \label{lem:ordering-body}
  The ordered strategy
  \(
  \mathfrak{F}((\csm + R)^* \cdot \fa^* \cdot \dup^* \cdot \textsf{U})
  \)
  is complete for $\mathfrak{F}((\csm + \fa + R +  \dup + \textsf{U})^*)$.
\end{lemma}
\begin{proof}
  First, we commute all the $\dup$ to the right, which yields $\mathfrak{F}((\csm + R + \fa)^* \cdot \dup^* \cdot \textsf{U})$. Then, we commute all $\fa$ to the right, stopping at the first $\dup$.
\end{proof}

\paragraph{Splitting the $\fa$ Rule}
To go further, we split $\fa$ as follows: if the deconstructed symbol is $\symite$ then we denote the function application by $\fa(b,b')$, where $b,b'$ are the involved conditionals; if the deconstructed symbol $f$ is in $\ssig$, then we denote the function application by $\fa_f$. We give below the two new rules:
\begin{gather*}
  \begin{gathered}[c]
    \infer[\fa(b,b')]{
      \begin{alignedat}{2}
        &\vec w,\ite{a}{u}{v}\\[-0.3em]
        \sim\;&
        \vec r,\ite{b}{s}{t}
      \end{alignedat}
    }{
      \vec w,a,u,v \sim \vec r,b,s,t
    }
  \end{gathered}
  \;\;\;\;
  \begin{gathered}[c]
    \infer[\fa_f]{
      \vec u,f(\vec v) \sim \vec s,f(\vec t)
    }{
      \vec u,\vec v \sim \vec s,\vec t
    }
  \end{gathered}
\end{gather*}
The set of rule names is now infinite, since there exists one rule $\fa(b,b')$ for every pair of ground terms $b,b'$.

\paragraph{Further Commutations}
\label{ex:ex5maybe}
Intuitively, we want to use $R$ at the beginning of the proof only. This is helpful since, as we observed earlier, all the other rules are decreasing (i.e. premises are smaller than the conclusion). The problem is that we cannot fully commute $\cs$ and $R$. For example, in:
\begin{gather*}
  \infer[\cs]{\ite{a}{u}{v} \sim \ite{b}{s}{t}}
  {
    \infer[R]{a,u \sim b,s}{a',u' \sim b',s'}
    \qquad & \qquad
    \infer[R]{a,v \sim b,t}{a'',v' \sim b'',t'}
  }
\end{gather*}
we can commute the rewritings on $u,v,s$ and $t$, but not on $a$ and $b$ because they appear twice in the premises, and $a'$ and $a''$ may be different (same for $b'$ and $b''$).

\paragraph{New Rules}
We handle this problem by adding new rules to track relations between branches. We give only simplified versions here, the full rules are in
\iffull
Appendix~\ref{app-section:ordering}.
\else
the long version~\cite{long}.
\fi
For every $a,c$ in $\mathcal{T}(\ssig,\Nonce)$ in $R$-normal form, we have the~rules:
\begin{gather*}
  \infer[\twobox^{\text{s}}]{ \vec u, C[a] \sim \vec v, C'[c]}
  {
    \vec u, C\big[\splitbox{a}{a}{a}\big]
    \sim
    \vec v, C'\big[\splitbox{c}{c}{c}\big]
  }\\
  \infer[\csmb^{\text{s}}]
  {
    \ite{\splitbox{a_1}{a_2}{a}}{u}{v}
    \sim
    \ite{\splitbox{c_1}{c_2}{c}}{s}{t}}
  {
    a_1, u \sim  c_1, s
    \qquad & \qquad
    a_2, v \sim  c_2, t
  }
\end{gather*}
where $\splitbox{\phantom{n}}{\phantom{n}}{a}$ is a new symbol of sort $\bool^2 \ra \bool$, and of fixed semantics: it ignores its arguments and has the semantics $\sem{a}$. Intuitively, $\splitbox{a_1}{a_2}{a}$ stands for the conditional $a$, and $a_1,a_2$ are, respectively, the left and right versions of $a$.

Remark that for the $\csmb$ rule to be sound we need $\sem{a_1}$, $\sem{a_2}$ and $\sem{a}$ to be equal, up to a negligible number of samplings (same for $c_1,c_2$ and $c$). This is not enforced by the rules, so it has to be an invariant of our strategy. We denote $\esig$ the set of new function symbols. We need the functions in $\esig$ to block the if-homomorphism to ensure that for all $\splitbox{a}{c}{b} \in \st(t)$, $\sem{a}= \sem{c} = \sem{b}$. Therefore the TRS $R_2$ is \emph{not} extended to $\esig$. For example we have:
\[
  \splitbox{\ite{a}{c}{d}}{e}{b} \not \ra_R^* \ite{a}{\splitbox{c}{e}{b}}{\splitbox{d}{e}{b}}
\]
The $R$ rule is replaced by $\rs$ which has an extra side-condition. $\rs$ can rewrite $u[s]$ into $u[t]$ as long as:
\[
  \left\{\splitbox{a}{c}{b} \in \st(t)\right\} \subseteq \left\{\splitbox{a}{c}{b} \in \st(u[s])\right\}
\]
This ensures that no new arbitrary $\splitbox{a}{c}{b}$ is introduced. New boxed conditionals are only introduced through the $\twobox$ rule. Similarly, the $FA$ axiom is \emph{not} extended to~$\esig$.
\begin{definition}
  A term $t$ is \emph{well-formed} if for every $\splitbox{a}{c}{b} \in \st(t)$, \( a =_R c =_R  b\). We lift this to formulas as expected.
\end{definition}

\begin{proposition}
  The following rules preserve well-formedness:
  \[
    \rs,\twobox,\csmb,\fas,\{\fa(b,b')\},\dup
  \]
  Besides, $\rs$, $\csmb$ and $\twobox$ are sound on well-formed formulas.
\end{proposition}
\begin{proof}
  The only rule not obviously preserving well-formedness is $\rs$, but its side-conditions guarantee the well-formedness invariant. The only rule that is not always sound is $\csmb$, and it is trivially sound on well-formed formulas.
\end{proof}

\paragraph{Ordered Strategy}
We have new rule commutations.
{
  \[
  \begin{array}{|ccc|}
    \hline
    \vphantom{\Big|}
    \fas\cdot\fa(b,b')&\;\Rightarrow\;& R\cdot\fa(b,b')\cdot\fas^*\cdot\dup
    \\\hline\hline
    \vphantom{\Big|}\csmb\cdot\rs &\;\Rightarrow\;& \rs\cdot\csmb
    \\[0.2em]
    \csmb\cdot\twobox &\;\Rightarrow\;& \rs\cdot\twobox\cdot\csmb
    \\\hline
  \end{array}
\]}

\begin{lemma}
  \label{lem:boxed-commute-body}
  All the rule commutations above are correct.
\end{lemma}
\begin{proof}
  The proof can be found in
  \iffull
  Appendix~\ref{app-section:ordering}.
  \else
  the long version~\cite{long}.
  \fi
\end{proof}
\noindent
This allows to have $\rs$ rules only at the beginning of the~proof.
\begin{lemma}
  \label{lem:ordered-strat-body}
  The ordered strategy:
  \[
    \mathfrak{F}((\twobox + \rs)^* \cdot \csmb^* \cdot \{\fa(b,b')\}^* \cdot \fas^* \cdot \dup^* \cdot \textsf{U})
  \]
  is complete for $\mathfrak{F}((\csm + \fa + R + \dup + \textsf{U})^*)$.
\end{lemma}
\begin{proof}
  We start from the result of Lemma~\ref{lem:ordering-body}, split the $\fa$ rules and commute rules until we get:
  \[
    \mathfrak{F}((\csm + R)^*  \cdot \{\fa(b,b')\}^* \cdot \fas^* \cdot \dup^* \cdot \textsf{U})
  \]
  We then replace all applications of $\csm$ by $\twobox.\csmb$. All $\splitbox{a}{a}{a}$ introduced are immediately ``opened'' by a $\csmb$ application, hence we know that the side-conditions of $\rs$ hold every time we apply $R$. Therefore we can replace all applications of $R$ by $\rs$,  which yields:
  \[
    \mathfrak{F}((\csmb + \twobox + \rs)^* \cdot \{\fa(b,b')\}^* \cdot \fas^* \cdot \dup^*  \cdot \textsf{U})
  \]
  Finally we commute the $\csmb$ applications to the right.
\end{proof}

\subsection{The Freeze Strategy}
We now show that we can restrict the terms on which the rules in $\{\fa(b,b')\}$ can be applied: when we apply a rule in $\{\fa(b,b')\}$, we ``freeze'' the conditionals $b$ and $b'$ to forbid any further applications of $\{\fa(b,b')\}$ to them.

\begin{example}
  \label{ex:freeze-cut}
  Let $a_i \equiv \ite{b_i}{c_i}{d_i}$ ($i \in \{1,2\}$), we want to forbid the following partial derivation to appear:
  \begin{equation*}
    \begin{array}[c]{c}
      \infer[\fa({a_1},{a_2})]{
        \ite{a_1}{u_1}{v_1} \sim \ite{a_2}{u_2}{v_2}}
      {
        \infer[\fa({b_1},{b_2})]{
          \phantom{b} a_1, u_1, v_1 \sim a_2, u_2, v_2 \phantom{b}}
        {
          b_1,c_1,d_1, u_1, v_1 \sim b_2,c_2,d_2, u_2, v_2
        }
      }
    \end{array}
  \end{equation*}
\end{example}
\paragraph{Freeze Strategy}
We let $\lrtbox{\phantom{v}}$ be a new function symbol of arity one, and for every ground term $s$ we let $\xxtbox{s}$ be the term:
\[
  \xxtbox{s} \equiv
  \begin{cases}
    \ite{\lrtbox{b}}{u}{v} & \text{ if } s \equiv \ite{b}{u}{v}\\
    s & \text{ if } s \in \mathcal{T}(\ssig,\Nonce)
  \end{cases}
\]
Moreover we replace every $\fa(b_1,b_2)$ rule by the rule:
\[
  \small
  \infer[\bfa({b_1},{b_2})]
  {
    \vec w_1, \ite{b_1}{u_1}{v_1}
    \sim \vec w_2, \ite{b_2}{u_2}{v_2}
  }
  {
    \vec w_1, \xxtbox{b_1}, u_1, v_1
    \sim \vec w_2, \xxtbox{b_2}, u_2, v_2
  }
\]
We let $\{\obfa({b_1},{b_2})\}$ be the restriction of $\{\bfa({b_1},{b_2})\}$ to the rules where ${b_1}$ and ${b_2}$ are not frozen conditionals. Finally, we add a new rule, $\unbox$, which unfreezes all conditionals: every $\lrtbox{b}$ is replaced by~$b$.

\begin{lemma}
  \label{lem:body-freeze}
  The following strategy:
  \[
    \mathfrak{F}((\twobox + \rs)^*  \cdot \csmb^* \cdot \{\obfa(b,b')\}^* \cdot \unbox \cdot \fas^* \cdot \dup^* \cdot \textsf{U})
  \]
  is complete for $\mathfrak{F}((\csm + \fa + R + \dup + \textsf{U})^*)$.
\end{lemma}
\begin{proof}
  Basically, the proof consists in eliminating all proof cuts of the shape given in Example~\ref{ex:freeze-cut}. The cut elimination is simple, though voluminous, and is given in
  \iffull
  Appendix~\ref{app-section:ordering}.
  \else
  the long version~\cite{long}.
  \fi
\end{proof}

%%% Local Variables:
%%% mode: latex
%%% TeX-master: "ms"
%%% End:

\section{Proof Form and Key Properties}
\label{section:decision}

The goal of this section is to show that we can assume w.l.o.g. that the terms appearing in the proof (following the ordered freeze strategy) after the $(\twobox + \rs)^*$ part have a particular form, that we call proof form. We also show properties of this restricted shape that allow more cut eliminations.

\subsection{Shape of the Terms}
Most of the completeness results shown before are for any set of unitary axioms closed under $\restr$. We now specialize these results to $\CCA$, to get some further restrictions.

When applying the unitary axioms $\CCA$, we would like to require that terms are in $R$-normal form, e.g. to avoid the application of \cca to terms with an unbounded component, such as $\pi_1(\pair{u}{v})$. Unfortunately, the side-conditions of \cca are not stable under $R$. E.g., consider the \cca instance:
\[
  \infer[\cca]
  {\enc{\ite{\eq{g(\nonce_u)}{\nonce_u}}{A}{B}}{\pk(\nonce)}{\nonce_r} \sim \enc{C}{\pk(\nonce)}{\nonce_r}}
  {}
\]
The $R$-normal form of the left term is:
\[
  \ite{\eq{g(\nonce_u)}{\nonce_u}}{\enc{A}{\pk(\nonce)}{\nonce_r}}{\enc{B}{\pk(\nonce)}{\nonce_r}}
\]
which cannot be used in a valid \cca instance, since the conditional $\eq{g(\nonce_u)}{\nonce_u}$ should be somehow ``hidden'' by the encryption. To avoid this difficulty, we use a different normal form for terms: we try to be as close as possible to the $R$-normal form, while keeping conditional branching below their encryption. First, we illustrate this on an example. The term:
\begin{gather*}
  \small
  \benc{
    \ite
    {(\ite{b}{a}{c})}
    {\enc{\ite{d}{u}{v}}{\pk}{\nonce_1}}
    {w}}
  {\pk}
  {\nonce_2}
\end{gather*}
is normalized as follows:
\begin{gather*}
  \lrenc{
    \begin{alignedat}{2}
      \ite{b}
      {&\ite{a} 
        {\enc{\ite{d}{u}{v}}{\pk}{\nonce_1}}
        {w}\\
      }
      {&\ite{c}
        {\enc{\ite{d}{u}{v}}{\pk}{\nonce_1}}
        {w}
      }
    \end{alignedat}}
  {\pk}
  {\nonce_2}
\end{gather*}

\paragraph{Basic Terms} 
We omit the rewriting strategy here
\iffull
(C.f. Appendix~\ref{app-section:proof-form})
\else
\fi, and describe instead the properties of the normalized terms. We let $\mathcal{A}_\succ$ be the ordered strategy from~Lemma~\ref{lem:body-freeze}, and $\mathcal{A}_{\csmb}$ be its restriction to proofs with an empty $(\twobox + \rs)^*$ part. The rule $\csmb$ is the only branching rule, therefore, after applying all the $\csmb$ rules, we can associate to each branch~$l$ of the proof an instance $\ek_l = (\keys_l,\rands_l,\encs_l,\decs_l)$  of the $\CCA$ axiom, where $\keys_l$, $\rands_l$, $\encs_l$ and $\decs_l$ are the sets of, respectively, secret keys, encryption randomness, encryptions and decryptions. We use $\ek_l$ to define a normal form for the terms appearing in branch~$l$. This is done through four mutually inductive definitions: $\ek_l$-\emph{encryption oracle calls} are well-formed encryptions; $\ek_l$-\emph{decryption oracle calls} are well-formed decryptions; $\ek_l$-\emph{normalized basic terms} are terms built using function symbols in $\ssig$ and well-formed encryptions and decryptions; and $\ek_l$-\emph{normalized simple terms} are combinations of normalized basic terms using $\symite$. We give only the definition of $\ekl$-normalized basic terms (the full definitions are in
\iffull
Appendix~\ref{app-section:proof-form}).
\else
the long version~\cite{long}).
\fi
\begin{definition}
  A $\ek_l$-\emph{normalized basic term} is a term $t$ of the form
  \(
  U[\vec w,(\alpha_j)_j,(\dec_k)_k]
  \) where:
  \begin{itemize}
  \item $U$ and $\vec w$ are if-free and $\rands_l,\keys_l$ do not appear in $\vec w$.
  \item $U[\vec w, (\enc{[]_j}{\pk_j}{\nonce_j})_j, (\dec([]_k,\sk_k))_k]$ is in $R$-normal form.
  \item $(\alpha_j)_j$ are $\ek_l$-\emph{encryption oracle calls} under $(\pk_j,\sk_j)_j$.
  \item $(\dec_k)_k$ are $\ek_l$-\emph{decryption oracle calls} under $(\pk_k,\sk_k)_k$.
  \end{itemize}
  If $t$ is of sort bool, we say that it is a $\ek_l$-\emph{normalized basic conditional}.
\end{definition}

\paragraph{Normalized Proof Form} 
Every application of $\csmb$:
\[
  \infer[\csmb]
  { 
    \ite{\splitbox{a_1}{a_2}{a}}{u}{v} 
    \sim
    \ite{\splitbox{b_1}{b_2}{b}}{s}{t}}
  {
    a_1, u \sim  b_1, s 
    \quad & \quad
    a_2, v \sim  b_2, t
  }  
\]
is such that if we extract the sub-proof of $a_i \sim b_i$ (for $i \in \{1,2\}$), we get a proof in $\mathcal{A}_{\csmb}$. Therefore, we can check that terms after $(\twobox + \rs)^*$ are of the form informally described in Fig.~\ref{fig:proof-term-restr-body}. We define a normal form for such proofs, called \emph{normalized proof form}, and we define $\npfproof$ by $P \npfproof t \sim t'$ if and only if $P \vdash t \sim t'$, the proof $P$ is in $\mathcal{A}_{\succ}$ and is in \emph{normalized proof form}. We do not give the full definition, but one of the key ingredients is to require that for every term $s$ appearing in a branch $l$ of the proof $P$, if $s$ is the conclusion of a sub-proof in the fragment $\mathfrak{F}(\fas^* \cdot \dup^* \cdot \textsf{U})$ then $s$ is a $\ek_l$-normalized basic term. 

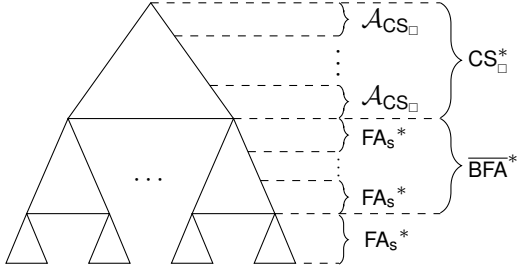
\begin{figure}[tb]
  \begin{center}
    \begin{tikzpicture}[scale=0.55]
      \draw[name path=RR] (3,2.8) node (root){} -- ++(2,-2.8) --
      node[above=1cm,midway]  {}
      ++(-4,0) -- cycle;

      \draw (1,0) -- ++(1,-2.3) 
      % First child
      -- ++(0.5,-1.2) -- 
      node[above,midway]  {}
      ++(-1,0) -- ++(0.5,1.2)
      % Child end
      -- 
      node[above=0.3cm,midway]  {}
      ++(-2,0) 
      % First child
      -- ++(0.5,-1.2) -- 
      node[above,midway]  {}
      ++(-1,0) -- ++(0.5,1.2)
      % Child end
      -- cycle;

      \node at (3,-1.5) {$\cdots$};

      \draw[name path=RC] (5,0) -- ++(1,-2.3) 
      % First child
      -- ++(0.5,-1.2) -- 
      node[above,midway]  {}
      ++(-1,0) -- ++(0.5,1.2)
      % Child end
      -- 
      node[above=0.3cm,midway]  {}
      ++(-2,0) 
      % First child
      -- ++(0.5,-1.2) -- 
      node[above,midway]  {}
      ++(-1,0) -- ++(0.5,1.2)
      % Child end
      -- cycle;

      % Intersections naming
      \path[name path=D1]  (5,-1.5) -- (7,-1.5);
      \path[name path=D2]  (5,-0.8) -- (7,-0.8);

      \path[name path=D3] (0,0.8) -- (7.55,0.8);
      \path[name path=D4] (0,2) -- (4,2);

      \path [name intersections={of=RC and D1,by=DI1}];
      \path [name intersections={of=RC and D2,by=DI2}];
      \path [name intersections={of=RR and D3,by=DI3}];
      \path [name intersections={of=RR and D4,by=DI4}];

      % Braces, first column
      \draw [decorate,decoration={brace,mirror,amplitude=3.5pt}]
      (7.55,-3.52) -- ++(0,1.18) node [midway,right=0.6em] {$\fas^*$};
      \draw[dashed] (6.7,-3.5) -- (7.55,-3.5);

      \draw [decorate,decoration={brace,amplitude=3.5pt}]
      (7.55,-0) -- ++(0,-0.8) node [midway,right=0.5em,yshift=-0.1em] {$\fas^*$};
      \node[rotate=90,scale=0.8] at (7.55,-1.15) {\footnotesize{$\cdots$}};
      \draw [decorate,decoration={brace,amplitude=3.5pt}]
      (7.55,-1.5) -- ++(0,-0.8) node [midway,right=0.5em] {$\fas^*$};
      \draw[dashed] (DI1) -- (7.55,-1.5);
      \draw[dashed] (DI2) -- (7.55,-0.8);

      \draw [decorate,decoration={brace,amplitude=3.5pt}]
      (7.55,2.8) -- ++(0,-0.8) node [midway,right=0.5em,yshift=-0.1em] {$\mathcal{A}_{\csmb}$};
      \node[rotate=90] at (7.55,1.4) {$\cdots$};
      \draw [decorate,decoration={brace,amplitude=3.5pt}]
      (7.55,0.8) -- ++(0,-0.78) node [midway,right=0.5em,yshift=0.1em] {$\mathcal{A}_{\csmb}$};
      \draw[dashed] (DI3) -- (7.55,0.8);
      \draw[dashed] (DI4) -- (7.55,2);

      % Braces, second column
      \draw [decorate,decoration={brace,amplitude=8pt}]
      (10,2.8) -- ++(0,-2.78) node [midway,right=0.7em] {$\csmb^*$};
      \draw[dashed] (3,2.8) -- ++(7,0);
      \draw[dashed] (5,0) -- (10,0);

      \draw [decorate,decoration={brace,amplitude=8pt}]
      (10,0) -- ++(0,-2.28) node [midway,right=0.7em] {$\obfa^*$};
      \draw[dashed] (6,-2.3) -- (10,-2.3);
    \end{tikzpicture}
  \end{center}
  \caption{\label{fig:proof-term-restr-body} The shape of the term is determined by the proof.}
\end{figure}

%%% Local Variables:
%%% mode: latex
%%% TeX-master: "ms"
%%% End:

%
\begin{lemma}
  \label{lem:body-proof-form}
  Every formula in $\mathfrak{F}((\csm + \fa + R + \dup + \CCA)^*)$ is provable using the strategy $\npfproof$.
\end{lemma}

\begin{proof}(\emph{sketch})
  The full proof is in
  \iffull
  Appendix~\ref{app-section:proof-form}.
  \else
  the long version~\cite{long}.
  \fi
  First, we rewrite terms by pulling conditionals upward without crossing an encryption function symbol, and without modifying decryption guards. Then, we remove all redexes from $R_1$ (e.g. $\pi_1(\pair{u}{v}) \ra u$) using a cut elimination procedure. E.g., the following cut can be eliminated using Lemma~\ref{lem:body-restr-elim}:
  \[
    \infer[R]{u \sim u'}
    {
      \infer[\fa_{\pair{}{}}]{\pi_1(\pair{u}{v}) \sim \pi_1(\pair{u'}{v'})}
      {
        u,v \sim u',v'
      }
    }
    \qedhere
  \]
\end{proof}

\subsection{Key Properties}

A term in $R$-normal form is in the following grammar:
\[
  t ::= u \in \mathcal{T}(\ssig,\Nonce) \mid \ite{b}{t}{t} \ (\text{with}\; b \in \mathcal{T}(\ssig,\Nonce))
\]
Given a term $t$ in $R$-normal form, we let $\condst(t)$ be its set of conditionals, and $\leavest(t)$ its set of leaves.

\paragraph{Characterization of Basic Terms} 
We give a key characterization proposition for basic terms: if two $\ekl$-normalized basic terms $\beta$ and $\beta'$ are such that, when $R$-normalizing them, they share a leaf term, then they are identical.

\begin{proposition}
  \label{prop:bas-cond-charac-body}
  For all $\ekl$-normalized basic terms $\beta,\beta'$, if we have
  \(
    \leavest(\beta\downarrow_R) \cap \leavest(\beta'\downarrow_R) \ne \emptyset
  \)
  then $\beta \equiv \beta'$.
\end{proposition}

\begin{proof}(\emph{sketch})
  The full proof is in
  \iffull
  Appendix~\ref{app-section:prop-basic-terms}.
  \else
  the long version~\cite{long}.
  \fi
  We give the intuition: since they are $\ekl$-normalized basic terms, we know that $\beta \equiv U[\vec w,(\alpha_j)_j,(\dec_k)_k]$, $\beta' \equiv U'[\vec w',(\alpha'_j)_j,(\dec'_k)_k]$ and:
  \begin{gather*}
    U[\vec w, (\enc{[]_j}{\pk_j}{\nonce_j})_j, (\dec([]_k,\sk_k))_k]\\
    U'[\vec w', (\enc{[]'_j}{\pk'_j}{\nonce_j})_j, (\dec([]'_k,\sk'_k))_k]
  \end{gather*}
  are in $R$-normal form. Using the fact that $U,U',\vec w, \vec w'$ are if-free, and the hypothesis that $\beta$ and $\beta'$ share a leaf term, we first show that we can assume $U \equiv U'$ and $\vec w \equiv \vec w'$ by induction on the number of positions where $U$ and $U'$ differ. Take $p$ where they differ, w.l.o.g. assume $\beta'_{|p}$ to be a hole of $U'$ (otherwise swap $\beta$ and $\beta'$). We have three cases: i) if $\beta_{|p}$ is in $\vec w$, we simply change $U$ to include everything up to $p$; ii) if $\beta_{|p}$ is in some encryption $\alpha_j\equiv \enc{m}{\pk}{\nonce}$, then we know that $\nonce$ appears in $\vec w$, which is not possible since, as $\beta$ is a $\ekl$-normalized basic term, $\nonce \in \rands_l$ does not appear in $\vec w$; iii) if  $\beta_{|p}$ is in some decryption $\dec_k\equiv \dec(u_k,\sk_k)$ then, similarly to the previous case, we have $\sk_k$ appearing in $\vec w$, which contradicts the fact that $\sk_k \in \keys_l$ do not appear in $\vec w$.

  Knowing that $U \equiv U'$ and $\vec w \equiv \vec w'$, it only remains to show that the encryptions $(\alpha_j)_j$ and $(\alpha'_j)_j$, and the decryptions $(\dec_k)_k$ and $(\dec'_k)_k$ are identical. The former follows from the fact that, for a given encryption randomness $\nonce \in \rands_l$, there exists a unique $m$ such $\enc{m}{\_}{\nonce} \in \encs_l$; and the latter follows from the fact that there is a unique way to guard a decryption in $\decs_l$ (this is not obvious, and relies on $\CCA$ side-conditions).
\end{proof}

\paragraph{Proofs of { $b \sim \false$} or {$\true$}}
Using the previous proposition, we can show that for all $b$, if $b$ is if-free then there is no derivation of $b \sim \true$ or $b \sim \false$ in $\mathcal{A}_\succ$. Such derivations would be problematic since $\true$ and $\false$ are conditionals of constant size, but $b$ could be of any size (and we are trying to bound all conditionals appearing in a proof). Also, the $\textsf{else}$ branch of a $\true$ conditional can contain anything and is, a priori, not bounded by the proof conclusion.

\begin{proposition}
  \label{prop:iffree-false-body}
  Let $b$ an if-free conditional in $R$-normal form, with $b \not \equiv \false$ (resp. $b \not \equiv \true$). Then there exists no derivation of $b \sim \false$ (resp. $b \sim \true$) in $\mathcal{A}_\succ$.

\end{proposition}
\begin{proof}
  This is shown by induction on the size of the derivation. The full proof is in
  \iffull
  Appendix~\ref{app-section:if-free-conds},
  \else
  the long version~\cite{long},
  \fi
  and relies on Proposition~\ref{prop:bas-cond-charac-body}.
\end{proof}

%%% Local Variables:
%%% mode: latex
%%% TeX-master: "ms"
%%% End:

\section{Bounding the Proof and Decision Procedure}
\label{section:bounding-body}

\label{subsection:cut-elim-cs-fa}
We give here two similar proof cut eliminations, one used on $\obfa$ conditionals and the other on $\csmb$ conditionals.

\paragraph{$\obfa$ Rule}
We already used this cut elimination to deal with Example~\ref{example:complex-cut} for conditionals involved in $\obfa$ applications. The cuts we want to eliminate are of the form:
\begin{equation}
  \label{eq:cut-example3}
  \begin{array}[c]{c}
    \infer[\obfa^{(2)}]{
      \begin{array}[c]{c}
        \underbrace{
          \begin{tikz}
            \tikzstyle{level 1}=[level distance=1.8em, sibling distance=3em]
            \tikzstyle{level 2}=[level distance=1.8em, sibling distance=2em]
            \tikzstyle{level 3}=[level distance=1.8em, sibling distance=2em]
            \node[anchor=base] (a) at (0,0){$a_{\tnum{1}}$}
            child {
              node {$a_{\tnum{2}}$}
              child {node{$u_{\tnum{3}}$}}
              child {node{$v_{\tnum{4}}$}}
            }
            child {node{$w_{\tnum{5}}$}};
          \end{tikz}
        }_{\sigma}
      \end{array}
      \quad\sim\quad
      \begin{array}[c]{c}
        \underbrace{
          \begin{tikz}
            \tikzstyle{level 1}=[level distance=1.8em, sibling distance=3em]
            \tikzstyle{level 2}=[level distance=1.8em, sibling distance=2em]
            \tikzstyle{level 3}=[level distance=1.8em, sibling distance=2em]
            \node[anchor=base] (a) at (0,0){$b_{\tnum{1}}$}
            child {
              node {$c_{\tnum{2}}$}
              child {node{$r_{\tnum{3}}$}}
              child {node{$s_{\tnum{4}}$}}
            }
            child {node{$t_{\tnum{5}}$}};
          \end{tikz}
        }_{\tau}
      \end{array}
    }
    {
      a_{\tnum{1}},a_{\tnum{2}},u_{\tnum{3}},v_{\tnum{4}},w_{\tnum{5}}
      \sim
      b_{\tnum{1}},c_{\tnum{2}},r_{\tnum{3}},s_{\tnum{4}},t_{\tnum{5}}
    }
  \end{array}
\end{equation}
Using Lemma~\ref{lem:body-restr-elim}, we extract a proof of $a_{\tnum{1}},a_{\tnum{2}} \sim b_{\tnum{1}},c_{\tnum{2}}$, which, thanks to the ordered strategy, is in $\mathfrak{F}(\fas^*\cdot\dup^*\cdot\cca)$. From Lemma~\ref{lem:cond-equiv-body} we get that $b \equiv c$. We then replace \eqref{eq:cut-example3} by: 
\begin{equation*}
  \label{eq:cut-example3-cut-elim}
  \begin{array}[c]{c}
    \infer[R]{\vphantom{L}\sigma \sim \tau}
    {
      \infer[\obfa]{
        \begin{array}[c]{c}
          \begin{tikz}
            \tikzstyle{level 1}=[level distance=1.8em, sibling distance=3em]
            \tikzstyle{level 2}=[level distance=1.8em, sibling distance=2em]
            \tikzstyle{level 3}=[level distance=1.8em, sibling distance=2em]
            \node[anchor=base] (a) at (0,0){$a_{\tnum{1}}$}
            child {node{$u_{\tnum{3}}$}}
            child {node{$w_{\tnum{5}}$}};
          \end{tikz}\\[-0.6em]
        \end{array}
        \quad\sim\quad
        \begin{array}[c]{c}
          \begin{tikz}
            \tikzstyle{level 1}=[level distance=1.8em, sibling distance=3em]
            \tikzstyle{level 2}=[level distance=1.8em, sibling distance=2em]
            \tikzstyle{level 3}=[level distance=1.8em, sibling distance=2em]
            \node[anchor=base] (a) at (0,0){$b_{\tnum{1}}$}
            child {node{$r_{\tnum{3}}$}}
            child {node{$t_{\tnum{5}}$}};
          \end{tikz}\\[-0.6em]
        \end{array}
      }
      {
        a_{\tnum{1}},u_{\tnum{3}},w_{\tnum{5}}
        \sim
        b_{\tnum{1}},r_{\tnum{3}},t_{\tnum{5}}
      }
    }
  \end{array}
\end{equation*}
We retrieve a proof in $\mathcal{A}_\succ$ by pulling $R$ to the beginning of the~proof.

\paragraph{$\protect\csmb$ Rule}
The $\csmb$ case is more complicated. E.g., take two boxed $\csmb$ conditionals for the same if-free conditional $a$, and two arbitrary $\csmb$ conditionals on the right side:
\[
  a^{{\smallsquare}}_{\tnum{i}} \equiv \splitbox{a^l_{\tnum{i}}}{a^r_{\tnum{i}}}{a}\ (i \in \{1,2\})
  \;\;\;
  b^{{\smallsquare}}_{\tnum{1}} \equiv \splitbox{b^l_{\tnum{1}}}{b^r_{\tnum{1}}}{b}
  \;\;\;
  c^{{\smallsquare}}_{\tnum{2}} \equiv \splitbox{c^l_{\tnum{2}}}{c^r_{\tnum{2}}}{c}
\]
Consider the following cut:
\begin{equation*}
  \small
  \infer[\csmb^{(2)}]{
    \begin{array}[c]{c}
      \underbrace{
        \begin{tikz}
          \tikzstyle{level 1}=[level distance=1.8em, sibling distance=3em]
          \tikzstyle{level 2}=[level distance=1.8em, sibling distance=2em]
          \tikzstyle{level 3}=[level distance=1.8em, sibling distance=2em]
          \node[anchor=base,outer sep=-0.1em] (a) at (0,0){$a^{{\smallsquare}}_{\tnum{1}}$}
          child {
            node[outer sep=-0.1em] {$a^{{\smallsquare}}_{\tnum{2}}$}
            child {node{$u_{\tnum{3}}$}}
            child {node{$v_{\tnum{4}}$}}
          }
          child {node{$w_{\tnum{5}}$}};
        \end{tikz}
      }_{\sigma}
    \end{array}
    \quad\sim\quad
    \begin{array}[c]{c}
      \underbrace{
        \begin{tikz}
          \tikzstyle{level 1}=[level distance=1.8em, sibling distance=3em]
          \tikzstyle{level 2}=[level distance=1.8em, sibling distance=2em]
          \tikzstyle{level 3}=[level distance=1.8em, sibling distance=2em]
          \node[anchor=base,outer sep=-0.1em] (a) at (0,0){$b^{{\smallsquare}}_{\tnum{1}}$}
          child {
            node[outer sep=-0.1em] {$c^{{\smallsquare}}_{\tnum{2}}$}
            child {node{$r_{\tnum{3}}$}}
            child {node{$s_{\tnum{4}}$}}
          }
          child {node{$t_{\tnum{5}}$}};
        \end{tikz}
      }_{\tau}
    \end{array}
  }
  {
    \infer*[(A)]{
      a^l_{\tnum{1}},a^l_{\tnum{2}},u_{\tnum{3}}
      \sim
      b^l_{\tnum{1}},c^l_{\tnum{2}},r_{\tnum{3}}
    }{}
    &
    \infer*[(B)]{
      a^l_{\tnum{1}},a^r_{\tnum{2}},v_{\tnum{4}}
      \sim
      b^l_{\tnum{1}},c^r_{\tnum{2}},s_{\tnum{4}}
    }{}
    &
    \infer*[(C)]{
      a^r_{\tnum{1}},w_{\tnum{5}}
      \sim
      b^r_{\tnum{1}},t_{\tnum{5}}
    }{}
  }
\end{equation*}
As we did for $\obfa$, we can extract from $(A)$, using Lemma~\ref{lem:body-restr-elim}, a proof of $a^l_{\tnum{1}},a^l_{\tnum{2}} \sim b^l_{\tnum{1}},c^l_{\tnum{2}}$. But using the ordered strategy, we get that this proof is in $\mathcal{A}_{\csmb}$, which we recall is the fragment:
\[
  \csmb^* \cdot \{\obfa(b,b')\}^* \cdot \unbox \cdot \fas^* \cdot \dup^* \cdot \CCA
\]
Therefore we cannot apply Lemma~\ref{lem:cond-equiv-body}. To deal with this cut, we generalize Lemma~\ref{lem:cond-equiv-body} to the case where the proof is in $\mathcal{A}_{\csmb}$. For this, we need the extra assumptions that $a^l_{\tnum{1}},a^l_{\tnum{2}},b^l_{\tnum{1}},c^l_{\tnum{2}}$ are if-free, which is a side-condition of $\csmb$.

\begin{lemma}
  \label{lem:cond-equiv-bis-body}
  For all $a,a',b,c$ such that their $R$-normal form is if-free and $a =_R a'$, if $P \npfproof a,a' \sim b,c$ then $b =_R c$.
\end{lemma}

\begin{proof}(\emph{sketch})
  The full proof is given in
  \iffull
  Appendix~\ref{app-section:if-free-conds}.
  \else
  the long version~\cite{long}.
  \fi
  It uses Proposition~\ref{prop:iffree-false-body} to obtain a proof $P'$ of $a,a' \sim b,c$ without any $\false$ and $\true$, and also relies on Proposition~\ref{prop:bas-cond-charac-body} and Lemma~\ref{lem:cond-equiv-body}.
\end{proof}

We now deal with the cut above. Using Lemma~\ref{lem:cond-equiv-bis-body}, we know that $b =_R c$. Since $b,c$ are in $R$-normal form, $b \equiv c$ and therefore $b^{\smallsquare}_{\tnum{1}} =_{\rs} b =_{\rs} c^{\smallsquare}_{\tnum{2}}$ (using well-formedness). Similarly $a^{\smallsquare}_{\tnum{1}} =_{\rs} a =_{\rs} a^{\smallsquare}_{\tnum{2}}$. This yields the (cut-free) proof:
\begin{equation*}
  \begin{array}[c]{c}
    \infer[\rs]{\vphantom{L}\sigma \sim \tau}
    {
      \infer[\csmb]{
        \begin{array}[c]{c}
          \begin{tikz}
            \tikzstyle{level 1}=[level distance=1.8em, sibling distance=3em]
            \tikzstyle{level 2}=[level distance=1.8em, sibling distance=2em]
            \tikzstyle{level 3}=[level distance=1.8em, sibling distance=2em]
            \node[anchor=base] (a) at (0,0){$a^{\smallsquare}_{\tnum{1}}$}
            child {node{$u_{\tnum{3}}$}}
            child {node{$w_{\tnum{5}}$}};
          \end{tikz}\\[-0.6em]
        \end{array}
        \quad\sim\quad
        \begin{array}[c]{c}
          \begin{tikz}
            \tikzstyle{level 1}=[level distance=1.8em, sibling distance=3em]
            \tikzstyle{level 2}=[level distance=1.8em, sibling distance=2em]
            \tikzstyle{level 3}=[level distance=1.8em, sibling distance=2em]
            \node[anchor=base] (a) at (0,0){$b^{\smallsquare}_{\tnum{1}}$}
            child {node{$r_{\tnum{3}}$}}
            child {node{$t_{\tnum{5}}$}};
          \end{tikz}\\[-0.6em]
        \end{array}
      }
      {
        \infer*[(A')]{
          a^l_{\tnum{1}},u_{\tnum{3}}
          \sim
          b^l_{\tnum{1}},r_{\tnum{3}}
        }{}
        \qquad & \qquad
        \infer*[(C)]{
          a^r_{\tnum{1}},w_{\tnum{5}}
          \sim
          b^r_{\tnum{1}},t_{\tnum{5}}
        }{}
      }
    }
  \end{array}
\end{equation*}
where $(A')$ is extracted from $(A)$ by Lemma~\ref{lem:body-restr-elim}. Finally, to get a proof in $\mathcal{A}_\succ$, we commute the $\rs$ rewriting to the beginning.

\subsection{Decision Procedure} 
Now, we explain how we obtain a decision procedure for our logic. Because the proofs and definitions are long and technical, we omit most of the details and focus instead on giving a high level sketch of the proof and decision procedure. 

\paragraph{Spurious Conditionals} 
A conditional $b$ without $\symite$ and in $R$-normal form is said to be \emph{spurious} in $t$ if, when $R$-normalizing $t$, the conditional $b$ disappears. Formally, $b$ is spurious in $t$ if $b \not \in \condst(t\downarrow_R)$. E.g., the conditional $\eq{\nonce_0}{\nonce_1}$ is spurious in:
\[
  \ite{\eq{\nonce_0}{\nonce_1}}{g(\nonce)}{g(\nonce)}
\]
We say that a basic conditional $\beta$, which may not be if-free, is spurious in $t$ if all its leaf terms are spurious in $t$ (i.e. $\leavest{(\beta\downarrow_R)}\cap\condst(t\downarrow_R) = \emptyset$). As we saw in Example~\ref{example:proof-base}, we may need to introduce spurious basic conditionals to carry out a proof. Still, we need to bound such terms. To do this, we characterize the basic conditionals that \emph{cannot} be removed: basically, a basic conditional is \abounded in a proof of $t \sim t'$ if it is not spurious in $t$ or $t'$, or if it is a guard for a decryption appearing in a \abounded conditional of $t \sim t'$ (indeed, we cannot remove a decryption's guards, as this would not yield a valid \cca instance).

We let $\aproof$ be the restriction of $\npfproof$ to proofs such that all basic conditionals appearing in the derivation are \abounded. Using the cut eliminations we introduced earlier, 
\iffull
plus some additional cut eliminations that are given in Appendix~\ref{app-section:if-free-conds},
\else
plus some additional cut eliminations,
\fi
we can show the following completeness result (the full proof is in
\iffull
Appendix~\ref{app-section:bounding}).
\else
the long version~\cite{long}).
\fi

\begin{lemma}
  \label{lem:bil-abounded-body}
  $\aproof$ is complete with respect to $\npfproof$.
\end{lemma}

\paragraph{Bounding \abounded Basic Conditionals}
Finally, it remains to bound the size of \abounded basic conditionals. Since basic conditionals can be nested (e.g. a basic conditional can contain decryption guards, which are themselves basic conditionals etc), we need to bound the length of sequences of nested basic conditionals.

Given a sequence of nested basic conditionals $\beta_1 <_{\st} \dots <_{\st} \beta_n$, (where $u <_{\st} v$ iff $u \not \equiv v$ and $u \in \st(v)$), we show that we can associate to each $\beta_i$ a ``frame term'' $\lambda_i \in \mathcal{B}(t,t')$ (where $\mathcal{B}(t,t')$ is a set of terms of bounded size w.r.t. $|t| + |t'|$). Basically, $\lambda_i$ is obtained from $\beta_i$ by ``flattening'' it: we remove all decryption guards, and replace the content of every encryption $\enc{m}{\pk}{\nonce}$ by a term $\enc{\tilde m}{\pk}{\nonce}$, where $\tilde m$ is if-free  and in $\mathcal{B}(t,t')$. Moreover, we show that, for every $\ekl$-normalized basic terms $\beta,\gamma$ and their associated frame terms $\lambda,\mu$, if $\lambda \equiv \mu$ then $\beta \equiv \gamma$ (this result is similar to Proposition~\ref{prop:bas-cond-charac-body}).

Since the $\beta_i$s are all pair-wise distinct (as $<_{\st}$ is strict), and since for every $i$, the frame term $\lambda_i$ uniquely characterizes $\beta_i$, we know that the $\lambda_i$s are pair-wise distinct. Using a pigeon-hole argument, this shows that $n \le |\mathcal{B}(t,t')|$. Then, by induction on the number of nested basic conditionals, we show a triple exponential upper-bound in $|t| + |t'|$ on the size of the basic conditionals appearing in a cut-free proof of~$t \sim t'$.

\paragraph{Decision Procedure} 
To conclude, we show that there exists a non-deterministic procedure that, given two terms $t$ and $t'$, non-deterministically guesses a set of \abounded basic terms that can appear in a proof $P$ of $P \aproof t \sim t'$ (in triple exponential time in $|t| + |t'|$). Then the procedure guesses the rule applications, and checks that the candidate derivation is a valid proof (in polynomial time in the candidate derivation size). This yields a $3$-\textsc{NExpTime} decision procedure that shows the decidability of our problem. 
\begin{theorem*}[Main Result]
  The following problem is decidable:\\
  \textbf{Input:} A ground formula $\vec u \sim \vec v$.\\
  \textbf{Question:} Is $\textsf{Ax} \wedge \vec u \not \sim \vec v$ unsatisfiable?
\end{theorem*}

%%% Local Variables:
%%% mode: latex
%%% TeX-master: "ms"
%%% End:

\section{Related Works}
\label{sec:related}

In \cite{DBLP:conf/ccs/BartheCGKLSB13}, the authors design a set of inference rules to prove $\textsf{CPA}$ and $\textsf{CCA}$ security of asymmetric encryption schemes in the Random Oracle Model. The paper also presents an attack finding algorithm. The authors of \cite{DBLP:conf/ccs/BartheCGKLSB13} do not provide decision algorithm for the designed inference rules. However, they designed proof search heuristics and implemented an automated tool, called \textsf{ZooCrypt}, to synthesize new $\textsf{CCA}$ encryption schemes. For small schemes, this procedure can show $\textsf{CCA}$ security or find an attack in more than $80\%$ of the cases. In $20\%$ of the cases, security remains undecided. Additionally, \textsf{ZooCrypt} automatically generates concrete security bounds.

As seen in the introduction, the problem of showing $\textsf{CPA}$ security can be cast into the BC logic. Take a candidate encryption scheme $x \mapsto t[x]$, where $t[]$ is a context built using, e.g., pairs, a one-way permutation $f$ using public key $\pk(\nonce)$, hash functions and xor. Then this scheme is $\textsf{CPA}$ if the following formula is valid in every computational model satisfying some implementation assumptions (mostly, $f$ is \textsf{OW-CPA} and the hash functions are $\textsf{PRF}$):
\[
  t[\pi_1(f(\pk(\nonce)))] \sim t[\pi_2(f(\pk(\nonce)))]
\]
This formula has a particular shape, which stems from the limitations on the adversary's interactions: the adversary can only interact with the (candidate) encryption scheme through the $\textsf{CPA}$ or $\textsf{CCA}$ game. There is no complex and arbitrary interactions with the adversary, as it is the case with a security protocol. We don't have such restrictions.

In \cite{10.1007/978-3-642-33167-1_33}, the authors study proof automation in the UC framework~\cite{Canetti:2001:UCS:874063.875553}. They design a complete procedure for deciding the existence of a simulator, for ideal and real functionalities using if-then-else, equality, random samplings and xor. Therefore their algorithm cannot be used to analyse functionalities relying on more complex functions (e.g., public key encryption), or stateful functionalities. This restricts the protocols that can be checked. Still, their method is \emph{semantically} complete (while we are complete w.r.t. a fixed set of inference rules): if there exists a simulator, they will find it.

In \cite{10.1007/978-3-642-17511-4_4}, the authors show the decidability of the problem of the equality of two distributions, for a \emph{specific} equational theory (concatenation, projection and xor). Then, for \emph{arbitrary} equational theories, they design a proof system for proving the equality of two distributions. This second contribution has similarities with our work, but differ in two~ways.

First, the proof system of \cite{10.1007/978-3-642-17511-4_4} shares some rules with ours, e.g. the $R$, $\dup$ and $\fa$ rules. But it does not allow for reasoning on terms using \symite. E.g., they do not have a counterpart to the $\cs$ rule. This is a major difference, as most of the difficulties encountered in the design of our decision procedure result from the \symite conditionals. Moreover, there are no rules corresponding to cryptographic assumptions, as our $\cca$ rules. Because of this and the lack of support for reasoning on branching terms, the analysis of security protocols is out of the scope of \cite{10.1007/978-3-642-17511-4_4}.

Second, the authors do not provide a decision procedure for their inference rules, but instead rely on heuristics.

\section{Conclusion}
\label{section:conclusion}

We designed a decision procedure for the Bana-Comon indistinguishability logic. This allows to automatically verify that a security protocol satisfies some security property. Our result can be reinterpreted, in the cryptographic game transformation setting, as a cut elimination procedure that guarantees that all intermediate games introduced in a proof are of bounded size w.r.t. the protocol studied.

A lot of work remains to be done. First, our decision procedure is in $3$-\textsc{NExpTime}\ch{, which is a high complexity. But, as we do not have any lower-bound, there may exist a more efficient decision procedure. Finding such a lower-bound is another interesting direction of research.} Then, our completeness result was proven for $\CCA$ only. We believe it can be extended to more primitives and cryptographic assumptions. For example, signatures and \textsc{euf-cma} are very similar to asymmetric encryption and $\textsc{ind-cca}_2$, and should be easy to handle (even combined with the $\CCA$ axioms).

\section*{Acknowledgment}
We thank Hubert Comon for his help and useful comments.

This research has been partially funded by the French National Research Agency (ANR) under the project TECAP (ANR-17-CE39-0004-01).

%%% Local Variables:
%%% mode: latex
%%% TeX-master: "ms"
%%% End:

\iffull\else
\IEEEtriggercmd{\enlargethispage{-7cm}}
\IEEEtriggeratref{10}
\fi

\bibliography{biblio}

\iffull
\appendices
\clearpage

\onecolumn
\tableofcontents

\newpage

\section{The Term Rewriting System \texorpdfstring{$R$}{R}}
\label{app-section:trs}

\subsection{Notations}
\begin{definition}
  A \emph{position} is a word in $\mathbb{N}^*$. The value of a term $t$ at a position $p$, denoted by $(t)_{|p}$, is the partial function defined inductively as follows:
  \[
    \begin{array}{lcl}
      (t)_{|\epsilon} &=& t\\
      (f(u_0,\dots,u_{n-1}))_{|i.p} &=& 
      \begin{cases}
        (u_i)_{|p} & \text{ if } i < n\\
        \text{undefined} & \text{ otherwise}
      \end{cases}
    \end{array}
  \]
  We say that a position in \emph{valid} is $t$ if $(t)_{|p}$ is defined. The set of positions of a term is the set of positions which are valid in $t$.
\end{definition}

\begin{definition}
  A context $D[]_{\vec x}$ (sometimes written $D$ when there is no confusion) is a term in $\mathcal{T}(\sig,\Nonce,\{[]_y \mid y \in \vec x\})$ where $\vec x$ are distinct special variables called holes. 

  For all contexts $D[]_{\vec x},C_0,\dots,C_{n-1}$ with $|\vec x| = n$, we let $D[(C_i)_{i<n}]$ be the context $D[]_{\vec x}$ in which we substitute, for all $0 \le i < n$, all occurrences of the hole $[]_{x_i}$ by $C_i$.

  A one-holed context is a context with one hole (in which case we write $D[]$ where $[]$ is the only variable).
\end{definition}

Often, we want to distinguish between holes that contain ``internal'' conditionals, and holes that contain terms appearing at the leaves. To do this we introduce the notion of if-context:
\begin{definition}
  For all distinct variables $\vec x, \vec y$, an if-context $D[]_{\vec x\diamond \vec y}$ is a context in $\mathcal{T}(\ite{\_}{\_}{\_},\{[]_z \mid z \in \vec x \cup \vec y\})$ such that for all position $p$, $D_{|p} \equiv \ite{b}{u}{v}$ implies:
  \begin{itemize}
  \item $b \in \{ []_z\mid z \in \vec x\}$
  \item $u,v \not \in  \{ []_z\mid z \in \vec x\}$
  \end{itemize}
  % Given two vector of terms $\vecu,\vec v$ of length, respectively, $|\vec x|,|\vec y|$ we
\end{definition}

\begin{example}
  Let $\vec x = x_1,x_2,x_3$ and $\vec y = y_1,y_2,y_3,y_4$, we give below two representations of the same if-context $D[]_{\vec x \diamond \vec y}$ (the term on the left, and the labelled tree on the right):

  \begin{minipage}{0.5\linewidth}
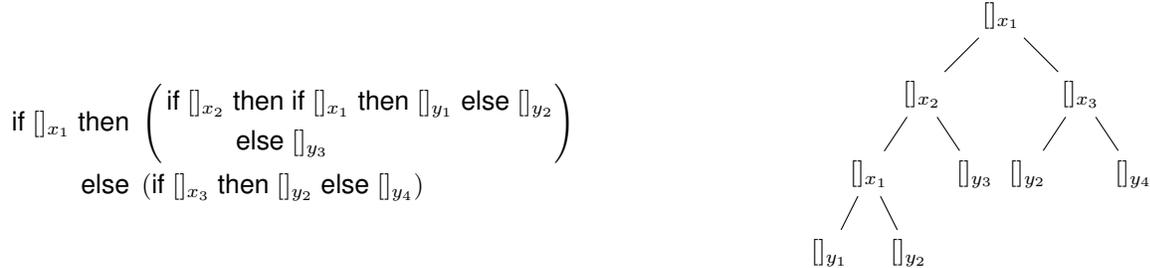
    \begin{alignat*}{2}
      \ite{[]_{x_1}}
      {&
        \left(\vcenter{\hbox{$\displaystyle
              \begin{alignedat}[t]{2}
                \ite{[]_{x_2}}
                {&\ite{[]_{x_1}}{[]_{y_1}}{[]_{y_2}}\\}
                {&[]_{y_3}}
              \end{alignedat}$}}
        \right)
        \\}
      {&\left(\ite{[]_{x_3}}{[]_{y_2}}{[]_{y_4}}\right)}
    \end{alignat*}
  \end{minipage}
  \begin{minipage}{0.5\linewidth}
    \begin{center}
      \begin{tikzpicture}[sibling distance=5em,level distance=3em]
        \tikzstyle{level 1}=[sibling distance=6em] 
        \tikzstyle{level 2}=[sibling distance=4em] 
        \tikzstyle{level 3}=[sibling distance=3em] 

        \node at (0,0) {$[]_{x_1}$}
        child {
          node  {$[]_{x_2}$}
          child {
            node  {$[]_{x_1}$}
            child {node {$[]_{y_1}$}}
            child {node {$[]_{y_2}$}}
          }
          child {node {$[]_{y_3}$}}
        }
        child {
          node  {$[]_{x_3}$}
          child {node {$[]_{y_2}$}}
          child {node {$[]_{y_4}$}}
        };
      \end{tikzpicture}
    \end{center}
  \end{minipage}
\end{example}

\begin{definition}
  For every term $t$, we let $\st(t)$ be the set of subterms of $t$.

  If $t \equiv C[\vec b \diamond \vec u]$ where $\vec b,\vec u$ are if-free terms then we let $\condst(t)$ be the set of conditionals $\vec b$, and $\leavest(t)$ be the set of terms $\vec u$. 
\end{definition}

\begin{definition}
  A \emph{directed path} $\dpath{\vec{\rho}}$ is a sequence $(b_0,d_0),\dots,(b_n,d_n)$ where $b_0,\dots,b_n$ are conditionals and $d_0,\dots,d_n$ (the directions) are in $\{\textsf{then},\textsf{else}\}$.
  
  Two directed paths $\dpath{\vec{\rho}}$ and $\dpath{\pvec{\rho}}'$ are said to have the same directions if:
  \begin{itemize}
  \item  they have the same length.
  \item the sequences of \emph{directions} $d_0,\dots,d_n$ and $d'_0,\dots,d'_n$ extracted from, respectively, $\dpath{\vec{\rho}}$ and $\dpath{\pvec{\rho}}'$, are equal.
  \end{itemize}

  Given a directed path $\dpath{\vec \rho}$, we let $\vec \rho$ stands for the sequence of \emph{conditionals} extracted from $\dpath{\vec \rho}$.
\end{definition}

\subsection{Convergence of \texorpdfstring{$R$}{R}}
\label{subsection:trs-convergent}

\paragraph{Lexicographic Path Ordering:} Let $\succ_f$ be a total precedence over function symbols. The lexicographic path ordering associated with $\succ_f$ is the pre-order defined by:
\[ s = f(s_1,\dots,s_n) \succ t = g(t_1,\dots,t_m) \text{ iff }
  \begin{cases}
    \exists i \in \llbracket 1,n\rrbracket \text{ s.t. } s_i \succeq t\\
    f = g \wedge  \forall j \in \llbracket 1,m\rrbracket, s \succ t_j \wedge s_1,\dots,s_n \succ_{lex} t_1,\dots,t_n\\
    f \succ_f g \wedge  \forall j \in \llbracket 1,m\rrbracket, s \succ t_j
  \end{cases}
\]

Let $\succ_f$ be a total precedence on $\sig,\Nonce$ such that $\symite$ is the smallest element (elements of $\Nonce$ are treated as function symbols of arity zero). Let $\succ$ be the lexicographic path ordering on $\mathcal{T}(\sig,\Nonce)$ using precedence $\succ_f$. Let $\succ_u$ be a user-chosen total order on if-free conditionals in $R$-normal form. We define the total ordering $\succ_c$ on conditionals as follows:
\[
  b \succ_c a = 
  \begin{cases}
    b \succ_u a & \text{ if $a$ and $b$ are if-free and $R$-irreducible}\\
    b \succ a & \text{ if $a$ and $b$ are not if-free or not $R$-irreducible}\\
    \true & \text{ if $a$ is if-free and $R$-irreducible, and $b$ is not}\\
    \false & \text{ if $b$ is if-free and $R$-irreducible, and $a$ is not}
  \end{cases}
\]
We then order $\ra_{R^{\succ_u}_4}$ as follows:
\begin{gather*}
  \ite{b}{(\ite{a}{x}{y})}{z} \ra \ite{a}{(\ite{b}{x}{z})}{(\ite{b}{y}{z})} \qquad \text{ when } b \succ_c a \\
  \ite{b}{x}{(\ite{a}{y}{z})} \ra \ite{a}{(\ite{b}{x}{y})}{(\ite{b}{x}{z})} \qquad \text{ when } b \succ_c a
\end{gather*}
Let $\ra_{R^{\succ_u}} = \ra_{R_1} \cup \ra_{R_2} \cup \ra_{R_3} \cup \ra_{R_4^{\succ_u}}$. The term rewriting system $\ra_{R^{\succ_u}}$ is an orientation of the rules given in Fig.~\ref{fig:trs}. When we do not care about the choice of total ordering on if-free conditionals in $R$-normal form $\succ_u$, we write $\ra_R$.

% \begin{figure}
%   \[
%     \begin{array}{ll}
%       \pi_i(\pair{x_1}{x_2}) \ra x_i & i \in \{1,2\}\\
%       \dec(\enc{x}{\pk(\nonce)}{z},\sk(\nonce)) \ra x\\
%       \ite{b_i}{(\ite{b_i}{x}{y})}{z} \ra \ite{b_i}{y}{z} & 1 \le i \le 5\\
%       \ite{b_i}{x}{(\ite{b_i}{y}{z})} \ra \ite{b_i}{x}{z} & 1 \le i \le 5 \\
%       \ite{b_j}{(\ite{b_i}{x}{y})}{z} \ra \ite{b_i}{(\ite{b_j}{x}{z})}{(\ite{b_j}{y}{z})} & 1 \le j < i \le 5 \\
%       \ite{b_j}{x}{(\ite{b_i}{y}{z})} \ra \ite{b_i}{(\ite{b_j}{x}{y})}{(\ite{b_j}{x}{z})} & 1 \le j < i \le 5 \\
%       \ite{(\ite{x}{y}{z})}{u}{v} \ra \ite{x}{(\ite{y}{u}{w})}{(\ite{z}{u}{w})}\\
%       \ite{y}{x}{x} \ra x\\
%       \eq{x}{x} \ra \true()\\
%       \ite{\true}{x}{y} \ra x\\
%       \ite{\false}{x}{y} \ra y\\
%       f(\vec{u},\ite{b_i}{x}{y},\vec v) \ra \ite{b_i}{f(\vec{u},x,\vec v)}{f(\vec{u},y,\vec v)} & 1 \le i \le 5
%     \end{array}
%   \]
%   \textbf{Convention:} $f \in \{ \pk(\_),\sk(\_),\enc{\_}{\_}{\_},\dec(\_,\_),\pi_1(\_),\pi_2(\_),\pair{\_}{\_},\eq{\_}{\_},g(\_,\_,\_)\}$
%   \caption{\label{figure:conf-spec-rules} The Term Rewriting System $\ra_{R_{\textsf{const}}}$ used for local confluence}
% \end{figure}

\begin{figure}
  \[\begin{array}{ll}
      \ra_{R_2'}
      &\left\{\begin{array}{lr}
          f(\vec{u},\fite{b}{x}{y},\vec v) \ra \fite{b}{f(\vec{u},x,\vec v)}{f(\vec{u},y,\vec v)} & (f \in \ssig)\\
          \fite{(\fite{b}{a}{c})}{x}{y} \ra \fite{b}{(\fite{a}{x}{y})}{(\fite{c}{x}{y})}
        \end{array}\right.\\[1em]
      \ra_{R_3'}
      & \left\{\begin{array}{l}
          \fite{\true}{x}{y} \ra x \\
          \fite{\false}{x}{y} \ra y \\
          \fite{b}{x}{x} \ra x \\
          \fite{b}{(\fite{b}{x}{y})}{z} \ra \fite{b}{x}{z}\\
          \fite{b}{x}{(\fite{b}{y}{z})} \ra \fite{b}{x}{z}
        \end{array}\right.\\[3em]
      \ra_{R_4^0}
      & \left\{\begin{array}{lrr}
          \lefteqn{\ite{b}{(\ite{a}{x}{y})}{z} \ra}\\
          & \ite{a}{(\ite{b}{x}{z})}{(\ite{b}{y}{z})} \quad& (b \succ a, \text{ $a$,$b$ not if-free or not in $R$-normal form})\\
          \lefteqn{\ite{b}{x}{(\ite{a}{y}{z})} \ra}\\
          & \ite{a}{(\ite{b}{x}{y})}{(\ite{b}{x}{z})} \quad& (b \succ a, \text{ $a$,$b$ not if-free or not in $R$-normal form})
        \end{array}\right.\\[2em]
      \ra_{R_4^1}
      & \left\{\begin{array}{lrr}
          \lefteqn{\ite{b}{(\fite{a}{x}{y})}{z} \ra}\\
          & \fite{a}{(\ite{b}{x}{z})}{(\ite{b}{y}{z})} \qquad\qquad\;\;& (\text{$b$ not if-free or not in $R$-normal form})\\
          \lefteqn{\ite{b}{x}{(\fite{a}{y}{z})} \ra}\\
          & \fite{a}{(\ite{b}{x}{y})}{(\ite{b}{x}{z})} \qquad\qquad\;\;& (\text{$b$ not if-free or not in $R$-normal form})
        \end{array}\right. \\[2em]
      \ra_{R_4^2}
      & \left\{\begin{array}{lr}
          \fite{b}{(\fite{a}{x}{y})}{z} \ra
          \fite{a}{(\fite{b}{x}{z})}{(\fite{b}{y}{z})} \quad\qquad\;\;& (b \succ_u a)\\
          \fite{b}{x}{(\fite{a}{y}{z})} \ra
          \fite{a}{(\fite{b}{x}{y})}{(\fite{b}{x}{z})} \quad\qquad\;\;& (b \succ_u a)
        \end{array}\right. \\[1em]
      \ra_{R^i}
      & \left\{\begin{array}{lr}
          \ite{b}{u}{v} \ra \fite{b}{u}{v} \qquad\qquad\qquad\qquad\qquad\quad\;\;\;& (\text{$b$ if-free and in $R$-normal form})
        \end{array}\right.
    \end{array}\]
  \caption{\label{fig:trs2} The Relations $\ra_{R_2'},\ra_{R_3'}, \ra_{R_4^0},\ra_{R_4^1},\ra_{R_4^2}$ and $\ra_{R^i}$ used for termination}
\end{figure}

\begin{theorem}
  \label{thm:trs-convergent}
  For all $\succ_u$, the term rewriting system $\ra_{R^{\succ_u}}$ is convergent on ground terms.
\end{theorem}

\begin{proof} 
  We show that $\ra_{R^{\succ u}}$ is locally confluent and terminating, and conclude by Newman's lemma.
  \paragraph{Local Confluence} 
  We show that all critical pairs are joinable. Normally, we would rely on some automated checker for local confluence. Unfortunately, as we rely on a side-condition to orient $R_4$ (using a LPO), writing down the rules in a tool is not straightforward. By consequence we believe it is simpler to manually check that every critical pair is joinable. We give below the most interesting critical pairs, and show how we join them. For every critical pair, we underline the starting term.
  \begin{itemize}
  \item \textbf{Critical Pairs $R_1/(R_1\cup R_2 \cup R_3 \cup R_4)$:} we only show the critical pairs involving $\pi_1(\_)$ (the critical pairs with $\pi_2(\_)$ are similar), and for $\eq{\_}{\_}$. The critical pairs involving $\dec(\_,\_)$ are similar to the critical pairs involving $\pi_1(\_)$.
    \[
      \ite{b}{u}{v}
      \la^2
      \ite{b}{\pi_1(\pair{u}{w})}{\pi_1(\pair{v}{w})}
      \la
      \underline{\pi_1(\pair{\ite{b}{u}{v}}{w})}
      \ra \ite{b}{u}{v}
    \]
    \[
      w
      \la
      \ite{b}{w}{w}
      \la^2
      \ite{b}{\pi_1(\pair{w}{u})}{\pi_2(\pair{w}{v})}
      \la
      \underline{\pi_1(\pair{w}{\ite{b}{u}{v}})}
      \ra w
    \]
    \begin{alignat*}{3}
      &&&\true
      &\la \\
      &&&\underline{\eq{\ite{b}{u}{v}}{\ite{b}{u}{v}}}\\
      &\ra\;\;&&
      \ite{b}{\eq{u}{\ite{b}{u}{v}}}{\eq{v}{\ite{b}{u}{v}}}\\
      &\ra&&
      \ite{b}{(\ite{b}{\eq{u}{u}}{\eq{u}{v}})}{\eq{v}{\ite{b}{u}{v}}}\\
      &\ra&&
      \ite{b}{\eq{u}{u}}{\eq{v}{\ite{b}{u}{v}}}\\
      &\ra&&
      \ite{b}{\true}{\eq{v}{\ite{b}{u}{v}}}\\
      &\ra^*&&
      \ite{b}{\true}{\true}\\
      &\ra&&
      \true
    \end{alignat*}

  \item \textbf{Critical Pairs $R_2/R_2$:} we assume that $b \succ_c c$. The other possible orderings are handled in the same fashion.
    \begin{alignat*}{3}
      &&\ite{c}{(\ite{b}{f(u,s)}{f(v,s)})}{(\ite{b}{f(u,t)}{f(v,t)})}
      &\;\;\la^2&\\
      &&\ite{c}{f(\ite{b}{u}{v},s)}{f(\ite{b}{u}{v},t)}
      &\;\;\la&\\
      &&\underline{f(\ite{b}{u}{v}, \ite{c}{s}{t})}
      \\ &\ra\;\;&
      \ite{b}{f(u,\ite{c}{s}{t})}{f(v,\ite{c}{s}{t})}\\
      &\ra^2\;\;&
      \ite{b}{(\ite{c}{f(u,s)}{f(u,t)})}{(\ite{c}{f(v,s)}{f(v,t)})}\\
      &\ra^*\;\;&
      \ite{c}{(\ite{b}{f(u,s)}{f(v,s)})}{(\ite{b}{f(u,t)}{f(v,t)})}
    \end{alignat*}
    
  \item \textbf{Critical Pairs $R_2/R_3$:}
    \[
      f(u, w)
      \la
      \underline{f(\ite{\true}{u}{v}, w)}
      \ra 
      \ite{\true}{f(u,w)}{f(v,w)}
      \ra 
      f(u,w)
    \]
    \[
      f(u, v)
      \la
      \underline{f(\ite{b}{u}{u}, v)}
      \ra 
      \ite{b}{f(u,v)}{f(u,v)}
      \ra 
      f(u,v)
    \]
    \begin{alignat*}{3}
        &&&\ite{b}{f(u,s)}{f(w,s)}
        &\;\;\la\\
        &&&f(\ite{b}{u}{w}, s)
        &\;\;\la\\
        &&&\underline{f(\ite{b}{(\ite{b}{u}{v})}{w}, s)}\\
        &\ra\;\; &&
        \ite{b}{f(\ite{b}{u}{v},s)}{f(w,s)}\\
        &\ra &&
        \ite{b}{(\ite{b}{f(u,s)}{f(v,s)})}{f(w,s)}\\
        &\ra&&
        \ite{b}{f(u,s)}{f(w,s)}
    \end{alignat*}

  \item \textbf{Critical Pairs $R_2/R_4$:} we assume that $a \succ_c b \succ_c c \succ_c d$. The other possible orderings are handled in the same fashion.
    \begin{alignat*}{3}
      &&&\ite{d}{(\ite{b}{(\ite{a}{u}{v})}{w})}{(\ite{c}{(\ite{a}{u}{v})}{w})}
      &\;\;\la^*\\
      &&&\ite{a\!\!\begin{array}[t]{l}}
      {
        \ite{d}
        {(\ite{b}{u}{w})}
        {(\ite{c}{u}{w})}
        \\}
      {
        \ite{d}
        {(\ite{b}{v}{w})}
        {(\ite{c}{v}{w})}
        \end{array}}
      &\;\;\la^2\\
      &&&\ite{a}
      {(
        \ite
        {(\ite{d}{b}{c})}
        {u}
        {w}
        )}
      {(
        \ite
        {(\ite{d}{b}{c})}
        {v}
        {w}
        )}
      &\;\;\la\\
      &&&\underline{\ite{(\ite{d}{b}{c})}{(\ite{a}{u}{v})}{w}}\\
      &\ra\;\;&&
      \ite{d}{(\ite{b}{(\ite{a}{u}{v})}{w})}{(\ite{c}{(\ite{a}{u}{v})}{w})}
    \end{alignat*}

  \item \textbf{Critical Pairs $R_3/R_3$:}
    \[
      u \la
      \underline{\ite{\true}{u}{u}}
      \ra u
    \]
    \[
      u
      \la
      \ite{\true}{u}{v}
      \la
      \underline{\ite{\true}{(\ite{\true}{u}{v})}{w}}
      \ra
      \ite{\true}{u}{w}
      \ra
      u
    \]
    \[
      \ite{b}{u}{v}
      \la
      \underline{\ite{b}{(\ite{b}{u}{v})}{(\ite{b}{u}{v})}}
      \begin{array}[t]{l}
        \ra
        \ite{b}{u}{(\ite{b}{u}{v})}\\
        \ra
        \ite{b}{u}{v}
      \end{array}
    \]

  \item \textbf{Critical Pairs $R_3/R_4$:}
    \begin{alignat*}{3}
      &&&\ite{a}{u}{v}
      &\;\;\la\\
      &&&\underline{\ite{b}{(\ite{a}{u}{v})}{(\ite{a}{u}{v})}}\\
      &\ra\;\;&&
      \ite{a}{(\ite{b}{u}{(\ite{a}{u}{v})})}{(\ite{b}{v}{(\ite{a}{u}{v})})}\\
      &\ra^2&&
      \ite{a\!\!\begin{array}[t]{l}}
      {
        \ite{a}{(\ite{b}{u}{u})}{(\ite{b}{u}{v})}
        \\}
      {
        \ite{a}{(\ite{b}{v}{u})}{(\ite{b}{v}{v})}
        \end{array}}\\
      &\ra^2&&
      \ite{a}{(\ite{b}{u}{u})}{(\ite{b}{v}{v})}\\
      &\ra^2&&
      \ite{a}{u}{v}
    \end{alignat*}

  \item \textbf{Critical Pairs $R_4/R_4$:} we assume that $a \succ_c b \succ_c c$. The other possible orderings are handled in the same fashion.
    \begin{alignat*}{3}
      &&&\ite{c\!\!\begin{array}[t]{l}}
      {
        \ite{b}{(\ite{a}{u}{s})}{(\ite{a}{v}{s})}
        \\}
      {
        \ite{b}{(\ite{a}{u}{t})}{(\ite{a}{v}{t})}
        \end{array}}
      &\;\;\;\la^2\\
      &&&\ite{c}
      {(\ite{a}{(\ite{b}{u}{v})}{s})}
      {(\ite{a}{(\ite{b}{v}{u})}{t})}
      &\;\;\la\\
      &&&\underline{\ite{a}{(\ite{b}{u}{v})}{(\ite{c}{s}{t})}}\\
      &\ra\;\;&&
      \ite{b}
      {(\ite{a}{u}{(\ite{c}{s}{t})})}
      {(\ite{a}{v}{(\ite{c}{s}{t})})}\\
      &\ra^2\;\;&&
      \ite{b\!\!\begin{array}[t]{l}}
      {
        \ite{c}{(\ite{a}{u}{s})}{(\ite{a}{u}{t})}
        \\}
      {
        \ite{c}{(\ite{a}{v}{s})}{(\ite{a}{v}{t})}
        \end{array}}\\
      &\ra^*\;\;&&
      \ite{c\!\!\begin{array}[t]{l}}
      {
        \ite{b}{(\ite{a}{u}{s})}{(\ite{a}{v}{s})}
        \\}
      {
        \ite{b}{(\ite{a}{u}{t})}{(\ite{a}{v}{t})}
        \end{array}}   
    \end{alignat*}
  \end{itemize}

  % First we describe the TRS given in \autoref{figure:conf-spec-rules}. This TRS includes five boolean constants, $b_1,\dots,b_5$. Most rules from $\ra_{R^{\succ_u}}$ involving $\symite$ have been replaced by specialized version using only $b_1,\dots,b_5$, where $b_j \succ b_i$ if $i < j$. We also have only a finite number of function symbols. These restriction allow us to write this TRS using a finite number of rules without side-conditions (it translated into 118 rules). 

  % We then checked that the term rewrite system of \autoref{figure:conf-spec-rules} is locally confluent. For this we used an automatized checker, the CRC local confluence checker from the Maude Formal Environment~\cite{mfetool}.

  % Any critical pair of $\ra_{R^{\succ u}}$ involves at most 5 conditionals. Moreover any critical pair involving a function $f$ or arity larger than three relies on at most two of its fields. Therefore all critical pairs of $\ra_{R^{\succ u}}$ are joinable using the local confluence of the TRS described in \autoref{figure:conf-spec-rules}.

  \paragraph{Termination} To prove termination we add to $\sig$ a symbol $\fite{b}{}{}$ for all if-free conditional $b$ in $R$-normal form. We also extend the precedence $\succ_f$ on function symbol by having the function symbols $\{\fite{b}{}{}\}$ be smaller than all the other function symbols, and $\fite{b}{}{} \succ_f \fite{a}{}{}$ if and only if $b \succ_u a$. Observe that the extended precedence is still a total order.

  We then consider the term rewriting system $\ra_{R'}$, defined by removing $\ra_{R_4}$ from $\ra_{R}$ and adding all the rules in Fig.~\ref{fig:trs2}:
  \[
    \ra_{R'} = \ra_{R_1} \cup \ra_{R_2} \cup \ra_{R_2'} \cup \ra_{R_3} \cup \ra_{R_3'}\cup\ra_{R_4^0}\cup\ra_{R_4^1}\cup\ra_{R_4^2}\cup\ra_{R^i}
  \]

  One can easily (but tediously) check that $\succ$ is compatible with $\ra_{R'}$: the only non-trivial cases are the cases in $\ra_{R_2}$ (the first rule is decreasing because $f \succ_f \symite$, the second rule using the lexicographic order), in $\ra_{R_2'}$ (same arguments than for $R_2$) and the cases in $\ra_{R_4^0},\ra_{R_4^1},\ra_{R_4^2}$ (where we use the side conditions $b \succ a$, $b \succ_u a$ \dots).

  Since $\succ$ is a lexicographic path ordering we know that it is total and well-founded on ground-terms. Therefore $\ra_{R'}$ is a terminating TRS on ground terms.
  
  To conclude, one just has to observe that for every ground terms $u,v$ and integer $n$, if $u \ra^{(n)}_R v$ then there exist $u',v'$ such that $u \ra_{R^i}^! u'$, $v \ra_{R^i}^! v'$ and $u' \ra_{R'}^{(\ge n)} v'$. That is, we have the following diagram (black edges stand for universal quantifications, red edges for existentials):
  \begin{center}
    \begin{tikzpicture}
      \node[minimum size=3em] (u) at (0,0){$u$};
      \node[minimum size=3em,minimum width=4em] (v) at (3,0){$v$};
      \node[minimum size=3em] (u') at (0,-2){$u'$};
      \node[minimum size=3em,minimum width=4em] (v') at (3,-2){$v'$};
      
      \draw[->] (u) -- (v) 
      node[above,xshift=-1.5em]{\footnotesize{$*$}} 
      node[below,xshift=-1.5em]{\footnotesize{$R$}};
      \draw[red,->] (u) -- (u') 
      node[right,yshift=1.5em]{$!$}
      node[left,yshift=1.5em]{$R^i$};
      \draw[red,->] (u') -- (v')
      node[above,xshift=-1.5em]{\footnotesize{$*$}}
      node[below,xshift=-1.5em]{\footnotesize{$R'$}};
      \draw[red,->] (v) -- (v') 
      node[right,yshift=1.5em]{$!$}
      node[left,yshift=1.5em]{$R^i$};
    \end{tikzpicture}
  \end{center}
  This result can be proved by induction on $n$. Since $\ra_{R'}$ is terminating on ground terms, and since any infinite sequence for $\ra_{R}$ can  be translated into an infinite sequence for $\ra_{R'}$, it follows easily that $\ra_{R}$ is terminating on ground terms.
\end{proof}

\subsection{Property of \texorpdfstring{$R$}{R}}

\begin{proposition}
  \label{prop:trs-prec}
  Let $\succ_u$ and $\succ_u'$ be two total orderings on if-free conditionals in $R$-normal form. Then for every ground term $t$ we have:
  \[
    \leavest(t\downarrow_{R^{\succ_u}}) = \leavest(t\downarrow_{R^{\succ_u'}})  \qquad \text{and} \qquad
    \condst(t\downarrow_{R^{\succ_u}}) = \condst(t\downarrow_{R^{\succ_u'}})
  \]

\end{proposition}

\begin{proof}
  Let $\vec b = \leavest(t\downarrow_{R^{\succ_u}})$ and $\vec u = \condst(t\downarrow_{R^{\succ_u}})$, we know that there exists a if-context $C$ such that $t \downarrow_{R^{\succ_u}} \equiv C[\vec b \diamond \vec u]$. It is then easy to show by induction on the length of the reduction that for all $n$, if  $C[\vec b \diamond \vec u] \ra^{(n)}_{R^{\succ_u'}} v$ then there exists an if-context $C'$ such that $v \equiv C'[\vec b \diamond \vec u]$. The wanted result follows immediately.
\end{proof}

%%% Local Variables:
%%% mode: latex
%%% TeX-master: "ms"
%%% End:

\newpage

\section{The \texorpdfstring{$\cca$}{CCA2} Axioms}
\label{app-section:cca-axiom}

We define and prove correct a recursive set of axioms for an $\textsc{ind-cca}_2$ encryption scheme. For the sack of simplicity, we first ignore all length constraints. We explain how length constraints are added and handled to the logic in Section~\ref{subsection:length}.

\paragraph{Multi-Users $\textsc{ind-cca}_2$ Game}
Consider the following multi-users $\textsc{ind-cca}_2$ game: the adversary receives $n$ public-keys. For each key $\pk_i$, he has access to a left-right oracle $\mathcal{O}_{\textsf{LR}}(\pk_i,b)$ that takes two messages $m_0,m_1$ as input and returns $\enc{m_b}{\pk_i}{\nonce_r}$, where $b$ is an internal random bit uniformly drawn at the beginning by the challenger (the same $b$ is used for all left-right oracles) and $\nonce_r$ is a fresh nonce. Moreover, for all key pairs $(\pk_i,\sk_i)$, the adversary has access to an $\sk_i$ decryption oracle $\mathcal{O}_{\textsf{dec}}(\sk_i)$, but cannot call $\mathcal{O}_{\textsf{dec}}(\sk_i)$ on a cipher-text returned by $\mathcal{O}_{\textsf{LR}}(\pk_i,b)$ (to do this, the two oracles use a shared memory where all encryption requests are logged). The advantage of an adversary against this game and the multi-user $\textsc{ind-cca}_2$ security are defined as usual.

It is known that if an encryption scheme is $\textsc{ind-cca}_2$ then it is also multi-users $\textsc{ind-cca}_2$ (see~\cite{musu}). Therefore, we allow multiple key pairs to appear in the $\cca$ axioms, and multiple encryptions over different terms using the same public key (each encryption corresponds to one call to a left-right oracle).

\paragraph{Decryption Guards}
If we want the following to hold in any computational model
\[
  \dec\big(
  \underbrace{
    t\big[
    \enc{u_1}{\pk}{\nonce_1},\dots,\enc{u_n}{\pk}{\nonce_n}
    \big]}_{s},\sk
  \big)
  \sim
  \dec\big(
  \underbrace{
    t\big[
    \enc{v_1}{\pk}{\nonce_1},\dots,\enc{v_n}{\pk}{\nonce_n}
    \big]}_{s'},\sk
  \big)
\]
then we need to make sure that $s$ is different from all $\enc{u_i}{\pk}{\nonce_i}$ and that $s'$ is different from all $\enc{v_i}{\pk}{\nonce_i}$. This is done by introducing all the unwanted equalities in $\symite$ tests and making sure that we are in the $\textsf{else}$ branch of all these tests, so as to have a ``safe call'' to the decryption oracle. Moreover, the adversary is allowed to use values obtained from previous calls to the decryption oracle in future calls.

To do this, we use the following function:
\begin{definition} We define the function $\elses$ by induction:
  \begin{align*}
    &\elses(\emptyset, x) \equiv x\\
    &\elses\left((\eq{a}{b}) :: \Gamma, x\right) \equiv \ite{\eq{a}{b}}{\zero(x)}{\elses(\Gamma,x)}
  \end{align*}
\end{definition}

\begin{example}
  Let $u \equiv t[\enc{v_1}{\pk}{\nonce_r^1},\enc{v_2}{\pk}{\nonce_r^2}]$. Then:
  \begin{multline*}
    \elses\big(
    \big(
    \eq{u}{\enc{v_1}{\pk}{\nonce_r^1}},
    \eq{u}{\enc{v_2}{\pk}{\nonce_r^2}}
    \big),
    \dec(u,\sk)
    \big) \equiv\\
    \ite{\eq{u}{\enc{v_1}{\pk}{\nonce_r^1}}}
    {\zero(\dec(u,\sk))}
    {\ite
      {\eq{u}{\enc{v_2}{\pk}{\nonce_r^2}}}
      {\zero(\dec(u,\sk))}
      {\dec(u,\sk)}}
  \end{multline*}
  Morally, this represents a safe call to the decryption oracle.
\end{example}

\paragraph{Definition of $\CCA$}
We use the following notations: for any finite set $\keys$ of valid private keys, $\keys \decpos \pvec{u}$ holds if for all $\sk \in \keys$, the secret key $\sk$ appears only in decryption position in $\pvec{u}$; $\nodec(\keys,\pvec{u})$ denotes that for all $\sk(\nonce) \in \keys$, the only occurrences of $\nonce$ are in subterms $\pk(\nonce)$;  $\hiddenr(\pvec{r};\pvec{u})$ denotes that for all $\nonce_r \in \pvec{r}$, $\nonce_r$ appears only in encryption randomness position and is not used with two distinct plaintexts.

We are now going to define by induction the $\cca$ axiom. In order to do this we define by induction a binary relation $R_{\cca^a}^\keys$ on $\cca$ executions, where $\keys$ is the finite set of private keys used in the terms (corresponding to the public keys sent by the challenger).
\begin{definition}
  Let $\keys$ be a set of private keys. $(\phi,\encvar,\decvar,\randmap,\encmap,\decmap)$ is a $\cca$ execution if:
  \begin{itemize}
  \item $\phi$ is a vector of ground terms in $\calt(\sig,\Nonce)$.
  \item $\encvar$ and $\decvar$ are two disjoint sets of variables used as handles for, respectively, encryptions and decryptions.
  \item $\randmap$ is a substitution from $\encvar$ to $\Nonce$.
  \item $\encmap$ and $\decmap$ are substitutions from, respectively, $\encvar$ and $\decvar$, to ground terms in $\calt(\sig,\Nonce)$.
  \end{itemize}
\end{definition}
$\randmap$,  $\encmap$ and $\decmap$ co-domains are the sets of, respectively, encryption randomness, encryption oracle calls and decryption oracle calls in $\phi$. Intuitively, we have:
\[
  (\phi,\encvar,\decvar,\randmap,\encmap,\decmap)
  R_{\cca^a}^\keys
  (\psi,\encvar,\decvar,\randmap',\encmap',\decmap')
\]
when we can build $\phi$ and $\psi$ using function symbols, matching encryption oracle calls and matching decryption oracle calls.
\begin{definition}
  Let $\keys$ be a finite set of private keys. We define the binary relation $R_{\cca^a}^\keys$ by induction:
  \begin{enumerate}
  \item \textbf{No Call to the Oracles:} if $\keys \decpos \phi$ then $(\phi,\emptyset,\emptyset,\emptyset,\emptyset,\emptyset) R_{\cca^a}^\keys (\phi,\emptyset,\emptyset,\emptyset,\emptyset,\emptyset)$ for every sequence $\phi$ of ground terms in $\calt(\sig,\Nonce)$ such that $\nodec(\keys;\phi)$.

  \item \textbf{Encryption Case:} Let $x$ a fresh variable that does not appear in $\encvar \cup \decvar$, $\sk$ be a secret key in $\keys$ and $\pk$ the corresponding public key. Then:
    \begin{multline*}
      \big(
      (\phi,\enc{u}{\pk}{\nonce_r}),
      \encvar \cup \{x\},\decvar,\randmap \cup \{x \mapsto \nonce_r\},
      \encmap \cup \{x \mapsto \enc{u}{\pk}{\nonce_r}\} ,\decmap
      \big) \\
      R_{\cca^a}^\keys
      \big(
      (\psi,\enc{v}{\pk}{\nonce_r'}),
      \encvar \cup \{x\},\decvar,\randmap' \cup \{x \mapsto \nonce_r'\},
      \encmap' \cup \{x \mapsto \enc{v}{\pk}{\nonce'_{r}}\} ,\decmap'
      \big)
    \end{multline*}
    if there exist $t,t' \in \mathcal{T}(\nzsig, \Nonce, \mathcal{\encvar})$ such that:
    \begin{itemize}
    \item $(\phi,\encvar,\decvar,\randmap,\encmap,\decmap) R_{\cca^a}^\keys (\psi,\encvar,\decvar,\randmap',\encmap',\decmap')$
    \item $u \equiv t\decmap, \,v \equiv t'\decmap'$
    \item $\nodec(\keys;t,t')$, which ensures that the only decryptions are calls to the oracle.
    \item $\fresh{\nonce_r,\nonce_r'}{\phi,u,\psi,v}$ and $\hiddenr(\encvar\randmap \cup \encvar\randmap';\phi,u,\psi,v)$
    \end{itemize}

  \item \textbf{Decryption Case:} Let $\sk \in \keys$, $\pk$ the corresponding public key and $z$ be a fresh variable. Then:
    \begin{multline*}
      \big(
      \left(\phi,\elses(l,\dec(u,\sk))\right),
      \encvar,\decvar \cup \{z\},\randmap, \encmap,
      \decmap\cup \{z \mapsto \elses(l,\dec(u,\sk))\}
      \big) \\
      R_{\cca^a}^\keys
      \big(
      \left(\psi,\elses(l',\dec(v,\sk))\right),
      \encvar,\decvar \cup \{z\},\randmap', \encmap',
      \decmap'\cup \{z \mapsto \elses(l',\dec(v,\sk))\}
      \big)
    \end{multline*}
    if there exists $t \in \mathcal{T}(\nizsig,\Nonce,\encvar,\decvar)$ such that:
    \begin{itemize}
    \item $(\phi,\encvar,\decvar,\randmap,\encmap,\decmap) R_{\cca^a}^\keys (\psi,\encvar,\decvar,\randmap',\encmap',\decmap')$
    \item $u \equiv t\encmap\decmap$ and $v \equiv t\encmap'\decmap'$.
    \item Consider the set $\caly_u$ of variables $x \in \encvar$ such that the encryption binded to $x$ directly appears in $u$, i.e. appears outside of another encryption. That is, $x$ must appear in the term $u$ where we substituted every encryption $\enc{\_}{\pk}{\nonce_x} \in \codom(\encmap)$ by $\enc{0}{\pk}{\nonce_x}$:
      \[
        x\randmap \in
        u\bigl\{\enc{0}{\pk}{\nonce_x} / \enc{\_}{\pk}{\nonce_x}
        \mid \enc{\_}{\pk}{\nonce_x} \in \codom(\encmap)\bigr\}\downarrow_R
      \]
      Then $l$ is the sequence of guards $l \equiv (\eq{u}{y_1},\dots,\eq{u}{y_m})$ where $(y_1,\dots,y_m) = \textsf{sort}(\caly_u\encmap)$.

      Similarly, $l' \equiv (\eq{v}{y'_1},\dots,\eq{v}{y'_{m}})$ where $(y'_1,\dots,y'_m) = \textsf{sort}(\caly_u\encmap')$\footnote{Remark that we use, for $v$, the set $\caly_u$ defined using $u$. As we will see later, this is not a problem because $\caly_u = \caly_v$.}.
    \item $\nodec(\keys;t) \text{ and } \hiddenr(\encvar\randmap \cup \encvar\randmap';\phi,u,\psi,v)$
    \end{itemize}
  \end{enumerate}
  where $\textsf{sort}$ is a deterministic function sorting terms according to an arbitrary linear order.
\end{definition}

\begin{remark}
  \label{rem:necessary-guards}
  In the decryption case, we add a guard only for encryption that appear directly in $u$. Without this restriction, we would add one guard $\eq{u}{x\encmap}$ for every $x \in \encvar$ such that $x\encmap$ is an encryption using public-key $\pk$.

  For example, if $\encvar = \{x_0,x_1,x_2\}$ and $\encmap = \{x_0 \mapsto \alpha_0, x_1 \mapsto \alpha_1, x_2 \mapsto \alpha_2 \}$ where:
  \begin{mathpar}
    \alpha_0 \mapsto \enc{m_0}{\pk}{\nonce_0}

    \alpha_1 \mapsto \enc{m_1}{\pk}{\nonce_1}

    \alpha_2\mapsto \enc{\alpha_1}{\pk}{\nonce_2}
  \end{mathpar}
  then to guard $\dec(g(\alpha_2),\sk)$, we need to add three guards, $\eq{g(\alpha_2)}{\alpha_0}$, $\eq{g(\alpha_2)}{\alpha_1}$ and $\eq{g(\alpha_2)}{\alpha_2}$. This yields the term:
  \begin{alignat*}{4}
    &\ite{&\eq{g(\alpha_2)}{\alpha_0}&&}{&&
      \zero(\dec(g(\alpha_2),\sk))\\
      &\!\!}{\ite{&\eq{g(\alpha_2)}{\alpha_1}&&}{&&
        \zero(\dec(g(\alpha_2),\sk))\\
        &\!\!}{\ite{&\eq{g(\alpha_2)}{\alpha_2}&&}{&&
          \zero(\dec(g(\alpha_2),\sk))\\
          &\!\!}{&&&&&\dec(g(\alpha_2),\sk)\phantom{)}}}}
  \end{alignat*}
  But here, the adversary, represented by the adversarial function $g$, is computing the query to the decryption oracle using only $\alpha_2$. Hence, it cannot use $\alpha_1$, which is hidden by the encryption, nor $\alpha_0$ which does not appear at all. Therefore, there is no need to add the guards $\eq{g(\alpha_2)}{\alpha_0}$ and $\eq{g(\alpha_2)}{\alpha_1}$, since $g$ has a negligible probability of returning $\alpha_0$ or $\alpha_1$.

  To remove unnecessary guards when building the decryption oracle call $\dec(u,\sk)$, we require that $\eq{u}{\alpha}$ is added to the list of guards if and only if $\alpha \equiv \enc{\_}{\pk}{\nonce}$ appears directly in $u$. This yields smaller axioms, e.g. the term $\dec(g(\alpha_2),\sk)$ is guarded by:
  \begin{alignat*}{2}
    \ite{&&\eq{g(\alpha_2)}{\alpha_2}}{&\zero(\dec(g(\alpha_2),\sk))\\
      &&}{&\dec(g(\alpha_2),\sk)}
  \end{alignat*}
  Finally, the $\textsf{sort}$ function is used to ensure that guards are always in the same order, which guarantees that two calls with the same terms are guarded in the same way.
\end{remark}

We can now define the recursive set of axioms $\cca^a$ and show their validity. We also state and prove a key property of these axioms.
\begin{definition}
  $\cca^a$ is the set of unitary axioms $\phi \sim \psi \mu$, where $\mu$ is a renaming of names in $\Nonce$ and there exist two $\cca$ executions $\caly,\caly'$ such that:
  \begin{mathpar}
    \caly =
    (\phi,\encvar,\decvar,\randmap,\encmap,\decmap)

    \caly' =
    (\psi,\encvar,\decvar,\randmap',\encmap',\decmap')

    \caly \mathbin{R_{\cca^a}^\keys} \caly'
  \end{mathpar}
  In that case, we say that $(\caly,\caly')$ is a valid $\cca^a$ application, and $\phi \sim \psi\mu$ is a valid $\cca^a$ instance.
\end{definition}

\begin{proposition}
  All formulas in $\cca^a$ are computationally valid if the encryption scheme is $\textsc{ind-cca}_2$.
\end{proposition}

\begin{proof}
  First, $\phi \sim \psi\mu$ is computationally valid if and only if $\phi \sim \psi$ is computationally valid. Hence, w.l.o.g. we consider $\mu$ empty.  Let $\cmodel$ be a computational model where the encryption and decryption symbol are interpreted as an $\textsc{ind-cca}_2$ encryption scheme. Let $\phi \sim \psi$ be a valid instance of $\cca^a$ such that $\sem{\phi} \not \approx_{\cmodel} \sem{\psi}$ i.e. there is a $\pptm$ $\mathcal{A}$ that has a non-negligible advantage of distinguishing these two distributions.

  Since $\phi \sim \psi$ is an instance of $\cca$ we know that there exist two $\cca$ executions such that:
  \[
    (\phi,\encvar,\decvar,\randmap,\encmap,\decmap)
    R_{\cca^a}^\keys
    (\psi,\encvar,\decvar,\randmap',\encmap',\decmap')
  \]

  We are going to build from $\phi$ and $\psi$ a winning attacker against the multi-user $\textsc{ind-cca}_2$ game. This attacker has access to a $LR$ oracle and a decryption oracle for all keys in $\keys$. We are going to build by induction on $R_{\cca^a}^\keys$ a algorithm $\mathcal{B}$ that samples from $\sem{\phi}$ or $\sem{\psi}$ (depending on the oracles internal bit). The algorithm $\mathcal{B}$ uses a memoisation technique: it builds a store whose keys are subterms of $\phi, \psi$ already encountered and variable in $\encvar \cup \decvar$, and values are elements of the $\cmodel$ domain.
  \begin{enumerate}
  \item $(\phi,\emptyset,\emptyset,\emptyset,\emptyset,\emptyset) R_{\cca^a}^\keys (\phi,\emptyset,\emptyset,\emptyset,\emptyset,\emptyset)$: for every term $t$ in the vector $\phi$, $\mathcal{B}$ samples from $\sem{t}$ by induction as follows:
    \begin{itemize}
    \item if $t$ is in the store then $\mathcal{B}$ returns its value.
    \item nonce $\nonce$: $\mathcal{B}$ draws $\nonce$ uniformly at random and stores the drawn value.

      Remark that $\nodec(\keys,\phi)$ ensures that $\nonce$ is not used in a secret key $\sk$ appearing in $\keys$, which we could not compute. If it is a public key $\pk$,  either the corresponding secret key $\sk$ is such that $\sk \in \keys$ and the challenger sent us a random sample from $\sem{\pk}$, or $\sk$ does not appear in $\keys$ and then $\mathcal{B}$ can draw the corresponding key pair itself.

    \item $f(t_1,\dots,t_n)$, then $\mathcal{B}$ inductively samples the function arguments $\left(\sem{t_1},\dots,\sem{t_1} \right)$ and then samples from  $\sem{f} \left(\sem{t_1},\dots,\sem{t_1} \right)$. $\mathcal{B}$ stores the value at the key $f(t_1,\dots,t_n)$.
    \end{itemize}

  \item \textbf{Encryption Case}:
    \begin{multline*}
      \big(
        (\phi,\enc{u}{\pk}{\nonce_r}),\encvar \cup \{x\},
        \decvar,\randmap \cup \{x \mapsto \nonce_r\},
        \encmap \cup \{x \mapsto \enc{u}{\pk}{\nonce_r}\} ,\decmap
      \big) \\
      R_{\cca^a}^\keys
      \big(
        (\psi,\enc{v}{\pk}{\nonce_r'}),\encvar \cup \{x\},
        \decvar,\randmap' \cup \{x \mapsto \nonce_r'\},
        \encmap' \cup \{x \mapsto \enc{v}{\pk}{\nonce_r'}\} ,\decmap'
      \big)
    \end{multline*}

    Since we have $\fresh{\nonce_r,\nonce_r'}{\phi,u,\psi,v}$ we know that the top-level terms do not appear in the store. It is easy to check that $\mathcal{B}$ inductive definition is such that $\mathcal{B}$ store has a value associated with every variable in $\encvar \cup \decvar$ and that, if $x \in \encvar$, then the store value of $x$ is either sampled from $\sem{x\encmap}$ or from $\sem{x\encmap'}$ (depending on the challenger internal bit), and that if $x \in \decvar$ then the store value of $x$ is either sampled from $\sem{x\decmap}$ or from $\sem{x\decmap'}$ (depending on the challenger internal bit). We also observe that if the challenger internal bit is $0$ then for all $w$:
    \[
      \mathcal{O}_{\textsf{LR}}(\pk,b)(\sem{u},\sem{v}) =
      \mathcal{O}_{\textsf{LR}}(\pk,b)(\sem{u},w)
    \]

    Similarly if the challenger internal bit is $1$ then for all $w$:
    \[
      \mathcal{O}_{\textsf{LR}}(\pk,b)(\sem{u},\sem{v}) =
      \mathcal{O}_{\textsf{LR}}(\pk,b)(w,\sem{v})
    \]

    $\mathcal{B}$ samples two values $\alpha, \beta$ such that if the challenger internal bit is $0$ then $\alpha$ is sampled from $\sem{u}$ and if the challenger internal bit is $1$ then $\beta$ is sampled from $\sem{v}$. Therefore whatever the challenger internal is bit, $\mathcal{O}_{\textsf{LR}}(\pk,b)(\alpha,\beta)$ is sampled from $\mathcal{O}_{\textsf{LR}}(\pk,b)(\sem{u},\sem{v})$:
    \begin{itemize}
    \item $\alpha$ is sampled from $\sem{u}$ using the case 1 algorithm. Remark that when we encounter a decryption under $\sk' \in \keys$, we know that it was already sampled and can therefore retrieve it from the store.
    \item similarly, $\beta$ is sampled from $\sem{v}$ using the case 1 algorithm.
    \end{itemize}
    The condition $\nodec(\keys;t,t')$ ensures that no secret key from $\keys$ appears in $u,v$ anywhere else than in decryption positions for already queried oracle calls (which can therefore be retrieved from the store), and the two conditions $\fresh{\nonce_r,\nonce_r'}{\phi,u,\psi,v}$ and $\hiddenr(\encvar\randmap \cup \encvar\randmap';\phi,u,\psi,v)$ ensure that all randomness used by the challenger left-right oracles do not appear anywhere else than in encryption randomness position for the corresponding left-right oracle calls.

    We store the result of the left-right oracle call at key $x$.
  \item \textbf{Decryption Case}:
    \begin{multline*}
      \left(\left(\phi,\elses(l,\dec(u,\sk))\right),\encvar,\decvar \cup \{z\},\randmap, \encmap,\decmap\cup \{z \mapsto \elses(l,\dec(u,\sk))\}  \right) \\R_{\cca^a}^\keys  \left(\left(\psi,\elses(l',\dec(v,\sk))\right),\encvar,\decvar \cup \{z\},\randmap', \encmap',\decmap'\cup \{z \mapsto \elses(l',\dec(v,\sk))\}\right)
    \end{multline*}

    We know that $u \equiv t\encmap\decmap$ and $v \equiv t\encmap'\decmap'$. $\mathcal{B}$ uses the case 1 algorithm to sample $\gamma$ from $\sem{t\encmap\decmap}$ or $\sem{t\encmap'\decmap'}$ depending on the challenger internal bit. $\nodec(\keys;t)$ ensures that no call to the decryption oracles are needed and $\hiddenr(\encvar\randmap \cup \encvar\randmap';\phi,u,\psi,v)$ guarantee that the randomness drawn by the challenger for $LR$ oracle encryptions do not appear in $t$.

    Observe that all calls to $\mathcal{O}_{\textsf{LR}}(\pk,b)$ have already been stored. Let $x_1\encmap,\dots,x_p\encmap$ be the corresponding keys in the store. Hence if $\gamma$ is equal to any of the values stored at keys $x_1\encmap,\dots,x_p\encmap$ then $\mathcal{B}$ return $\sem{\zero}(\gamma)$, otherwise $\mathcal{B}$ can call the decryption oracle $\mathcal{O}_{\textsf{dec}}(\sk)$ on $\gamma$.

    As we observed in Remark~\ref{rem:necessary-guards}, if the challenger internal bit is $0$, checking whether $\gamma$ is different from the values sampled from $\sem{x_1\encmap},\dots,\sem{x_p\encmap}$  amounts to checking whether $\gamma$ is different from the values sampled from $\sem{y_1},\dots,\sem{y_m}$, except for a negligible number of samplings. Therefore we are sampling from the correct distribution (up to a negligible number of samplings).

    Moreover, the set of variables $x \in \encvar$ such that the encryption binded to $x$ in $\encmap$ appears directly in the \emph{left decryption} $u$:
    \[
      x\randmap \in
      u\bigl\{\enc{0}{\pk}{\nonce_x} / \enc{\_}{\pk}{\nonce_x}
      \mid \enc{\_}{\pk}{\nonce_x} \in \codom(\encmap)\bigr\}\downarrow_R
    \]
    is exactly the set of variables $x$ such that the encryption binded to $x$ in $\encmap'$ appears directly in the \emph{right decryption }$v$:
    \[
      x\randmap \in
      v\bigl\{\enc{0}{\pk}{\nonce_x} / \enc{\_}{\pk}{\nonce_x}
      \mid \enc{\_}{\pk}{\nonce_x} \in \codom(\encmap')\bigr\}\downarrow_R
    \]
    Hence, if the internal bit is $1$ then checking whether $\gamma$ is different from the values sampled from $\sem{x_1\encmap'},\dots,\sem{x_p\encmap'}$  amounts to checking whether $\gamma$ is different from the values sampled from $\sem{y'_1},\dots,\sem{y'_m}$, except for a negligible number of samplings.

    We store the result at key $z$.
  \end{enumerate}

  The attacker against the multi-user $\textsc{ind-cca}_2$ game simply returns $\mathcal{A}(\mathcal{B})$. Since $\mathcal{B}$ samples either from $\sem{\phi}$ if $b = 0$ or from  $\sem{\psi}$ if $b = 1$ (up to a negligible number of samplings), and since $\mathcal{A}$ has a non-negligible advantage of distinguishing  $\sem{\phi}$ from  $\sem{\psi}$ we know that the attacker has a non-negligible advantage against the multi-user $\textsc{ind-cca}_2$ game.
\end{proof}

\subsection{Closure Under \texorpdfstring{$\restr$}{Restr}}
\label{sub:restr-cca}
To close our logic under $\restr$, we need the unitary axioms to be closed. Therefore, we let $\CCA$ be the closure of $\cca^a$ under $\restr$.
\begin{definition}
  $\cca$ is the set of formula $\phi \sim\psi$ such that we have the derivation:
  \[
    \infer[\restr]{
      \phi \sim \psi
    }{
      \infer[\cca^a]{
        \phi' \sim \psi'
      }{}
    }
  \]
\end{definition}

The main contribution of this sub-section, given below, states that any instance $\pvec{u} \sim \pvec{v}$ of $\CCA$ can be automatically extended into an instance $\pvec{u}' \sim \pvec{v}'$ of $\cca^a$ of, at most, polynomial size.
\begin{proposition}
  \label{prop:cca-small-restr}
  For every instance $\pvec{u} \sim \pvec{v}$ of $\cca$, there exists $\pvec{u}_1,\pvec{v}_1$ such that $\pvec{u},\pvec{u}_1 \sim \pvec{v},\pvec{v}_1$ is an instance of $\cca^a$ (modulo $\perm$) and $|\pvec{u}_1| + |\pvec{v}_1|$ is of polynomial size in $|\pvec{u}| + |\pvec{v}|$. We let $\completion(\pvec{u} \sim \pvec{v})$ be the formula $\pvec{u},\pvec{u}_1 \sim \pvec{v},\pvec{v}_1$.
\end{proposition}

\begin{proof}
  We first show how to extend an instance of $\CCA$ into an instance of $\cca^a$. Let $(u_i)_{i \in I} \sim (v_i)_{i \in I}$ be an instance of $\cca^a$. Let $I' \subseteq I$, we want to extend $(u_i)_{i \in I'} \sim (v_i)_{i \in I'}$ into an instance of $\cca^a$. Let $\phi \equiv (u_i)_{i \in I}$, $\psi \equiv (v_i)_{i \in I}$, since $(u_i)_{i \in I} \sim (v_i)_{i \in I}$ is an instance of $\cca^a$ we have:
  \[
    (\phi,\encvar,\decvar,\randmap,\encmap,\decmap) R_{\cca^a}^\keys (\psi,\encvar,\decvar,\randmap',\encmap',\decmap')
  \]
  For all $x \in \encvar\cup \decvar$, we let $i_x\in I$ be the index corresponding to $x\encmap\decmap\sim x\encmap'\decmap'$. Moreover, for all $x \in \decvar$, we let $t_{i_x}$ be the context used for the decryption in the definition of $R_{\cca^a}^\keys$ (hence we have $x\decmap \equiv \elses(l,\dec(t_{i_x}\encmap\decmap),\sk)$).

  \paragraph{Outline}
  We are going to define $I^{lr},I^l,I^r \subseteq I$ and $(\tilde u_i)_{i \in J}$, $(\tilde v_i)_{i \in J}$ (where $J = I^{lr} \cup I^l \cup I^r$) such that:
  \begin{itemize}
  \item $I^{lr},I^l,I^r$ are pair-wise disjoints and $I' \subseteq I^{lr}$.
  \item  $(\tilde u_i)_{i \in J} \sim (\tilde v_i)_{i \in J}$ is an instance of $\cca^a$ of polynomial size with respect to $\sum_{i \in I'} |u_i| + |v_i|$.
  \end{itemize}
  Intuitively, $I^{lr}$ is the subset of indices of $I \backslash I'$ of the terms that are subterm of $(u_i)_{i \in I'} \sim (v_i)_{i \in I'}$ \emph{on the left and on the right}, i.e. for all $i \in I^{lr}$, $u_i \in \st((u_i)_{i \in I'})$ \emph{and} $v_i \in \st((v_i)_{i \in I'})$. The terms whose index is in $I^{lr}$ are easy to handle, as they are immediately bounded by the terms whose indices is in $I'$.

  Then, $I^l$ is the subset of indices of $I \backslash I'$ of the terms that are subterms of $(u_i)_{i \in I'} \sim (v_i)_{i \in I'}$ \emph{on the left only} (i.e. for every $i \in I^{l}$, we only know that $u_i \in \st((u_i)_{i \in I'})$). Terms with indices in $I^l$ are easy to bound on the left, but not on the right. To bound the right terms, we introduce dummy messages (by replace encryptions by encryption of $g()$, where $g$ is an adversarial function symbol in $\mathcal{G}$). Similarly $I^r$ is the subset of indices of $I \backslash I'$ of the terms that are subterms of $(u_i)_{i \in I'} \sim (v_i)_{i \in I'}$ \emph{on the right only}.

  First, we define $I^{lr},I^l,I^r$, and then we define the corresponding $\cca^a$ instance $(\tilde u_i)_{i \in J} \sim (\tilde v_i)_{i \in J}$.

  \paragraph{Inductive Definition of the Left and Right Appearance Sets}
  We define by induction on $i \in I'$ the sets $I^{l}_i,I^{r}_i \subseteq I$. Intuitively, $I^l_i$ is the set of indices of $I$ needed so that $u_i$ is well-defined  (same for $I^r_i$ and $v_i$). Let $i \in I'$, we do a case disjunction on the rule applied to $u_i,v_i$ in $R_{\cca^a}^\keys$:
  \begin{itemize}
  \item \textbf{No Call to the Oracles:} In that case we take $I^l_i = I^r_i = \{i\}$.
  \item \textbf{Encryption Case:} let $t,t' \in \mathcal{T}(\nzsig, \Nonce, \mathcal{\decvar})$ such that $u_i \equiv \enc{t\decmap}{\_}{\_}$ and $v_i \equiv \enc{t'\decmap'}{\_}{\_}$. To have $u_i$ well-defined, we need all the decryptions in $u_i$ to be well-defined (same for $v_i$). Hence let:
    \begin{mathpar}
      I^l_i = \{i\} \cup \bigcup_{x \in \decvar \cap \st(t)} I^l_{i_x}

      I^r_i = \{i\} \cup \bigcup_{x \in \decvar \cap \st(t')} I^r_{i_x}
    \end{mathpar}
  \item \textbf{Decryption Case:} recall that $u_i \equiv \elses(l,\dec(u,\sk))$ where $u \equiv t_{i}\encmap\decmap$. Therefore we need all encryption in $\encvar\cap\st(t_i)$ and decryption in $\decvar\cap\st(t_i)$ to be defined, on the left and on the right. Hence we let:
    \begin{mathpar}
      I^l_i = \{i\} \cup \bigcup_{x \in (\decvar\cup\encvar) \cap \st(t_i)} I^l_{i_x}

      I^r_i = \{i\} \cup \bigcup_{x \in (\decvar\cup\encvar) \cap \st(t_i)} I^r_{i_x}
    \end{mathpar}
  \end{itemize}
  We let:
  \begin{mathpar}
    I^{lr} = \bigcup_{i \in I'} I^l_i \cap  \bigcup_{i \in I'} I^r_i

    I^{l} = \bigcup_{i \in I'} I^l_i \cap  \overline{\bigcup_{i \in I'} I^r_i}

    I^{r} = \overline{\bigcup_{i \in I'} I^l_i} \cap  \bigcup_{i \in I'} I^r_i
  \end{mathpar}
  These three sets are disjoint and form a partition of $\bigcup_{i \in I'} I^l_i \cup I^r_i$. Remark that for every $i \in I^l_j$, $u_i$ is a subterm of $u_j$. Hence, for every $i \in I^{lr}\cup I^l$, there exists $j \in I'$ such that $u_i$ is a subterm of $u_j$.

  \paragraph{Building the New Instance}
  We define (by induction on $i$) the terms $(\tilde u_i)_{i \in J}$, by letting $\tilde u_i$ be:
  \begin{itemize}
  \item $u_i$ when  $i \in I^{lr} \cup I^l$.
  \item $\enc{g()}{\pk}{\nonce}$ when $i \in I^r$ and $u_i$ is an encryption, with $u_i \equiv \enc{\_}{\pk}{\nonce}$.
  \item $\elses(\tilde l,\dec(\tilde u,\sk))$ when $i \in I^r$ and $u_i$ is a decryption, where $u_i \equiv \elses(l,\dec(u,\sk))$, $u \equiv t_{i}\encmap\decmap$,  $l$ is the sequence of guards $l \equiv (\eq{u}{y_1},\dots,\eq{u}{y_m})$ where $(y_1,\dots,y_m) = \textsf{sort}(\caly_u\encmap)$. Then we take:
    \begin{itemize}
    \item $\tilde u \equiv t_{i}\rtilde \encmap \rtilde \decmap$, where $\rtilde \encmap = \{x \mapsto \tilde u_{i_x} \mid x \in \encvar\}$ and $\rtilde \decmap = \{x \mapsto \tilde u_{i_x} \mid x \in \decvar\}$.
    \item $\tilde l \equiv (\eq{\tilde u}{\tilde y_1},\dots,\eq{\tilde u}{\tilde y_m})$ where $(\tilde y_1,\dots,\tilde y_m) = \textsf{sort}(\caly_u\rtilde \encmap)$.
    \end{itemize}
  \end{itemize}
  Similarly, we define $\tilde v_i$ for every $i \in J$.

  \paragraph{Conclusion}
  Let $J = I^{lr}\cup I^l\cup I^r$. To conclude, we check that  $(\tilde u_i)_{i \in J} \sim (\tilde v_i)_{i \in J}$:
  \begin{itemize}
  \item is a $\cca^a$ instance. This is done by induction on $i \in J$.
  \item is of polynomial size w.r.t. $(u_i)_{i \in I'} \sim (v_i)_{i \in I'}$.
  \end{itemize}

  We omit the details of the proof of the first point.

  For the second point, we first show by induction on $i$ that $|I^l_i| \le |u_i|$ and $|I^r_i| \le |v_i|$. We deduce that:
  \[
    |J| = \big|\bigcup_{i \in I'} I^r_i \cup I^l_i\big|
    \le \sum_{i \in I'} |I^r_i| + |I^l_i|
    \le \sum_{i \in I'} |u_i| + |v_i|
  \]
  Let $i \in I^{lr}\cup I^l$, we know that there exists $j \in I'$ such that $u_i$ is a subterm of $u_j$. Since $\tilde u_i \equiv u_i$, we deduce that $|\tilde u_i| \le |u_j| \le \sum_{j \in I'} |u_j| + |v_j|$.

  Let $i \in I^r$. If $\tilde u_i$ is an encryption then it is of constant size. Assume $\tilde u_i$ is a decryption. Then $\tilde u_i$ is the decryption $v_i$ where any encryption whose index is in $I^{lr}$ has been replaced by its left counterpart, and any encryption whose index is in $I^{r}$ has been replaced by a dummy encryption (the case $I^l$ cannot happen, since $i \in I^r$). Since there are at most $|v_i| - 1$ such encryptions (as $v_i$ contain at least one occurrence of the $\dec$ function symbol), and since any encryption with index in $I^{lr}$ or $I^r$ is upper-bounded by $\sum_{j \in I'} |u_j| + |v_j|$, we get that:
  \[
    |\tilde u_i|
    \le |v_i| + (|v_i| - 1).\sum_{j \in I'} |u_j| + |v_j|
    \le |v_i|.\sum_{j \in I'} |u_j| + |v_j|
    \le \big(\sum_{j \in I'} |u_j| + |v_j|\big)^2
  \]
  We deduce that $(\tilde u_i)_{i \in J} \sim (\tilde v_i)_{i \in J}$ is of polynomial size in $\sum_{j \in I'} |u_j| + |v_j|$.
\end{proof}

\subsection{Length in the \texorpdfstring{$\cca$}{CCA2} Axioms}
\label{subsection:length}
If we want the formula $\enc{t}{\pk}{r} \sim \enc{t'}{\pk'}{r'}$ to be a valid application of the $\cca$ axioms, we need to make sure that $t$ and $t'$ are of the same length. Since the length of terms depend on implementation details (e.g. how is the pair $\pair{\_}{\_}$ implemented), we let the user supply implementation assumptions. We use a predicate symbol $\eql{\_}{\_}$ in the logic, together with some derivation rules $\mathcal{D}_{\textsf{L}}$ (supplied by the user), and we require that they verify the following properties:
\begin{itemize}
\item \textbf{Complexity:} for every $u,v$, we can decide whether $\eql{u}{v}$ is a consequence of $\mathcal{D}_{\textsf{L}}$ in polynomial time in $|u| + |v|$.
\item \textbf{Branch Invariance:} for all term $b,u,v,t$, if $\eql{\ite{b}{u}{v}}{t}$ is derivable using $\mathcal{D}_{\textsf{L}}$ then $\eql{u}{t}$ and $\eql{v}{t}$ are derivable using $\mathcal{D}_{\textsf{L}}$.
\end{itemize}
We add to all $\cca$ instances the side condition $\eql{m_l}{m_r}$ for every encryption oracle call on $(m_l,m_r)$. Then, we know that our $\cca$ instances are valid in any computational model $\cmodel$ where the encryption is interpreted as a $\textsc{ind-cca}_2$ encryption scheme, and where the following property holds: for every ground terms $u,v$, if $\eql{u}{v}$ is derivable using $\cald_\sfL$, then:
\begin{equation*}
  \label{eq:length-1}
  \sem{\lensymb(u)}_{\cmodel} =
  \sem{\lensymb(v)}_{\cmodel}
\end{equation*}

\paragraph{Example: Block Cipher}
We give here an example of derivation rules $\mathcal{D}_{\textsf{L}}$ that axiomatize the fact that the encryption function is built upon a block cipher, taking blocks of length $\lblock$ and returning blocks of length $\leblock$. The length constant $l_{\enc{}{}{}}$ is used to represent the constant length used, e.g., for the IV and the HMAC.

We let $\mathcal{L}$ be a set of length constants, and we define a length expression to be an expression of the form $\sum_{l \in L} k_l.l$, where $L$ is a finite subset of $\mathcal{L}$ and $(k_l)_{l \in L}$ are positive integers. We consider length expressions modulo commutativity (i.e. $3.l_1 + 4.l_2 \approx 4.l_2 + 3.l_1$), and we assume that for every length expression $l_e$, there exists a function symbol $\pad{l_e} \in \sig$. Intuitively $\pad{l_e}$ is function padding messages to length $l$: if the message is too long it truncates it, and if the message is too short it pads it. Similarly, we assume that for every $l_e$, we have a function symbol $0_{l_e} \in \sig$ or arity zero which, intuitively, returns $l_e$ zeroes. Also, we assume that $\mathcal{L}$ contains the following length constants: $l_{\pair{}{}}, l_{enc{}{}{}}, \lblock, l_\eta$.

We define the $\Length$ (partial) function on terms in Figure~\ref{fig:large-length-def}. Then, we let $\mathcal{D}_{\textsf{L}}$ be the (recursive) set of unitary axioms:
\[
  \infer[]{\eql{u}{v}}{\Length(u) = \Length(v) \ne \textsf{undefined}}
\]

\begin{figure}[t]
  \begin{mathpar}
    \Length(\nonce) = l_\eta

    \Length(0_{l_e}) = l_e

    \Length(u) = \Length(u') \text{ if } u =_R u' \text{ and } \Length(u),\Length(u') \text{ are not \textsf{undefined}}

    \Length(\pair{u}{v}) =   \Length(u) + \Length(v) + l_{\pair{}{}}

    \forall l_e. \Length(\pad{l_e}(u)) = l_e

    \forall k. \Length(\enc{u}{\pk}{\nonce}) = k.\leblock + l_{\enc{}{}{}} \text{ if } \Length(u) = k.\lblock

    \forall k. \Length(\dec(u,\sk)) = k.\lblock \text{ if } \Length(u) = k.\leblock + l_{\enc{}{}{}}

    \Length(\ite{b}{u}{v}) =
    \begin{cases}
      \Length(u) & \text{if } \Length(u) = \Length(v)\\
      \textsf{undefined} & \text{otherwise}
    \end{cases}
  \end{mathpar}

  \caption{Definition of the $\Length$ partial function.}
  \label{fig:large-length-def}
\end{figure}

\begin{proposition}
  The function $\Length$ is well defined, and the set of axioms $\mathcal{D}_{\textsf{L}}$ satisfies the branch invariance properties.
\end{proposition}

\begin{proof}
  To check that $\Length$ is well defined, one just need to look at the critical pairs in the definition and check that they are joinable. Soundness is easy, as $\llbracket \Length \rrbracket_{\cmodel}$ is just an under-approximation of $\llbracket\lensymb \rrbracket_{\cmodel}$ in every computational model $\cmodel$ where the encryption is interpreted as a block cipher, the padding functions are interpreted as expected etc.

  Finally, branch invariance follows directly from the definition of $  \Length(\ite{b}{u}{v})$.
\end{proof}

\begin{remark}
  We can allow the user to add any set of length equations, as long as the branch invariance property holds and the $\Length$ function is well-defined. E.g one may wish to add equations like $\Length(A) = \Length(B)  = \Length(C) = l_{\textsf{agent}}$.
\end{remark}

%%% Local Variables:
%%% mode: latex
%%% TeX-master: "ms"
%%% End:

%%% Local Variables:
%%% mode: latex
%%% TeX-master: "ms"
%%% End:

\newpage

\section{Rule Ordering and Freeze Strategy}
\label{app-section:ordering}

\FloatBarrier

In this section, we give the proofs of the $\restr$ elimination lemma (Lemma~\ref{lem:body-restr-elim}). We then show the rule commutations used to obtain a complete ordered strategy (Lemma~\ref{lem:rule-commute-body}, Lemma~\ref{lem:boxed-commute-body}). Finally we show the completeness of the freeze strategy (Lemma~\ref{lem:body-freeze}).

\subsection{Tracking Relations Between Branches}

\begin{figure}[t]
      \begin{itemize}
      \item $(Sym):$ $\sim$ is symmetric. 

      \item For any permutation $\pi$ of $1,\ldots, n$:
        \(
        \begin{array}[c]{c}
          \quad \infer[Perm]{x_1\ldots,x_n \sim y_1,\ldots, y_n}{x_{\pi(1)},\ldots,x_{\pi(n)} \sim y_{\pi(1)},\ldots,y_{\pi(n)}}
        \end{array}
        \)

      \item 
        \(
        \begin{array}[c]{c}
          \infer[\dup]{\vec u,t,t \sim \vec v,t',t'}{\vec u,t \sim \vec v,t'}
        \end{array}
        \)

      \item If $s =_R t$ and $\{\splitbox{a}{c}{b} \in \st(\vec u,t)\} \subseteq \{\splitbox{a}{c}{b} \in \st(\vec u,C[s])\}$ then:
        \(
        \begin{array}[c]{c}
          \quad \infer[\rs]{\vec u, C[s] \sim \vec v}{\vec u, C[t] \sim \vec v}
        \end{array}
        \)

      \item For all $f \in \sig$, 
        \( 
        \begin{array}[c]{c}
          \infer[\fa]{f(\vec{x}),\vec{y} \sim f(\vec{x'}), \vec{y'}}{\vec{x},\vec{y}\sim \vec{x'}, \vec{y'}}
        \end{array}
        \)

      \item For every $b,b' \in \mathcal{T}(\ssig,\Nonce)$:
        \[
          \infer[\csmb]{\vec w, \left(\ite{\splitbox{b_1}{b_2}{b}}{u_i}{v_i}\right)_i \sim \vec w',\left(\ite{\splitbox{b'_1}{b'_2}{b'}}{u_i'}{v_i'}\right)_i}
          {
            \vec w, b_1, (u_i)_i \sim \vec w', b_1', (u_i')_i \quad & \quad
            \vec w, b_2, (v_i)_i \sim \vec w', b_2', (v_i')_i
          }
        \]

      \item $\unbox$ unfreezes all conditionals.
        
      \item For every $b,b' \in \mathcal{T}(\ssig \cup \esig,\Nonce)$:
        \[
          \infer[\twobox]{ \vec u, C[b] \sim \vec u', C'[b']}
          {
            \vec u, C\left[\splitbox{b}{b}{\tberase(b)\downarrow_R}\right] 
            \sim
            \vec u', C'\left[\splitbox{b'}{b'}{\tberase(b')\downarrow_R}\right]
          }
        \]
      \end{itemize}
  \caption{\label{figure:axioms-strat} Summary of the strategy axioms.}
\end{figure}

We introduce the following erasure function, defined on if-free ground terms inductively as follows:
\[
  \tberase(t) \equiv 
  \begin{cases}
    f(\tberase(t_1),\dots,\tberase(t_n)) & \text{ if } t \equiv f(t_1,\dots,t_n) \wedge f \in \ssig\\
    \tberase(b)& \text{ if } t \equiv \splitbox{b_1}{b_2}{b}\\
    \nonce & \text{ if } t \equiv \nonce \wedge \nonce \in \Nonce
  \end{cases}
\]
This function is used to define the full (not simplified) versions of $\unbox$ and $\twobox$, which are given in Fig.~\ref{figure:axioms-strat}, together with a summary of all the axioms introduced for the complete strategy.
\begin{remark}
  We modify the definition of $\condst(t)$ as follows: for all $t$, $\condst(t) = \condst(\tberase(t))$. 
\end{remark}

\subsection{Proof Ordering}

\begin{figure}[t]
  \[
    \begin{array}{|ccc|}
      \hline
      \vphantom{\Big|}\dup\cdot R &\;\Rightarrow\;& R\cdot\dup
      \\[0.2em]
      \dup\cdot\fa&\;\Rightarrow\;& \fa^*\cdot\dup
      \\[0.2em]
      % \dup\cdot\ift&\;\Rightarrow\;& \ift^*\cdot\dup
      % \\[0.2em]
      \dup\cdot \cs&\;\Rightarrow\;& \cs\cdot\dup
      \\\hline\hline
      % \fa\cdot\ift&\;\Rightarrow\;& \ift\cdot\fa
      % \\[0.2em]
      \vphantom{\Big|}\fa\cdot R &\;\Rightarrow\;& R\cdot\fa
      \\[0.2em]
      \fa\cdot \cs&\;\Rightarrow\;& R\cdot \cs\cdot\fa
      \\\hline\hline
      \vphantom{\Big|}
      \fas\cdot\fa(b,b')&\;\Rightarrow\;& R\cdot\fa(b,b')\cdot\fas^*\cdot\dup
      \\\hline\hline
      % \cs\cdot\ift &\;\Rightarrow\;& \ift\cdot \cs
      % \\\[0.2em]
      % \csm\cdot R &\;\Rightarrow\;& R\cdot\csm\cdot\rfree
      % \\[0.2em]
      \vphantom{\Big|}\csmb\cdot\rs &\;\Rightarrow\;& \rs\cdot\csmb
      \\[0.2em]
      \csmb\cdot\twobox &\;\Rightarrow\;& \rs\cdot\twobox\cdot\csmb
      \\\hline
    \end{array}
  \]
  \textbf{Explanation:} Each entry $w \Rightarrow w'$ means that a derivation in $w$ can be rewritten into a derivation in $w'$.
  \caption{\label{fig:recap} Summary of all the rule commutations}
\end{figure}

We now show that all the rule commutations given in Fig.~\ref{fig:recap} are correct. Observe that this subsumes Lemma~\ref{lem:rule-commute-body} and Lemma~\ref{lem:boxed-commute-body}.
\begin{lemma} 
  All the rule commutations in Fig.~\ref{fig:recap} are correct.
\end{lemma}

\begin{proof}
  We split the proof depending on the left-most rule we are commuting.
\paragraph{Delay $\dup$} 
\label{subsub:delay-dup}
\begin{itemize}
\item If the $R$ rules involves a term which is not duplicated then this is trivial. Assume the $R$ rewriting involves a duplicated term, and that $t =_Rs $ and $t' =_Rs'$:
  \begin{gather*}
    \begin{array}{lcr}
      \begin{array}[c]{c}
        \infer[\dup]{\vec u,\vec v,t, \vec v,t \sim \vec u',\vec v',t', \vec v',t'}
        {
          \infer[R]{\vec u, \vec v,t \sim \vec u', \vec v',t'}
          {
            {\vec u, \vec v,s \sim \vec u', \vec v',s'}
          }
        }
      \end{array}
      &\Rightarrow&
      \begin{array}[c]{c}
        \infer[R]{\vec u,\vec v,t, \vec v,t \sim \vec u',\vec v',t', \vec v',t'}
        {
          \infer[\dup]{\vec u,\vec v,s, \vec v,s \sim \vec u',\vec v',s', \vec v',s'}
          {
            {\vec u, \vec v,s \sim \vec u', \vec v',s'}
          }
        }
      \end{array}
    \end{array}
  \end{gather*}
  
\item Similarly if the $\fa$ rules does not involve a duplicated term then this is trivial. Otherwise:
  \begin{gather*}
    \begin{array}{lcr}
      \begin{array}[c]{c}
        \infer[\dup]{\vec u,\vec v,f(\vec w), \vec v,f(\vec w) \sim \vec u',\vec v',f(\vec w'), \vec v',f(\vec w')}
        {
          \infer[\fa]{\vec u,\vec v,f(\vec w) \sim \vec u',\vec v',f(\vec w')}
          {
            {\vec u,\vec v,\vec w \sim \vec u',\vec v',\vec w'}
          }
        }
      \end{array}
      &\Rightarrow&
      \begin{array}[c]{c}
        \infer[\fa]{\vec u,\vec v,f(\vec w), \vec v,f(\vec w) \sim \vec u',\vec v',f(\vec w'), \vec v',f(\vec w')}
        {
          \infer[\fa]{\vec u,\vec v,f(\vec w), \vec v,\vec w \sim \vec u',\vec v',f(\vec w'), \vec' v,\vec w'}
          {
            \infer[\dup]{\vec u,\vec v,\vec w, \vec v,\vec w \sim \vec u',\vec v',\vec w', \vec v',\vec w'}
            {
              {\vec u \sim \vec u'}
            }
          }
        }
      \end{array}
    \end{array}
  \end{gather*}
% \item Commutation of $\dup$ with $IFT$ is easy.
\item Commutation of $\dup$ with $\csm$ is easy. 
\end{itemize}

% \begin{remark}
%   Commutation of $\dup$ with the simple $CS_s$ does not work:
%   \[
%     \infer[\dup]{\ite{b}{u}{v}, \ite{b}{u}{v} \sim \ite{b'}{u'}{v'}, \ite{b'}{u'}{v'}}
%     {
%       \infer[CS_s]{\ite{b}{u}{v} \sim \ite{b'}{u'}{v'}}
%       {b,u \sim b',u' & b,v \sim b',v'}
%     }
%   \]
%   One may try to apply $CS_s$ twice, but this does not yield the wanted proof (the two middle premises do not appear above):
%   \[
%     \infer[CS_s]{\ite{b}{u}{v}, \ite{b}{u}{v} \sim \ite{b'}{u'}{v'}, \ite{b'}{u'}{v'}}
%     {
%       \infer[CS_s]{b,u,\ite{b}{u}{v} \sim b',u',\ite{b'}{u'}{v'}}
%       {
%         \infer[\dup]{b,u,b,u \sim b',u',b',u'}
%         {b,u \sim b',u'}
%         & \infer[\dup]{b,u,b,v \sim b',u',b',v'}
%         {b,u,v \sim b',u',v'}
%       }
%       &
%       \infer[CS_s]{b,v,\ite{b}{u}{v} \sim b',v',\ite{b'}{u'}{v'}}
%       {
%         \infer[\dup]{b,v,b,u \sim b',v',b',u'}
%         {b,v,u \sim b',v',u'}
%         & 
%         \infer[\dup]{b,v,b,v \sim b',v',b',v'}
%         {b,v \sim b',v'}
%       }
%     }
%   \]
% \end{remark}

% \begin{remark}
%   It is possible to replace some applications of the $\csm$ rule by the $CS_s$ rule, simply by rewriting the goal:
%   \begin{equation} 
%     \label{eq:cs-m-simple1}
%     t_1,\dots,t_n \sim t'_1,\dots,t'_n
%   \end{equation}
%   into:
%   \begin{equation} 
%     \label{eq:cs-m-simple2}
%     \pair{t_1}{\pair{\dots}{t_n}\dots} \sim \pair{t'_1}{\pair{\dots}{t'_n}\dots}
%   \end{equation}
% \end{remark}

\paragraph{Delay $\fa$}

\begin{itemize}
% \item Assume $w,t \in \mathcal{T}(\ssig,\Nonce)$ and $w \emb (t\downarrow_R)$:
%   \[
%     \infer[\fa]{\vec w, f(\vec u, C[\ite{\eq{t}{w}}{C'[t]}{v}]) \sim \vec w', f(\vec u')}
%     {
%       \infer[IFT]{\vec w, \vec u, C[\ite{\eq{t}{w}}{C'[t]}{v}] \sim \vec w', \vec u'}
%       {\vec w, \vec u, C[\ite{\eq{t}{w}}{C'[w]}{v}] \sim \vec w', \vec u'}
%     }
%   \]
%   Can be rewritten into:
%   \[
%     \infer[IFT]{\vec w, f(\vec u, C[\ite{\eq{t}{w}}{C'[t]}{v}]) \sim \vec w', f(\vec u')}
%     {
%       \infer[\fa]{\vec w, f(\vec u, C[\ite{\eq{t}{w}}{C'[w]}{v}]) \sim \vec w', f(\vec u')}
%       {\vec w, \vec u, C[\ite{\eq{t}{w}}{C'[w]}{v}] \sim \vec w', \vec u'}
%     }
%   \]
\item For every $b,b' \in \mathcal{T}(\ssig,\Nonce)$:
  \[
    \infer[\fa]
    {
      \begin{array}{rc}
        &\vec w_1, (\ite{b}{u_i}{v_i})_{i \in I}, f(\vec w_2,(\ite{b}{u_i}{v_i})_{i\in J}) \\
        \sim& \vec w'_1, (\ite{b'}{u'_i}{v'_i})_{i \in I}, f(\vec w'_2,(\ite{b'}{u'_i}{v'_i})_{i\in J})
      \end{array}
    }
    {
      \infer[\csm]{\vec w_1,\vec w_2, (\ite{b}{u_i}{v_i})_{i \in I \cup J} \sim \vec w'_1,\vec w'_2, (\ite{b'}{u'_i}{v'_i})_{i \in I \cup J}}
      {
        \vec w_1,\vec w_2, b, (u_i)_{i \in I \cup J} \sim \vec w'_1,\vec w'_2, b', (u'_i)_{i \in I \cup J} \quad & \quad
        \vec w_1,\vec w_2, b, (v_i)_{i \in I \cup J} \sim \vec w'_1,\vec w'_2, b', (v'_i)_{i \in I \cup J}
      }
    }
  \]
  Can be rewritten into:
  \[
    \infer[R]
    {
      \begin{array}{rc}
        &\vec w_1, (\ite{b}{u_i}{v_i})_{i \in I}, f(\vec w_2,(\ite{b}{u_i}{v_i})_{i\in J}) \\
        \sim& \vec w'_1, (\ite{b'}{u'_i}{v'_i})_{i \in I}, f(\vec w'_2,(\ite{b'}{u'_i}{v'_i})_{i\in J})
      \end{array}
    }
    {
      \infer[\csm]
      {
        \begin{array}{rc}
          &\vec w_1, (\ite{b}{u_i}{v_i})_{i \in I}, \ite{b}{f(\vec w_2,(u_i)_{i \in J})}{f(\vec w_2,(v_i)_{i \in J})} \\
          \sim & \vec w'_1, (\ite{b'}{u'_i}{v'_i})_{i \in I}, \ite{b'}{f(\vec w'_2,(u'_i)_{i \in J})}{f(\vec w'_2,(v'_i)_{i \in J})}
        \end{array}
      }
      {
        \infer[\fa]
        {
          \begin{array}{rc}
            &\vec w_1, b, (u_i)_{i \in I}, f(\vec w_2,(u_i)_{i \in J}) \\
            \sim & \vec w'_1, b', (u'_i)_{i \in I}, f(\vec w'_2,(u'_i)_{i \in J})
          \end{array}
        }
        {\vec w_1,\vec w_2, b, (u_i)_{i \in I \cup J} \sim \vec w'_1,\vec w'_2, b', (u'_i)_{i \in I \cup J}}
        &
        \infer[\fa]
        {
          \begin{array}{rc}
            &\vec w_1, b, (v_i)_{i \in I}, f(\vec w_2,(v_i)_{i \in J}) \\
            \sim & \vec w'_1, b', (v'_i)_{i \in I}, f(\vec w'_2,(v'_i)_{i \in J})
          \end{array}
        }
        {\vec w_1,\vec w_2, b, (v_i)_{i \in I \cup J} \sim \vec w'_1,\vec w'_2, b', (v'_i)_{i \in I \cup J}}
      }
    }
  \]
\item Assume that $\vec u,\vec v, \vec u', \vec v' =_R \vec u_1,\vec v_1, \vec u_1', \vec v_1'$:
  \begin{gather*}
    \begin{array}{lcr}
      \begin{array}[c]{c} 
        \infer[\fa]{\vec u,f(\vec v) \sim \vec u', f(\vec v')}
        {
          \infer[R]{ \vec u,\vec v \sim \vec u', \vec v'}
          { \vec u_1,\vec v_1 \sim \vec u_1', \vec v_1'}
        }
      \end{array}
      &\Rightarrow& 
      \begin{array}[c]{c} 
        \infer[R]{\vec u,f(\vec v) \sim \vec u', f(\vec v')}
        {
          \infer[\fa]{\vec u_1,f(\vec v_1) \sim \vec u_1', f(\vec v_1')}
          { \vec u_1,\vec v_1 \sim \vec u_1', \vec v_1'}
        }
      \end{array}
    \end{array}
  \end{gather*}
\item For all $f,b,b'$, one can always apply $\fa_f$  after $\fa(b,b')$: 
  \[
    \infer[\fa_f]{\vec u, f(\vec v,\ite{b}{s}{t}) \sim \vec u', f(\vec v',\ite{b'}{s'}{t'})}
    {
      \infer[\fa(b,b')]{\vec u, \vec v,\ite{b}{s}{t} \sim \vec u', \vec v',\ite{b'}{s'}{t'}}
      {\vec u, \vec v,b,s,t \sim \vec u', \vec v',b',s',t'}
    }
  \]
  Then we can rewrite this proof as follows:
  \[
    \infer[R]{\vec u, f(\vec v,\ite{b}{s}{t}) \sim \vec u', f(\vec v',\ite{b'}{s'}{t'})}
    {
      \infer[\fa(b,b')]{\vec u, \ite{b}{f(\vec v,s)}{f(\vec v,t)}\sim \vec u', \ite{b'}{f(\vec v',s')}{f(\vec v',t')}}
      {
        \infer[\fa_f]{\vec u, b,f(\vec v,s),f(\vec v,t) \sim \vec u', b',f(\vec v',s'),f(\vec v',t')}
        {
          \infer[\fa_f]{\vec u, b,\vec v,s,f(\vec v,t) \sim \vec u', b',\vec v',s',f(\vec v',t')}
          {
            \infer[\dup]{\vec u, b,\vec v,s,\vec v,t \sim \vec u', b',\vec v',s',\vec v',t'}
            {
              {\vec u, b,s,\vec v,t \sim \vec u', b',s',\vec v',t'}
            }
          }
        }
      }
    }
  \]

\item $\fa(b,b') - \fa(a,a') $ commutation: assume that $u =_R \ite{a}{s}{t}$ and that $u' =_R \ite{a'}{s'}{t'}$.
  \[
    \infer[\fa(b,b')]{\vec w,\ite{b}{u}{v} \sim \vec w',\ite{b'}{u'}{v'}}
    {
      \infer[R]{\vec w,b,u,v \sim \vec w',b',u',v'}
      {
        \infer[\fa(a,a')]{\vec w,b,\ite{a}{s}{t},v \sim \vec w',b',\ite{a'}{s'}{t'},v'}
        {\vec w,b,a,s,t,v \sim \vec w',b',a',s',t',v'}
      }
    }
  \]
  Then we can rewrite this proof as follows:
  \[
    \infer[R]{\vec w,\ite{b}{u}{v} \sim \vec w',\ite{b'}{u'}{v'}}
    {
      \infer[\fa(a,a')]
      {
        \vec w,
        \begin{alignedat}[t]{2}
          \ite{a}{&\ite{b}{s}{v}\\}{&\ite{b}{t}{v}} 
        \end{alignedat}
        \;\sim \;
        \vec w',
        \begin{alignedat}[t]{2}
          \ite{a'}{&\ite{b'}{s'}{v'}\\}{&\ite{b'}{t'}{v'}}
        \end{alignedat}
      }
      {
        \infer[\fa(b,b')]
        {
          \vec w,a,\ite{b}{s}{v},\ite{b}{t}{v} 
          \sim
          \vec w',a',\ite{b'}{s'}{v'},\ite{b'}{t'}{v'}
        }
        {
          \infer[\fa(b,b')]
          {
            \vec w,a,b,s,v,\ite{b}{t}{v} 
            \sim
            \vec w',a',b',s',v',\ite{b'}{t'}{v'}
          }
          {
            \infer[\dup]
            {\vec w,a,b,s,v,b,t,v 
              \sim
              \vec w',a',b',s',v',b',t',v'
            }
            {
              \vec w,a,b,s,t,v 
              \sim
              \vec w',a',b',s',t',v'
            }
          }
        }
      }
    }
  \]
\end{itemize}

\paragraph{Delay $\csm$}
\begin{itemize}
% \item Assume $s,t \in \mathcal{T}(\ssig,\Nonce)$ and $s \emb (t\downarrow_R)$:
%   \[
%     \infer[\csm]{\vec w, \ite{b}{C[\ite{\eq{t}{s}}{C'[t]}{v}]}{v_0},(\ite{b}{u_i}{v_i})_i \sim \vec w', (\ite{b'}{u'_i}{v'_i})_i}
%     {
%       \infer[IFT]{\vec w, b, C[\ite{\eq{t}{s}}{C'[t]}{v}], (u_i)_i \sim \vec w', b', (u'_i)_i}
%       {\vec w, b, C[\ite{\eq{t}{s}}{C'[s]}{v}], (u_i)_i \sim \vec w', b', (u'_i)_i}
%       &
%       \vec w, b, v_0,(v_i)_i \sim \vec w', b', (v'_i)_i
%     }
%   \]
%   Can be rewritten into:
%   \[
%     \infer[IFT]{\vec w, \ite{b}{C[\ite{\eq{t}{s}}{C'[t]}{v}]}{v_0},(\ite{b}{u_i}{v_i})_i \sim \vec w', (\ite{b'}{u'_i}{v'_i})_i}
%     {
%       \infer[\csm]{\vec w, \ite{b}{C[\ite{\eq{t}{s}}{C'[s]}{v}]}{v_0},(\ite{b}{u_i}{v_i})_i \sim \vec w', (\ite{b'}{u'_i}{v'_i})_i}
%       {
%         \vec w, b, C[\ite{\eq{t}{s}}{C'[s]}{v}], (u_i)_i \sim \vec w', b', (u'_i)_i
%         \quad & \quad
%         \vec w, b, v_0,(v_i)_i \sim \vec w', b', (v'_i)_i      }
%     }
%   \] 
\item For all $b,b' \in \mathcal{T}(\ssig,\Nonce)$, the rule application:
  \[
    \infer[\csm]{(w_n)_n, (\ite{b}{u_i}{v_i})_i \sim (w'_n)_n, (\ite{b'}{u'_i}{v'_i})_i}
    {
      \infer[R]{(w_n)_n, b, (u_i)_i \sim (w'_n)_n, b',(u_i')_i}
      {(w_n)_n^0, b^0, (u^0_i)_i \sim (w'^0_n)_n, b'^0,(u'^0_i)_i}
      \quad & \quad
      (w_n)_n, b, (v_i)_i \sim (w'_n)_n, b',(v_i')_i
    }
  \]
  can be rewritten into:
  \[
    \infer[R]{(w_n)_n, (\ite{b}{u_i}{v_i})_i \sim (w_n')_n, (\ite{b'}{u'_i}{v'_i})_i}
    {
      \infer[\csm]{(\ite{b}{w^0_n}{w_n})_n, (\ite{b}{u^0_i}{v_i})_i \sim (\ite{b'}{w'^0_n}{w'_n})_n, (\ite{b'}{u'^0_i}{v'_i})_i}
      {
        \infer[\rfree]{(w_n)_n^0, b, (u^0_i)_i \sim (w_n)_n'^0, b',(u'^0_i)_i}
        {(w_n)_n^0, b^0, (u^0_i)_i \sim (w_n)_n'^0, b'^0,(u'^0_i)_i}
        \quad & \quad
        (w_n)_n, b, (v_i)_i \sim (w_n)_n', b',(v_i')_i
      }
    }
  \]
\end{itemize}

% \paragraph{$\rfree$ Commutations}
% We study here the commutation of $\rfree$ with $\csm$. 
% \begin{itemize}
% \item If the $R$ rewriting occurs in a part of the formula which is not involved in the case study then the commutation is trivial: for all $a,a',b,b' \in \mathcal{T}(\ssig,\Nonce)$, the rule application:
%   \[
%     \infer[\rfree]{a,\vec w, (\ite{b}{u_i}{v_i})_i \sim a',\vec w', (\ite{b'}{u'_i}{v'_i})_i}
%     {
%       \infer[\csm]{a^0,\vec w, (\ite{b}{u_i}{v_i})_i \sim a'^0,\vec w', (\ite{b'}{u'_i}{v'_i})_i}
%       {
%         a^0,\vec w, b, (u_i)_i \sim a'^0,\vec w', b', (u'_i)_i
%         \quad & \quad
%         a^0,\vec w, b, (v_i)_i \sim a'^0,\vec w', b', (v'_i)_i
%       }
%     }
%   \]
%   can be rewritten into:
%   \[
%     \infer[\csm]{a,\vec w, (\ite{b}{u_i}{v_i})_i \sim a',\vec w', (\ite{b'}{u'_i}{v'_i})_i}
%     {
%       \infer[\rfree]{a,\vec w, b, (u_i)_i \sim a',\vec w', b', (u'_i)_i}
%       {a^0,\vec w, b, (u_i)_i \sim a'^0,\vec w', b', (u'_i)_i}
%       \quad & \quad
%       \infer[\rfree]{a,\vec w, b, (v_i)_i \sim a',\vec w', b', (v'_i)_i} 
%       {a^0,\vec w, b, (v_i)_i \sim a'^0,\vec w', b', (v'_i)_i} 
%     }
%   \]
% \end{itemize}

\paragraph{Delay $\protect\csmb$}
\begin{itemize}
\item The following proof:
  \[
    \infer[\csmb]{(w_j)_j, (\ite{\splitbox{a_1}{a_2}{b}}{u_i}{v_i})_i \sim (w_j')_j, (\ite{\splitbox{a_1'}{a_2'}{b'}}{u'_i}{ v'_i})_i}
    {
      \infer[\rs]{(w_j)_j, a_1, (u_i)_i \sim (w_j')_j, a_1', (u'_i)_i}
      {
        (w_j^1)_j, b_1, (u^1_i)_i \sim (w_j'^1)_j, b'_1, (u'^1_i)_i
      }
      & 
      \infer[\rs]{(w_j)_j, a_2, (v_i)_i \sim (w_j')_j, a_2', (v'_i)_i}
      {
        (w_j^2)_j, b_2, (v^1_i)_i \sim (w_j'^2)_j, b'_2, (v'^1_i)_i
      }
    }
  \]
  can be rewritten into:
  \[
    \infer[\rs]{(w_j)_j, (\ite{\splitbox{a_1}{a_2}{b}}{u_i}{v_i})_i \sim (w_j')_j, (\ite{\splitbox{a_1'}{a_2'}{b'}}{u'_i}{v'_i})_i}
    {
      \infer[\csmb]
      {
        \begin{array}{ll}
          &(\ite{\splitbox{b_1}{b_2}{b}}{w_j^1}{w_j^2})_j, (\ite{\splitbox{b_1}{b_2}{b}}{u^1_i}{v^1_i})_i \\
          \sim & (\ite{\splitbox{b_1'}{b_2'}{b'}}{w_j'^1}{w_j'^2})_j, (\ite{\splitbox{b_1'}{b_2'}{b'}}{u'^1_i}{v'^1_i})_i
        \end{array}
      }
      {
        (w_j^1)_j, b_1, (u^1_i)_i \sim (w_j'^1)_j, b'_1, (u'^1_i)_i
        & 
        (w_j^2)_j, b_2, (v^1_i)_i \sim (w_j'^1)_j, b'_2, (v'^1_i)_i
      }
    }
  \]
\item Similarly we can commute $\csmb$ with $\twobox$. Let $b,b' \in \mathcal{T}(\ssig \cup \esig,\Nonce)$, and let:
  \[
    b_\square \equiv \splitbox{b}{b}{\tberase(b)\downarrow_R} \qquad \wedge \qquad 
    b'_\square \equiv \splitbox{b'}{b'}{\tberase(b')\downarrow_R}
  \]
  Then the following proof: 
  \[
    \infer[\csmb]
    {
      \begin{array}{ll}
        &(w_j[b])_j, \left(\ite{\splitbox{a_1[b]}{a_2[b]}{a}}{u_i[b]}{v_i[b]}\right)_i\\
        \sim& (w_j'[b'])_j, \left(\ite{\splitbox{a_1'[b']}{a_2'[b']}{a'}}{u'_i[b']}{ v'_i[b']}\right)_i
      \end{array}
    }
    {
      \infer[\twobox]
      {
        \begin{array}{ll}
          &(w_j[b])_j, a_1[b], (u_i[b])_i\\
          \sim& (w_j'[b'])_j, a_1'[b'], (u'_i[b'])_i
        \end{array}
      }
      {
        \begin{array}{ll}
          &(w_j[b_\square])_j, a_1[b_\square], (u_i[b_\square])_i\\
          \sim& (w_j'[b_\square'])_j, a_1'[b_\square'], (u'_i[b_\square'])_i
        \end{array}
      }
      & 
      \begin{array}{ll}
        &(w_j[b])_j, a_2[b], (v_i[b])_i\\
        \sim& (w_j'[b'])_j, a_2'[b'], (v'_i[b'])_i
      \end{array}
    }
  \]
  can be rewritten into:
  \[
    \infer[\rs]
    {
      \begin{array}{ll}
        &(w_j[b])_j, \left(\ite{\splitbox{a_1[b]}{a_2[b]}{a}}{u_i[b]}{v_i[b]}\right)_i\\
        \sim& (w_j'[b'])_j, \left(\ite{\splitbox{a_1'[b']}{a_2'[b']}{a'}}{u'_i[b']}{ v'_i[b']}\right)_i
      \end{array}
    }
    {
      \infer[\twobox]
      {
        \begin{array}{ll}
          &\left(\ite{\splitbox{a_1[b]}{a_2[b]}{a}}{w_j[b]}{w_j[b]}\right)_j,
            \left(\ite{\splitbox{a_1[b]}{a_2[b]}{a}}{u_i[b]}{v_i[b]}\right)_i\\
          \sim& \left(\ite{\splitbox{a_1'[b']}{a_2'[b']}{a'}}{w_j'[b']}{w_j'[b']}\right)_j, 
                \left(\ite{\splitbox{a_1'[b']}{a_2'[b']}{a'}}{u'_i[b']}{ v'_i[b']}\right)_i
        \end{array}
      }
      {
        \infer[\csmb]
        {
          \begin{array}{ll}
            &\left(\ite{\splitbox{a_1[b_\square]}{a_2[b]}{a}}{w_j[b_\square]}{w_j[b]}\right)_j,
              \left(\ite{\splitbox{a_1[b_\square]}{a_2[b]}{a}}{u_i[b_\square]}{v_i[b]}\right)_i\\
            \sim& \left(\ite{\splitbox{a_1'[b_\square']}{a_2'[b']}{a'}}{w_j'[b_\square']}{w_j'[b']}\right)_j, 
                  \left(\ite{\splitbox{a_1'[b_\square']}{a_2'[b']}{a'}}{u'_i[b_\square']}{ v'_i[b']}\right)_i
          \end{array}
        }
        {
          \begin{array}{ll}
            &(w_j[b_\square])_j, a_1[b_\square], (u_i[b_\square])_i\\
            \sim& (w_j'[b_\square'])_j, a_1'[b_\square'], (u'_i[b_\square'])_i
          \end{array}
          & 
          \begin{array}{ll}
            &(w_j[b])_j, a_2[b], (v_i[b])_i\\
            \sim& (w_j'[b'])_j, a_2'[b'], (v'_i[b'])_i
          \end{array}
        }
      }
    }
  \]
  The commutation with an application of $\twobox$ in the right branch is exactly the same.\qedhere
\end{itemize}

%%% Local Variables:
%%% mode: latex
%%% TeX-master: "ms"
%%% End:

\end{proof}

% \begin{lemma}
%   \label{lem:ordering}
%   Let $\textsf{U}$ be a set of unitary axioms closed under~$\restr$. Then the following ordered strategy:
%   \[ 
%     \mathfrak{F}((\csm + R + \ift)^* \cdot \fa^* \cdot \dup^* \cdot \textsf{U})
%   \]
%   is complete for $\mathfrak{F}((\ift + \csm + \fa + R +  \dup + \textsf{U})^*)$.
% \end{lemma}

% \begin{proof}
%   We start with:
%   \[
%     (R + \fa + \csm + \ift + \dup)^* \cdot \textsf{U}
%   \]
%   We start by commuting all $\fa$ and $\dup$ applications to the right, which yield a proof of the form:
%   \[
%     (R + \csm + \ift )^* \cdot \fa^* \cdot \dup^* \cdot \textsf{U}
%   \]
%   We then split $\fa$ as follows:
%   \[
%     (R + \csm + \ift )^* \cdot (\{\fa(b,b')\} + \fas)^* \cdot \dup^* \cdot \textsf{U}
%   \]
%   We commute all $\fas$ applications, which yield a proof of the form:
%   \[
%     (R + \csm + \ift )^* \cdot (\{\fa(b,b')\} + R)^* \cdot \fas^*  \cdot \dup^* \cdot \textsf{U}
%   \]
%   We commute the $R$ applications in the right part of the proof to get:
%   \[
%     (R + \csm + \ift)^* \cdot \{\fa(b,b')\}^* \cdot \fas^* \cdot \dup^* \cdot \textsf{U}
%   \]
% \end{proof}

\subsection{Restr Elimination}

We show in the following lemma that any proof using $\restr$ can be rewritten into a (no larger) proof without the $\restr$ rule. In other word, the $\restr$ rule is admissible in our logic. Remark that this $\restr$ elimination result subsumes Lemma~\ref{lem:body-restr-elim}.
\begin{lemma}[$\restr$ Elimination]
  \label{lem:restrelim}
  If $P \vdash \vec u \sim \vec v$ with $P$ in $(\csmb + R + \twobox + \fa + \dup + \cca + \restr)^*$ then there exists $P'$ such that $P' \vdash \vec u \sim \vec v$ and $P'$ contains no $\restr$ applications. Moreover the height of $P'$ is no larger than the height of $P$.
\end{lemma}

\begin{proof}
  We do a proof by induction on the height of the derivation $P$ of $\vec u \sim \vec v$. For the inductive case, assume that we have a derivation $P$ of ${\vec u \sim \vec v}$ where the last rule applied is $\restr$:
  \[
    \infer[\restr]{\vec u \sim \vec v}
    {\vec u, \vec t \sim \vec v, \vec t}
  \]
  We discriminate on the second last rule applied:
  \begin{itemize}
  \item If it is a unitary axiom  we conclude easily using the fact that unitary axioms are closed under $\restr$.
  \item If it is a $\fa$ axiom and $\vec t$ is not involved in this function application then $P$ is of the form:
    \[ 
      \begin{array}[c]{c}
        \infer[\restr]{f(\vec u),\vec u' \sim f(\vec v),\vec v'}
        {
          \infer[\fa]{f(\vec u),\vec u',\vec t \sim f(\vec v),\vec v', \vec t'}
          {
            P_0
          }
        } 
      \end{array}
      \qquad \wedge \qquad
      P_0 \vdash \vec u,\vec u',\vec t \sim \vec v,\vec v', \vec t'
    \]
    By applying the induction hypothesis on the following derivation:
    \[
      \infer[\restr]{\vec u,\vec u'\sim \vec v,\vec v'}{P_0}
    \]
    we have a derivation $P' \vdash \vec u,\vec u' \sim \vec v,\vec v'$ in the wanted fragment. We conclude by applying the $\fa$ rule: $\vcenter{\infer[\fa]{f(\vec u),\vec u' \sim f(\vec v),\vec v'}{P'}}$.

  \item If it is a $\fa$ axiom and $\vec t$ is  involved in this function application then $P$ is of the form:
    \[ 
      \begin{array}[c]{c}
        \infer[\restr]{\vec u \sim \vec v}
        {
          \infer[\fa]{\vec u,\vec u',f(\vec u'') \sim \vec v,\vec v',f(\vec v'')}
          {
            P_0
          }
        } 
      \end{array}
      \qquad \wedge \qquad
      P_0 \vdash \vec u,\vec u',\vec u'' \sim \vec v,\vec v',\vec v''
    \]
    By applying the induction hypothesis on the following derivation:
    \[
      \infer[\restr]{\vec u \sim \vec v}{P_0}
    \]
    We get a derivation $P' \vdash \vec u \sim \vec v$ in the wanted fragment.
  \item The $\csmb$ axiom is handled similarly to $\fa$.
  \item The $\dup, \twobox$ and $R$ axioms are trivial to handle.\qedhere
    \qedhere
  \end{itemize}
\end{proof}

\paragraph{Sub-Proof Extraction Functions $\extractl$ and $\extractr$} It follows that, given a proof $P \vdash \vec u \sim \vec v$ and a position $h$ in the proof $P$ such that:
\[
  P_{|h} = 
  \begin{array}[c]{c}
    \infer[\csmb]{\vec w, \left(\ite{\splitbox{b_1}{b_2}{b}}{u_i}{v_i}\right)_i \sim \vec w',\left(\ite{\splitbox{b'_1}{b'_2}{b'}}{u_i'}{v_i'}\right)_i}
    {
      \vec w, b_1, (u_i)_i \sim \vec w', b_1', (u_i')_i \quad & \quad
      \vec w, b_2, (v_i)_i \sim \vec w', b_2', (v_i')_i
    }
  \end{array}
\]
we can extract from $P$ the left (resp. right) proof of $b_1 \sim  b_1'$ (resp. $b_2 \sim  b_2'$) using the $\restr$ elimination procedure described in the proof of Lemma~\ref{lem:restrelim}. We let $\extractl(h,P)$ be proof of $b_1 \sim  b_1'$ extracted from $P_{|h}$, and $\extractr(h,P)$ be proof of $b_2 \sim  b_2'$ extracted from $P_{|h}$.

\subsection{Completeness of the Freeze Strategy}
We give here a proof of Lemma~\ref{lem:body-freeze}, which we recall below. 
\begin{lemma*}[\ref{lem:body-freeze}]
  Let $\textsf{U}$ be a set of unitary axioms closed under $\restr$. Then the following strategy:
  \[
    \mathfrak{F}((\twobox + \rs)^*  \cdot \csmb^* \cdot \{\obfa(b,b')\}^* \cdot \unbox \cdot \fas^* \cdot \dup^* \cdot \textsf{U})
  \]
  is complete for $\mathfrak{F}(\csm + \fa + R + \dup + \textsf{U})$.
\end{lemma*}

  Before starting the proof, we need to define the induction ordering.
  \paragraph{Proof ordering} 
  Let us consider the following well-founded order on proofs: a proof is interpreted by the multi-set of pair $(b,b')$  appearing as (potentially frozen) labels of $\bfa$  applications where we erased the function symbol $\lrtbox{\phantom{a}}$. We then order these  multi-set using the multi-set ordering $\succ_{\text{mult}}$, which is induced by the product ordering $\succ_{\times}$, which itself is built upon an arbitrary total rewrite ordering on ground terms without boxes $\succ$ (e.g a LPO for some arbitrary precedence over function symbols). 

  \paragraph{Example} Assume that  $b_1 \equiv \ite{b}{a}{c}$ and $b_2 \equiv \ite{b'}{a'}{c'}$. We let $P_1$ be the derivation:
  \[
    \infer[\bfa({b_1},{b_2})]{ \ite{b_1}{u_1}{v_1} \sim \ite{b_2}{u_2}{v_2}}
    {
      \infer[\bfa(\lrtbox{b},\lrtbox{b'})]{ \xxtbox{b_1}, u_1, v_1 \sim \xxtbox{b_2}, u_2, v_2}
      {
        \lrtbox{b},{a},{c}, u_1, v_1 \sim \vec x',\lrtbox{b'},{a'},{c'}, u_1, v_1
      }
    }
  \]
  And $P_2$ be  the derivation:
  \[
    \small
    \infer[R]{\ite{b_1}{u_1}{v_1} \sim \ite{b_2}{u_2}{v_2}}
    {
      \infer[\bfa({b},{b'})]{ \ite{b}{(\ite{a}{u_1}{v_1})}{(\ite{c}{u_1}{v_1})} \sim \ite{b'}{(\ite{a'}{u_2}{v_2})}{(\ite{c'}{u_2}{v_2})}}
      {
        \infer[\bfa({a},{a'})]{ \xxtbox{b},\ite{a}{u_1}{v_1},\ite{c}{u_1}{v_1} \sim  \xxtbox{b'},\ite{a'}{u_2}{u_2},\ite{c'}{u_2}{v_2}}
        {
          \infer[\bfa({c},{c'})]{\xxtbox{b},\xxtbox{a},u_1,v_1,\ite{c}{u_1}{v_1} \sim  \xxtbox{b'},\xxtbox{a'},u_2,u_2,\ite{c'}{u_2}{v_2}}
          {
            \infer[\dup]{ \xxtbox{b},\xxtbox{a},u_1,v_1,\xxtbox{c},u_1,v_1 \sim \xxtbox{b'},\xxtbox{a'},u_2,u_2,\xxtbox{c'},u_2,v_2}
            {\xxtbox{b},\xxtbox{a},\xxtbox{c}, u_1, v_1 \sim \xxtbox{b'},\xxtbox{a'},\xxtbox{c'}, u_2, v_2}
          }
        }
      }
    }
  \]

  $P_1$ and $P_2$ are respectively interpreted as the multi-sets $\{(b_1,b_2),(b,b')\}$ and $\{(b,b'),(a,a'),(c,c')\}$ (observe that we unfroze the conditionals). $b,a,c$ (resp. $b',a',c'$) are strict subterms of $b_1$ (resp. $b_2$), therefore we have $(b_1,b_2) \succ_{\times} (b,b')$, $(b_1,b_2) \succ_{\times} (a,a')$ and $(b_1,b_2) \succ_{\times} (c,c')$. Therefore we have:
  \[
    \{(b_1,b_2),(b,b')\} \succ_{\text{mult}} \{(b,b'),(a,a'),(c,c')\}
  \]
  By consequence $P_2$ is a smaller proof of $\ite{b_1}{u_1}{v_1} \sim \ite{b_2}{u_2}{v_2}$ than $P_1$.

  \begin{proof}[Proof of Lemma~\ref{lem:body-freeze}]
  First we are going to show a cut elimination strategy to get rid of the deconstruction of frozen conditionals introduced by:
  \[
    \infer[\bfa({b_1},{b_2})]{\vec w_1, \ite{b_1}{u_1}{v_1} \sim \vec w_2, \ite{b_2}{u_2}{v_2}}
    {
      \vec w_1, \xxtbox{b_1}, u_1', v_1' \sim \vec w_2, \xxtbox{b_2}, u'_2, v'_2
    }
  \]

  Assume now that $u \sim v$ is not provable without deconstructing frozen conditionals introduced as described above. We consider a proof $P_1$ of $u \sim v$ that we suppose minimal for $\succ_{\text{mult}}$. We are going to consider the first conditionals $(b_1,b_2)$ (starting from the bottom) which are deconstructed. We let $b_1 \equiv \ite{b}{a}{c}$ and $b_2 \equiv \ite{b'}{a'}{c'}$, we know that our proof has the following shape:
  \[
    \infer[R]{u \sim v}
    {
      \infer*[(A_1)]{C[\ite{b_1}{u_1}{v_1}] \sim C[\ite{b_2}{u_2}{v_2}]}
      {
        \infer[\bfa({b_1},{b_2})]{\vec w_1, \ite{b_1}{u_1}{v_1} \sim \vec w_2, \ite{b_2}{u_2}{v_2}}
        {
          \infer*[(A_2)]{\vec w_1, \xxtbox{b_1}, u_1, v_1 \sim \vec w_2, \xxtbox{b_2}, u_2, v_2}
          {
            \infer[\bfa(\lrtbox{b},\lrtbox{b'})]{\vec x,\xxtbox{b_1},\vec y \sim \vec x',\xxtbox{b_2},\vec y'}
            {\infer*[(A_3)]
              {
                \vec x,\lrtbox{b},{a},{c},\vec y 
                \sim
                \vec x',\lrtbox{b'},{a'},{c'},\vec y'}{}}
          }
        }
      }
    }
  \]
  Where $C$ is a one-hole context. Since $(b_1,b_2)$ are the first conditionals deconstructed in this proof we know that $C$ is such that the hole does not appear in a conditional branch. This proof can be rewritten as the following proof $P_2$:
  {\small
    \[
      \infer[R]{u \sim v}
      {
        \infer[R]{C[\ite{b_1}{u_1}{v_1}] \sim C[\ite{b_2}{u_2}{v_2}]}
        {
          \infer*[(A_1)]{C[\ite{b}{(\ite{a}{u_1}{v_1})}{(\ite{c}{u_1}{v_1})}] \sim C[\ite{b'}{(\ite{a'}{u_2}{v_2})}{(\ite{c'}{u_2}{v_2})}]}
          {
            \infer[\bfa({b},{b'})]
            {
              \vec w_1, \ite{b\!\!\begin{array}[t]{l}}{\ite{a}{u_1}{v_1}\\}{\ite{c}{u_1}{v_1}\end{array}} 
              \sim
              \vec w_2, \ite{b'\!\!\begin{array}[t]{l}}{\ite{a'}{u_2}{v_2}\\}{\ite{c'}{u_2}{v_2}\end{array}}
            }
            {
              \infer[\bfa({a},{a'})]{\vec w_1, \xxtbox{b},\ite{a}{u_1}{v_1},\ite{c}{u_1}{v_1} \sim \vec w_2, \xxtbox{b'},\ite{a'}{u_2}{u_2},\ite{c'}{u_2}{v_2}}
              {
                \infer[\bfa({c},{c'})]{\vec w_1, \xxtbox{b},\xxtbox{a},u_1,v_1,\ite{c}{u_1}{v_1} \sim \vec w_2, \xxtbox{b'},\xxtbox{a'},u_2,u_2,\ite{c'}{u_2}{v_2}}
                {
                  \infer[\dup]{\vec w_1, \xxtbox{b},\xxtbox{a},u_1,v_1,\xxtbox{c},u_1,v_1 \sim \vec w_2, \xxtbox{b'},\xxtbox{a'},u_2,u_2,\xxtbox{c'},u_2,v_2}
                  {
                    \infer*[(A_2)]{\vec w_1,\xxtbox{b},\xxtbox{a},\xxtbox{c}, u_1, v_1 \sim \vec w_2 ,\xxtbox{b'},\xxtbox{a'},\xxtbox{c'}, u_2, v_2}
                    {
                      \infer*[(A_3)]{\vec x,\xxtbox{b},\xxtbox{a},\xxtbox{c},\vec y \sim \vec x',\xxtbox{b'},\xxtbox{a'},\xxtbox{c'},\vec y'}{}
                    }
                  }
                }
              }
            }
          }
        }
      }
    \]
  }
  One can check that $A_1$ remains the same in the second proof tree since the hole in $C$ is not in a conditional branch.

  The $A_1,A_2,A_3$ parts are the same in both proofs, so let $M$ be the interpretation of $A_1,A_2,A_3$ as a multi-set. Then the interpretation of $P_1$ (resp. $P_2$) is $M \cup \{(b_1,b_2),(b,b')\}$ (resp. $M \cup \{(b,b'),(a,a'),(c,c')\}$). Therefore $P_2$ is a strictly smaller proof of $u \sim v$ than $P_1$ (this is almost the same multi-sets than in the example above). Absurd.
\end{proof}

%%% Local Variables:
%%% mode: latex
%%% TeX-master: "ms"
%%% End:

% \section{Undecidability of the  $\fafrag$ Fragment for Arbitrary TRS}
% \input{fafrag}

\newpage

\section{Proof Form}

\label{app-section:proof-form}

In this section, we define what are the early proof form and the normal proof form. This is rather technical and lengthy, as the definition of normal proof form relies on four mutually recursive definitions: $\ek_l$-\emph{encryp\-tion oracle calls} are well-formed encryptions; $\ek_l$-\emph{decryption oracle calls} are well-formed decryptions; $\ek_l$-\emph{normalized basic terms} are terms built using well-formed encryptions and decryptions as well as function symbols different from $\symite$; and $\ek_l$-\emph{normalized simple terms} are combinations of normalized basic terms using $\symite$. 

We then show Lemma~\ref{lem:body-proof-form}, which is a weak normalization result: it describes a procedure that, given a proof $P$ of $\vec u \sim \vec v$ following the ordered freeze strategy of Lemma~\ref{lem:body-freeze}, computes a proof $P'$ of $\vec u \sim \vec v$ such that $P'$ is in normal proof form. This procedure is a careful bottom-up rewriting of all the sub-terms appearing in $P$.

We also give a proof of Lemma~\ref{lem:cond-equiv-body}.

\subsection{Early Proof Form} 
We showed in Lemma~\ref{lem:body-freeze} that:
\[ 
  (\twobox + \rs)^* \cdot \csmb^* \cdot \{\obfa(b,b')\}^* \cdot \unbox \cdot \fas^* \cdot \dup^* \cdot \cca \tag{$\mathcal{A}_\succ$}
\]
is complete for $\csm + \fa + R +  \dup + \cca$. Let us consider a proof $P$ following this ordering. From now on we will use $\mathcal{A}_\succ$ to denote this fragment. Moreover we let $\mathcal{A}_{\csmb}$ and $\mathcal{A}_{\obfa}$ be, respectively, the fragments:
\begin{gather*}
  \csmb^* \cdot \{\obfa(b,b')\}^* \cdot \unbox \cdot \fas^* \cdot \dup^* \cdot \cca \tag{$\mathcal{A}_{\csmb}$}\\
  \{\obfa(b,b')\}^* \cdot \unbox \cdot \fas^* \cdot \dup^* \cdot \cca \tag{$\mathcal{A}_{\obfa}$}
\end{gather*}
The only branching rule is the $\csmb$ rule, which has two premises. Hence after having completed all the $\csmb$ applications we know that the proof will be non-branching and in $\mathcal{A}_{\obfa}$. We want to name each branch of the proof tree, and its corresponding instance of the $\cca$ axiom. To do so, we index each branch of the proof tree $P$ by some $l \in L$ where $L$ is a set of labels, and we let $\lproof$ be the proof system $\vdash$ with branch annotations. When $P \lproof t \sim t'$, we let $\prooflabel(P)$ be the set of labels $L$ annotating the branches in $P$, and for all $l\in L$, we let $\instance(P,l)$ be the instance of $\cca$ obtained using Proposition~\ref{prop:cca-small-restr} from the instance of $\cca$ used in branch $l$:
\[
  \instance(P,l)  = \infer[\cca]{\vec w, (\alpha_i)_{i \in I},(\dec_j)_{j \in J} \sim \vec w, (\alpha'_i)_{i \in I},(\dec'_j)_{j \in J}}{}
\]
We also define $\encs^P_l = \{\alpha_i \mid i\}$, $\decs^P_l = \{\dec_j \mid j \in J\}$ and $\keys^P_l$ to be the sets of, respectively, encryptions, decryptions and keys used in the $\cca$ application of the branch $l$ of proof $P$, on the left side. Similarly we define $\encs'^P_l$, $\decs'^P_l$ and $\keys'^P_l$ for the right side.

\begin{definition}
  For all terms $t,t'$ and proofs $P$ such that $P \lproof_{\mathcal{A}_{\csmb}} t \sim t'$, we say that $P$ proof in \emph{early proof form} if $t$ and $t'$ are of the following form:
  \begin{gather*}
    t \equiv C\left[
      \left(\splitbox{b^{h_\sfl}}{b^{h_\sfr}}{b^h}\right)_{h \in H} 
      \diamond 
      \left( u_l
      \right)_{l \in \prooflabel(P)}
    \right]
    \qquad \wedge \qquad
    t' \equiv C\left[
      \left(\splitbox{b'^{h_\sfl}}{b'^{h_\sfr}}{b'^h}\right)_{h \in H} 
      \diamond 
      \left( u'_l
      \right)_{l \in \prooflabel(P)}
    \right]
  \end{gather*}
  where $H$ is a set of positions in $P$ (we let $\cspos(P) \equiv H$) such that:
  \begin{itemize}
  \item for all $h \in H$, the rule applied at position $h$ in $P$ is a $\csmb$ rule on the conditionals:
    \[
      \left(\splitbox{b^{h_\sfl}}{b^{h_\sfr}}{b^h},\splitbox{b'^{h_\sfl}}{b'^{h_\sfr}}{b'^h}\right)
    \]
  \item $(b^h)_{h \in H}$ are if-free conditionals in $R$-normal form and for all $h \in H, b^{h_\sfl} =_R b^{h_\sfr} =_R b^h$ (same for $b'^{h_\sfl},b'^{h_\sfr},b^h$).
  \item Let $P^{h_\sfl} = \extractl(h,P)$ and $P^{h_\sfr} = \extractr(h,P)$, then:
    \[
      P^{h_\sfl} \lproof_{\mathcal{A}_{\csmb}} b^{h_\sfl} \sim  b'^{h_\sfl}
      \qquad \wedge \qquad
      P^{h_\sfr} \lproof_{\mathcal{A}_{\csmb}} b^{h_\sfr} \sim  b'^{h_\sfr}
    \]
    and these two proofs are in \emph{early proof form}.
  \item  $\prooflabel(P^{h_\sfl}) \subseteq \prooflabel(P)$, and for all $l \in \prooflabel(P^{h_\sfl})$, $\instance(P^{h_\sfl},l)$ is subsumed by $\instance(P,l)$ (same for $\prooflabel(P^{h_\sfr})$).
  \item For all $l \in \prooflabel(P)$, we know that the extraction from $P$ of the sub-proof of $u_l \sim u'_l$ is in the fragment $\mathcal{A}_{\obfa}$.
  \end{itemize}
\end{definition}

\begin{proposition}
  \label{prop:epf-labelling}
  For all terms $t,t'$ and proofs $P$ such that $P \vdash_{\mathcal{A}_{\csmb}} t \sim t'$, there exists a labelling $P'$ of $P$ such that $P' \lproof_{\mathcal{A}_{\csmb}} t \sim t'$ and $P'$ is in early proof form.
\end{proposition}

\begin{proof}
  We can check that the proof $P$ has the wanted shape and is properly labelled by induction on the size of the proof, by observing that for all $h \in \cspos(P)$ and $\sfx \in \{\sfl,\sfr\}$, $\extractx(h,P)$ is of size strictly smaller that $P$. We only need to perform some $\alpha$-renaming to have the labelling of the sub-proofs coincide.
  
  Finally we can check that the resulting proof $Q$ is such that for all $h \in \cspos(Q), \sfx \in \{\sfl,\sfr\}$, for all $l \in \prooflabel(\extractx(h,P))$, the $\CCA$ instance $\instance(\extractx(h,P),l)$ is subsumed by $\instance(P,l)$. This follows from the fact that $\extractx(h,P)$ is obtained through the $\restr$ elimination procedure from $P$.
\end{proof}

We define below the set $\setindex(P)$ of positions of $P$, which is the set of \emph{all} positions of $P$ where a $\csmb$ rule is applied. This set is naturally ordered using the prefix ordering on positions. Moreover we can define the ``depth'' of a position $h$ in $P$ to be, intuitively, the number of nested applications of the $\csmb$ rule. 
\begin{definition} 
  Let $P \lproof_{\mathcal{A}_{\csmb}} t \sim t'$ in early proof form.
  \begin{itemize}
  \item We let $\setindex(P)$ be the set of indices where $\csmb$ rules occur in the proof $P$:
    \[
      \setindex(P) = \cspos(P) \cup \left(\bigcup_{h \in \cspos(P)} \setindex\left(\extractl(h,P)\right) \cup \setindex\left(\extractr(h,P)\right)  \right)
    \]
  \item For all $h,h' \in \setindex(t,P)$, we let $<$ be the ancestor relation, defined by $h < h'$ if and only if $h$ is a prefix of $h'$.
  \item For all $h \in \setindex(P)$, we let $\ifdepth_P(h)$ be the depth of $h$ in $P$, defined as follows:
    \[
      \ifdepth_P(h) = 
      \begin{cases}
        0 & \text{if } h \in \cspos(P)\\
        1 + \ifdepth_{P^\sfl}(h) &\text{if } \exists g \in \cspos(P) \text{ such that } h \in \setindex(\extractl(g,P)))\\
        1 + \ifdepth_{P^\sfr}(h) &\text{if } \exists g \in \cspos(P) \text{ such that } h \in \setindex(\extractr(g,P)))
      \end{cases}
    \]  
  \end{itemize}
\end{definition}

For all $\sfh = h_\sfx$, where $h \in \setindex(P)$ and $\sfx \in \{\sfl,\sfr\}$, we let $\cspos_P(\sfh) = \cspos(\extractx(h,P))$. When there is no ambiguity on the proof $P$, we write $\cspos(\sfh)$ instead of $\cspos_P(\sfh)$. 

\begin{definition}
  Let $P \lproof_{\mathcal{A}_{\csmb}} t \sim t'$ in early proof form. For all $l \in \prooflabel(P)$, we define:
  \[
    \hbranch(l) = \left\{h_\sfx \mid h \in \setindex(P) \wedge \sfx \in \{\sfl,\sfr\} \wedge l \in \prooflabel(\extractx(h,P))\right\} \cup \{\epsilon\}
  \]
  We abuse the notation and say that $h \in \hbranch(l)$ if there exists $\sfx \in \{\sfl,\sfr\}$ such that $h_\sfx \in \hbranch(l)$. In that case, we say that $\sfx$ is the direction taken at $h$ in $l$.
\end{definition}
Morally, $\hbranch(l)$ is the set of positions of $P$ where a $\csmb$ rule is applied on a given branch. Of course for all $l \in \prooflabel(P)$, $\epsilon \in \hbranch(l)$ since $\epsilon$ is the index of the toplevel proof $P$.

\subsection{Shape of the Terms} 
\label{subsection:proof-det-shape}
For all proofs in $\mathcal{A}_\succ$, all $R$ rewritings are done at the beginning of the proofs in the $(\twobox + \rs)^*$ part, and, afterwards, all rules (apart from $\dup$) only ``peel off'' terms by removing the top-most function symbol. Therefore the terms just after $(\twobox + \rs)^*$ characterize the shape of the subsequent proof. This observation is illustrated in Fig.~\ref{fig:proof-term-restr}. Recall that for all $P \lproof_{\mathcal{A}_{\csmb}} t \sim t'$ in early proof form, we have:  
\begin{gather*}
    t \equiv C\left[
      \left(\splitbox{b^{h_\sfl}}{b^{h_\sfr}}{b^h}\right)_{h \in H} 
      \diamond 
      \left( u_l
      \right)_{l \in \prooflabel(P)}
    \right]
    \qquad \wedge \qquad
    t' \equiv C\left[
      \left(\splitbox{b'^{h_\sfl}}{b'^{h_\sfr}}{b'^h}\right)_{h \in H} 
      \diamond 
      \left( u'_l
      \right)_{l \in \prooflabel(P)}
    \right]
  \end{gather*}
  where for all $l \in \prooflabel(P)$, the extraction from $P$ of the sub-proof of $u_l \sim u'_l$ is in the fragment $\mathcal{A}_{\obfa}$. This means that for all $l$:
  \begin{gather*}
    u_l \equiv D_l
    \left[
      \left(
        B_{i,l}[\vec w_{i,l},(\alpha_{i,l}^j)_{j \in J^0_{i,l}},(\dec_{i,l}^k)_{K^0_{i,l}}]
      \right)_{i \in I}
      \diamond
      \left(
        U_{m,l}[\vec w_{m,l},(\alpha_{m,l}^j)_{j \in J^1_{i,l}},(\dec_{m,l}^k)_{k \in K^1_{i,l}}]
      \right)_{m \in M}
    \right] \\
    u'_l \equiv D_l
    \left[
      \left(
        B_{i,l}[\vec w_{i,l},(\alpha_{i,l}'^j)_{j \in J^0_{i,l}},(\dec_{i,l}'^k)_{K^0_{i,l}}]
      \right)_{i \in I}
      \diamond
      \left(
        U_{m,l}[\vec w_{m,l},(\alpha_{m,l}'^j)_{j \in J^1_{i,l}},(\dec_{m,l}'^k)_{k \in K^1_{i,l}}]
      \right)_{m \in M}
    \right]
  \end{gather*}
  where $D_l$ is an if-context, $(B_{i,l})_i$ and $(U_{m,l})_m$ are if-free contexts, the encryptions appear in $\encs^P_l$:
  \begin{gather*}
    \left\{\alpha_{i,l}^j \mid i \in I,j \in J^0_{i,l} \right\} \cup \left\{\alpha_{m,l}^j \mid m \in M,j \in J^0_{m,l} \right\} \;\subseteq\; \encs_l^P\\
    \left\{\alpha_{i,l}'^j \mid i \in I,j \in J^0_{i,l} \right\} \cup \left\{\alpha_{m,l}'^j \mid m \in M,j \in J^0_{m,l} \right\} \;\subseteq\; \encs'^P_l
  \end{gather*}
  and the decryptions appear in $\decs^P_l$:
  \begin{gather*}
    \left\{\dec_{i,l}^k \mid i \in I,k \in K^0_{i,l} \right\} \cup \left\{\dec_{m,l}^k \mid m \in M,k \in K^0_{m,l} \right\} \;\subseteq\; \decs^P_l\\
    \left\{\dec_{i,l}'^k \mid i \in I,k \in K^0_{i,l} \right\} \cup \left\{\dec_{m,l}'^k \mid m \in M,k \in K^0_{m,l} \right\} \;\subseteq\; \decs'^P_l
  \end{gather*}

\begin{figure}[t]
  \begin{center}
    \begin{tikzpicture}
      \draw (3,4.5) -- ++(3,-4.5) --
      node[above=1cm,midway]  {$\left(\splitbox{b^l_h}{b^r_h}{b_h}\right)_{h \in H}$}
      ++(-6,0) -- cycle;

      \draw (0,0) -- ++(1.5,-2.3) 
      % First child
      -- ++(1,-1.7) -- 
      node[above,midway]  {$U_{n_0,0}[]$}
      ++(-2,0) -- ++(1,1.7)
      % Child end
      -- 
      node[above=0.3cm,midway]  {$\left(B_{i,0}[]\right)_i$}
      ++(-3,0) 
      % Second child
      -- ++(1,-1.7) -- 
      node[above,midway]  {$U_{0,0}[]$}
      ++(-2,0) -- ++(1,1.7)
      % Child end
      -- cycle;

      \node at (3,-1.5) {$\cdots$};

      \draw (6,0) -- ++(1.5,-2.3) 
      % First child
      -- ++(1,-1.7) -- 
      node[above,midway]  {$U_{n_m,m}[]$}
      ++(-2,0) -- ++(1,1.7)
      % Child end
      -- 
      node[above=0.3cm,midway]  {$\left(B_{i,{m}}[]\right)_i$}
      ++(-3,0) 
      % Second child
      -- ++(1,-1.7) -- 
      node[above,midway]  {$U_{0,m}[]$}
      ++(-2,0) -- ++(1,1.7)
      % Child end
      -- cycle;

      % Braces, first column
      \draw [decorate,decoration={brace,mirror,amplitude=7pt}]
      (9,-4.02) -- ++(0,1.68) node [midway,right=0.6em] {$\fas^*$};
      \draw[dashed] (8.5,-4) -- (9,-4);

      \draw [decorate,decoration={brace,amplitude=3.5pt}]
      (9,-0) -- ++(0,-0.8) node [midway,right=0.6em] {$\fas^*$};
      \node[rotate=90] at (9,-1.15) {{\footnotesize{$\cdots$}}};
      \draw [decorate,decoration={brace,amplitude=3.5pt}]
      (9,-1.5) -- ++(0,-0.8) node [midway,right=0.6em] {$\fas^*$};
      \draw[dashed] (6.98,-1.5) -- (9,-1.5);
      \draw[dashed] (6.52,-0.8) -- (9,-0.8);

      \draw [decorate,decoration={brace,amplitude=3.5pt}]
      (9,4.5) -- ++(0,-0.8) node [midway,right=0.6em] {$\mathcal{A}_{\csmb}$};
      \node[rotate=90] at (9,2.25) {$\cdots$};
      \draw [decorate,decoration={brace,amplitude=3.5pt}]
      (9,0.8) -- ++(0,-0.78) node [midway,right=0.6em] {$\mathcal{A}_{\csmb}$};
      \draw[dashed] (5.48,0.8) -- (9,0.8);
      \draw[dashed] (3.52,3.7) -- (9,3.7);

      % Braces, second column
      \draw [decorate,decoration={brace,amplitude=8pt}]
      (11,4.5) -- ++(0,-4.48) node [midway,right=1em] {$\csmb^*$};
      \draw[dashed] (3,4.5) -- (11,4.5);
      \draw[dashed] (6,0) -- (11,0);

      \draw [decorate,decoration={brace,amplitude=8pt}]
      (11,0) -- ++(0,-2.28) node [midway,right=1em] {$\{\obfa(b,b')\}^*$};
      \draw[dashed] (7.5,-2.3) -- (11,-2.3);
    \end{tikzpicture}
  \end{center}
  \caption{\label{fig:proof-term-restr} The shape of the term is determined by the proof.}
\end{figure}
  
\begin{figure}[ht]
  \[
    \infer[% (\twobox + \rs)^*
    ]{t \sim t'}
    {
      \infer[]{
        \begin{array}{c}
          C\left[\left(\splitbox{b^l_h}{b^r_h}{b_h}\right)_h \diamond \left(
              D_l
              \left[\left(B_{i,l}[\vec w_{i,l},(\alpha_{i,l}^j)_j,(\dec_{i,l}^k)_k]\right)_i
                \diamond\left(U_{m,l}[\vec w_{m,l},(\alpha_{m,l}^j)_j,(\dec_{m,l}^k)_k]\right)_m\right]
            \right)_l\right] \\
          \sim\\
          C\left[\left(\splitbox{b'^l_h}{b'^r_h}{b'_h}\right)_h\diamond \left(
              D_l
              \left[\left(B_{i,l}[\vec w_{i,l},(\alpha_{i,l}'^j)_j,(\dec_{i,l}'^k)_k]\right)_i
                \diamond\left(U_{m,l}[\vec w_{m,l},(\alpha_{m,l}'^j)_j,(\dec_{m,l}'^k)_k]\right)_m\right]
            \right)_l\right] 
        \end{array}
      }
      {
        \forall l \in L,\qquad \infer*[\csmb^*]{}{
          \infer[]{}{
            \infer{
              \begin{array}{c}
                \left(
                  D^{{\sfh}}_l
                  \left[\left(B^{{\sfh}}_{i,l}[\vec w_{i,l}^{{\sfh}},(\alpha_{i,l}^{{{\sfh}},j})_j,(\dec_{i,l}^{{{\sfh}},k})_k]\right)_i
                    \diamond\left(U^{{\sfh}}_{m,l}[\vec w_{m,l}^{{\sfh}},(\alpha_{m,l}^{{{\sfh}},j})_j,(\dec_{m,l}^{{{\sfh}},k})_k]\right)_m\right]
                \right)_{{{\sfh}} \in \hbranch(l)}\\
                \sim\\
                \left(
                  D_l^{{\sfh}}
                  \left[\left(B_{i,l}^{{\sfh}}[\vec w_{i,l}^{{\sfh}},(\alpha_{i,l}'^{{{\sfh}},j})_j,(\dec_{i,l}'^{{{\sfh}},k})_k]\right)_i
                    \diamond\left(U_{m,l}^{{\sfh}}[\vec w_{m,l}^{{\sfh}},(\alpha_{m,l}'^{{{\sfh}},j})_j,(\dec_{m,l}'^{{{\sfh}},k})_k]\right)_m\right]
                \right)_{{{\sfh}} \in \hbranch(l)}
              \end{array}
            }
            {
              \infer*[\{\obfa(b,b')\}^*]{}{
                \infer[]{}{
                  \infer[]{
                    \begin{array}{c}
                      \left(
                        \left(B^{{\sfh}}_{i,l}[\vec w_{i,l}^{{\sfh}},(\alpha_{i,l}^{{{\sfh}},j})_j,(\dec_{i,l}^{{{\sfh}},k})_k]\right)_i,
                        \left(U^{{\sfh}}_{m,l}[\vec w_{m,l}^{{\sfh}},(\alpha_{m,l}^{{{\sfh}},j})_j,(\dec_{m,l}^{{{\sfh}},k})_k]\right)_m
                      \right)_{{{\sfh}} \in \hbranch(l)}\\
                      \sim\\
                      \left(
                        \left(B_{i,l}^{{\sfh}}[\vec w_{i,l}^{{\sfh}},(\alpha_{i,l}'^{{{\sfh}},j})_j,(\dec_{i,l}'^{{{\sfh}},k})_k]\right)_i,
                        \left(U_{m,l}^{{\sfh}}[\vec w_{m,l}^{{\sfh}},(\alpha_{m,l}'^{{{\sfh}},j})_j,(\dec_{m,l}'^{{{\sfh}},k})_k]\right)_m
                      \right)_{{{\sfh}} \in \hbranch(l)}
                    \end{array}
                  }
                  {
                    \infer*[\fas^* \cdot \dup^*]{}{
                      \infer[]{}{
                        \infer[\bcca]
                        {
                          \begin{array}{c}
                            \left(
                              \left(
                                \vec w_{i,l}^{{\sfh}},(\alpha_{i,l}^{{{\sfh}},j})_j,(\dec_{i,l}^{{{\sfh}},k})_k
                              \right)_i,
                              \left(
                                \vec w_{m,l}^{{\sfh}},(\alpha_{m,l}^{{{\sfh}},j})_j,(\dec_{m,l}^{{{\sfh}},k})_k
                              \right)_m
                            \right)_{{{\sfh}} \in \hbranch(l)}\\
                            \sim\\
                            \left(
                              \left(
                                \vec w_{i,l}^{{\sfh}},(\alpha_{i,l}'^{{{\sfh}},j})_j,(\dec_{i,l}'^{{{\sfh}},k})_k
                              \right)_i,
                              \left(
                                \vec w_{m,l}^{{\sfh}},(\alpha_{m,l}'^{{{\sfh}},j})_j,(\dec_{m,l}'^{{{\sfh}},k})_k
                              \right)_m
                            \right)_{{{\sfh}} \in \hbranch(l)}
                          \end{array}
                        }           
                        {}     
                      }
                    }        
                  }
                }
              }
            }
          }
        }
      }
    }
  \]
  \caption{\label{fig:shape-proof} Shape of a full proof (for simplicity, we omitted the boxes in terms and related rules).}
\end{figure}

Using these notation, we give some definitions:
\begin{definition}
  Let $P \lproof_{\mathcal{A}_{\csmb}} t \sim t'$. Then for all $l \in \prooflabel(P)$, we define the following relations:
  \begin{itemize}
  \item $(b,b') \lesimcs^{\epsilon,l} (t \sim t',P)$ (resp. $b \lecs^{\epsilon,l} (t,P)$, $b' \lecs^{\epsilon,l} (t',P)$) if and only if there exists $h_0 \in \setindex(P)$ such that:
    \[
      b \equiv b^{h_0}
      \quad \wedge \quad
      b' \equiv b'^{h_0}
    \]
  \item $(\beta,\beta') \lesimcond^{\epsilon,l} (t \sim t',P)$ (resp. $\beta \lecond^{\epsilon,l} (t,P)$, $\beta' \lecond^{\epsilon,l} (t',P)$) if and only if there exists $i \in I$ such that:
    \[
      \beta \equiv B_{i,l}[\vec w_{i,l},(\alpha_{i,l}^{j})_{j \in J^0_{i,l}},(\dec_{i,l}^{k})_{k \in K^0_{i,l}}]
      \quad \wedge \quad
      \beta' \equiv B_{i,l}[\vec w_{i,l},(\alpha_{i,l}'^{j})_{j \in J^0_{i,l}},(\dec_{i,l}'^{k})_{k \in K^0_{i,l}}]
    \]
  \item $(\gamma,\gamma') \lesimleave^{\epsilon,l} (t \sim t',P)$ (resp. $\gamma \leleave^{\epsilon,l} (t,P)$, $\gamma' \leleave^{\epsilon,l} (t',P)$) if and only if there exists $m \in M$ such that:
    \[
      \gamma \equiv U_{m,l}[\vec w_{m,l},(\alpha_{m,l}^{j})_{j \in J^1_{i,l}},(\dec_{m,l}^{k})_{k \in K^1_{i,l}}]
      \quad \wedge \quad
      \gamma' \equiv U_{m,l}[\vec w_{m,l},(\alpha_{m,l}'^{j})_{j \in J^1_{i,l}},(\dec_{m,l}'^{k})_{k \in K^1_{i,l}}]
    \]
  \end{itemize}
\end{definition}

\begin{definition}
  Let $P \lproof_{\mathcal{A}_{\csmb}} t \sim t'$ in early proof form. For all $h \in \setindex(P)$, $\sfx \in \{\sfl,\sfr\}$:
  \begin{itemize}
  \item For all $\Delta \in \{\textsf{c$\sim$c},\textsf{l$\sim$l},\textsf{cs$\sim$cs}\}$, we define $\lesymbol_{\Delta}^{h_{\sfx},l} (t \sim t',P)$ as follows:
    \[
      \forall s,s'.\, (s,s') \lesymbol_{\Delta}^{h_{\sfx},l} (t \sim t',P) \qquad\text{ if and only if }\qquad (s,s') \lesymbol_{\Delta}^{\epsilon,l} (b \sim b',\extractx(h,P))
    \]
    where $\extractx(h,P)$ is a proof of $b \sim b'$.
  \item For all $\Delta \in \{\textsf{c},\textsf{l},\textsf{cs}\}$, we define $\lesymbol_{\Delta}^{h_{\sfx},l} (t,P)$ as follows:
    \[
      \forall s.\, s \lesymbol_{\Delta}^{h_{\sfx},l} (t,P) \qquad\text{ if and only if }\qquad s \lesymbol_{\Delta}^{\epsilon,l} (b,\extractx(h,P))
    \]
    where $\extractx(h,P)$ is a proof of $b \sim b'$.
  \end{itemize}
\end{definition}

\begin{remark}
  We extend these notations to proofs $P$ such that $P \lproof_{\mathcal{A}_{\succ}} t \sim t'$. Let $P'$ be such that:
  \[
    P \equiv 
    \begin{array}[c]{c}
      \infer[(\twobox + \rs)^*]{t \sim t'}{P'}
    \end{array}
  \]
  and $P' \lproof_{\mathcal{A}_{\csmb}} t_0 \sim t_0'$, then $(s,s') \lesymbol_{\Delta}^{\sfh,l} (t \sim t',P)$ if and only if $(s,s') \lesymbol_{\Delta}^{\sfh,l} (t_0 \sim t'_0,P')$ where $\Delta \in \{\textsf{c$\sim$c},\textsf{l$\sim$l},\textsf{cs$\sim$cs}\}$. Similarly  $s \lesymbol_{\Delta}^{\sfh,l} (t,P)$ if and only if $s \lesymbol_{\Delta}^{\sfh,l} (t_0,P')$ where $\Delta \in \{\textsf{c},\textsf{l},\textsf{cs}\}$.

  Extending these notations to $B^\sfh_l[],U^\sfh_l \dots$, we describe the shape of a complete proof in Fig.~\ref{fig:shape-proof}.
\end{remark}

\subsection{Simple Terms} 
% \label{subsection:simple-terms}
A public/private key pair is valid if the same name has been used to generate the keys.
\begin{definition}
  A valid public/private key pair is a pair of terms $(\pk(\nonce),\sk(\nonce))$ where $\nonce$ is a name.
\end{definition}

We will now formally define the normal form for terms used in the strategy. This is done through four mutually inductive definitions: the normal forms of well-formed encryptions and of well-formed decryptions; the normal form of basic terms built using well-formed encryptions and decryptions, as well as function symbols different from $\symite$; and finally the normal form of terms with conditionals.

The next step will be to prove that all intermediate terms in the proofs can be assumed to be in these normal forms. To keep the proof tractable, this will be done in two steps. Therefore we introduce two versions of some forms, e.g. we will define \emph{simple terms} to be terms having a particular form, and \emph{normalized simple terms} to be \emph{simple terms} satisfying some further properties. Consider a instance of $\CCA_a$:
\[
  (\phi,\encvar,\decvar,\randmap,\encmap,\decmap) R_{\CCA_a}^\keys (\_,\_,\_,\_,\_,\_)
\]
Let $\encs = \encvar \encmap$ be the set of encryptions, $\decs = \decvar \decmap$ be set of decryptions and $\rands = \encvar \randmap$ the set of encryption randomness used. We also let $\ek = (\keys,\rands,\encs,\decs)$.
\begin{definition}
  A \ek-\emph{encryption oracle call} is a term $t$ of the form $\enc{u}{\pk}{r}$ where:
  \begin{itemize}
  \item $\enc{u}{\pk}{r} \in \encs$, $r \in \rands$, $(\pk,\sk)$ is a valid public/private key pair and with $\sk \in \keys$.
  \item $u$ is a \ek-\emph{normalized simple terms}.
  \end{itemize}
\end{definition}
Similarly, a \ek-\emph{decryption oracle calls} $t$ is valid decryption in $\decs$ under secret key $\sk \in \keys$ such that all other encryptions and decryptions appearing directly in $t$, either in guards or in the decrypted term, are themselves \ek-\emph{encryption oracle calls} and \ek-\emph{decryption oracle calls}.
\begin{definition}
  A \ek-\emph{decryption oracle call} is a term of the form
  \(
    C\left[\vec g \diamond (s_i)_{i \le p}\right]
  \) in $\decs$
  where:
  \begin{itemize}
  \item $(\pk,\sk)$ is valid public/private key pair and $\sk \in \keys$.
  \item There exists a context $u$ if-free and in $R$-normal form, and a term $t$ such that:
    \begin{mathpar}
      t\equiv u[(\alpha_{j})_j,(\dec_{k})_k]
      
      \forall i < p,\,
      s_i \equiv \zero(\dec(t,\sk))

      s_p \equiv \dec(t,\sk)

      \forall g \in \vec g,\,
      g \equiv \eq{t}{\alpha_j}
    \end{mathpar}
  \item For all $j$, $\alpha_{j}$ is a \ek-\emph{encryption oracle call}.
  \item For all $k$, $\dec_{k}$ is a \ek-\emph{decryption oracle call}.
  \end{itemize}
  $(\alpha_{j})_{j}$ are called $u$'s encryptions. We often write $(\dec_k)_k$ to denote a vector of decryption oracle calls.
\end{definition}
Figure~\ref{fig:enc-dec-tikz} gives a visual representation of the shapes of encryption and decryption oracle calls.

A \ek-\emph{basic term} is a term build using \ek-\emph{encryption oracle calls}, \ek-\emph{decryption oracle calls}, function symbols in $\nizsig$ and names in $\Nonce$, with some restrictions. More precisely, we require that:
\begin{itemize}
\item We do not use names in $\rands$, as this would contradict $\cca$ randomness side-conditions.
\item We do not decrypt terms using secret keys in $\keys$.
\end{itemize}
\begin{definition}
  A \ek-\emph{basic term} is a term of the form
  \(
  U[\vec w,(\alpha_j)_j,(\dec_k)_k]
  \) where:
  \begin{itemize}
  \item $U$ and $\vec w$ are if-free, $U$ does not contain $\zero(\_)$, $\fresh{\rands}{\vec w}$ and $\nodec(\keys,\vec w)$.
  \item $(\alpha_j)_j$ are \ek-\emph{encryption oracle calls}.
  \item $(\dec_k)_k$ are \ek-\emph{decryption oracle calls}.
  \end{itemize}
  A \ek-\emph{basic conditional} is a \ek-\emph{basic term} of sort $\bool$.
\end{definition}

A \ek-\emph{normalized basic term} is a a \ek-\emph{basic term} that has been built without introducing any $R$-redex.
\begin{definition}
  A \ek-\emph{normalized basic term} is a \ek-\emph{basic term} of the form
  \(
  U[\vec w,(\alpha_j)_j,(\dec_k)_k]
  \)
  where:
  \begin{itemize}
  \item $(\alpha_j)_j$ are encryptions under $(\pk_j,\sk_j)_j$, and $(\dec_k)_k$ are decryptions under $(\pk_k,\sk_k)_k$.
  \item $U[\vec w, (\enc{[]_j}{\pk_j}{0})_j, (\dec([]_k,\sk_k))_k]$ is in $R$-normal form.
  \end{itemize}
  A \ek-\emph{normalized basic conditional} is a \ek-\emph{normalized basic term} of sort $\bool$.
\end{definition}
Finally, a \ek-\emph{simple term} is a term build using only \ek-\emph{basic term} and the $\symite$ function symbols. Moreover, if we use only \ek-\emph{normalized basic term}, then we get an a \ek-\emph{normalized simple term}.
\begin{definition}
  A \ek-\emph{simple term} (resp. \ek-\emph{normalized simple term}) is a term of the form $C[ \vec b \diamond \vec u]$ where:
  \begin{itemize}
  \item $C$ is an if-context.
  \item $\vec b$ are \ek-\emph{basic conditionals} (resp. \ek-\emph{normalized basic conditionals}).
  \item $\vec u$ are \ek-\emph{basic terms} (resp. \ek-\emph{normalized basic terms}).
  \end{itemize}
\end{definition}

\begin{figure}[tp]
  \begin{gather*}
    \begin{array}{cc}
      \lrenc{
        \begin{array}{c}
          \begin{tikzpicture}
            \draw (-1.5,0) -- (0,2) -- (1.5,0) -- cycle;
            \draw (-1.5,0) -- ++(0.5,-1) -- node[above]{$\textsf{t}_1$} ++(-1,0) -- cycle;
            \draw (1.5,0) -- ++(0.5,-1) -- node[above]{$\textsf{t}_n$} ++(-1,0) -- cycle;
            \node at (0,-0.6) {$\cdots$};
            \node at (0,0.9) {$\vec b$};
          \end{tikzpicture}
        \end{array}}
      {\pk}{\nonce_r} \qquad&\qquad
      \begin{array}{c}
        \begin{tikzpicture}[sibling distance=5em,level distance=4em]
          \tikzstyle{level 1}=[level distance=4em]
          \tikzstyle{level 2}=[level distance=2em, sibling distance=2em]
          \tikzstyle{level 3}=[level distance=5em, sibling distance=5em]
          \node at (0,0) {$\eq{\textsf{t}}{\alpha_1}$}
          child {node {$\zero(\dec(\textsf{t},\sk))$}}
          child {
            node {$\cdots$}
            child {node{}edge from parent[draw=none]}
            child {
              node {$\eq{\textsf{t}}{\alpha_n}$}
              child {node {$\zero(\dec(\textsf{t},\sk))$}}
              child {node {$\dec(\textsf{t},\sk)$}}
            }
          };
        \end{tikzpicture}
      \end{array}\\[5em]
      \textbf{Encryption Oracle Call} &
      \textbf{Decryption Oracle Call}
    \end{array}
  \end{gather*}
  \textbf{Convention:} $\alpha_1,\dots,\alpha_n$ are the encryptions of $\encs$ under $\pk$ appearing in $\textsf{t}$.
  \caption{\label{fig:enc-dec-tikz} Shapes of Encryption and Decryption Oracle Calls}
\end{figure}

\begin{remark}
  For all term $u$, the guards of a $\ekl$-decryption oracle calls are $\ekl$-normalized basic terms. But the \emph{leaves} of $\ek$-decryption oracle calls \emph{are not} $\ek$-normalized basic terms, because they do not satisfy the condition $\nodec(\keys,\cdot)$.
\end{remark}

The inductive definition of $\ek$-normalized basic terms naturally gives us a relation $\lest^{\ek}$ between $\ek$-normalized basic terms, $\ek$-normalized simple terms, $\ek$-decryption oracle calls and $\ek$-encryption oracle calls.
\begin{definition}
  $\lest^{\ek}$ is the reflexive and transitive closure of the relation $<^{\ek}$ defined as:
  \begin{itemize}
  \item For all \ek-encryption oracle call $t \equiv \enc{u}{\pk}{r}$, $u <^{\ek} t$.
  \item For all \ek-decryption oracle call:
    \[
      t \equiv C\left[\vec g[(\alpha_{j})_j,(\dec_{k})_k] \diamond (s_i[(\alpha_{j})_j,(\dec_{k})_k])_{i \le p}\right]
    \]
    for all $j$, $\alpha_{j} <^{\ek} t$ and for all $k$, $\dec_{k} <^{\ek} t$.
  \item For all \ek-normalized basic term $t \equiv U[\vec w,(\alpha_j)_j,(\dec_k)_k]$ where:
    for all $j$, $\alpha_{j} <^{\ek} t$ and for all $k$, $\dec_{k} <^{\ek} t$.
  \item For all \ek normalized simple term $t \equiv C[ \vec b \diamond \vec u]$, $\forall b \in \vec b, b <^{\ek} t$ and  $\forall u \in \vec u, u <^{\ek} t$.
  \end{itemize}
\end{definition}

We let $\lebt^{\ek}$ be union of the restriction of $\lest^{\ek}$ to the instances where the left term is a $\ek$-normalized basic term, and the set of guards appearing in the right-term. Formally:
\begin{definition}
  Let $\lest'^{\ek}$ be the reflexive and transitive closure of the order $<'^{\ek}$, which has the same definition than $<^{\ek}$, apart for the $\ek$-decryption oracle call:
  \begin{itemize}
  \item For all \ek-decryption oracle call:
    \[
      t \equiv C\left[\vec g \diamond (s_i[(\alpha_{j})_j,(\dec_{k})_k])_{i \le p}\right]
    \]
    for all $j$, $\alpha_{j} <'^{\ek} t$; for all $k$, $\dec_{k} <'^{\ek} t$; and for all $b \in \vec g$, $b <'^{\ek} t$.
  \end{itemize}
  We finally define $\lebt^{\ek}$: for every terms $u,v$:
  \[
    u \lebt^{\ek} v \quad \text{ iff } \quad u \lest'^{\ek} v \;\wedge\; u \text{ is a $\ek$-normalized basic term}
  \]
\end{definition}

\subsection{Proof Form and Normalized Proof Form} 

\begin{definition}
  Let $P \lproof_{\mathcal{A}_{\csmb}} t \sim t'$ in early proof form. We say that this proof is in \emph{proof form} (resp. \emph{normalized proof form}) if:
  \begin{gather*}
    t \equiv C\left[\left(\splitbox{b^{h_\sfl}}{b^{h_\sfr}}{b^h}\right)_{h \in H} \diamond \left(
        D_l\left[
          \left(\beta\right)_{\beta \lecond^{\epsilon,l} (t,P)}
          \diamond\left(\gamma\right)_{\gamma \leleave^{\epsilon,l} (t,P)}
        \right]
      \right)_{l \in L}\right]\\
    t' \equiv C\left[\left(\splitbox{b'^{h_\sfl}}{b'^{h_\sfr}}{b'^h}\right)_{h \in H} \diamond \left(
        D_l\left[
          \left(\beta'\right)_{\beta' \lecond^{\epsilon,l} (t',P)}
          \diamond\left(\gamma'\right)_{\gamma' \leleave^{\epsilon,l} (t',P)}
        \right]
      \right)_{l \in L}\right]
  \end{gather*}
  and it satisfies the following properties:
  \begin{itemize}
  \item $(b^{h_\sfl})_{h \in H},(b^{h_\sfr})_{h \in H}$ are terms in \emph{proof forms} (resp. \emph{normalized proof forms}).
  \item For all $l$,
    $
    D_l\left[
      \left(\beta\right)_{\beta \lecond^{\epsilon,l} (t,P)}
      \diamond\left(\gamma\right)_{\gamma \leleave^{\epsilon,l} (t,P)}
    \right]
    $
    is a $(\keys^P_l,\encs^P_l)$-\emph{simple term} (resp. $(\keys^P_l,\encs^P_l)$-\emph{normalized simple term}).
  \item For all $l$,
    $
    D_l\left[
      \left(\beta'\right)_{\beta' \lecond^{\epsilon,l} (t',P)}
      \diamond\left(\gamma'\right)_{\gamma' \leleave^{\epsilon,l} (t',P)}
    \right]
    $
    is a $(\keys'^P_l,\encs'^P_l)$-\emph{simple term} (resp. $(\keys'^P_l,\encs'^P_l)$-\emph{normalized simple term}).
  \end{itemize}
  We let $P \npfproof t \sim t'$ if and only if $P \vdash_{\mathcal{A}_{\csmb}} t \sim t'$ and the proof is in \emph{normalized proof form}.
\end{definition}

Let $P \npfproof t \sim t'$,  we already defined the set of conditionals $\lecond^{\sfh,l} (t,P)$ used in the $\obfa$ rules in the sub-proof $P$ of at index $\sfh$ and branch $l$. In the case of proof in normalized proof form, these conditionals are normalized basic conditional. Similarly the set of leave terms $\leleave^{\sfh,l} (t,P)$ in the sub-proof of $P$ of at index $\sfh$ and branch $l$ is a set of normalized basic terms. Recall that a basic term may contain other basic terms in its subterm. Hence we can define the set of all normalized basic terms appearing in the subterms of $\lecond^{\sfh,l} (t,P) \cup \leleave^{\sfh,l} (t,P)$.
\begin{definition}
  For all $P \npfproof t \sim t'$, we define $\lebt^{\sfh,l}(t,P)$ as follows: for all term $s$, $s \lebt^{\sfh,l} (t,P)$ if and only if there exists $u (\lecond^{\sfh,l} \cup \leleave^{\sfh,l}) (t,P)$ such that $s \lebt^{\ekl} u$.
\end{definition}

\subsection{Eager Reduction for \texorpdfstring{$\fas^* \cdot \dup^* \cdot \cca$}{\{FAf* . Dup* . CCA\}}}
% \label{subsection:eager-red}

Before proving that we can restrict ourselves to term in proof forms we need several auxiliary results about the $\fas^* \cdot \dup^* \cdot \cca$ fragment, which we state and prove here.
\begin{proposition}
  \label{prop:fashape}
  For all $b,b' \in \mathcal{T}(\sig,\Nonce)$, if $b \sim b'$ is derivable in $\fas^* \cdot \dup^* \cdot \cca$ then $b \equiv C[\vec w,(\alpha_i)_i,(\dec_j)_j]$, $b' \equiv C[\vec w,(\alpha_i')_i,(\dec_j')_j]$ and the applied \cca axiom is:
  \[
    \vec w,(\alpha_i)_i,(\dec_j)_j \sim \vec w,(\alpha_i')_i,(\dec_j')_j
  \]
\end{proposition}

\begin{proof}
  This is easy to show by induction on the proof derivation.
\end{proof}

We now give the proof of Lemma~\ref{lem:cond-equiv-body}, which we recall below:
\begin{lemma*}[\ref{lem:cond-equiv-body}]
  For all $b,b',b''$, if $b,b \sim b',b''$ is in the fragment $\mathfrak{F}(\fas^*\cdot\dup^*\cdot\cca)$ then $b' \equiv b''$.
\end{lemma*}

\begin{proof}
  From Proposition~\ref{prop:fashape} we have:
  \begin{align*}
    b \equiv C^l[\vec w^l,(\alpha^l_i)_{i \in I^l},(\dec^l_j)_{j \in J^l}] 
    &&b' \equiv C^l[\vec w^l,(\alpha'^l_i)_{i \in I^l},(\dec'^l_j)_{j \in J^l}]\\
    b \equiv C^r[\vec w^r,(\alpha^r_i)_{i \in I^r},(\dec^r_j)_{j \in J^r}]
    &&b'' \equiv C^r[\vec w^r,(\alpha'^r_i)_{i \in I^r},(\dec'^r_j)_{j \in J^r}]
  \end{align*}

  Assume that $C^l \not \equiv C^r$. Let $p$ be the position of a hole of $C^l$ such that $p$ is a valid position but not a hole position in $C^r$ (if this is not the case, invert $b'$ and $b''$). Then we have three cases:
  \begin{itemize}
  \item If the hole at $b_{|p}$ is mapped to a term $u \in \vec w^l$, then we can rewrite the proof such that $p$ is an hole position in both terms.
  \item If the hole at $b_{|p}$ is mapped to an encryption oracle call $\enc{m}{\pk(\nonce)}{r}$ in $b$ and $\enc{m'}{\pk(\nonce)}{r}$ in $b'$. Since $\enc{m}{\pk(\nonce)}{r}$ is an encryption in the \cca application we know from the freshness side-condition that $r$ does not appear in $\vec w^r$.

    Then there exists a context $A$ such that $A$ is not a hole, $m \equiv A[\vec w^r,(\alpha^r_i)_{i \in I^r},(\dec^r_j)_{j \in J^r}]$ and $C^r_{|p} \equiv A$. By consequence we know that $r \in \vec w^r$. Absurd.
  \item If the hole at $b_{|p}$ is mapped to a decryption oracle call $\dec^l_{i_0}$ in $b$. We let $\dec(m,\sk(\nonce))$ be such that $\dec(m,\sk(\nonce))$ is well-guarded in $\dec^l_{i_0}$. Since $\dec^l_{i_0}$ is a decryption in the \cca application we know from the key-usability side-condition that $\sk(\nonce)$ appears only in decryption position in $\vec w^r$. Then there exists a context $A$ such that $A$ is not a hole, $b'_{|p} \equiv A[\vec w^r,(\alpha^r_i)_{i \in I^r},(\dec^r_j)_{j \in J^r}]$ and $A$ is if-free. By consequence we know that $\fa_{\dec}$ is applied on the right-side, which implies that either $\nonce \in \vec w^r$ or $\sk(\nonce) \in \vec w^r$. Absurd.
    \qedhere
  \end{itemize}
\end{proof}

\paragraph{Eager Reduction} We state here a key result about the $\fas^* \cdot \dup^* \cdot \cca$ fragment, which deals with the following problem: when trying to prove that $u \sim u'$ holds, one may rewrite $u$ and $u'$ into $\pair{u}{v}$ and $\pair{u'}{v'}$ using $R$. The problem here is that $v$ and $v'$ are arbitrary large terms, which makes the proof space unbounded. E.g. this is the case in the following proof:
\[
  \infer[R]{u \sim u'}
  {
    \infer[\fa_{\pair{}{}}]{\pi_1(\pair{u}{v}) \sim \pi_1(\pair{u'}{v'})}
    {
      \infer{u,v \sim u',v'}
      {
        \infer*[(P)]{}{}
      }
    }
  }
\]
Of course there is a shortcut here: since $(P)$ is a proof of $u,v \sim u',v'$ using the $\restr$ rule we have a proof of $u \sim u'$. Moreover the $\restr$ elimination Lemma~\ref{lem:restrelim} allows us to get rid of $v$ and $v'$, and to get a (no larger) proof $P_{\textsf{cut}}$ of $u \sim u'$.

One may wish to generalize this, and to prove that we can restrict ourselves to proofs where all intermediate terms are in $R$-normal form. As we saw this is not possible (terms in proof form are not necessarily in $R$-normal form). Therefore we prove a slightly different result. For all basic terms  $C[\vec w, (\alpha_{i})_{i \in I}, (\dec_{j})_{j \in J}]$ and $C'[\vec w, (\alpha'_{i})_{i \in I}, (\dec'_{j})_{j \in J}]$, for all proof:
\[
  \infer[\fas^* \cdot \dup^*]{C[\vec w, (\alpha_{i})_{i \in I}, (\dec_{j})_{j \in J}] \sim C[\vec w, (\alpha'_{i})_{i \in I}, (\dec'_{j})_{j \in J}]}{
    \infer*{}{
      \infer[\cca]{\vec w, (\alpha_{i})_{i \in I}, (\dec_{j})_{j \in J} \sim \vec w, (\alpha'_{i})_{i \in I}, (\dec'_{j})_{j \in J}}{}
    }
  }
\]
we are going to prove that we can assume that there are no redexes in $C$. This shows that we can assume the basic terms $C[\vec w, (\alpha_{i})_{i \in I}, (\dec_{j})_{j \in J}]$ and $C'[\vec w, (\alpha'_{i})_{i \in I}, (\dec'_{j})_{j \in J}]$ to be \emph{normalized} basic terms.

\paragraph{Formal Statement} 
We are going to prove that we can guarantee that $C$ does not contain any redexes and that some further technical properties holds. These properties (that we discuss below) are used to deal with the fact that $\sim$ is not a congruence: they allow to compose applications of the cut-elimination lemma. We start by discussing the properties needed to compose these cut-eliminations, then give the composition proposition and finally we will state the cut-elimination lemma.

Let $\left(C^k[\vec w^k, (\alpha^k_{i})_{i \in I^k}, (\dec^k_{j})_{j \in J^k}]\right)_k$ and $\left(C'^k[\vec w^k, (\alpha'^k_{i})_{i \in I^k}, (\dec'^k_{j})_{j \in J^k}]\right)_k$ be basic terms, and assume that we have the proof:
\begin{equation}
  \label{eq:prop-stich}
  \begin{array}[c]{c}
    \infer[\fas^* \cdot \dup^*]
    {
      \left(C^k[\vec w^k, (\alpha^k_{i})_{i \in I^k}, (\dec^k_{j})_{j \in J^k}]\right)_k
      \sim
      \left(C'^k[\vec w^k, (\alpha'^k_{i})_{i \in I^k}, (\dec'^k_{j})_{j \in J^k}]\right)_k
    }
    {
      \infer*{}{
        \infer[\cca]{
          \left(\vec w^k, (\alpha^k_{i})_{i \in I^k}, (\dec^k_{j})_{j \in J^k} \right)_k
          \sim
          \left(\vec w^k, (\alpha'^k_{i})_{i \in I^k}, (\dec'^k_{j})_{j \in J^k}\right)_k
        }{}
      }
    }
  \end{array}
\end{equation}
Moreover assume that, for all $k$, there exists basic terms $\tilde C^k$, $\tilde{\vec w}^k$ and $\tilde I^k $, $\tilde J^k$ such that we can rewrite the sub-proof of:
\[
  \tilde C^k[\vec w^k, ( \alpha^k_{i})_{i \in I^k}, (\dec^k_{j})_{j \in  J^k}]
  \sim
  \tilde C'^k[\vec w^k, ( \alpha'^k_{i})_{i \in I^k}, ( \dec'^k_{j})_{j \in J^k}]    
\]
into the following proof:
\begin{equation}
  \label{eq:prop-stich2}
  \begin{array}[c]{c}
    \infer[R]
    {
      C^k[\vec w^k, ( \alpha^k_{i})_{i \in I^k}, (\dec^k_{j})_{j \in  J^k}]
      \sim
      C'^k[\vec w^k, ( \alpha'^k_{i})_{i \in I^k}, ( \dec'^k_{j})_{j \in J^k}]    
    }
    {
      \infer[\fas^* \cdot \dup^*]
      {
        \tilde C^k[\tilde{\vec w}^k, ( \alpha^k_{i})_{i \in \tilde I^k}, (\dec^k_{j})_{j \in \tilde J^k}]
        \sim
        \tilde C'^k[\tilde{\vec w}^k, ( \alpha'^k_{i})_{i \in \tilde I^k}, ( \dec'^k_{j})_{j \in \tilde J^k}]
      }
      {
        \infer*{}{
          \infer[\cca]{
            \tilde{\vec w}^k, ( \alpha^k_{i})_{i \in \tilde I^k}, ( \dec^k_{j})_{j \in \tilde J^k}
            \sim
            \tilde{\vec w}^k, ( \alpha'^k_{i})_{i \in \tilde I^k}, ( \dec'^k_{j})_{j \in \tilde J^k}
          }{}
        }
      }
    }
  \end{array}
\end{equation}
Then we can recombine the instances of the \CCA axiom into one instance, as long as we did not introduce new encryptions and new decryptions (i.e.  $\tilde I^k \subseteq I^k$ and $\tilde J^k \subseteq J^k$), and as long as $\tilde{\vec w}^k$ does not contain new encryptions randomness or secret keys etc ... A sufficient condition to that ensure the latter property holds is to require that $\tilde{\vec w}^k \subseteq \st(\vec w^k)$. Putting everything together one get the following proposition:
\begin{proposition}
  \label{prop:stitch}
  For all basic terms $\left(C^k[\vec w^k, (\alpha^k_{i})_{i \in I^k}, (\dec^k_{j})_{j \in J^k}]\right)_k$ and $\left(C'^k[\vec w^k, (\alpha'^k_{i})_{i \in I^k}, (\dec'^k_{j})_{j \in J^k}]\right)_k$ such that the proof displayed in Equation~\eqref{eq:prop-stich} is valid, if for all $k$, there exists basic terms $\tilde C^k$, $\tilde{\vec w}^k$ and $\tilde I^k $, $\tilde J^k$ such that:
  \begin{itemize}
  \item $\st(\tilde{\vec w}^k) \subseteq \st(\vec w^k)$.
  \item $\tilde I^k \subseteq I^k$ and $\tilde J^k \subseteq J^k$
  \item The derivation in \eqref{eq:prop-stich2} is valid.
  \end{itemize}
  Then we have:
  \[
    \infer[R]
    {
      \left(C^k[\vec w^k, (\alpha^k_{i})_{i \in I^k}, (\dec^k_{j})_{j \in J^k}]\right)_k
      \sim
      \left(C'^k[\vec w^k, (\alpha'^k_{i})_{i \in I^k}, (\dec'^k_{j})_{j \in J^k}]\right)_k
    }
    {
      \infer[\fas^* \cdot \dup^*]
      {
        \left(\tilde C^k[\tilde{\vec w}^k, ( \alpha^k_{i})_{i \in \tilde I^k}, (\dec^k_{j})_{j \in \tilde J^k}]\right)_k
        \sim
        \left(\tilde C'^k[\tilde{\vec w}^k, ( \alpha'^k_{i})_{i \in \tilde I^k}, ( \dec'^k_{j})_{j \in \tilde J^k}]\right)_k
      }
      {
        \infer*{}{
          \infer[\cca]{
            \left(\tilde{\vec w}^k, ( \alpha^k_{i})_{i \in \tilde I^k}, ( \dec^k_{j})_{j \in \tilde J^k} \right)_k
            \sim
            \left(\tilde{\vec w}^k, ( \alpha'^k_{i})_{i \in \tilde I^k}, ( \dec'^k_{j})_{j \in \tilde J^k}\right)_k
          }{}
        }
      }
    }
  \]
\end{proposition}

\begin{proof}
  Axioms $\fas$ and $\dup$ verify a kind of frame property. If we have the derivation:
  \[
    \infer[Ax]{\vec u \sim \vec v}{\vec u' \sim \vec v'}
  \]
  then for all $\vec w,\vec w'$ of same length, the following derivation is valid:
  \[
    \infer[Ax]{\vec w, \vec u \sim \vec w', \vec v}{\vec w, \vec u' \sim \vec w', \vec v'}
  \]
  The only problem comes from the \cca axiom, which does not verify the frame property. But thanks to the hypothesis $\st(\tilde{\vec w}^k) \subseteq \st(\vec w^k)$ and $\tilde I^k \subseteq I^k$, $\tilde J^k \subseteq J^k$, we know that the \cca application:
  \[
    \infer[\cca]{
      \tilde{\vec w}^k, (\alpha^k_{i})_{i \in \tilde I^k}, (\dec^k_{j})_{j \in \tilde J^k}
      \sim
      \tilde{\vec w}^k, ( \alpha'^k_{i})_{i \in \tilde I^k}, ( \dec'^k_{j})_{j \in \tilde J^k}
    }{}
  \]
  is ``included'' in the application:
  \[
    \infer[\cca]{
      \left(\vec w^k, (\alpha^k_{i})_{i \in I^k}, (\dec^k_{j})_{j \in J^k} \right)_k
      \sim
      \left(\vec w^k, (\alpha'^k_{i})_{i \in I^k}, (\dec'^k_{j})_{j \in J^k}\right)_k
    }{}
  \]
  Therefore we can combine all proofs, using $\dup$ to remove duplicates, to get the wanted proof.
\end{proof}

\begin{lemma} \label{lem:eagerfaiffree}
  For all basic terms  $C[\vec w, (\alpha_{i})_{i \in I}, (\dec_{j})_{j \in J}]$ and $C'[\vec w, (\alpha'_{i})_{i \in I}, (\dec'_{j})_{j \in J}]$, if we have a derivation:
  \[
    \infer[\fas^* \cdot \dup^*]{C[\vec w, (\alpha_{i})_{i \in I}, (\dec_{j})_{j \in J}] \sim C[\vec w, (\alpha'_{i})_{i \in I}, (\dec'_{j})_{j \in J}]}{
      \infer*{}{
        \infer[\cca]{\vec w, (\alpha_{i})_{i \in I}, (\dec_{j})_{j \in J} \sim \vec w, (\alpha'_{i})_{i \in I}, (\dec'_{j})_{j \in J}}{}
      }
    }
  \]
  then there exists $\tilde C$, $\tilde{\vec w}$ and $\tilde I $, $\tilde J$ such that:
  \begin{itemize}
  \item $\st(\tilde{\vec w}) \subseteq \st(\vec w)$.
  \item $\tilde I \subseteq I$, $\tilde J \subseteq J$
  \item $\tilde C[\tilde{\vec w}, ( \alpha_i)_{i \in \tilde I}, ( \dec_{j})_{j \in \tilde J}]$ and  $\tilde C[\tilde{\vec w}, ( \alpha'_{i})_{i \in \tilde I}, ( \dec'_{j})_{j \in \tilde J}]$ are normalized basic terms.
  \item We have the following derivation:
    \[
      \infer[R]{C[\vec w, (\alpha_{i})_{i \in I}, (\dec_{j})_{j \in J}] \sim C[\vec w, (\alpha'_{i})_{i \in I}, (\dec'_{j})_{j \in J}]}{
        \infer[\fas^* \cdot \dup^*]{\tilde C[\tilde{\vec w}, ( \alpha_i)_{i \in \tilde I}, ( \dec_{j})_{j \in \tilde J}] 
          \sim \tilde C[\tilde{\vec w}, ( \alpha'_{i})_{i \in \tilde I}, ( \dec'_{j})_{j \in \tilde J}]}{
          \infer*{}{
            \infer[\cca]{\tilde{\vec w}, ( \alpha_{i})_{i \in \tilde I}, ( \dec_{j})_{j \in \tilde J} \sim \tilde{\vec w}, ( \alpha'_{i})_{i \in \tilde I}, ( \dec'_{j })_{j \in \tilde J}}{}
          }
        }
      }
    \]
  \end{itemize}
\end{lemma}

\begin{proof}
  We start by observing that if we have a derivation:
  \[
    \infer[\fas^* \cdot \dup^*]{(C[\vec w^k, (\alpha^k_{i})_{i \in I}, (\dec^k_{j})_{j \in J}])_k \sim (C[\vec w^k, (\alpha'^k_{i})_{i \in I}, (\dec'^k_{j})_{j \in J}])_k}{
      \infer*{}{
        \infer[\cca]{(\vec w^k, (\alpha^k_{i})_{i \in I}, (\dec^k_{j})_{j \in J})_k \sim (\vec w^k, (\alpha'^k_{i})_{i \in I}, (\dec'^k_{j})_{j \in J})_k}{}
      }
    }
  \]
  Then we can apply the lemma for each $k$ and recombine the proofs together using Proposition~\ref{prop:stitch}. 

  We prove the lemma by induction on the context $C$. Each time we say we have a shortcut we use Lemma~\ref{lem:restrelim} to get ride of the $\restr$ application introduced by the shortcut.
  \begin{itemize}
  \item Both left and right side can be reduced by $\pi_i(\pair{x_1}{x_2}) \ra x_i$. W.l.o.g we assume $i = 1$, therefore we have:
    \[
      \infer[\fa_{\pi_1}]{
        \begin{array}{lc}
          &\pi_1\left(\pair{C^1[\vec w^1, (\alpha^1_{i})_{i \in I}, (\dec^1_{j})_{j \in J}]}{C^2[\vec w^2, (\alpha^2_{i})_{i \in I}, (\dec^2_{j})_{j \in J}]} \right)\\
          \sim& \pi_1\left(\pair{C^1[\vec w^1, (\alpha'^1_{i})_{i \in I}, (\dec'^1_{j})_{j \in J}]'}{C^2[\vec w^2, (\alpha'^2_{i})_{i \in I}, (\dec'^2_{j})_{j \in J}]}\right)
        \end{array}
      }{
        \begin{array}{lc}
          &\pair{C^1[\vec w^1, (\alpha^1_{i})_{i \in I}, (\dec^1_{j})_{j \in J}]}{C^2[\vec w^2, (\alpha^2_{i})_{i \in I}, (\dec^2_{j})_{j \in J}]} \\
          \sim& \pair{C^1[\vec w^1, (\alpha'^1_{i})_{i \in I}, (\dec'^1_{j})_{j \in J}]'}{C^2[\vec w^2, (\alpha'^2_{i})_{i \in I}, (\dec'^2_{j})_{j \in J}]}
        \end{array}}
    \] 
    By induction hypothesis we have a derivation of the premise in which terms are normalized basic terms. Observe that this implies that the normalized basic terms start with a pair symbol, therefore we have:
    \[
      \infer[R]{
        \begin{array}{lc}
          &\pi_1\left(\pair{C^1[\vec w^1, (\alpha^1_{i})_{i \in I}, (\dec^1_{j})_{j \in J}]}{C^2[\vec w^2, (\alpha^2_{i})_{i \in I}, (\dec^2_{j})_{j \in J}]} \right)\\
          \sim& \pi_1\left(\pair{C^1[\vec w^1, (\alpha'^1_{i})_{i \in I}, (\dec'^1_{j})_{j \in J}]'}{C^2[\vec w^2, (\alpha'^2_{i})_{i \in I}, (\dec'^2_{j})_{j \in J}]}\right)
        \end{array}
      }{
        \infer[\fa_{\pi_1}]{
          \begin{array}{lc}
            &\pi_1\left(\pair{\tilde C^1[\tilde{\vec w}^1, (\tilde{\alpha}^1_{i})_{i \in I}, (\tilde{\dec}^1_{j})_{j \in J}]}{\tilde C^2[\tilde{\vec w}^2, (\tilde{\alpha}^2_{i})_{i \in I}, (\tilde{\dec}^2_{j})_{j \in J}]} \right)\\
            \sim& \pi_1\left(\pair{\tilde C^1[\tilde{\vec w}^1, (\tilde{\alpha}'^1_{i})_{i \in I}, (\tilde{\dec}'^1_{j})_{j \in J}]'}{\tilde C^2[\tilde{\vec w}^2, (\tilde{\alpha}'^2_{i})_{i \in I}, (\tilde{\dec}'^2_{j})_{j \in J}]}\right)
          \end{array}
        }{
          \begin{array}{lc}
            &\pair{\tilde C^1[\tilde{\vec w}^1, (\tilde{\alpha}^1_{i})_{i \in I}, (\tilde{\dec}^1_{j})_{j \in J}]}{\tilde C^2[\tilde{\vec w}^2, (\tilde{\alpha}^2_{i})_{i \in I}, (\tilde{\dec}^2_{j})_{j \in J}]} \\
            \sim& \pair{\tilde C^1[\tilde{\vec w}^1, (\tilde{\alpha}'^1_{i})_{i \in I}, (\tilde{\dec}'^1_{j})_{j \in J}]'}{\tilde C^2[\tilde{\vec w}^2, (\tilde{\alpha}'^2_{i})_{i \in I}, (\tilde{\dec}'^2_{j})_{j \in J}]}
          \end{array}}}
    \]
    We look at the next rule:
    \begin{itemize}
    \item Either it is an is a unitary axioms and both terms are the same, in which case we can construct directly a derivation (by induction over $P$) of:
      \[
        \infer[R]{
          \begin{array}{lc}
            &\pi_1\left(\pair{C^1[\vec w^1, (\alpha^1_{i})_{i \in I}, (\dec^1_{j})_{j \in J}]}{C^2[\vec w^2, (\alpha^2_{i})_{i \in I}, (\dec^2_{j})_{j \in J}]} \right)\\
            \sim& \pi_1\left(\pair{C^1[\vec w^1, (\alpha'^1_{i})_{i \in I}, (\dec'^1_{j})_{j \in J}]'}{C^2[\vec w^2, (\alpha'^2_{i})_{i \in I}, (\dec'^2_{j})_{j \in J}]}\right)
          \end{array}}
        {
          \infer[\fa_{\pi_1}]{
            \tilde C^1[\tilde{\vec w}^1, (\tilde{\alpha}^1_{i})_{i \in I}, (\tilde{\dec}^1_{j})_{j \in J}]
            \sim\tilde C^1[\tilde{\vec w}^1, (\tilde{\alpha}'^1_{i})_{i \in I}, (\tilde{\dec}'^1_{j})_{j \in J}]'
          }{
            \begin{array}{c}
              \tilde C^1[\tilde{\vec w}^1, (\tilde{\alpha}^1_{i})_{i \in I}, (\tilde{\dec}^1_{j})_{j \in J}]
              \sim \tilde C^1[\tilde{\vec w}^1, (\tilde{\alpha}'^1_{i})_{i \in I}, (\tilde{\dec}'^1_{j})_{j \in J}]'
            \end{array}
          }
        }
      \]
    \item Or it is a function application: it can only be a function application on the pair symbol, hence we have a shortcut.
    \end{itemize}

  \item Only one side can be reduced by $\pi_i(\pair{x_1}{x_2}) \ra x_i$. This is impossible since, at all positions in the proof, corresponding terms start with the same function symbol. Absurd.

  \item Both sides can be reduced by $\dec(\enc{x}{\pk(n)}{\nonce_r},\sk(n)) \ra x$: 
    \[ \infer{\dec(\enc{u}{\pk(n)}{r},\sk(n)) \sim \dec(\enc{u'}{\pk(n')}{r'},\sk(n'))}
      {\enc{u}{\pk(n)}{r},\sk(n) \sim \enc{u'}{\pk(n')}{r'},\sk(n')} \]
    Using the induction hypothesis we know that we have a derivation of $\enc{u}{\pk(n)}{r},\sk(n) \sim \enc{u'}{\pk(n')}{r'},\sk(n')$ where intermediate terms are normalized basic conditionals. We look at the next rule applied on $\enc{u}{\pk(n)}{r},\_ \sim \enc{u'}{\pk(n')}{r'},\_$ which is not $\dup$. If it is a function application then we have a shortcuts, if it is a unitary axioms then we have two cases:
    \begin{itemize}
    \item $\enc{u}{\pk(n)}{r}$ is a renaming of $\enc{u'}{\pk(n')}{r'}$, in which case we can build by induction a proof of $u \sim u'$ whose intermediate terms are normalized basic conditionals.
    \item $\enc{u}{\pk(n)}{r}$ is a not renaming of $\enc{u'}{\pk(n')}{r'}$, in which case the IND-CCA2 axiom is used. This means that at the root of the proof tree we know that $\sk$ appears only in decryption position. By induction we show that this is not the case. Absurd.
    \end{itemize}

  \item Only one side can be reduced by $\dec(\enc{x}{\pk(n)}{r},\sk(n)) \ra x$. Observe that it is necessarily of the form:
    \[ 
      \infer{\dec(\enc{t}{\pk(n)}{r},\sk(n)) \sim \dec(\enc{t'}{p'}{r'},\sk'(n'))}
      {\enc{t}{\pk(n)}{r},\sk(n) \sim \enc{t'}{p'}{r'},\sk'(n')} 
    \]
    We look at the next rule applied to $\enc{t}{\pk(n)}{r}$ and $\enc{t'}{p'}{r'}$ it which is not $\dup$:
    \begin{itemize}
    \item If it is a unitary axiom, then necessarily $p' \equiv \pk(n)$ and $n' \equiv n$. Therefore the right side can be reduced by $\dec(\enc{x}{\pk(n)}{r},\sk(n)) \ra x$. Absurd.
    \item If it is $\fa_{\enc{}{}{}}$ then there is a proof of $\pk(n),\sk(n) \sim p',\sk(n')$, which implies that $p' \equiv \pk(n)$ and $n' \equiv n$. Therefore the right side can be reduced by $\dec(\enc{x}{\pk(n)}{r},\sk(n)) \ra x$. Absurd.
    \end{itemize}

  \item Both side can be reduced by $\eq{x}{x} \ra \true$. In this case the cut is trivial.

  \item Only one side can be reduced by $\eq{x}{x} \ra \true$. Therefore we have a proof of the form:
    \[
      \infer[\fa_{\eq{}{}}]{\eq{t}{t} \sim \eq{t'}{t''}}
      {t,t \sim t',t''}
    \]
    Using Lemma~\ref{lem:cond-equiv-body} we know that $t' \equiv t''$, therefore both side can be reduced by $\eq{x}{x} \ra \true$. Absurd.
    \qedhere
  \end{itemize}
\end{proof}

%%% Local Variables:
%%% mode: latex
%%% TeX-master: "ms"
%%% End:

\subsection{Restriction to Proofs in Normalized Proof Form}
% \label{subsection:proof-form-restr-lemma}

\begin{definition}
  We let $\bcca$ be the restriction of \cca to cases $\vec w, (\alpha_i)_i, (\dec_j)_j \sim \vec w', (\alpha'_i)_i, (\dec'_j)_j$ where:
  \begin{itemize}
  \item $(\alpha_j)_j, (\alpha'_j)_j$ are encryption oracle calls.
  \item $(\dec_j)_j, (\dec'_j)_j$ are decryption oracle calls.
  \end{itemize}
\end{definition}

\begin{lemma}
  \label{lem:complete-strat-pre}
  The following strategy is complete for $\mathfrak{F}((\csm + \fa + R + \dup + \cca)^*)$:
  \begin{equation*}
    \mathfrak{F}((\twobox + \rs)^*  \cdot \csmb^* \cdot \{\obfa(b,b')\}^* \cdot \unbox \cdot \fas^* \cdot \dup^* \cdot \bcca )
  \end{equation*}

\end{lemma}

\begin{proof}
  We showed in Lemma~\ref{lem:body-freeze} that the following strategy:
  \[
    \mathfrak{F}((\twobox + \rs)^*  \cdot \csmb^* \cdot \{\obfa(b,b')\}^* \cdot \unbox \cdot \fas^* \cdot \dup^* \cdot \cca)
  \]
  is complete for $\csm + \fa + R + \dup + \cca$. Let $P$ be a proof of $t \sim t'$ in this fragment. Let $L^P = \prooflabel(P)$ the set of indices of the branch of the proof tree. Recall that $\encs_l^P,\decs_l^P$ and $\keys_l^P$ are the sets of, respectively, encryptions, decryptions and keys used in the \cca instance of branch $l$, and that $\ekl^P = (\keys_l^P,\rands_l^P,\encs_l^P,\decs_l^P)$. We define the order $<$ as follows: for all $u,u' \in \encs_l^P \cup \decs_l^P$, we let $u < u'$ hold if $u$ is a strict subterm of $u'$. 

  We show that for all proof $P$ of $t \sim t'$ in the above fragment, there is a proof $Q$ of $t \sim t'$ where for all $l \in \prooflabel(Q)$, all $u \in \encs_l^Q$ are $\ekl^Q$-encryption oracle calls, and all $u \in \decs_l^Q$ are $\ekl^Q$-decryption oracle calls (the same holds for $\encs_l'^Q$ and $\decs_l'^Q$). We prove this by induction on the number of elements of $\bigcup_{l} \encs_l^P \cup \decs_l^P$ that are not $\ekl^P$-encryption oracle calls or $\ekl^P$-decryption oracle calls, plus the number of elements of $\bigcup_{l} \encs_l'^P \cup \decs_l'^P$ that are not $\ekl'^P$-encryption oracle calls or $\ekl'^P$-decryption oracle calls.

  Let $P$ be a proof of $t \sim t'$, $l \in L^P$ and let $u$ maximal for $<$ which is not a $\ekl^P$-encryption oracle call or a $\ekl^P$-decryption oracle call.  
  \begin{itemize}
  \item If $u \in \encs_l^P$ is an encryption. We know that $u \equiv \enc{m}{\pk}{\nonce_r}$ where the corresponding secret key $\sk$ is in $\keys_l^P$. Let $(\alpha_k)_k$ be the set of elements of $\encs_l^P \cap \st(m)$, and let $(\dec_n)_n$ be the set of elements of $\decs_l^P \cap \st(m)$. We know that there exists a context $C$ such that:
    \[
      m \equiv C[(\alpha_k)_k,(\dec_n)_n]
    \]
    We let $A$ be an if-context and $(B_i[])_i$, $(U_m[])_m$ be if-free contexts in $R$-normal form such that $C[] =_R A[(B_i[])_i \diamond (U_m[])_m]$. Let $m_0$ be the term:
    \[
      m_0 \equiv A[(B_i[(\alpha_k)_k,(\dec_n)_n])_i \diamond (U_m[(\alpha_k)_k,(\dec_n)_n])_m]
    \]
    We know that $m_0 =_R m$. By maximality of $u$ we know that the $(\alpha_k)_k$ are $\ekl^P$-encryption oracle calls, and the $(\dec_n)_n$ are $\ekl^P$-decryption oracle calls. For all $k$ we know that $\alpha_k \equiv \enc{\_}{\pk_k}{\nonce_k}$, and for all $l$ let $\sk_n$ be the secret key of $\dec_n$. Assume that there is some $i$ such that:
    \[
      \tilde m \equiv B_i[(\enc{[]_k}{\pk_k}{\nonce_k})_k,(\dec([]_n,\sk_n))_n]
    \]
    is not in $R$-normal form. Since $B_i[]$ is in $R$-normal form, this means that there exists some $k$ such that $\dec(\enc{[]_k}{\pk_k}{\nonce_k},\sk_k)$ is a subterm of $\tilde m$. This implies that $\sk_k$ is a subterm of $B_i[]$. But $\sk_k \in \keys_l^P$, and therefore $B_i$ cannot contain $\sk_k$ as a subterm. Absurd. The same me reasoning applies to $U_m[(\alpha_k)_k,(\dec_n)_n]$.

    Therefore for all $k$ (resp. for all $m$), $B_i[(\alpha_k)_k,(\dec_n)_n]$ (resp. $U_m[(\alpha_k)_k,(\dec_n)_n]$) is an $\ekl^P$-normalized basic term. Hence $m_0$ is a $\ekl^P$-normalized simple term. We then rewrite, using $R$, all occurrences of $\enc{m}{\sk}{\nonce_r}$ by $\enc{m'}{\sk}{\nonce_r}$ in branch $l$, i.e in every:
    \[
      D^{{\sfh}}_l
      \left[
        \left(B^{{\sfh}}_{i,l}[\vec w_{i,l}^{{\sfh}},(\alpha_{i,l}^{{{\sfh}},j})_j,(\dec_{i,l}^{{{\sfh}},k})_k]\right)_i
        \diamond
        \left(U^{{\sfh}}_{m,l}[\vec w_{m,l}^{{\sfh}},(\alpha_{m,l}^{{{\sfh}},j})_j,(\dec_{m,l}^{{{\sfh}},k})_k]\right)_m
      \right]
    \]
    with $\sfh \in \hbranch(l)$. We can check that this yields a new proof $Q$ of $t \sim t'$ with a smaller number of terms in $\encs_l^Q \cup \decs_l^Q$ which are not $\ekl^Q$-encryption oracle calls or $\ekl^Q$-decryption oracle calls: the only difficulty lies in making sure that the side-conditions of the decryptions still holds. This is the case, for example look at one of the conditions under which a encryption $\alpha_0 \equiv \enc{m_0}{\pk}{\nonce_0}$ must be guarded in $\dec(u_0,\sk)$: we require that $\nonce_0 \in \st(u_0\downarrow_R)$, which is indeed stable under any $R$ rewriting (hence in particular the rewriting of $\enc{m}{\sk}{\nonce_r}$ into $\enc{m'}{\sk}{\nonce_r}$).

    Since all other branches $l' \in L_P \backslash \{l\}$ are left unchanged, and since the right part of the proof (corresponding to $t'$) is also left unchanged we can conclude using the induction hypothesis.
    
  \item One can easily check that the case where $u \equiv C[(g_e)_e \diamond (s_a)_{a \le p}] \in \decs_l^P$ is a decryption cannot happen.
    \qedhere
  \end{itemize}
\end{proof}

We are now ready to give the proof of Lemma~\ref{lem:body-proof-form}, which we recall below.
\begin{lemma*}[\ref{lem:body-proof-form}]
  The restriction of the fragment $\mathcal{A}_{\succ}$ to formulas provable in $\npfproof$ is complete for $\mathfrak{F}((\csm + \fa + R + \dup + \CCA)^*)$.
\end{lemma*}

\begin{proof}
  Using Lemma~\ref{lem:complete-strat-pre} we know that the strategy:
  \begin{equation*}
    \mathfrak{F}((\twobox + \rs)^* \cdot \csmb^* \cdot \{\obfa(b,b')\}^* \cdot \unbox \cdot \fas^* \cdot \dup^* \cdot \bcca )
  \end{equation*}
  is complete for $\mathfrak{F}((\csm + \fa + R + \dup + \bcca)^*)$. 

  First we show that this strategy remains complete even if with restrict it to proofs such that the terms after $(\twobox + \rs)^*$ are in proof form. Let $P$ be such that $P \vdash_{\mathcal{A}_{\csmb}} t \sim t'$, we want to find $t_0 =_R t, t_0' =_R t'$ and $P'$ such that $P' \lproof_{\mathcal{A}_{\csmb}} t \sim t'$ is in proof form.

  Assume that $P_0 \vdash_{\mathcal{A}_{\csmb}} t \sim t'$, using Proposition~\ref{prop:epf-labelling} we know that there exists $P$ such that $P \lproof_{\mathcal{A}_{\csmb}} t \sim t'$.  Let $h \in \setindex(P), \sfx \in \{\sfl,\sfr\}$, $\sfh = h_{\sfx}$. We know that there exists $b^\sfh,b'^\sfh$ such that $\extractx(h,P) \lproof_{\mathcal{A}_{\csmb}} b^\sfh \sim b'^\sfh$.  To get a proof with terms in proof form, we need to show that for all $\sfh,l$, for all $(\beta,\beta') (\lesimcond^{\sfh,l} \cup \lesimleave^{\sfh,l}) (t \sim t',P)$, $\beta,\beta'$ are of the form:
  \[
    \beta \equiv B[\vec w,(\alpha_j)_j,(\dec_k)_k]
    \quad \wedge \quad
    \beta' \equiv B[\vec w,(\alpha'_j)_j,(\dec'_k)_k]
  \]
  the contexts $B$ is if-free. Assume that this is not the case. Then there exists contexts $B_e,B_c,B_0,B_1$ such that:
  \[
    B \equiv B_e[\ite{B_c}{B_0}{B_1}]] =_R \ite{B_c}{B_e[B_0]}{B_e[B_1]}
  \]
  Let $t_0$ be the term obtained from $t$ by replacing this occurrence of $\beta$ by:
  \[
    \ite{B_c[\vec w,(\alpha_j)_j,(\dec_k)_k]}{(B_e[B_0])[\vec w,(\alpha_j)_j,(\dec_k)_k]}{(B_e[B_1])[\vec w,(\alpha_j)_j,(\dec_k)_k]}
  \]
  Similarly we define $t_0'$ by replacing $\beta'$ by the corresponding term. Then $t_0 =_R t$ and $t_0' =_R t'$. Moreover it is easy to check that the formula $t_0 \sim t'_0$ is provable in $\lproof_{\mathcal{A}_{\csmb}}$, as we replaced one $\obfa$ application by three $\obfa$ applications (without changing the encryptions, decryptions or branches of the proof etc ...). 

  Moreover we replaced $B$ by three terms $B_c,B_e[B_0],B_e[B_1]$ containing strictly less $\ite{}{}{}$ applications. Therefore we can show by induction that we can ensure that all the contexts $B$ are if-free by repeating the proof rewriting above.

  To show that there a proof of $t \sim t'$ such that the terms after $(\twobox + \rs)^*$ are in \emph{normalized} proof form, we only have to apply the Lemma~\ref{lem:eagerfaiffree} to all branches $l$, and to commute the new $R$ rewriting to the bottom of the proof.
\end{proof}

%%% Local Variables:
%%% mode: latex
%%% TeX-master: "ms"
%%% End:

\newpage

\section{Restrictions on the Basic Conditionals Part}
\label{app-section:prop-basic-terms}
In this section, we give the proof of Proposition~\ref{prop:bas-cond-charac-body}.

\subsection{Properties of Normalized Basic Terms}
% \label{subsection:prop-norm-bt}

\begin{definition}
  We call a \emph{conditional context} a context $C[]_{\vec x}$ such that all holes appear in the conditional part of an $\ite{}{}{}$. Formally, for all position $p$, if $C_{|p}$ is a hole $[]_x$ then $p = p'.0$ and there exists $u,v$ such that:
  \[
    C_{|p'} \equiv \ite{[]_x}{u}{v}
  \]
  We say that $u$ is an \emph{almost conditional context} if $u$ a conditional context or a hole.
\end{definition}

The main goal of this subsection is to show the following lemma.
\begin{lemma}
  \label{lem:bas-cond-restr-spurious2}
  For all $P \npfproof t \sim t'$, for all $\sfh,l$ and $\beta,\beta' \lebt^{\sfh,l} (t,P)$, there exists an almost conditional context $\tilde \beta'[]$ such that:
  \[
    \beta' \equiv \tilde \beta'\left[\beta\right]
    \quad \wedge \quad
    \leavest(\beta\downarrow_R) \cap \condst\left(\tilde \beta'[] \downarrow_R \right) = \emptyset
  \]
\end{lemma}

Before delving in the proof, we would like to remark that the above lemma is not entirely satisfactory. Consider the following example:
\begin{mathpar}
  \beta_{0} \equiv \eq{\enc{\ite{b}{s}{t}}{\pk(\nonce)}{\nonce_r}}{0}
  =_R
  \ite{b}
  {\underbrace{\eq{\enc{s}{\pk(\nonce)}{\nonce_r}}{0}}_{\beta_0^0}}
  {\underbrace{\eq{\enc{t}{\pk(\nonce)}{\nonce_r}}{0}}_{\beta_0^1}}

  \beta_{1} \equiv \eq{\enc{\ite{\beta_{0}^0}{u}{u}}{\pk(\nonce)}{\nonce'_r}}{0}
\end{mathpar}
where $\beta_{0}^0,\beta_{0}^1 \not \in \condst(u\downarrow_R)$ and $s \ne_R t$. Then $\beta_0^0,\beta_{0}^1 \not \in \condst(\beta_1\downarrow_R)$, because $\beta_0^0$ disappear using the rule $\ite{x}{y}{y} \ra y$ in $R$. Hence, Lemma~\ref{lem:bas-cond-restr-spurious2} could choose $\tilde \beta_1 \equiv \beta_1$. Of course this situation cannot occur, as we cannot have $\beta_0^0$ be a subterm of $\beta_1$ (this contradicts the freshness side-condition of encryptions' randomnesses in the $\cca$ axiom). But we cannot rule this situation out simply by applying the lemma, we have to make a more in-depth analysis. We would like to a stronger version of this lemma that somehow directly ``includes'' the above observation.

To do this we introduce over-approximations of $\condst(\cdot\downarrow_R)$ and $\leavest(\cdot\downarrow_R)$, show that Lemma~\ref{lem:bas-cond-restr-spurious2} holds for the over-approximations of $\condst(\cdot\downarrow_R)$ and $\leavest(\cdot\downarrow_R)$.

\begin{definition}
  We define the function $\aleavest$ from the set of terms to the set of if-free terms in $R$-normal form:
  \begin{mathpar}
    \aleavest(u_0,\dots,u_n) = \cup_{ i \le n} \aleavest(u_i)

    \aleavest(\ite{b}{u}{v}) = \aleavest(u,v)

    \aleavest(f(u_0,\dots,u_n)) = \left\{f(v_0,\dots,v_n)\downarrow_R \mid \forall i \le n, v_i \in \aleavest(u_i)  \right\} \quad (\forall f \in \ssig \cup \Nonce)
  \end{mathpar}
  We define the function $\acondst$ from the set of terms to the set of if-free conditionals in $R$-normal form:
  \begin{mathpar}
    \acondst(u_0,\dots,u_n) = \cup_{ i \le n} \acondst(u_i)

    \acondst(f(\vec u)) = \acondst(\vec u) \quad (\forall f \in \ssig \cup \Nonce)

    \acondst(\ite{b}{u}{v}) = \acondst(b) \cup \aleavest(b) \cup \acondst(u,v)
  \end{mathpar}
\end{definition}

\begin{remark}
  The over-approximation is two-fold: for $\aleavest()$ there is a first over-approximation, and for $\acondst()$ there is an over-approximation, plus the over-approximation of $\aleavest()$ % \toadd{examples}
  .
\end{remark}

\begin{proposition}
  \label{prop:acondst-overapprox}
  $\aleavest$ and $\acondst$ are sound over-approximations:
  \begin{itemize}
  \item For all $u \ra_R^* u'$, $\aleavest(u) \supseteq \aleavest(u')$. Moreover $\aleavest(u\downarrow_R) = \leavest(u\downarrow_R)$.
  \item For all $u \ra_R^* u'$, $\acondst(u) \supseteq \acondst(u')$. Moreover $\acondst(u\downarrow_R) = \condst(u\downarrow_R)$.
  \end{itemize}
\end{proposition}

\begin{proof}
  The facts that $\aleavest(u\downarrow_R) = \leavest(u\downarrow_R)$ and  $\acondst(u\downarrow_R) = \condst(u\downarrow_R)$ are straightforward to show. Let us prove by induction on $\ra_R^*$ that for all $u \ra_R^* u'$, $\aleavest(u) \supseteq \aleavest(u')$. If $u \equiv u'$ this is immediate, assume that $u \ra_R v \ra_R^* u'$. By induction hypothesis we know that $\aleavest(v) \supseteq \aleavest(u')$. We then have a case disjunction (we omit the redundant or obvious cases):
  \begin{itemize}
  \item $u \equiv \ite{b}{(\ite{b}{s}{t})}{w}$ and $v \equiv  \ite{b}{s}{w}$ then:
    \begin{alignat*}{2}
      \aleavest(u) &\;=\;&& \aleavest(s) \cup \aleavest(t) \cup \aleavest(w) \\
      &\;\supseteq\;&& \aleavest(s) \cup \aleavest(w) = \aleavest(v) \\
      &\;\supseteq\;&& \aleavest(u')
    \end{alignat*}
  \item $u \equiv \ite{b}{s}{s}$ and $v \equiv s$ then:
    \begin{equation*}
      \aleavest(u) = \aleavest(s) = \aleavest(v)
    \end{equation*}
  \item $u \equiv \ite{(\ite{b}{a}{c})}{s}{t}$ and $v \equiv  \ite{b}{(\ite{a}{s}{t})}{(\ite{c}{s}{t})}$ then:
    \begin{equation*}
      \aleavest(u) = \aleavest(s) \cup\aleavest(t) = \aleavest(v)
    \end{equation*}
  \item $u \equiv \ite{b}{(\ite{a}{s}{t})}{w}$ and $v \equiv  \ite{a}{(\ite{b}{s}{w})}{(\ite{b}{t}{w})}$ then:
    \begin{equation*}
      \aleavest(u) = \aleavest(s) \cup\aleavest(t) \cup \aleavest(w) = \aleavest(v)
    \end{equation*}
  \item $u \equiv f(\vec w,\ite{b}{\vec s}{\vec t})$ and $v \equiv \ite{b}{f(\vec w,\vec s)}{f(\vec w,\vec t)}$ then:
    \begin{alignat*}{3}
      \aleavest(u) &\;=\;&&\{ f(\vec w',\vec w'')\downarrow_R \mid \forall i, w'_i \in \aleavest(w_i) \wedge \forall j, w''_j \in \aleavest(s_j) \cup \aleavest(t_j) \}\\
      &\;\supseteq\;& &\{ f(\vec w',\vec w'')\downarrow_R \mid \forall i, w'_i \in \aleavest(w_i) \wedge \forall j, w''_j \in \aleavest(s_j) \}\\
      &&\;\cup\;&\{ f(\vec w',\vec w'')\downarrow_R \mid \forall i, w'_i \in \aleavest(w_i) \wedge \forall j, w''_j \in \aleavest(t_j) \}\\
      &\;\supseteq\;& & \aleavest(f(\vec w,\vec s)) \cup \aleavest(f(\vec w,\vec t))\\
      &\;\supseteq\;& & \aleavest(v)
    \end{alignat*}
  \item ($u \equiv \pi_i(\pair{s_1}{s_2})$, $v \equiv s_i$) and ($u \equiv \dec(\enc{m}{\pk(\nonce)}{\nonce_r},\sk(\nonce))$, $v \equiv m$) are trivial.
  \end{itemize}

  \paragraph{} Similarly, we show by induction on $\ra_R^*$ that for all $u \ra_R^* u'$, $\acondst(u) \supseteq \acondst(u')$. If $u \equiv u'$ this is immediate, assume that $u \ra_R v \ra_R^* u'$. By induction hypothesis we know that $\aleavest(v) \supseteq \aleavest(u')$. We then have a case disjunction (we omit the redundant or obvious cases):
  \begin{itemize}
  \item $u \equiv \ite{b}{(\ite{b}{s}{t})}{w}$ and $v \equiv  \ite{b}{s}{w}$ then:
    \begin{alignat*}{2}
      \acondst(u) &\;=\;&& \acondst(s,t,w) \cup \acondst(b) \cup \aleavest(b) \\
      &\;\supseteq\;&& \acondst(s,w) \cup \acondst(b) \cup \aleavest(b) \\
      &\;\supseteq\;&& \acondst(v)
    \end{alignat*}
  \item ($u \equiv \ite{b}{(\ite{a}{s}{t})}{w}$, $v \equiv  \ite{a}{(\ite{b}{s}{w})}{(\ite{b}{t}{w})}$)
    and ($u \equiv \ite{b}{s}{s}$, $v \equiv s$) are simple.
  \item $u \equiv \ite{(\ite{b}{a}{c})}{s}{t}$ and $v \equiv  \ite{b}{(\ite{a}{s}{t})}{(\ite{c}{s}{t})}$ then:
    \begin{equation*}
      \acondst(u) = \acondst(b,a,c,s,t) \cup\aleavest(b,a,c) = \acondst(v)
    \end{equation*}
  \item $u \equiv f(\vec w,\ite{b}{\vec s}{\vec t})$ and $v \equiv \ite{b}{f(\vec w,\vec s)}{f(\vec w,\vec t)}$ then:
    \begin{equation*}
      \acondst(u) = \acondst(b,\vec w,\vec s,\vec t) \cup\aleavest(b) = \acondst(v)
    \end{equation*}
  \item ($u \equiv \pi_i(\pair{s_1}{s_2})$, $v \equiv s_i$) and ($u \equiv \dec(\enc{m}{\pk(\nonce)}{\nonce_r},\sk(\nonce))$, $v \equiv m$) are trivial.
    \qedhere
  \end{itemize}
\end{proof}

Let us show the following helpful propositions:
\begin{proposition}
  \label{prop:base-cond-restr-dec}
  For all $\ekl$-normalized basic terms $\beta,\beta' $ if:
  \[
    \aleavest(\beta) \cap \aleavest(\beta') \ne \emptyset
  \]
  then we have $\ekl$-normalized basic terms $B[\vec w,(\alpha^{j})_j,(\delta^{k})_k], B[\vec w,(\alpha'^{j})_j,(\delta'^{k})_k]$ such that:
  \begin{gather*}
    \beta \equiv B[\vec w,(\alpha^{j})_j,(\delta^{k})_k] \quad \wedge \quad \beta' \equiv B[\vec w,(\alpha'^{j})_j,(\delta'^{k})_k]\\
    \forall j, \;\aleavest(\alpha^{j}) \cap \aleavest(\alpha'^{j}) \ne \emptyset
    \quad \wedge \quad
    \forall k, \;\aleavest(\delta^{k}) \cap \aleavest(\delta'^{k}) \ne \emptyset
  \end{gather*}

\end{proposition}

\begin{proof}
  We have $\ekl$-normalized basic terms $B[\vec w,(\alpha^{j})_j,(\delta^{k})_k], B'[\vec w',(\alpha'^{j})_j,(\delta'^{k})_k]$ such that:
  \begin{gather*}
    \beta \equiv B[\vec w,(\alpha^{j})_j,(\delta^{k})_k] \quad \wedge \quad \beta' \equiv B'[\vec w',(\alpha'^{j})_j,(\delta'^{k})_k]
  \end{gather*}
  Since $\beta,\beta'$ are $\ekl$-normalized basic terms, we know that:
  \begin{gather*}
    B[\vec w,(\enc{0}{\_}{\_})_j,(\dec(0,\_))_k] \quad \wedge \quad B'[\vec w',(\enc{[]_j}{\_}{\_})_j,(\dec([]_j,\_))_k]
  \end{gather*}
  are in $R$-normal form, and that $B,B',\vec w, \vec w'$ are if-free. Hence:
  \begin{gather*}
    \aleavest(\beta) = \left\{B[\vec w,(a^{j})_j,(d^{k})_k] \mid \forall j,a^{j} \in \aleavest(\alpha^j) \;\wedge\; \forall k,d^{k} \in \aleavest(\delta^k)\right\}\\
    \aleavest(\beta') = \left\{B'[\vec w',(a'^{j})_j,(d'^{k})_k] \mid \forall j,a'^{j} \in \aleavest(\alpha'^j) \;\wedge\; \forall k,d'^{k} \in \aleavest(\delta'^k)\right\}
  \end{gather*}

  Similarly to what we did in the proof of Lemma~\ref{lem:cond-equiv-body}, we prove that we can assume that $B \equiv  B'$ by induction on the number of hole positions in $B$ or $B'$ such that $(B)_{|p}$ differs from $(B')_{|p}$ (modulo hole renaming). Knowing that $B \equiv B'$, it is then straightforward to show that:
  \[
    \vec w \equiv \vec w'
    \quad \wedge \quad
    \forall j, \;\aleavest(\alpha^{j}) \cap \aleavest(\alpha'^{j}) \ne \emptyset
    \quad \wedge \quad
    \forall k, \;\aleavest(\delta^{k}) \cap \aleavest(\delta'^{k}) \ne \emptyset
  \]

  The base case is trivial, let us prove the inductive case. We let $p$ be the position of a hole in $B$ such that $p$ is a valid position in $B'$, but not a hole (if $p$ is not valid in $B'$, invert $B$ and $B'$). Let $B[\vec w,(a^{j})_j,(d^{k})_k]$ and $B'[\vec w',(a'^{j})_j,(d'^{k})_k]$ be such that:
  \[
    \forall j,k.\, a^{j} \in \aleavest(\alpha^j)\wedge d^{k} \in \aleavest(\delta^k)
    \quad \wedge \quad
    \forall j,k.\, a'^{j} \in \aleavest(\alpha'^j) \wedge d'^{k} \in \aleavest(\delta'^k)
  \]
  and:
  \[
    B[\vec w,(a^{j})_j,(d^{k})_k] \equiv B'[\vec w',(a'^{j})_j,(d'^{k})_k] \in \aleavest(\beta') \cap \aleavest(\beta)
  \]
  We then have three cases depending on $\beta_{|p}$:
  \begin{itemize}
  \item $B$ contains a hole $[]_x$ at position $p$ such that $\beta_{|p} \in \vec w$. Then let $\tilde B'$ be the context $B'$ in which we replaced the term at position $p$ by $[]_y$ (where $y$ is a fresh hole variable) and let $\tilde{\vec w}'$ be the terms $\vec w'$ extended by $\beta_{|p}$ (binded to $[]_y$). Then $B$ differs $\tilde B'$ on a smaller number of hole position, therefore we can conclude by induction hypothesis.
  \item $B$ contains a hole $[]_x$ at position $p$ such that $\beta_{|p}$ is an encryption oracle call $\enc{m}{\pk(\nonce_p)}{\nonce_r}$. Since $\enc{m}{\pk(\nonce_p)}{\nonce_r} \in \encs_l$ is an encryption in an instance of a \cca application, we know from the freshness side-condition that $\nonce_r$ does not appear in $\vec w$ and that $\nonce_r \in \rands_l$.

    Moreover since $\beta'$ is a $\ekl$-normalized basic term, we know that $\fresh{\rands_l}{\vec w'}$. But since $p$ is a valid non-hole position in $B'$, we have $\nonce_r \in \vec w'$. Absurd.
  \item Similarly if  $B$ contains a hole $[]_x$ at position $p$  such that $\beta_{|p}$ is a decryption oracle call \( \dec(m,\sk(\nonce)) \).
    Since $\dec(m,\sk(\nonce))$ is a decryption oracle call we know that $\sk(\nonce) \in\ \keys_l$. Moreover since $\beta'$ is a $\ekl$-normalized basic term, we know that $\nodec(\keys_l,\vec w')$. But since $p$ is a valid non-hole position in $B'$, we know that either $\sk(\nonce) \in \vec w'$ or $\nonce \in \vec w'$. Absurd.
    \qedhere
  \end{itemize}
\end{proof}

We can now state the following proposition, which subsumes Proposition~\ref{prop:bas-cond-charac-body}.

\begin{proposition}
  \label{prop:bas-cond-charac}
  For all $\ekl$-normalized basic terms $\beta,\beta'$, if:
  \[
    \aleavest(\beta) \cap \aleavest(\beta') \ne \emptyset
  \]
  then $\beta \equiv \beta'$.
\end{proposition}

\begin{proof}
  We show this by induction on $|\beta| + |\beta'|$. Using Proposition~\ref{prop:base-cond-restr-dec} we know that we have $\ekl$-normalized basic terms $B[\vec w,(\alpha^{j})_j,(\delta^{k})_k], B[\vec w,(\alpha'^{j})_j,(\delta'^{k})_k]$ such that:
  \begin{gather*}
    \beta \equiv B[\vec w,(\alpha^{j})_j,(\delta^{k})_k] \quad \wedge \quad \beta' \equiv B[\vec w,(\alpha'^{j})_j,(\delta'^{k})_k]\\
    \forall j, \;\aleavest(\alpha^{j}) \cap \aleavest(\alpha'^{j}) \ne \emptyset
    \quad \wedge \quad
    \forall k, \;\aleavest(\delta^{k}) \cap \aleavest(\delta'^{k}) \ne \emptyset
  \end{gather*}
  To conclude we only need to show that for all $j$, $\aleavest(\alpha^{j}) \cap \aleavest(\alpha'^{j}) \ne \emptyset$ implies that $\alpha^j \equiv \alpha'^j$ and that $\aleavest(\delta^{k}) \cap \aleavest(\delta'^{k}) \ne \emptyset$ implies that $\delta^k \equiv \delta'^k$. The former is immediate, as $\aleavest(\alpha^{j}) \cap \aleavest(\alpha'^{j}) \ne \emptyset$ implies that $\alpha^j \equiv \enc{m}{\pk(\nonce)}{\nonce_r}$ and $\alpha'^j \equiv \enc{m'}{\pk(\nonce)}{\nonce_r}$. Since $\alpha^j,\alpha'^j \in \encs_l$ and since there is \emph{as most one} $\ekl$-encryption oracle call with the same randomness, we have $m \equiv m'$. It only remains to show that for all $k$, $\delta^k \equiv \delta'^k$. Since $\delta^k$, $\delta'^k$ are $\ekl$-decryption oracle calls we know that
  \[
    \delta^k \equiv C\left[\vec g \diamond (s_i)_{i \le p}\right]
    \qquad\wedge\qquad
    \delta'^k \equiv C'\left[\vec g' \diamond (s'_i)_{i \le p'}\right]
  \]
  where:
  \begin{itemize}
  \item There exists contexts $u,u'$, if-free and in $R$-normal form, such that:
    \begin{mathpar}
      \forall i < p,\,
      s_i \equiv \zero(\dec(u[(\alpha_{j})_j,(\dec_{k})_k],\sk))

      s_p \equiv \dec(u[(\alpha_{j})_j,(\dec_{k})_k],\sk)

      \forall g \in \vec g,\,
      g \equiv \eq{u[(\alpha_{j})_j,(\dec_{k})_k]}{\alpha_j}
    \end{mathpar}
    \begin{mathpar}
      \forall i < p',\,
      s'_i \equiv \zero(\dec(u'[(\alpha'_{j})_j,(\dec'_{k})_k],\sk'))

      \!\!
      s'_p \equiv \dec(u'[(\alpha'_{j})_j,(\dec'_{k})_k],\sk')

      \!\!
      \forall g \in \pvec{g}',\,
      g \equiv \eq{u'[(\alpha'_{j})_j,(\dec'_{k})_k]}{\alpha'_j}
    \end{mathpar}
  \item $(\alpha_{j})_j,(\alpha'_{j})_j$ are $\ekl$-\emph{encryption oracle calls}.
  \item $(\dec_{k})_k,(\dec'_{k})_k$ are $\ekl$-\emph{decryption oracle call}.
  \end{itemize}
  Since $\aleavest(\delta^{k}) \cap \aleavest(\delta'^{k}) \ne \emptyset$, and since $u,u'$ are if-free and in $R$-normal form we know that $u \equiv u'$, for all $j$, $\aleavest(\alpha_j)\cap\aleavest(\alpha'_j)$ and for all $k$, $\aleavest(\dec_k)\cap\aleavest(\dec'_k)$. It follows, by induction hypothesis, that for all $j$, $\alpha_j \equiv \alpha_j'$ and for all $k$, $\dec_k \equiv \dec_k'$. We only have to check that the guards are the same. Since $\delta^k,\delta'^k \in \decs_l$, we know from the definition of the $\CCA$ axioms that $\delta^k$ (resp. $\delta'^k$) has one guard for every encryption $\alpha \in \encs_l$ such that $\alpha\equiv \enc{\_}{\pk}{\nonce}$ and $\nonce \in \st(s_p\downarrow_R)\}$ (resp. $\nonce \in \st(s'_p\downarrow_R)\}$). Since we showed that $s_p \equiv s'_p$, we deduce that $\delta^k,\delta'^k$ have the same guards. Since guards are sorted according to an arbitrary but fixed order (the $\textsf{sort}$ function in the definition of $R_{\CCA_a}^\keys$), we know that $\delta^k \equiv \delta'^k$.
\end{proof}

\begin{corollary}
  \label{cor:bas-cond-pull}
  For all $P \npfproof t \sim t'$, for all $\sfh,l$:
  \begin{enumerate}[(i)]
  \item \label{item:bas-cond-charac-beta} for all $\beta,\beta' \lecond^{\sfh,l} (t,P)$ if
    \(\leavest(\beta\downarrow_R) \cap \leavest(\beta'\downarrow_R) \ne \emptyset\)
    then \(\beta \equiv \beta'\).
  \item \label{item:bas-cond-charac-gamma} for all $\gamma,\gamma' \leleave^{\sfh,l} (t,P)$ if
    \(\leavest(\gamma\downarrow_R) \cap \leavest(\gamma'\downarrow_R) \ne \emptyset\)
    then \(\gamma \equiv \gamma'\).
  \item \label{item:bas-term-charac-delta} for all $\beta \lecond^{\sfh,l} (t,P)$, $\gamma \leleave^{\sfh,l} (t,P)$ if
    \(\leavest(\beta\downarrow_R) \cap \leavest(\gamma\downarrow_R) \ne \emptyset\)
    then \(\beta \equiv \gamma\).
  \end{enumerate}
\end{corollary}

\begin{lemma}
  \label{lem:bas-cond-restr-spurious3}
  For all $P \npfproof t \sim t'$, for all $\sfh,l$ and $\beta,\beta' \lebt^{\sfh,l} (t,P)$, there exists an almost conditional context $\tilde \beta'[]$ such that:
  \[
    \beta' \equiv \tilde \beta'\left[\beta\right]
    \quad \wedge \quad
    \leavest(\beta\downarrow_R) \cap \acondst\left(\tilde \beta'[]  \right) = \emptyset
  \]

\end{lemma}

\begin{proof}
    Let $l \in \prooflabel(P)$. We prove by mutual induction on the definition of $\ekl$-normalized simple terms, $\ekl$-normalized basic terms, $\ekl$-encryption oracle calls and $\ekl$-decryption oracle calls that for every $u \in \st(\beta')$ such that $u$ is in one of the four above cases, there exists a conditional context $u_c[]$ such that:
  \begin{gather*}
    u \equiv  u_c\left[\beta\right]
    \quad \wedge \quad
    \leavest(\beta\downarrow_R) \cap \acondst\left( u_c[] \right) = \emptyset
    \quad \wedge \quad
    \aleavest(\vec u_c) = \aleavest(\vec u)
  \end{gather*}
  Moreover if $u$ is a $\ekl$-normalized basic term then there exists an almost conditional context $u_d[]$ such that:
  \begin{gather*}
    u \equiv  u_d\left[\beta\right]
    \quad \wedge \quad
    \leavest(\beta\downarrow_R) \not \in \acondst\left( u_d[] \right) \cup \aleavest\left( u_d[] \right)
  \end{gather*}
  \begin{itemize}
  \item \textbf{Normalized Simple Term:} Let $u \equiv C[\vec b \diamond \vec s]$, where $\vec b$ are $\ekl$-normalized basic conditionals and $\vec s$ are $\ekl$-normalized basic terms. Let $\vec b_d[]$ and $\vec s_c[]$ be contexts obtained from $\vec b,\vec s$ by induction hypothesis such that $\vec b,\vec s \equiv \vec b_d[\beta],\vec s_c[\beta]$ and:
    \[
      \leavest(\vec s_c[]) = \leavest(\vec s)
      \quad \wedge \quad
      \leavest(\beta\downarrow_R) \cap \left(\acondst\left( \vec b_d[],\vec s_c[] \right)\cup \aleavest\left(\vec b_d[]\right)\right) = \emptyset
    \]
    Moreover:
    \begin{gather*}
      \acondst(C[\vec b_d[] \diamond \vec s_c[]]) = \acondst(\vec b_d[],\vec s_c[]) \cup \aleavest(\vec b_d[]) = \acondst(C[\vec b \diamond \vec s])\\
      \aleavest(C[\vec b_d[] \diamond \vec s_c[]]) = \aleavest(\vec s_c[]) = \aleavest(\vec s) = \aleavest(C[\vec b \diamond \vec s])
    \end{gather*}
    Hence we can take $\vec u_c \equiv C[\vec b_d[] \diamond \vec s_c[]]$.
  \item \textbf{Normalized Basic Term:} Let $u \equiv B[\vec w, (\alpha^i)_i,(\dec^j)_j]$ be a $\ekl$-normalized basic term. Let $(\alpha_c^i)_i,(\alpha_d^i)_i$ and $(\dec_c^j)_j,(\dec_d^j)_j$ be the contexts obtained by applying the induction hypothesis to $(\alpha^i)_i$ and $(\dec^j)_j$. Using the fact that:
    \[
      \aleavest\left((\alpha_c^i)_i,(\dec_c^j)_i\right) = \aleavest\left((\alpha^i)_i,(\dec^j)_i\right)
    \]
    and since $B$ and $\vec w$ are if-free, one can check that:
    \[
      \aleavest\left(B[\vec w, (\alpha_c^i)_i,(\dec_c^j)_j]\right) = \aleavest\left(B[\vec w, (\alpha^i)_i,(\dec^j)_j]\right)
    \]
    It is then immediate to check that $u_c \equiv B[\vec w, (\alpha_c^i)_i,(\dec_c^j)_j]$ satisfies the wanted properties. It remains to construct the context $u_d[]$: if for all, $\leavest(\beta\downarrow_R) \cap \aleavest(u) = \emptyset$ then $u_d \equiv u_c$ satisfies the wanted properties. Otherwise using Proposition~\ref{prop:bas-cond-charac} we know that $\beta \equiv u$, hence we can take $u_d \equiv []$.
  \item \textbf{Encryption Oracle Call:} The proof done for the normalized basic term case applies here.
  \item \textbf{Decryption Oracle Call:} The proof done for the normalized simple term case applies here.\qedhere
\end{itemize}
\end{proof}
Observe that this lemma subsumes Lemma~\ref{lem:bas-cond-restr-spurious2}.

\subsection{Well-nestedness}

\begin{definition}
  A simple term  $C[\vec a \diamond \vec b]$ is said to be \emph{flat} if $\vec a,\vec b$ are if-free terms in $R$-normal forms.
\end{definition}

\begin{definition}
  We let \emph{well-nested} be the smallest relation between sets $(\mathcal{C},\mathcal{D})$ of flat simple terms such that:
  \begin{itemize}
  \item[(a)] $(\mathcal{C},\mathcal{D})$ is well-nested if for every $C_0[\vec a_0 \diamond \vec b_0] \in \mathcal{C}$:
    \begin{alignat*}{2}
      &&\forall C[\vec a \diamond \vec b] \in  \mathcal{C},
      \quad \vec b_0 \cap \vec a = \emptyset\\
      \text{and }\quad&& \forall D[\vec c \diamond \vec t] \in \mathcal{D},
      \quad \vec b_0 \cap \vec c = \emptyset
    \end{alignat*}
  \item[(b)] $(\mathcal{C},\mathcal{D})$ is well-nested if for every $C_0[\vec a_0 \diamond \vec b_0] \in \mathcal{C}$:
    \begin{itemize}
    \item[(i)] For all $C[\vec a \diamond \vec b] \in \mathcal{C}$,
      there exist two if-contexts $C'^i, C''^i$ such that:
      \begin{equation*}
        C[\vec a \diamond \vec b] =_R \ite{C_0[\vec a_0 \diamond \vec b_0]}{C'^i[\vec a' \diamond \vec b']}{C''^i[\vec a'' \diamond \vec b'']}
      \end{equation*}
      where $\vec a',\vec a'' \subseteq \vec a \backslash \vec b_0$ and $\vec b', \vec b'' \subseteq \vec b$.
    \item[(ii)] For every $D[\vec c \diamond \vec t] \in \mathcal{D}$,
      there exist two if-contexts $D'^i, D'^i$ such that:
      \begin{equation*}
        D[\vec c \diamond \vec t] =_R \ite{C_0[\vec a_0 \diamond \vec b_0]}{D'^i[\vec c' \diamond \vec t']}{D''^i[\vec c'' \diamond \vec t'']}
      \end{equation*}
      where $\vec c',\vec c'' \subseteq \vec c \backslash \vec b_0$ and $\vec t', \vec t'' \subseteq \vec t$.
    \item[(iii)] The following couples of sets are well-nested:
      \begin{gather*}
        \left(
          \left\{C'^i[\vec a' \diamond \vec b'] \mid C[\vec a \diamond \vec b] \in \mathcal{C}\right\},
          \left\{D'^i[\vec c' \diamond \vec t'] \mid D[\vec c \diamond \vec t] \in \mathcal{D}\right\}
        \right)\\
        \left(
          \left\{C''^i[\vec a'' \diamond \vec b''] \mid C[\vec a \diamond \vec b] \in \mathcal{C}\right\},
          \left\{D''^i[\vec c'' \diamond \vec t''] \mid D[\vec c \diamond \vec t] \in \mathcal{D}\right\}
        \right)
      \end{gather*}
    \end{itemize}
  \end{itemize}
\end{definition}

\begin{proposition}
  If $(\mathcal{C},\mathcal{D})$ is such that for all $C_i[\vec a_i \diamond \vec b_i] \in \mathcal{C}$:
  \begin{alignat*}{2}
    &&\forall C_j[\vec a_j \diamond \vec b_j] \in \mathcal{C},
    \quad \vec b_i \cap \vec a_j = \emptyset\\
    \text{and }\quad&& \forall D_j[\vec c_j \diamond \vec t_j] \in \mathcal{D},
    \quad \vec b_i \cap \vec c_j = \emptyset
  \end{alignat*}
  Then $(\mathcal{C},\mathcal{D})$ verifies the properties (i),(ii) and (iii) above.
\end{proposition}

\begin{proof}
  Trivial by taking $C_j'^i \equiv C_j''^i \equiv C_j$.
\end{proof}

\paragraph{Main Lemma} We introduce now some tools used in the proof of the main lemma of this subsection, before stating and proving this lemma.
\begin{definition}
  We let $\pos(t)$ be the set of positions of $p$, and $\head$ be the partial function defined on terms such that for all $f\in \sig$, for all terms $\vec t$, $\head(f(\vec t)) \equiv f$.
\end{definition}

\begin{definition}
  For all conditional contexts $C_0,C_1$, we let $C_0 \cccup C_1$ be the conditional context, if it exists, defined as follows: $\pos(C_1 \cccup C_2) = \pos(C_0) \cap \pos(C_1)$ and for all position $p$ in $\pos(C_0 \cccup C_1)$:
  \[
    (C_0 \cccup C_1)_{|p} \equiv
    \begin{cases}
      a & \text{ if }     \head((C_0)_{|p}) \equiv  \head((C_1)_{|p}) \equiv a \qquad (a \in \sig \cup \Nonce)\\
      []_x & \text{ if }     (C_0)_{|p} \equiv []_x \wedge \left(\head((C_1)_{|p}) \equiv []_x \vee \head((C_1)_{|p}) \equiv a\right) \qquad (a \in \sig \cup \Nonce)\\
      []_x & \text{ if }     (C_1)_{|p} \equiv []_x \wedge \left(\head((C_0)_{|p}) \equiv []_x \vee \head((C_0)_{|p}) \equiv a \right) \qquad (a \in \sig \cup \Nonce)
    \end{cases}
  \]
  If such a conditional context  does not exist then we let $C_0 \cccup C_1$ be the special element $\textsf{undefined}$. We also let:
  \[
    \textsf{undefined} \cccup C_0 \equiv C_0 \cccup  \textsf{undefined} \equiv  \textsf{undefined}
  \]
\end{definition}

\begin{example}
  \label{ex:cccup}
  For all conditionals $a,b,c,d,e,f$ and terms $t_0,\dots,t_3$ , if we let:
  \begin{gather*}
    C_0 \equiv
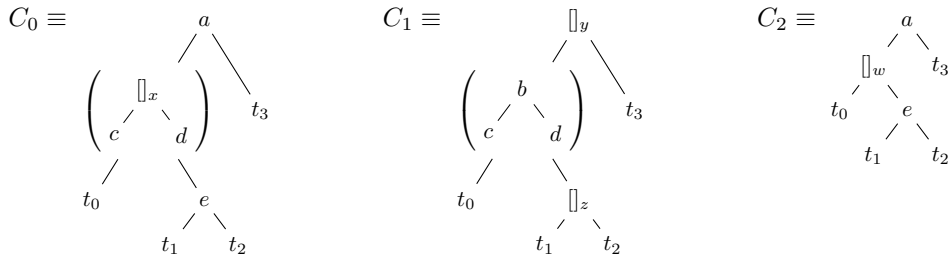
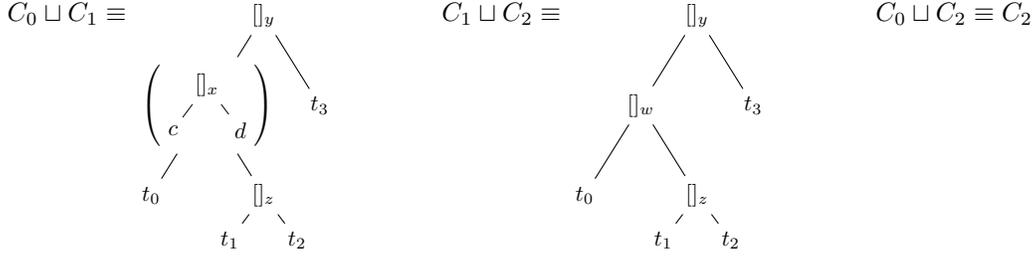
    \begin{tikz}[baseline,level distance=4em, sibling distance=5em,transform shape, scale=0.85]
      \tikzstyle{level 1}=[level distance=4em]
      \tikzstyle{level 2}=[level distance=4em]
      \tikzstyle{level 3}=[level distance=2em,sibling distance=3em]
      \node[anchor=base] (a) at (0,0){$a$}
      child {
        node {
          $
          \left(
            \begin{tikz}[baseline={([yshift=-.5ex]current bounding box.center)},level distance=2em, sibling distance=3em]
              \tikzstyle{level 1}=[level distance=2em]
              \tikzstyle{level 2}=[level distance=2em]
              \tikzstyle{level 3}=[level distance=2em]
              \node[anchor=base] (b) at (0,0){$[]_x$}
              child { node {$c$}}
              child { node {$d$}};
            \end{tikz}
          \right)
          $
        }
        child {
          node {$t_0$}
        }
        child {
          node {$e$}
          child { node {$t_1$}}
          child { node {$t_2$}}
        }
      }
      child {
        node {$t_3$}
      };
    \end{tikz}
    \qquad
    \qquad
    C_1 \equiv
    \begin{tikz}[baseline,level distance=4em, sibling distance=5em,transform shape, scale=0.85]
      \tikzstyle{level 1}=[level distance=4em]
      \tikzstyle{level 2}=[level distance=4em]
      \tikzstyle{level 3}=[level distance=2em,sibling distance=3em]
      \node[anchor=base] (a) at (0,0){$[]_y$}
      child {
        node {
          $
          \left(
            \begin{tikz}[baseline={([yshift=-.5ex]current bounding box.center)},level distance=2em, sibling distance=3em]
              \tikzstyle{level 1}=[level distance=2em]
              \tikzstyle{level 2}=[level distance=2em]
              \tikzstyle{level 3}=[level distance=2em]
              \node[anchor=base] (b) at (0,0){$b$}
              child { node {$c$}}
              child { node {$d$}};
            \end{tikz}
          \right)
          $
        }
        child {
          node {$t_0$}
        }
        child {
          node {$[]_z$}
          child { node {$t_1$}}
          child { node {$t_2$}}
        }
      }
      child {
        node {$t_3$}
      };
    \end{tikz}
    \qquad
    \qquad
    C_2 \equiv
    \begin{tikz}[baseline,level distance=2em, sibling distance=3em,transform shape, scale=0.85]
      \node[anchor=base] (a) at (0,0){$a$}
      child {
        node {$[]_w$}
        child {
          node {$t_0$}
        }
        child {
          node {$e$}
          child { node {$t_1$}}
          child { node {$t_2$}}
        }
      }
      child {
        node {$t_3$}
      };
    \end{tikz}
  \end{gather*}
  Then we have:
  \begin{gather*}
    C_0 \cccup C_1 \equiv
    \begin{tikz}[baseline,level distance=4em, sibling distance=5em,transform shape, scale=0.85]
      \tikzstyle{level 1}=[level distance=4em]
      \tikzstyle{level 2}=[level distance=4em]
      \tikzstyle{level 3}=[level distance=2em,sibling distance=3em]
      \node[anchor=base] (a) at (0,0){$[]_y$}
      child {
        node {
          $
          \left(
            \begin{tikz}[baseline={([yshift=-.5ex]current bounding box.center)},level distance=2em, sibling distance=3em]
              \tikzstyle{level 1}=[level distance=2em]
              \tikzstyle{level 2}=[level distance=2em]
              \tikzstyle{level 3}=[level distance=2em]
              \node[anchor=base] (b) at (0,0){$[]_x$}
              child { node {$c$}}
              child { node {$d$}};
            \end{tikz}
          \right)
          $
        }
        child {
          node {$t_0$}
        }
        child {
          node {$[]_z$}
          child { node {$t_1$}}
          child { node {$t_2$}}
        }
      }
      child {
        node {$t_3$}
      };
    \end{tikz}
    \qquad
    \qquad
    C_1 \cccup C_2 \equiv
    \begin{tikz}[baseline,level distance=4em, sibling distance=5em,transform shape, scale=0.85]
      \tikzstyle{level 1}=[level distance=4em]
      \tikzstyle{level 2}=[level distance=4em]
      \tikzstyle{level 3}=[level distance=2em,sibling distance=3em]
      \node[anchor=base] (a) at (0,0){$[]_y$}
      child {
        node {$[]_w$}
        child {
          node {$t_0$}
        }
        child {
          node {$[]_z$}
          child { node {$t_1$}}
          child { node {$t_2$}}
        }
      }
      child {
        node {$t_3$}
      };
    \end{tikz}
    \qquad
    \qquad
    C_0 \cccup C_2 \equiv C_2
  \end{gather*}
\end{example}

\begin{definition}
  We let $\ccsubseteq$ be the relation on conditional contexts defined as follows: for all conditional contexts $C_0,C_1$, we let $C_0 \ccsubseteq C_1$ hold if $\pos(C_1) \subseteq \pos(C_0)$ and for all position $p$ in $\pos(C_1)$:
  \[
    \text{if } \head((C_1)_{|p}) \equiv
    \begin{cases}
      a & \text{ then }     \head((C_0)_{|p}) \equiv a\qquad (a \in \sig \cup \Nonce)\\
      []_x  & \text{ then } \head((C_0)_{|p}) \equiv a \vee \head((C_0)_{|p}) \equiv []_x \qquad (a \in \sig \cup \Nonce)
    \end{cases}
  \]
  Moreover we let $C_0 \ccsubseteq \textsf{undefined}$ for all conditional context $C_0$ (and $\textsf{undefined} \ccsubseteq \textsf{undefined}$).
\end{definition}

\begin{example}
  Using the conditional contexts defined in Example~\ref{ex:cccup}, we have, for example, the following relations:
  
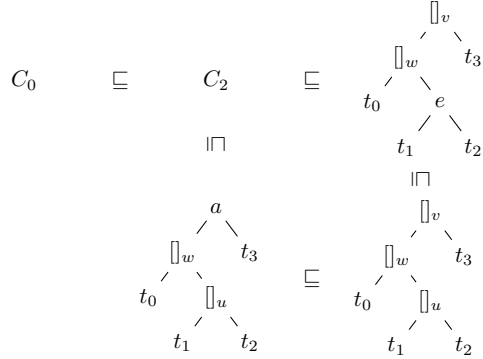
\begin{center}
    \begin{tikzpicture}[level distance=2em, sibling distance=3em,transform shape, scale=0.85]
      \path (0,0) node (a) {$C_0$}
      -- ++(1.5,0)
      node (b) {$\ccsubseteq$}
      -- ++(1.5,0)
      node (c) {$C_2$}
      -- ++(1.5,0)
      node (d) {$\ccsubseteq$}
      -- ++(2,0)
      node[yshift=3em] (e) {$[]_v$}
      child {
        node {$[]_w$}
        child {
          node {$t_0$}
        }
        child {
          node {$e$}
          child { node {$t_1$}}
          child { node {$t_2$}}
        }
      }
      child {
        node {$t_3$}
      };
      \path (c)
      -- ++(0,-1)
      node[rotate=-90] (ee) {$\ccsubseteq$}
      -- ++(0,-1)
      node[] (f) {$a$}
      child {
        node {$[]_w$}
        child {
          node {$t_0$}
        }
        child {
          node {$[]_u$}
          child { node {$t_1$}}
          child { node {$t_2$}}
        }
      }
      child {
        node {$t_3$}
      };

      \path (ee)
      -- ++(1.5,0)
      node[yshift=-6em] (ee) {$\ccsubseteq$}
      -- ++(1.5,0)
      node[yshift=-3em,xshift=1em] (g) {$[]_v$}
      child {
        node {$[]_w$}
        child {
          node {$t_0$}
        }
        child {
          node {$[]_u$}
          child { node {$t_1$}}
          child { node {$t_2$}}
        }
      }
      child {
        node {$t_3$}
      };

      \path (e.south) -- ++(-0.3,-2.3) node[rotate=-90]{$\ccsubseteq$};
    \end{tikzpicture}
  \end{center}
\end{example}

Let $\mathcal{S}_{cc}$ be the set of conditional contexts extended with $\textsf{undefined}$.
\begin{proposition}
  \label{prop:cupsubst1}
  $(\mathcal{S}_{cc},\cccup,\ccsubseteq)$ is a semi-lattice. That is we have the following properties:
  \begin{enumerate}[(i)]
  \item\label{prop-item:cupsubst0} $\cccup$ is associative, commutative, idempotent.
  \item\label{prop-item:cupsubst0p} $\ccsubseteq$ is an order (i.e. reflexive, transitive and antisymmetric).
  \item\label{prop-item:cupsubst1} For all $C_0,C_1 \in \mathcal{S}_{cc}$, we have $C_0 \ccsubseteq (C_0 \cccup C_1)$ and $C_1 \ccsubseteq (C_0 \cccup C_1)$. Moreover $(C_0 \cccup C_1)$ is the least upper-bound of $C_0$ and $C_1$.
  \end{enumerate}
\end{proposition}

\begin{proof}
  These properties are straightforward to show, we are only going to give the proof that $(C_0 \cccup C_1)$ is the least upper-bound of $C_0$ and $C_1$. Assume that there is $C$ such that:
  \[
    C_0 \ccsubseteq C \ccsubseteq C_0 \cccup C_1 \qquad \qquad
    C_1 \ccsubseteq C \ccsubseteq C_0 \cccup C_1
  \]
  If $C_0 \cccup C_1 \equiv \textsf{undefined}$ then one can check that $C \equiv \textsf{undefined}$. Otherwise we know that $\pos(C_0 \cccup C_1) = \pos(C_0) \cap \pos(C_1)$, and that:
  \[
    \pos(C_0) \supseteq \pos(C ) \supseteq \pos(C_0 \cccup C_1) \qquad \qquad
    \pos(C_1) \supseteq \pos(C ) \supseteq \pos(C_0 \cccup C_1)
  \]
  Hence $\pos(C) = \pos(C_0 \cccup C_1)$. Using the fact that $C \ccsubseteq C_0 \cccup C_1$ we know that for all position $p \in \pos(C)$, if $\head((C_0 \cccup C_1)_{|p}) = a$ (with $a \in \sig \cup \Nonce$) then $\head(C_{|p}) = a$. If $\head((C_0 \cccup C_1)_{|p}) = []_x$ then either $\head(C_{|p}) = []_x$ or $\head(C_{|p}) = a$ (with $a \in \sig \cup \Nonce$). In the former case there is nothing to show, in the the latter case observe that $\head((C_0 \cccup C_1)_{|p}) = []_x$ implies that either $\head((C_0)_{|p}) = []_x$ or $\head((C_1)_{|p}) = []_x$. W.l.o.g assume $\head((C_0)_{|p}) = []_x$. Then using the fact that $C_0 \ccsubseteq C$, we know that $\head((C_0)_{|p}) = []_x$ implies that $\head((C_0)_{|p}) = []_x$, absurd.

  Therefore for all $p \in \pos(C)$, $\head(C_{|p}) = \head((C_0 \cccup C_1)_{|p})$. Moreover $\pos(C) = \pos(C_0 \cccup C_1)$, hence $C \equiv C_0 \cccup C_1$.
\end{proof}

% Let $\ra_{R'}$ be the term rewriting system $\ra_{R}$ without the rule $\ra_{\textsf{same}}$ (where $\ra_{\textsf{same}}$ is the rule $\ite{x}{y}{y} \ra y$). We know that $\ra_{R'}$ is terminating (but not confluent), therefore we let $\ired_{R'}(t)$ be the set of $R'$-irreducible terms reachable from $t$.

\begin{proposition}
  \label{prop:cupsubst2}
  For all $C_0,C_1 \in \mathcal{S}_{cc}$, if $C_0 \ccsubseteq C_1$ and if:
  \[
    \forall p,p' \in \pos(C_1), (C_1)_{|p} \equiv (C_1)_{|p'} \equiv []_x \;\Rightarrow\; (C_0)_{|p} \equiv (C_0)_{|p'}
  \]
  then $\condst(C_1\downarrow_R) \cap \mathcal{T}(\ssig,\Nonce) \subseteq \acondst(C_0)$.
\end{proposition}

% \begin{remark}
%   This is indeed a stronger property, as $\ra^!_R \;= \; \ra_{R'}^!.\ra_{\textsf{same}}^!$, and because if $u \ra_{\textsf{same}}^* v$ then $\condst(u) \supseteq \condst(v)$.
% \end{remark}

\begin{proof}
  Assume that $C_0 \ccsubseteq C_1$, with $C_0,C_1 \ne \textsf{undefined}$ (the case $C_0 \ne \textsf{undefined}$ or $C_1 \ne \textsf{undefined}$ is easy to handle, with the convention that $\condst(\textsf{undefined}) = \emptyset$), and that:
  \begin{equation}
    \label{eq:condccsub}
    \forall p,p' \in \pos(C_1), (C_1)_{|p} \equiv (C_1)_{|p'} \equiv []_x \;\Rightarrow\; (C_0)_{|p} \equiv (C_0)_{|p'}
  \end{equation}
  First we show that we can extend this property as follows:
  \begin{equation}
    \label{eq:condccsub1}
    \forall p,p' \in \pos(C_1), (C_1)_{|p} \equiv (C_1)_{|p'} \; \Rightarrow \;(C_0)_{|p} \equiv (C_0)_{|p'}
  \end{equation}
  Let $q = p\cdot q_0$ and $q = p'\cdot q_0$ be positions in $\pos(C_1)$. Since $(C_0)_{|p} \equiv (C_0)_{|p'}$, we know that $\head((C_1)_{|q}) \equiv \head((C_1)_{|q'})$.
  \begin{itemize}
  \item If $\head((C_1)_{|q}) \equiv a$ (with $a \in \sig \cup \Nonce$) then, from the fact that $C_0 \ccsubseteq C_1$ we get that $\head((C_0)_{|q}) \equiv a$, and that $\head((C_0)_{|q'}) \equiv a$.
  \item If $\head((C_1)_{|q}) \equiv []_x$ then using \eqref{eq:condccsub} we get that $(C_0)_{|p} \equiv (C_0)_{|p'}$.
  \end{itemize}

  Let $\ra_{R'}$ be $\ra_R$ without the non left-linear rules:
  \begin{mathpar}
    \ite{x}{y}{y} \ra y

    \dec(\enc{x}{\pk(y)}{r},\sk(y)) \ra x

    \ite{w}{(\ite{w}{x}{y})}{z} \ra \ite{w}{x}{z}

    \ite{w}{x}{(\ite{w}{y}{z})} \ra \ite{w}{x}{z}
  \end{mathpar}
  We then mimic all reduction $\ra_{R}$ on $C_1$ by a reduction on $C_0$, while maintaining $\ccsubseteq$ and the invariant of \eqref{eq:condccsub}. Mimicking rules in $\ra_R$ is easy as they are left-linear. To mimic rules in $(\ra_{R} \backslash \ra_{R'})$, we use \eqref{eq:condccsub1}. Formally, we show by induction on the length of the reduction sequence that for all $C_1'$ such that $C_1 \ra_{R}^* C_1'$, there exists $C_0'$ such that $C_0' \ccsubseteq C_1'$, \eqref{eq:condccsub} holds for $C_0',C_1'$ and $C_0 \ra_{R}^* C_0'$.

  Therefore let $C_1'$ be in $R$-normal form such that $C_1 \ra_{R}^* C_1'$. Let $C_0'$ be such that $C_0' \ccsubseteq C_1'$,\eqref{eq:condccsub} holds for $C_0',C_1'$ and $C_0 \ra_{R}^* C_0'$. $C_1'$ is of the form $D[\vec b,\vec b_{[]} \diamond \vec u]$ where $\vec b,\vec u$ are if-free and in $R$-normal form, $\vec b$ does not contain any hole variable and $\vec b_{[]}$ is a vector of hole variables. Therefore $\condst(C_1\downarrow_R)\cap \mathcal{T}(\ssig,\Nonce) = \condst(C_1')\cap \mathcal{T}(\ssig,\Nonce) = \vec b$. We conclude by observing that $\vec b \subseteq \acondst(C_0')$, and that $\acondst(C_0') \subseteq \acondst(C_0)$ by Proposition~\ref{prop:acondst-overapprox}.
\end{proof}

\begin{lemma}\label{lem:well-nested}
  For all $P \npfproof t \sim t'$, for all $\sfh,l$, the following couple of sets is well-nested:
  \[
    \left(
      \left\{\beta\downarrow_R\mid \beta \lecond^{\sfh,l} (t,P) \right\},
      \left\{\gamma\downarrow_R\mid \gamma \leleave^{\sfh,l} (t,P) \right\}
    \right)
  \]

\end{lemma}

\begin{proof}
  We do this proof in the case $\sfh = \epsilon$. The other cases are identical.

  We consider an arbitrary ordering $(\beta_i)_{1 \le i \le i_{max}}$ of $\{\beta\mid \beta \lecond^{\sfh,l} (t,P)\}$ and $(\gamma_m)_{1 \le m \le m_{max}}$ of $\{\gamma\mid \gamma \leleave^{\sfh,l} (t,P) \}$.

  Using Lemma~\ref{lem:bas-cond-restr-spurious3}, we know that all $i \ne i_0$, there exists a conditional context $\tilde \beta_{i}$ such that:
  \[
    \beta_{i} \equiv \tilde \beta_{i}\left[\beta_{i_o}\right]
    \quad \wedge \quad
    \leavest(\beta_{i_0}\downarrow_R) \cap \acondst\left(\tilde \beta_{i,l}\right) = \emptyset
  \]
  From now on we use $\beta_{i}^{(i_0)}$ to denote this conditional context, and $[]_{i_0}$ the hole variable used in the conditional contexts $\{\beta_{i}^{(i_0)} \mid i\}$. We similarly define $\gamma_{m}^{(i_0)}$ and we have:
  \[
    \gamma_{m} \equiv \tilde \gamma_{m}\left[\beta_{i_o}\right]
    \quad \wedge \quad
    \leavest(\beta_{i_0}\downarrow_R) \cap \acondst\left(\tilde \gamma_{m}\right) = \emptyset
  \]
  We extend this notation by having $j$ range between $-1 \le j < n_{max}$ (resp. $-1 \le j < m_{max}$), and having $\beta_{i}^{(-1)} \equiv \beta_{i}$ (resp. $\gamma_{m}^{(-1)} \equiv \gamma_{m}$).

  \noindent Consider the following set~$\mathcal{S}$:
  \[
    \left\{\left( (\cccup_{j \le n} \beta^{(l_j)}_{i})_i, (\cccup_{j \le n} \gamma^{(l_j)}_{m})_m \right) \mid (l_j)_j \text{ distinct indices} \wedge l_0 \equiv -1\right\}
  \]
  Using Proposition~\ref{prop:cupsubst1}.\eqref{prop-item:cupsubst1} we know that for all $i \ne l_{j_0}$:
  \[
    \beta^{(l_{j_0})}_{i} \ccsubseteq \cccup_{j \le n} \beta^{(l_j)}_{i}
    \qquad \wedge \qquad
    \beta^{(l_{j_0})}_{i} \ccsubseteq \cccup_{j \le n} \gamma^{(l_j)}_{m}
  \]
  Using Proposition~\ref{prop:cupsubst2} we know that for all $j,o$ and for all $i \ne l_{j_0}$:
  \[
    \acondst\left(\beta^{(l_{j_0})}_{i}\right) \supseteq \condst\left(\cccup_{j \le n} \beta^{(l_j)}_{i}\downarrow_R\right)
    \qquad \wedge \qquad
    \acondst\left(\gamma^{(l_{j_0})}_{m}\right) \supseteq \condst\left(\cccup_{j \le n} \gamma^{(l_j)}_{m}\downarrow_R\right)
  \]
  Which implies that:
  \begin{equation}
    \label{eq:condcont1}
    \leavest(\beta_{i_0}\downarrow_R) \cap \condst\left(\cccup_{j \le n} \beta^{(l_j)}_{i}\downarrow_R\right) = \emptyset
    \qquad \wedge \qquad
    \leavest(\beta_{i_0}\downarrow_R) \cap \condst\left(\cccup_{j \le n} \gamma^{(l_j)}_{m}\downarrow_R\right) = \emptyset
  \end{equation}
  Moreover it is quite simple to show that for all $(l_j)_{j \le n+1}$, for all $i \ne l_{n+1}$:
  \[
    \cccup_{j \le n+1} \beta^{(l_j)}_{i} \equiv \left(\cccup_{j \le n} \beta^{(l_j)}_{i}\right)\{[]_{l_{n+1}} / \cccup_{j \le n} \beta^{(l_j)}_{n+1} \}
  \]
  Therefore:
  \begin{alignat*}{2}
    \cccup_{j \le n} \beta^{(l_j)}_{i} &\;=_R\;& \left(\cccup_{j \le n+1} \beta^{(l_j)}_{i}\right)\{ \cccup_{j \le n} \beta^{(l_j)}_{n+1} / []_{l_{n+1}}\} \\
    &\;=_R\;&
    \begin{alignedat}[t]{2}
      \ite{\left( \cccup_{j \le n} \beta^{(l_j)}_{n+1}\right)}
      {&\left(\cccup_{j \le n+1} \beta^{(l_j)}_{i}\right)\{\true / []_{l_{n+1}}\}\\}
      {&\left(\cccup_{j \le n+1} \beta^{(l_j)}_{i}\right)\{\false / []_{l_{n+1}}\}}
    \end{alignedat}\numberthis\label{eq:condcont2}
  \end{alignat*}
  Consider the following set~$\mathcal{S}'$:
  \[
    \left\{
      \left(
        \left(\cccup_{j \le n} \beta^{(l_j)}_{i}\{e_j / []_{l_{j}}\}\downarrow_R\right)_i,
        \left(\cccup_{j \le n} \gamma^{(l_j)}_{m}\{e_j / []_{l_{j}}\} \downarrow_R\right)_m
      \right)
      \mid (l_j)_j \text{ distinct indices} \wedge (e_j)_j \in \{\true,\false\}^{n}
    \right\}
  \]
  We show by decreasing induction on $n$, starting from $n = i_{max} + 1$, that all the elements of $\mathcal{S}'$ are well-nested.

  \paragraph{Base case} If $n = n_{max} + 1$ then from \eqref{eq:condcont1} we get that for all sequence $(e_j)_j$ in $\{\true,\false\}^{n}$, for all $j \ne i$:
  \[
    \leavest(\beta_{j}\downarrow_R) \cap \condst\left(\left(\cccup_{j \le n} \beta^{(j)}_{i}\right)\{e_j / []_{j}\}\downarrow_R\right) = \emptyset
  \]
  Moreover we have:
  \[
    \leavest\left(\left(\cccup_{j \le n} \beta^{(j)}_{i}\right)\{e_j / []_{j}\}\downarrow_R\right)
    \subseteq \left\{ \beta_{i}^o\mid o\right\}
  \]
  Hence we get that the following set is well-nested (case (a)):
  \[
    \left(
      \left(\cccup_{j \le n} \beta^{(j)}_{i}\{e_j / []_{{j}}\}\downarrow_R\right)_i,
      \left(\cccup_{j \le n} \gamma^{(j)}_{m}\{e_j / []_{{j}}\} \downarrow_R\right)_m
    \right)
  \]

  \paragraph{Inductive Case} If $n \le n_{max}$ then from \eqref{eq:condcont2} we get that for all sequence $(l_j)_{j \le n+1}$, for all sequence $(e_{l_j})_j$ in $\{\true,\false\}^{n}$, for all $j \ne i$:
  \begin{multline*}
    \left(\cccup_{j \le n} \beta^{(l_j)}_{i}\right)\{e_{l_j} / []_{{l_j}} \mid j \le n\}
    \;=_R\;\\
    \begin{alignedat}[t]{2}
      \ite{\left(\left( \cccup_{j \le n} \beta^{(l_j)}_{l_{n+1}}\right)\{e_{l_j} / []_{{l_j}} \mid j \le n\}\right)}
      {&\left(\cccup_{j \le n+1} \beta^{(l_j)}_{i}\right)\{e_{l_j} / []_{{l_j}} \mid j \le n\}\{\true / []_{l_{n+1}}\}\\}
      {&\left(\cccup_{j \le n+1} \beta^{(l_j)}_{i}\right)\{e_{l_j} / []_{{l_j}} \mid j \le n\}\{\false / []_{l_{n+1}}\}}
    \end{alignedat}
  \end{multline*}
  Let $e_{l_{n+1}} \equiv \true$ (resp. $e_{l_{n+1}} \equiv \false$). We get from \eqref{eq:condcont1} that for all $o$ and $i \ne l_{n+1}$:
  \[
    \leavest(\beta_{l_{n+1}}\downarrow_R) \cap \not \in \condst\left(\left(\cccup_{j \le n+1} \beta^{(l_j)}_{i}\right)\{e_{l_j} / []_{l_j}\}\downarrow_R\right) = \emptyset
  \]
  We can do a similar reasoning on $\gamma_{i}$ to show that for all $o$:
  \[
    \leavest(\beta_{l_{n+1}}\downarrow_R) \cap \condst\left(\left(\cccup_{j \le n+1} \gamma^{(l_j)}_{i}\right)\{e_{l_j} / []_{l_j}\}\downarrow_R\right) = \emptyset
  \]
  Moreover by induction hypothesis we know that:
  \[
    \left(
      \left(\left(\cccup_{j \le n+1} \beta^{(l_j)}_{i}\right)\{e_{l_j} / []_{l_j}\}\downarrow_R\right)_i,
      \left(\left(\cccup_{j \le n+1} \gamma^{(l_j)}_{i}\right)\{e_{l_j} / []_{l_j}\}\downarrow_R\right)_i
    \right)
  \]
  is well-nested for $e_{l_{n+1}} \equiv \true$ and $e_{l_{n+1}} \equiv \false$. We deduce from this that the following set is well nested (case b):
  \[
    \left(
      \left(\left(\cccup_{j \le n} \beta^{(l_j)}_{i}\right)\{e_{l_j} / []_{l_j}\}\downarrow_R\right)_i,
      \left(\left(\cccup_{j \le n} \gamma^{(l_j)}_{i}\right)\{e_{l_j} / []_{l_j}\}\downarrow_R\right)_i
    \right)
  \]
  \paragraph{Conclusion} Recall that $\beta^{(l_0)}_{i} \equiv \beta^{(-1)}_{i} \equiv \beta_{i}$. We conclude the proof of this lemma by observing that
  \[
    \left(
      \left\{C^\sfh_{i}\left[\vec b^\sfh_{i} \diamond \{ \beta_{i}^{\sfh,o} \mid o\}\right]\mid i\right\},
      \left\{C^\sfh_{m}\left[\vec b^\sfh_{m} \diamond \{ \gamma_{m}^{\sfh,o} \mid o\}\right]\mid m\right\}
    \right)
  \]
  is the couple of sets:
  \[
    \left(
      \left(\left(\cccup_{j \le 0} \beta^{(l_j)}_{i}\right)\{e_{l_j} / []_{l_j}\}\downarrow_R\right)_i,
      \left(\left(\cccup_{j \le 0} \gamma^{(l_j)}_{i}\right)\{e_{l_j} / []_{l_j}\}\downarrow_R\right)_i
    \right)
  \]
  which is in $\mathcal{S}'$, and therefore well-nested.
\end{proof}

\subsection{Spurious Conditionals}

\begin{definition}
  An if-free conditional $b$ is said to be \emph{spurious} in a term $t$ if $b\downarrow_R \not \in \condst(t\downarrow_R)$.
\end{definition}

\begin{definition}
  A set of positions is said to be spurious in a term $t$ if it is non-empty and $t[\true/x \mid x \in S] =_R t[\false/x \mid x \in S] =_R t$. A spurious set is \emph{minimal} (resp. \emph{maximal}) if it has not strict spurious subset (resp. overset), and a spurious set is \emph{rooted} if there exists $p \in S$ such that $\forall p' \in S, p \le p'$ (i.e. is a common ancestor of all positions in $S$).
\end{definition}

\begin{example}
  Let $a \equiv \eq{A}{0}$ and $b \equiv \eq{B}{0}$ be two conditionals. Consider the following term $t$:
  \begin{alignat*}{2}
    \ite{b}
    {&
      \begin{alignedat}[t]{2}
        \ite{a}{&\ite{b}{T}{U}\\}{&V}
      \end{alignedat}
      \\}
    {&
      \begin{alignedat}[t]{2}
        \ite{a}{&T\\}{&\ite{a}{V}{V}}
      \end{alignedat}
    }
  \end{alignat*}
  Then the conditional $b$ is spurious in $t$, since $b$ is not a subterm of $t \downarrow_R \equiv \ite{a}{T}{V}$. Moreover the conditional $a$ is a subterm of $t \downarrow_R$, hence is spurious. Nonetheless, the set of position $S = \{ 220 \}$ is spurious. Indeed we have:
  \begin{alignat*}{2}
    \begin{alignedat}{2}
      \ite{b}
      {&
        \begin{alignedat}[t]{2}
          \ite{a}{&\ite{b}{T}{U}\\}{&V}
        \end{alignedat}
        \\}
      {&
        \begin{alignedat}[t]{2}
          \ite{a}{&T\\}{&\ite{\framebox{$a$}_{220}}{V}{V}}
        \end{alignedat}
      }
    \end{alignedat}
    &\qquad=_R\qquad&
    \begin{alignedat}{2}
      \ite{b}
      {&
        \begin{alignedat}[t]{2}
          \ite{a}{&\ite{b}{T}{U}\\}{&V}
        \end{alignedat}
        \\}
      {&
        \begin{alignedat}[t]{2}
          \ite{a}{&T\\}{&\ite{\framebox{$\true$}_{220}}{V}{V}}
        \end{alignedat}
      }
    \end{alignedat}\\[2em]
    &\qquad=_R\qquad&
    \begin{alignedat}{2}
      \ite{b}
      {&
        \begin{alignedat}[t]{2}
          \ite{a}{&\ite{b}{T}{U}\\}{&V}
        \end{alignedat}
        \\}
      {&
        \begin{alignedat}[t]{2}
          \ite{a}{&T\\}{&\ite{\framebox{$\false$}_{220}}{V}{V}}
        \end{alignedat}
      }
    \end{alignedat}
  \end{alignat*}
\end{example}

\paragraph{Spurious Conditionals to Spurious Sets}
Knowing that a conditional $a$ is spurious in a term $t$ does not necessarily mean that we know a spurious set of positions $S$ such that for all $p \in S$, $t_{|p} \equiv a$. If $b$ is in $R$-normal form this is easy, but terms in proof form are not in $R$-normal form. The following proposition shows that such a set of positions exists, under some conditions.

\begin{proposition}
  \label{prop:spurious-replace}
  Let $\vec a,\vec b,\vec c$ be if-free conditionals in $R$-normal form. Let $t$ be the term:
  \[
    t \equiv B\left[\vec c \diamond \left(\vec w,\ite{C[\vec b \diamond \vec a]}{u}{v}\right)\right]
  \]
  Let $a \in \vec a$ be a spurious conditional in $t$ such that:
  \begin{itemize}
  \item $a \not \in \vec b \cup \{ \true,\false\} \cup \condst(u\downarrow_R) \cup \condst(v\downarrow_R)$.
  \item $a \not \in \rho$ where $\rho$ is the set of conditionals appearing on the path from the root to $(\ite{C[\vec b \diamond \vec a]}{u}{v})$.
  \end{itemize}
  Then we have:
  \[
    B\left[\vec c \diamond \left(\vec w,\ite{C[\vec b \diamond \vec a]}{u}{v}\right)\right] =_R B\left[\vec c \diamond \left(\vec w,\ite{C[\vec b \diamond \vec a',\true]}{u}{v}\right)\right]
  \]
  where $\vec a' = \vec a \backslash \{a\}$.
\end{proposition}

\begin{proof} We recall that:
  \[
    t \equiv B\left[\vec c \diamond \left(\vec w,\ite{C[\vec b \diamond \vec a]}{u}{v}\right)\right]
  \]
  We start with the simple observation that:
  \begin{alignat*}{2}
    \ite{C[\vec b \diamond \vec a]}{u}{v} &\;=_R&\;
    \begin{alignedat}[t]{2}
      \ite{a}
      {&\ite{C[\vec b \diamond \vec a',\true]}{u}{v}\\}
      {&\ite{C[\vec b \diamond \vec a',\false]}{u}{v}}
    \end{alignedat}
  \end{alignat*}
  Let $C_u[\vec b_u \diamond \vec t_u]$ and $C_v[\vec b_v \diamond \vec t_v]$ be the $R$-normal forms of $u$ and $v$. Let $C_l,C_r$ be such that :
  \begin{alignat*}{2}
    \ite{C[\vec b \diamond \vec a',\true]}{u}{v} =_R C_l[\vec b_u ,  \vec b_v ,  \vec b ,  \vec a'\diamond \vec t_u,\vec t_v]\\
    \ite{C[\vec b \diamond \vec a',\false]}{u}{v} =_R C_r[\vec b_u ,  \vec b_v ,  \vec b ,  \vec a'\diamond \vec t_u,\vec t_v]
  \end{alignat*}

  Since $a \not \in \condst(u\downarrow_R) ,  \condst(v\downarrow_R)$ we know that $a \not \in \vec b_u ,  \vec b_v$.  Moreover since  $\vec a' = \vec a \backslash \{a\}$ and $a \not \in \vec b$ we know that $a \not \in \vec b_u ,  \vec b_v ,  \vec b ,  \vec a'$. Therefore:
  \begin{gather}\label{eq:spurious-replace-eq}
    a \not \in \condst(C_l[\vec b_u ,  \vec b_v ,  \vec b ,  \vec a'\diamond \vec t_u,\vec t_v]) \quad \text{ and } \quad
    a \not \in \condst(C_r[\vec b_u ,  \vec b_v ,  \vec b ,  \vec a'\diamond \vec t_u,\vec t_v])
  \end{gather}
  In a second time we get rid in $C_l$ and $C_r$ of all the conditionals appearing in $\rho$. We let $\pvec a^{\textsf{l}}$ and $\pvec a^{\textsf{r}}$ be such that:
  \begin{equation}\label{eq:spurious-replace-eq0}
    \pvec a^{\textsf{l}} \subseteq \vec b_u ,  \vec b_v ,  \vec b ,  \vec a' \backslash \rho \quad \wedge \quad
    \pvec a^{\textsf{r}} \subseteq \vec b_u ,  \vec b_v ,  \vec b ,  \vec a' \backslash \rho
  \end{equation}
  and $C_l'$, $C_r'$ such that:
  \begin{gather}
    \label{eq:spurious-replace-eq1}
    B
    \left[
      \vec c \diamond
      \left(
        \vec w,C_l[\vec b_u ,  \vec b_v ,  \vec b ,  \vec a'\diamond \vec t_u,\vec t_v]
      \right)
    \right]
    =_R
    B
    \left[
      \vec c \diamond
      \left(
        \vec w,C_l'[\pvec a^{\textsf{l}} \diamond \vec t_u, \vec t_v ]
      \right)
    \right]\\
    \label{eq:spurious-replace-eq2}
    B
    \left[
      \vec c \diamond
      \left(
        \vec w,C_r[\vec b_u ,  \vec b_v ,  \vec b ,  \vec a'\diamond \vec t_u,\vec t_v]
      \right)
    \right]
    =_R
    B
    \left[
      \vec c \diamond
      \left(
        \vec w,C_r'[\pvec a^{\textsf{r}} \diamond \vec t_u, \vec t_v ]
      \right)
    \right]
  \end{gather}
  Therefore we deduce from \eqref{eq:spurious-replace-eq} and \eqref{eq:spurious-replace-eq0} that $a \not \in \pvec a^{\textsf{l}}$ and $a \not \in \pvec a^{\textsf{r}}$.

  \paragraph{Case 1} If there exists $c_0 \in \vec c$ such that the path $\rho$ from the root of $t$ to $\ite{C[\vec b \diamond \vec a]}{u}{v}$ contains one of the following shapes, where solid edges represent one element of the path $\rho$, and dotted edges represent a summary of a part of the path $\rho$.
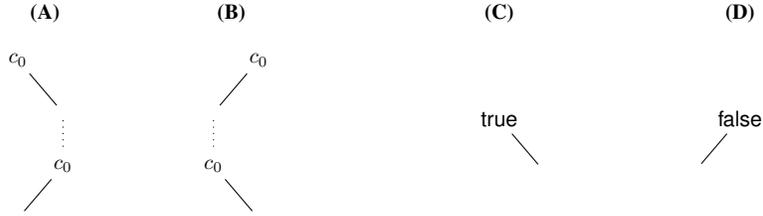
  \begin{center}
    \begin{tikzpicture}[level distance=2.5em, transform shape, scale=0.8]
      \node[above] at (0,0) {$c_0$}
      child {node{}edge from parent[draw=none]}
      child {
        node{}
        child {
          node{$c_0$}edge from parent[dotted]
          child {node{}edge from parent[solid]}
          child {node{}edge from parent[draw=none]}
        }
      };
      \node at (0.45,1) {\textbf{(A)}};

      \node[above] at (4,0) {$c_0$}
      child {
        node{}
        child {
          node{$c_0$}edge from parent[dotted]
          child {node{}edge from parent[draw=none]}
          child {node{}edge from parent[solid]}
        }
      }
      child {node{}edge from parent[draw=none]};
      \node at (3.55,1) {\textbf{(B)}};

      \node[above] at (8,-1) {$\true$}
      child {node{}edge from parent[draw=none]}
      child {node{}edge from parent[solid]};
      \node at (8,1) {\textbf{(C)}};

      \node[above] at (12,-1) {$\false$}
      child {node{}edge from parent[solid]}
      child {node{}edge from parent[draw=none]};
      \node at (12,1) {\textbf{(D)}};
    \end{tikzpicture}
  \end{center}
  In these four cases the result is easy to show. Since the proof are very similar we only describe case \textbf{(A)}: in that case we know that there exists a decomposition of $B$,$\vec c$ and $\vec w$ into, respectively, $B_1,\dots,B_5$, $\vec c_1,\dots,\vec c_5$ and $\vec w_1,\dots,\vec w_5$ such that:
  \begin{multline*}
    B\left[\vec c \diamond \left(\vec w,\ite{C[\vec b \diamond \vec a]}{u}{v}\right)\right] \equiv \\
    B_1\left[\vec c_1 \diamond
      \left(
        \vec w_1,
        \begin{alignedat}{2}
          \ite{c_0}
          {&B_2\left[\vec c_2 \diamond \vec w_2\right]\\}
          {&
            B_3\left[\vec c_3 \diamond
              \left(
                \vec w_3,
                \begin{alignedat}{2}
                  \ite{c_0}
                  {&
                    \tikz[baseline]\node[anchor=base,draw=red]{
                      $B_4\left[\vec c_4 \diamond \left(\vec w_4,\ite{C[\vec b \diamond \vec a]}{u}{v}\right)\right]$
                    };
                    \\}
                  {&B_5\left[\vec c_5 \diamond \vec w_5\right]}
                \end{alignedat}
              \right)
            \right]
          }
        \end{alignedat}
      \right)
    \right]
  \end{multline*}
  We can then rewrite the term $B_4\left[\vec c_4 \diamond \left(\vec w_4,\ite{C[\vec b \diamond \vec a]}{u}{v}\right)\right]$ using:
  \[
    \ite{b}{u}{\left(\ite{b}{\mathbf{v}}{w}\right)} \ra_R^*  \ite{b}{u}{\left(\ite{b}{\mathbf{v'}}{w}\right)} \tag{ for all term $\mathbf{v'}$}
  \]
  which yields the following term (we framed in red the part where the rewriting occurs):
  \begin{multline*}
    B\left[\vec c \diamond \left(\vec w,\ite{C[\vec b \diamond \vec a]}{u}{v}\right)\right] =_R \\
    B_1\left[\vec c_1 \diamond
      \left(
        \vec w_1,
        \begin{alignedat}{2}
          \ite{c_0}
          {&B_2\left[\vec c_2 \diamond \vec w_2\right]\\}
          {&
            B_3\left[\vec c_3 \diamond
              \left(
                \vec w_3,
                \begin{alignedat}{2}
                  \ite{c_0}
                  {&
                    \tikz[baseline]\node[anchor=base,draw=red]{
                      $B_4\left[\vec c_4 \diamond \left(\vec w_4,\ite{C[\vec b \diamond \vec a',\true]}{u}{v}\right)\right]$
                    };
                    \\}
                  {&B_5\left[\vec c_5 \diamond \vec w_5\right]}
                \end{alignedat}
              \right)
            \right]
          }
        \end{alignedat}
      \right)
    \right]
  \end{multline*}

  \paragraph{Case 2:} Let $s$ be such that $t \equiv s[\ite{C[\vec b \diamond \vec a]}{u}{v}]$.
  If none of the shapes of \textbf{Case 1} occurs, then we know that there exists $B'$ such that $s =_R B'\left[\vec c \diamond \big(\vec w,[]\big)\right]$ and the path $\rho'$ from the root to $[]$ is a subset of $\rho$ and does not contain duplicates, $\true$ and $\false$. The existence of such a $B'$ is proved by induction on the number of duplicate conditionals, $\true$ and $\false$ occurring on $\rho'$: indeed since the shape (A) and (B) (resp. (C) and (D)) are forbidden in $\rho$, we know that if we have a duplicate (resp. $\true$ or $\false$) then we can always rewrite $B$ such that the hole containing $s$ does not disappear.

  Let $\rho' = c_1,\dots,c_n$. In a second time we are going to take $B'$ as small as possible, i.e. only a branch $c_1,\dots,c_n$.

  \begin{minipage}[]{0.30\linewidth}
    \begin{center}
      \underline{\textbf{Example of if-context $B'$:}}

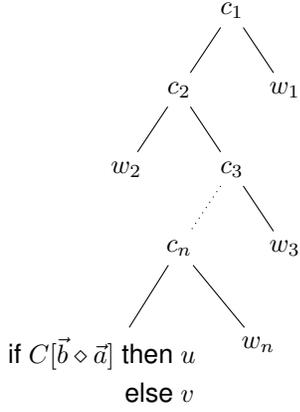
      \begin{tikzpicture}[level distance=3em,sibling distance=4em]
        \tikzstyle{level 4}=[sibling distance=6em];
        \node at (0,0) {$c_1$}
        child {
          node {$c_2$}
          child {node {$w_2$}}
          child {
            node {$c_3$}
            child {
              node{$c_n$}
              edge from parent[dotted]
              child {
                node[below] {$\begin{alignedat}{2}\ite{C[\vec b \diamond \vec a]}{&u\\}{&v}\end{alignedat}$}
                edge from parent[solid]
              }
              child {
                node[below] {$w_n$}
                edge from parent[solid]
              }
            }
            child {node {$w_3$}}
          }
        }
        child {node {$w_1$}};
      \end{tikzpicture}
    \end{center}
  \end{minipage}
  \begin{minipage}[]{0.65\linewidth}
    Wet let $\vec w = w_1,\dots,w_n$, and we have:
    \[
      s =_R B'\left[c_1,\dots,c_n\diamond w_1,\dots,w_n,
        []
      \right]
    \]
    We let $\prec_u$ be a total ordering on if-free conditional in $R$-normal form such that the $n+1$ maximum elements are $c_1 \prec_u \dots \prec c_n \prec_u a$. For all $1 \le i \le n$, we let $W_i[\vec d_i \diamond \vec e_i]$ be the $R_{\prec_u}$-normal form of $w_i$. We have:
    \[
      s =_R B'\left[c_1,\dots,c_n\diamond \left(W_i[\vec d_i \diamond \vec e_i]\right)_{i\le n},
        []
      \right]
    \]
    For all $i$, we let $W'_i[\pvec d'_i \diamond \pvec e'_i]$ be terms in $R$-normal form such that $\pvec d'_i \cap \{c_j \mid j \le i \} = \emptyset$ and:
    \[
      s =_R B'\left[c_1,\dots,c_n\diamond \left(W'_i[\pvec d'_i \diamond \pvec e'_i]\right)_{i\le n},
        []
      \right]
    \]
  \end{minipage}

  \noindent Using \eqref{eq:spurious-replace-eq1} and \eqref{eq:spurious-replace-eq2} we get:
  \[
    t =_R B'\left[c_1,\dots,c_n\diamond \left(W'_i[\pvec d'_i \diamond \pvec e'_i]\right)_{i\le n},
      \begin{alignedat}{2}
        \ite{a}{&C_l'[\pvec a^{\textsf{l}} \diamond \vec t_u, \vec t_v ]\\}{&C_r'[\pvec a^{\textsf{r}} \diamond \vec t_u, \vec t_v ]}
      \end{alignedat}
    \right]
  \]

  It is then quite easy to show by induction on the length of the reduction sequence that there exists a sequence $1 \le i_1 < \dots < i_k \le n$ and an if-context $B''$ such that:
  \begin{alignat*}{2}
    &\left(B'\left[c_1,\dots,c_n\diamond \left(W'_i[\pvec d'_i \diamond \pvec e'_i]\right)_{i\le n},
        \begin{alignedat}{2}
          \ite{a}{&C_l'[\pvec a^{\textsf{l}} \diamond \vec t_u, \vec t_v ]\\}{&C_r'[\pvec a^{\textsf{r}} \diamond \vec t_u, \vec t_v ]}
        \end{alignedat}
      \right]\right)\downarrow_{R_{\prec_u}}\\
    =_R\;\;&
    B''\left[c_{i_1},\dots,c_{i_k}\diamond \left(W'_{i_j}[\pvec d'_{i_j} \diamond \pvec e'_{i_j}]\right)_{j \le k},
      \left(
        \begin{alignedat}{2}
          \ite{a}{&C_l'[\pvec a^{\textsf{l}} \diamond \vec t_u, \vec t_v ]\\}{&C_r'[\pvec a^{\textsf{r}} \diamond \vec t_u, \vec t_v ]}
        \end{alignedat}
      \right)\downarrow_{R_{\prec_u}}
    \right]
  \end{alignat*}
  We deduce from this that  $a$ is spurious in:
  \begin{equation*}
    \ite{a}{C_l'[\pvec a^{\textsf{l}} \diamond \vec t_u, \vec t_v ]}{C_r'[\pvec a^{\textsf{r}} \diamond \vec t_u, \vec t_v ]}
  \end{equation*}
  Since $a$ will stay the top-most conditional in the $R$-normal form of this term (because of the order $\prec_u$ we chose), and since $a \ne \true$ $a \ne \false$ and $a \not \in \pvec a^{\textsf{l}},\pvec a^{\textsf{r}}$, there is only one rule that can be applied: $\ite{a}{x}{x} \ra x$. Consequently:
  \[
    C_l'[\pvec a^{\textsf{l}} \diamond \vec t_u, \vec t_v ]\;=_R\;C_r'[\pvec a^{\textsf{r}} \diamond \vec t_u, \vec t_v ]
  \]
  Hence:
  \begin{alignat*}{2}
    t =_R \;& B'\left[c_1,\dots,c_n\diamond \left(W'_i[\pvec d'_i \diamond \pvec e'_i]\right)_{i\le n},
      C_l'[\pvec a^{\textsf{l}} \diamond \vec t_u, \vec t_v ]
    \right]\\
    =_R\;& s\left[C_l'[\pvec a^{\textsf{l}} \diamond \vec t_u, \vec t_v ]\right]\\
    =_R\;&
    B
    \left[
      \vec c \diamond
      \left(
        \vec w,C_l'[\pvec a^{\textsf{l}} \diamond \vec t_u, \vec t_v ]
      \right)
    \right]
  \end{alignat*}
  Hence using \eqref{eq:spurious-replace-eq1} we get:
  \begin{equation*}
    t =_R
    B
    \left[
      \vec c \diamond
      \left(
        \vec w,C_l[\vec b_u ,  \vec b_v ,  \vec b ,  \vec a'\diamond \vec t_u,\vec t_v]
      \right)
    \right]
    =_R
    B
    \left[
      \vec c \diamond
      \left(
        \vec w, \ite{C[\vec b \diamond \vec a',\true]}{u}{v}
      \right)
    \right]\tag*{\qedhere}
  \end{equation*}
\end{proof}

\paragraph{Properties of  $R$}

\begin{proposition}
  \label{prop:split-leavest}
  For all simple term:
  \[
    B\left[\left( C_i[\vec a_i,a \diamond \vec b_i,a] \right)_i \diamond \left (D_j[\vec c_j,a \diamond \vec t_j] \right)_j\right]
  \]
  such that $a,(\vec a_i,\vec b_i)_i,(\vec c_j, \vec t_j)_j$ are if-free and in $R$-normal form and $a \not \in \vec a_i \cup \vec b_i \cup \vec c_j$, if:
  \begin{equation*}
    t \in \leavest\left(\left(
        B\left[\left( C_i[\vec a_i,a \diamond \vec b_i,a] \right)_i \diamond \left (D_j[\vec c_j,a \diamond \vec t_j] \right)_j\right]
      \right)\downarrow_R\right)
  \end{equation*}
  then:
  \begin{alignat*}{2}
    &&t \in \leavest\left(\left(
        B\left[\left( C_i[\vec a_i,\true \diamond \vec b_i,\true] \right)_i \diamond \left (D_j[\vec c_j,\true \diamond \vec t_j] \right)_j\right]
      \right)\downarrow_R\right)\\
    \text{or } \quad&&
    t \in \leavest\left(\left(
        B\left[\left( C_i[\vec a_i,\false \diamond \vec b_i,\false] \right)_i \diamond \left (D_j[\vec c_j,\false \diamond \vec t_j] \right)_j\right]
      \right)\downarrow_R\right)
  \end{alignat*}
\end{proposition}

\begin{proof} We know that:
  \begin{alignat*}{2}
    & B\left[\left( C_i[\vec a_i,a \diamond \vec b_i,a] \right)_i \diamond \left (D_j[\vec c_j,a \diamond \vec t_j] \right)_j\right]\\
    =_R \quad &
    \begin{alignedat}[t]{2}
      \ite{a}
      {&
        \tikz[baseline][anchor=base]\node[draw=red] (a) {$\displaystyle
          B\left[\left( C_i[\vec a_i,\true \diamond \vec b_i,\true] \right)_i \diamond \left (D_j[\vec c_j,\true \diamond \vec t_j] \right)_j\right]
          $} node[right,red] at (a.east) {$B_{\true}$};
        \\}
      {&
        \tikz[baseline][anchor=base]\node[draw=red] (a) {$\displaystyle
          B\left[\left( C_i[\vec a_i,\false \diamond \vec b_i,\false] \right)_i \diamond \left (D_j[\vec c_j,\false \diamond \vec t_j] \right)_j\right]
          $} node[right,red] at (a.east) {$B_{\false}$};
      }
    \end{alignedat}
  \end{alignat*}
  Let $\succ_u$ be a total order on if-free conditionals in $R$-normal form such that $a$ is minimal. It is quite simple to show that:
  \begin{alignat*}{2}
    & \left(\left(B\left[\left( C_i[\vec a_i,a \diamond \vec b_i,a] \right)_i \diamond \left (D_j[\vec c_j,a \diamond \vec t_j] \right)_j\right]\right)\downarrow_{R_{\succ_u}}\right)\\
    \equiv \quad &
    \begin{cases}
      \left(
          B\left[\left( C_i[\vec a_i,\true \diamond \vec b_i,\true] \right)_i \diamond \left (D_j[\vec c_j,\true \diamond \vec t_j] \right)_j\right]
        \right)
        \downarrow_{R_{\succ_u}} & \text{if } B_\true =_R B_\false\\[1em]
      \begin{alignedat}[t]{2}
        \ite{a}
        {&
          \left(\left(
              B\left[\left( C_i[\vec a_i,\true \diamond \vec b_i,\true] \right)_i \diamond \left (D_j[\vec c_j,\true \diamond \vec t_j] \right)_j\right]
            \right)
            \downarrow_{R_{\succ_u}}\right)
          \\}
        {&
          \left(
            \left(
              B\left[\left( C_i[\vec a_i,\false \diamond \vec b_i,\false] \right)_i \diamond \left (D_j[\vec c_j,\false \diamond \vec t_j] \right)_j\right]
            \right)
            \downarrow_{R_{\succ_u}}
          \right)
        }
      \end{alignedat} & \text{otherwise}
    \end{cases}
  \end{alignat*}
  The wanted result follows easily from Proposition~\ref{prop:trs-prec}
\end{proof}

\begin{proposition}
  \label{prop:if-same-st}
  For all simple terms:
  \[
    C[\vec a \diamond \vec b]
    \quad \quad
    B^l\left[\left( C^l_i[\vec a^l_i \diamond \vec b^l_i] \right)_i \diamond \left (D^l_j[\vec c^l_j \diamond \vec t^l_j] \right)_j\right]
    \quad \quad
    B^r\left[\left( C^r_i[\vec a^r_i \diamond \vec b^r_i] \right)_i \diamond \left (D^r_j[\vec c^r_j \diamond \vec t^r_j] \right)_j\right]
  \]
  such that:
  \begin{itemize}
  \item For all $x \in \{l,r\}$, for all $i$, $(\vec a^x_i,\vec b^x_i,\vec c^x_i, \vec t^x_i)_i$ are if-free and in $R$-normal form.
  \item $\vec a,\vec b$ are if-free, in $R$-normal form and $(\vec a \cup \vec b) \cap \{\true,\false\} = \emptyset$.
  \item $\vec b \cap (\bigcup_{x \in \{l,r\},i} \vec a^n_i,\vec b^n_i,\vec c^n_i) = \emptyset$.
  \item $\vec a \cap \vec b = \emptyset$.
  \end{itemize}
  we have that for all $x \in \{l,r\}$:
  \begin{align*}
    &\quad  t \in \leavest\left(\left(
      B^x\left[\left( C^x_i[\vec a^x_i \diamond \vec b^x_i] \right)_i \diamond \left (D^x_j[\vec c^x_j \diamond \vec t^x_j] \right)_j\right]\right)\downarrow_R\right)\\
    \implies&\quad  t \in \leavest\left(\left(\begin{alignedat}{2}
          &\ite{C[\vec a \diamond \vec b]&&}{
            B^l\left[\left( C^l_i[\vec a^l_i \diamond \vec b^l_i] \right)_i \diamond \left (D^l_j[\vec c^l_j \diamond \vec t^l_j] \right)_j\right]\\
            &&&}{B^r\left[\left( C^r_i[\vec a^r_i \diamond \vec b^r_i] \right)_i \diamond \left (D^r_j[\vec c^r_j \diamond \vec t^r_j] \right)_j\right]}
        \end{alignedat}\right)\downarrow_R\right)
  \end{align*}
\end{proposition}

\begin{proof} We prove this by induction on $|\vec a|$.
  \paragraph{Base Case}  The case $x = l$ and $x = r$ are exactly the same, therefore we assume that $x = l$. We have $C[\vec a \diamond \vec b] \equiv b$, where $b$ is an if-free conditional. Let $\succ_u$ be any total order on if-free conditionals in $R$-normal form such that $b$ is minimal. We then let $D^l[\vec a^l \diamond \vec t^l]$ and $D^r[\vec a^r \diamond \vec t^r]$ be the $R_{\succ_u}$-normal form of:
  \begin{alignat*}{2}
    B^l\left[\left( C^l_i[\vec a^l_i \diamond \vec b^l_i] \right)_i \diamond \left (D^l_j[\vec c^l_j \diamond \vec t^l_j] \right)_j\right] &
    \text{ and }& B^r\left[\left( C^r_i[\vec a^r_i \diamond \vec b^r_i] \right)_i \diamond \left (D^r_j[\vec c^r_j \diamond \vec t^r_j] \right)_j\right]
  \end{alignat*}
  Since:
  \[
    t \in \leavest\left(\left(
        B^l\left[\left( C^l_i[\vec a^l_i \diamond \vec b^l_i] \right)_i \diamond \left (D^l_j[\vec c^l_j \diamond \vec t^l_j] \right)_j\right]
      \right)\downarrow_R\right)
  \]
  we know by Proposition~\ref{prop:trs-prec} that:
  \begin{equation}\label{eq:trs-1}
    t \in \leavest\left(\left(D^l[\vec a^l \diamond \vec t^l]\right)\downarrow_{R_{\succ_u}}\right)
  \end{equation}
  Using the fact that $(\vec a^l_i,\vec b^l_i,\vec c^l_i, \vec t^l_i)_i$ are if-free and in $R$-normal form, it is simple to show by induction on the length of the reduction that $\vec a^l \subseteq (\vec a^l_i,\vec b^l_i,\vec c^l_i)_i$. Together with the fact that $b \not \in(\bigcup_{x \in \{l,r\},i} \vec a^n_i,\vec b^n_i,\vec c^n_i)$, this shows that $b \not \in \vec a^l$. Similarly $\vec a^r \subseteq (\vec a^r_i,\vec b^r_i,\vec c^r_i)_i$ and $b \not \in \vec a^r$.

  We know that:
  \[
    \begin{alignedat}{2}
      &\ite{b&&}{
        B^l\left[\left( C^l_i[\vec a^l_i \diamond \vec b^l_i] \right)_i \diamond \left (D^l_j[\vec c^l_j \diamond \vec t^l_j] \right)_j\right]\\
        &&&}{B^r\left[\left( C^r_i[\vec a^r_i \diamond \vec b^r_i] \right)_i \diamond \left (D^r_j[\vec c^r_j \diamond \vec t^r_j] \right)_j\right]}
    \end{alignedat} \quad =_R \quad
    \underbrace{\begin{alignedat}{2}
        &\ite{b&&}{
          D^l[\vec a^l \diamond \vec t^l]\\
          &&&}{D^r[\vec a^r \diamond \vec t^r]}
      \end{alignedat}}_{s}
  \]
  Since $b$ is and if-free conditional in $R$-normal form minimal for $\succ_u$, since $D^l[\vec a^l \diamond \vec t^l]$ and $D^r[\vec a^r \diamond \vec t^r]$ are in $R_{\succ_u}$-normal form and since $b \not \in \vec a^l \cup \vec a^r$, there is only one rule that may be applicable to $s$: $\ite{b}{x}{x} \ra x$.

  If the rule is not applicable then $s$ is in $R_{\succ_u}$-normal form, \eqref{eq:trs-1} implies that $t \in \leavest(s\downarrow_{R_{\succ_u}})$, which by Proposition~\ref{prop:trs-prec} shows that:
  \[
    t \in \leavest\left(\left(\begin{alignedat}{2}
          &\ite{C[\vec a \diamond \vec b]&&}{
            B^l\left[\left( C^l_i[\vec a^l_i \diamond \vec b^l_i] \right)_i \diamond \left (D^l_j[\vec c^l_j \diamond \vec t^l_j] \right)_j\right]\\
            &&&}{B^r\left[\left( C^r_i[\vec a^r_i \diamond \vec b^r_i] \right)_i \diamond \left (D^r_j[\vec c^r_j \diamond \vec t^r_j] \right)_j\right]}
        \end{alignedat}\right)\downarrow_R\right)
  \]
  If the rule is applicable then $s\downarrow_{R_{\succ_u}} \equiv D^l[\vec a^l \diamond \vec t^l]$. \eqref{eq:trs-1} implies that $t \in \leavest(s\downarrow_{R_{\succ_u}})$, which by Proposition~\ref{prop:trs-prec} shows the wanted result.

  \paragraph{Inductive Case} Assume that the result holds for $m$, and consider $\vec a$ of length $m + 1$. Again w.l.o.g. we can take $x = l$. Let $a \in \vec a$, and $\vec a_0 = \vec a \backslash a$. We know that:
  \begin{itemize}
  \item There exist $C'[\vec a' \diamond \vec b']$ and $C''[\vec a'' \diamond \vec b'']$ such that:
    \[
      C[\vec a \diamond \vec b] =_R \ite{a}{C'[\vec a' \diamond \vec b']}{C''[\vec a'' \diamond \vec b'']}
    \]
    with $\vec a'\cup \vec a'' \subseteq \vec a_0$ and $\vec b'\cup\vec b'' \subseteq \vec b$.
  \item For all $x \in \{l,r\}$, there exist $C'^x_i[\vec a'^x_i \diamond \vec b'^x_i]$ and $C''^x_i[\vec a''^x_i \diamond \vec b''^x_i]$ such that:
    \[
      C^x_i[\vec a^x_i \diamond \vec b^x_i] =_R \ite{a}{C'^x_i[\vec a'^x_i \diamond \vec b'^x_i]}{C''^x_i[\vec a''^x_i \diamond \vec b''^x_i]}
    \]
    with $\vec a'^x_i\cup \vec a''^x_i \subseteq \vec a^x_i\backslash \{a\}$ and $\vec b'^x_i\cup\vec b''^x_i \subseteq \vec b^x_i \cup \{\true,\false\} \backslash \{a\}$.
  \item For all $x \in \{l,r\}$, there exist $D'^x[\vec c'^x_j \diamond \vec t'^x_j]$ and $D''^x[\vec c''^x_j \diamond \vec t''^x_j]$ such that:
    \[
      D^x_j[\vec c^x_j \diamond \vec t^x_j] =_R \ite{a}{D'^x_j[\vec c'^x_j \diamond \vec t'^x_j]}{D''^x_j[\vec c''^x_j \diamond \vec t''^x_j]}
    \]
    with $\vec c'^x_j\cup \vec c''^x_j \subseteq \vec c^x_j\backslash \{a\}$ and $\vec t'^x_j\cup\vec t''^x_j \subseteq \vec t^x_j \cup \{\true,\false\} \backslash \{a\}$.
  \end{itemize}
  Using Proposition~\ref{prop:split-leavest} we know that:
  \begin{alignat}{2}
    &&  t \in \leavest\left(\left(
        B^l\left[\left( C'^l_i[\vec a'^l_i \diamond \vec b'^l_i] \right)_i \diamond \left (D'^l_j[\vec c'^l_j \diamond \vec t'^l_j] \right)_j\right]
      \right)\downarrow_R\right)\label{eq:trs-proof1}\\
    \text{or } \quad &&  t \in \leavest\left(\left(
        B^l\left[\left( C''^l_i[\vec a''^l_i \diamond \vec b''^l_i] \right)_i \diamond \left (D''^l_j[\vec c''^l_j \diamond \vec t''^l_j] \right)_j\right]
      \right)\downarrow_R\right)\label{eq:trs-proof2}
  \end{alignat}
  Assume that we are in Case~\eqref{eq:trs-proof1} (the other case is exactly the same). We can then rewrite the initial term as follows:
  \begin{alignat*}{2}
    &&\left.\begin{alignedat}{2}
        &\ite{C[\vec a \diamond \vec b]&&}{
          B^l\left[\left( C^l_i[\vec a^l_i \diamond \vec b^l_i] \right)_i \diamond \left (D^l_j[\vec c^l_j \diamond \vec t^l_j] \right)_j\right]\\
          &&&}{B^r\left[\left( C^r_i[\vec a^r_i \diamond \vec b^r_i] \right)_i \diamond \left (D^r_j[\vec c^r_j \diamond \vec t^r_j] \right)_j\right]}
      \end{alignedat}\right\}s\\
    &=_R\quad&
    \left.\begin{alignedat}{2}
        &\ite{a&&}
        {
          \tikz[baseline]\node[anchor=base,draw=red] (a) {$\displaystyle\begin{alignedat}[t]{2}
              &\ite{C'[\vec a' \diamond \vec b']&&}{
                B^l\left[\left( C'^l_i[\vec a'^l_i \diamond \vec b'^l_i] \right)_i \diamond \left (D'^l_j[\vec c'^l_j \diamond \vec t'^l_j] \right)_j\right]\\
                &&&}{B^r\left[\left( C'^r_i[\vec a'^r_i \diamond \vec b'^r_i] \right)_i \diamond \left (D'^r_j[\vec c'^r_j \diamond \vec t'^r_j] \right)_j\right]}
            \end{alignedat}$} node[right,red] at (a.east) {$s_l$};
          \\
          &&&}
        {
          \begin{alignedat}[t]{2}
            &\ite{C''[\vec a'' \diamond \vec b'']&&}{
              B^l\left[\left( C''^l_i[\vec a''^l_i \diamond \vec b''^l_i] \right)_i \diamond \left (D''^l_j[\vec c''^l_j \diamond \vec t''^l_j] \right)_j\right]\\
              &&&}{B^r\left[\left( C''^r_i[\vec a''^r_i \diamond \vec b''^r_i] \right)_i \diamond \left (D''^r_j[\vec c''^r_j \diamond \vec t''^r_j] \right)_j\right]}
          \end{alignedat}
        }
      \end{alignedat}\right\}s'
  \end{alignat*}
  We start by checking that the induction hypothesis can be applied to the red framed term $s_l$. The first two conditions are trivial, let us check the last two ones:
  \begin{itemize}
  \item Since $\vec a' \subseteq \vec a$ and $\vec b' \subseteq \vec b$, it is easy to check that $\vec a' \cap \vec b' = \emptyset$.
  \item Since:
    \[
      \vec a'^x_i \subseteq \vec a^x_i \qquad \vec b'^x_i \subseteq \vec b^x_i \cup \{\true,\false\} \qquad \vec c'^x_j \subseteq \vec c^x_j
    \]
    we know that:
    \[\textstyle
      \left(\bigcup_{i,x \in \{l,r\}} \vec a'^x_i,\vec b'^x_i,\vec c'^x_i\right) \subseteq \left(\bigcup_{i,x \in \{l,r\}} \vec a^x_i,\vec b^x_i,\vec c^x_i\right) \cup \{ \true, \false\}
    \]
    From the fact that $\vec b \cap (\bigcup_{x \in \{l,r\},i} \vec a^x_i,\vec b^x_i,\vec c^x_i) = \emptyset$ and $\vec b \cap \{\true,\false\} = \emptyset$ we deduce that:
    \[\textstyle
      \vec b \cap \left(\bigcup_{i,x \in \{l,r\}} \vec a'^x_i,\vec b'^x_i,\vec c'^x_i\right) = \emptyset
    \]
    Finally since $\vec b' \subseteq \vec b$ we get:
    \[\textstyle
      \vec b' \cap \left(\bigcup_{i,x \in \{l,r\}} \vec a'^x_i,\vec b'^x_i,\vec c'^x_i\right) = \emptyset
    \]
  \end{itemize}
  Hence by induction hypothesis:
  \begin{equation}\label{eq:trs-5}
    t \in \leavest\left(\left( \begin{alignedat}{2}
          &\ite{C'[\vec a' \diamond \vec b']&&}{
            B^l\left[\left( C'^l_i[\vec a'^l_i \diamond \vec b'^l_i] \right)_i \diamond \left (D'^l_j[\vec c'^l_j \diamond \vec t'^l_j] \right)_j\right]\\
            &&&}{B^r\left[\left( C''^r_i[\vec a''^r_i \diamond \vec b''^r_i] \right)_i \diamond \left (D''^r_j[\vec c''^r_j \diamond \vec t''^r_j] \right)_j\right]}
        \end{alignedat}\right) \downarrow_R\right)
  \end{equation}
  Moreover as $\vec a'\cup \vec a'' \subseteq \vec a_0 = \vec a \backslash \{a\}$ and $\vec a \cap \vec b = \emptyset$, we know that:
  \[\textstyle
    a \not \in \vec a' \cup \vec a'' \cup \vec b' \cup \vec b'' \cup  \left(\bigcup_{i,x \in \{l,r\}} \vec a^n_i,\vec b^n_i,\vec c^n_i\right)
  \]
  Since:
  \[
    \vec a'^x_i \cup \vec a''^x_i \subseteq \vec a^x_i
    \qquad \vec b'^x_i \cup  \vec b''^x_i \subseteq \vec b^x_i \cup \{\true,\false\}
    \qquad \vec c'^x_j \cup \vec c''^x_j \subseteq \vec c^x_j
  \]
  and using the fact that $a \not \in \{\true,\false\}$, we get from \eqref{eq:trs-5} that:
  \[\textstyle
    a \not \in \vec a' \cup \vec a'' \cup \vec b' \cup \vec b''
    \cup  \left(\bigcup_{i,x \in \{l,r\}} \vec a'^n_i,\vec b'^n_i,\vec c'^n_i\right)
    \cup \left(\bigcup_{i,x \in \{l,r\}} \vec a''^n_i,\vec b''^n_i,\vec c''^n_i\right)
  \]
  Hence we can apply again the induction hypothesis (with $m = 1$) to $s'$, which shows that $t \in \leavest(s' \downarrow_R) \equiv \leavest(s \downarrow_R)$.
\end{proof}

\paragraph{Sufficient Conditions for Non Spuriousness of Leaves}
We now give sufficient conditions to show that a leave term is not spurious.
\begin{proposition}
  \label{prop:uml-abounded-tech}
  For all simple term:
  \[
    s \equiv A\left[
      \vec d
      \diamond
      \left(
        B_l\left[
          \left(
            \beta_{i,l}
            % C_{i,l}[\vec a_{i,l} \diamond \vec b_{i,l}]
          \right)_i
          \diamond
          \left(
            \gamma_{j,l}
            % D_{j,l}[\vec c_{j,l} \diamond \vec t_{j,l}]
          \right)_j
        \right]
      \right)_l
    \right]
  \]
  such that:
  \begin{itemize}
  \item[(i)] $\vec d$ are if-free and in $R$-normal form, and for all $i,j,l$, $\condst(\beta_{i,l}\downarrow_R) \cap \vec \leavest(\beta_{i,l}\downarrow_R) = \emptyset$.
  \item[(ii)] $\left(\vec d \cup \bigcup_i \leavest(\beta_{i,l}\downarrow_R)\right) \cap \{\true,\false\} = \emptyset$.
  \item[(iii)] For every positions $p < p'$ in $A[\_\diamond (B_l)_l]$ such that $s_{|p} \equiv \zeta$  and $s_{|p'} \equiv \zeta'$, we have $\leavest(\zeta\downarrow_R) \cap \leavest(\zeta'\downarrow_R) = \emptyset$.
  \item[(iv)] For all $l$, for all $i,j$, $\leavest(\beta_{i,l}\downarrow_R) \cap \leavest(\beta_{j,l}\downarrow_R) \ne \emptyset$ implies that $\beta_{i,l} \equiv \beta_{j,l}$.
  \item[(v)] For all $l$, the following couple of sets is well-nested:
    \[
      \left(
        \left\{ \beta_{i,l}\downarrow_R \mid i\right\},
        \left\{\gamma_{j,l}\downarrow_R \mid j\right\}_j
      \right)
    \]
  \end{itemize}
  for all $l,j$, there exists $t \in \vec t_{j,l}$ such that $t \in \leavest(s \downarrow_R)$.
\end{proposition}

\begin{proof}
  For all $l,i,j$, we let $C_{i,l}[]$, $D_{j,l}[]$ be if-contexts and $\vec a_{i,l}$, $\vec b_{i,l},\vec c_{j,l}$, $\vec t_{j,l}$ be if-free terms in $R$-normal form such that:
  \begin{mathpar}
  \vec a_{i,l} \equiv \condst(\beta_{i,l}\downarrow_R)

  \vec b_{i,l} \equiv \leavest(\beta_{i,l}\downarrow_R)

  \vec c_{i,l} \equiv \condst(\gamma_{j,l}\downarrow_R)

  \vec t_{i,l} \equiv \leavest(\gamma_{j,l}\downarrow_R)\\

  \beta_{i,l} \downarrow_R \;\equiv\; C_{i,l}[\vec a_{i,l} \diamond \vec b_{i,l}]

  \gamma_{j,l} \downarrow_R \;\equiv\; D_{j,l}[\vec c_{j,l} \diamond \vec t_{j,l}]
  \end{mathpar}
  We start by showing that this is the case if $\vec d = \emptyset$ and $A \equiv []$ in the first part of the proof, and then will deal with the general case in the second part.
  \paragraph{Part 1}
  Since $\vec d = \emptyset$ we know that:
  \[
    s \equiv B\left[\left( C_{i}[\vec a_{i} \diamond \vec b_{i}] \right)_i \diamond \left (D_{j}[\vec c_{j} \diamond \vec t_{j}] \right)_j\right]
  \]
  satisfying conditions (i) to (v).

  We let $\textsf{nested-if}(B)$ be the maximum number of nested $\ite{}{}{}$, and $\vec a_0$ be the conditionals of the basic conditional at the root of $B$. We prove the proposition by induction on ($\textsf{nested-if}(B)$,$|\vec a_0|$), ordered with the lexicographic ordering.

  \paragraph{Part 1: Base Case} The base case is simple: it suffices to notice that since $\vec c,\vec t$ are if-free and in $R$-normal form:
  \[
    \leavest(s \downarrow_R) = \leavest(D[\vec c \diamond \vec t] \downarrow_R) \subseteq \vec t
  \]
  This is a simple proof by induction on the length of the reduction sequence.

  \paragraph{Part 1: First Inductive Case} Assume that the property holds for ($n, \omega$) and lets show that it holds for ($n+1,0$). Consider:
  \begin{alignat*}{2}
    &s \equiv \ite{b_0&&}{
      B^l\left[\left( C^l_i[\vec a^l_i \diamond \vec b^l_i] \right)_i \diamond \left (D^l_j[\vec c^l_j \diamond \vec t^l_j] \right)_j\right]\\
      &&&}{B^r\left[\left( C^r_i[\vec a^r_i \diamond \vec b^r_i] \right)_i \diamond \left (D^r_j[\vec c^r_j \diamond \vec t^r_j] \right)_j\right]}
  \end{alignat*}
  where $B^l$ and $B^r$ are such that $\textsf{nested-if}(B^l) \le n$ and $\textsf{nested-if}(B^r) \le n$. Using the well-nested condition, we know that for all $i \ne 0, x \in \{l,r\}$, there exists two if-context $C'^x_i, C''^x_i$ such that:
  \begin{equation*}
    C^x_i[\vec a_i^x \diamond \vec b_i^x] =_R \ite{b_0}{C'^x_i[\vec a_i'^x \diamond \vec b_i'^x]}{C''^x_i[\vec a_i''^x \diamond \vec b_i''^x]}
  \end{equation*}
  where $\vec a_i'^x,\vec a_i''^x \subseteq \vec a_i^x \backslash b_0$ and $\vec b_i'^x, \vec b_i''^x \subseteq \vec b_i^x$. Similarly for all $j,x \in \{l,r\}$, we know that there exists two if-context $D_j'^x, D_j''^x$ such that:
  \begin{equation*}
    D_j^x[\vec c_j^x \diamond \vec t_j^x] =_R \ite{b_0}{D_j'^x[\vec c_j'^x \diamond \vec t_j'^x]}{D''^x_j[\vec c_j''^x \diamond \vec t_j''^x]}
  \end{equation*}
  where $\vec c_j'^x,\vec c_j''^x \subseteq \vec c_j^x \backslash b_0$ and $\vec t_j'^x, \vec t_j''^x \subseteq \vec t_j^x$.  We can rewrite the term $s$ as follows:
  \begin{align*}
    s &\equiv
    \begin{alignedat}[t]{2}
      &\ite{b_0&&}{
        \tikz[baseline]\node[anchor=base,draw=red] (a) {$\displaystyle
          B^l\left[\left( C'^l_i[\vec a'^l_i \diamond \vec b'^l_i] \right)_i \diamond \left (D'^l_j[\vec c'^l_j \diamond \vec t'^l_j] \right)_j\right]
          $} node[right,red] at (a.east) {$s_l$};\\
        &&&}
      {
        \tikz[baseline]\node[anchor=base,draw=red] (a) {$\displaystyle
          B^r\left[\left( C''^r_i[\vec a''^r_i \diamond \vec b''^r_i] \right)_i \diamond \left (D''^r_j[\vec c''^r_j \diamond \vec t''^r_j] \right)_j\right]
          $} node[right,red] at (a.east) {$s_r$};
      }
    \end{alignedat}
  \end{align*}
  Using the induction hypothesis on the framed term $s_l$ (resp. $s_r$), we know that for all $j$, there exists $t \in \vec t'^l_j \subseteq \vec t^l_j$ (resp. $t \in \vec t'^r_j \subseteq \vec t^r_j$) such that:
  \begin{alignat*}{2}
    &t \in \leavest\left(B^l\left[\left( C'^l_i[\vec a'^l_i \diamond \vec b'^l_i] \right)_i \diamond \left (D'^l_j[\vec c'^l_j \diamond \vec t'^l_j] \right)_j\right]\right)\downarrow_R\\
    \Big(\text{resp. }\quad &
    t \in \leavest\left(B^r\left[\left( C''^r_i[\vec a''^r_i \diamond \vec b''^r_i] \right)_i \diamond \left (D''^r_j[\vec c''^r_j \diamond \vec t''^r_j] \right)_j\right]\right)\downarrow_R
    \Big)
  \end{alignat*}
  We now want to apply Proposition~\ref{prop:if-same-st} to show that $t \in  \leavest(s \downarrow_R)$. The only difficulty lies in showing that:
  \[\textstyle
    b_0 \cap \left(\bigcup_{i} \vec a'^l_i,\vec a''^r_i,\vec b'^l_i,\vec b''^r_i,\vec c'^l_i,\vec c''^r_i\right) = \emptyset
  \]
  We know that
  \(
  b_0 \cap \left(\bigcup_{i} \vec a'^l_i,\vec a''^r_i,\vec c'^l_i,\vec c''^r_i\right) = \emptyset
  \) (since $\vec a'^l_i \subseteq \vec a^l_i \backslash \{b_0\},\dots$), so it only remains to show that $b_0 \not \in \bigcup_{i} \vec b'^l_i,\vec b''^r_i$. This follows from the hypothesis (iii), since $b_0$ is at the root of $B$ and therefore for all $i$, $b_0 \not \in \vec b^l_i \supseteq \vec b'^l_i$ (resp. $b_0 \not \in\vec b^r_i \supseteq \vec b''^r_i$).

  \paragraph{Part 1: Second Inductive Case} Now assume that the property holds for ($n+1$,$k$) and lets show that it holds for ($n+1$,$k+1$). Consider:
  \begin{alignat*}{2}
    &s \equiv \ite{C_0[\vec a_0 \diamond \vec b_0]&&}{
      B^l\left[\left( C_i[\vec a_i \diamond \vec b_i] \right)_{i\in I^l} \diamond \left (D_j[\vec c_j \diamond \vec t_j] \right)_{j\in J^l}\right]\\
      &&&}{B^r\left[\left( C_i[\vec a_i \diamond \vec b_i] \right)_{i\in I^r} \diamond \left (D_j[\vec c_j \diamond \vec t_j] \right)_{j\in J^r}\right]}
  \end{alignat*}
  where $B^l$ and $B^r$ are such that of $\textsf{nested-if}(B^l) \le n$, $\textsf{nested-if}(B^r) \le n$, and $|\vec a_0| = k+1$.

  We are looking for $m$ such that for all $j$, $\vec a_{m} \cap \vec b_{j} = \emptyset$, $\vec b_m \subseteq \vec a_0$ and $\vec a_m,\vec b_m \subseteq \vec a_0,\vec b_0$.
  \begin{itemize}
  \item If there exists $k_0$ such that $\vec a_0 \cap \vec b_{k_0} \ne \emptyset$ then we know that $\vec a_{k_0},\vec b_{k_0} \subseteq \vec a_0$ and $\vec a_{k_0},\vec b_{k_0} \subset \vec a_0,\vec b_0$. We repeat this process and build a sequence $(k_l)_l$ such that for all $l$, $\vec a_{k_{l+1}},\vec b_{k_{l+1}} \subseteq \vec a_{k_l}$ and $\vec a_{k_{l+1}},\vec b_{k_{l+1}} \subset \vec a_{k_l},\vec b_{k_{l}}$.

    This sequence is necessarily finite. Let $l_{max}$ it length and let $m = k_{l_{max}-1}$. We know that for all $j$, $\vec a_{m} \cap \vec b_{j} = \emptyset$ (otherwise we could extend the sequence). Moreover we know that $\vec b_m \subseteq  \vec a_0$ and $\vec a_m,\vec b_m \subset \vec a_0,\vec b_0$.
  \item  If for all $k_0$, $\vec a_0 \cap \vec b_{k_0} = \emptyset$ then we take $m=0$.
  \end{itemize}

  Using the well-nested hypothesis, we know that for all $j \in I^l \cup I^r$, there exist two if-context $C_j', C_j''$ such that:
  \begin{equation*}
    C_j[\vec a_j \diamond \vec b_j] =_R \ite{C_m[\vec a_m \diamond \vec b_m]}{C_j'[\vec a_j' \diamond \vec b_j']}{C''_j[\vec a_j'' \diamond \vec b_j'']}
  \end{equation*}
  where $\vec a_j',\vec a_j'' \subseteq \vec a_j \backslash \vec b_m$ and $\vec b_j', \vec b_j'' \subseteq \vec b_j$. Similarly  there exist two if-context $D_j', D_j''$ such that:
  \begin{equation*}
    D_j[\vec c_j \diamond \vec t_j] =_R \ite{C_m[\vec a_m \diamond \vec b_m]}{D_j'[\vec c_j' \diamond \vec t_j']}{C''_j[\vec c_j'' \diamond \vec t_j'']}
  \end{equation*}
  where $\vec c_j',\vec c_j'' \subseteq \vec a_j \backslash \vec b_m$ and $\vec t_j', \vec t_j'' \subseteq \vec t_j$.

  We let $B^l_\true$ (resp. $B^r_\true$) be the if-context obtained from $B^l$ (resp. $B^r$) by replacing every conditional hole~$[]_i$ that is mapped to $C_m[\vec a_m \diamond \vec b_m]$ in $s$ by its \textsf{then} branch. Similarly we define $B^l_\false$ (resp. $B^r_\false$) by replacing every conditional hole $[]_i$ that is mapped to $C_m[\vec a_m \diamond \vec b_m]$ in $s$ by its \textsf{else} branch. By consequence:
  \begin{alignat*}{2}
    &s \equiv \ite{C_m[\vec a_m \diamond \vec b_m]&&}
    {
      \tikz[baseline]\node[anchor=base,draw=red] (a) {$\displaystyle
        \begin{alignedat}[t]{2}
          & \ite{C_0'[\vec a'_0 \diamond \vec b'_0]&&}{
            B^l_\true\left[\left( C'_i[\vec a'_i \diamond \vec b'_i] \right)_{i\in I^l_\true} \diamond \left (D'_j[\vec c'_j \diamond \vec t'_j] \right)_{j\in J^l_\true}\right]\\
            &&&}{B^r_\true\left[\left( C'_i[\vec a'_i \diamond \vec b'_i] \right)_{i\in I^r_\true} \diamond \left (D'_j[\vec c'_j \diamond \vec t'_j] \right)_{j\in J^r_\true}\right]}
        \end{alignedat}
        $} node[right,red] at (a.east) {$s_\true$};\\
      &&&}{
      \tikz[baseline]\node[anchor=base,draw=red] (a) {$\displaystyle
        \begin{alignedat}[t]{2}
          & \ite{C_0''[\vec a''_0 \diamond \vec b''_0]&&}{
            B^l_\false\left[\left( C''_i[\vec a''_i \diamond \vec b''_i] \right)_{i\in I^l_\false} \diamond \left (D''_j[\vec c''_j \diamond \vec t''_j] \right)_{j\in J^l_\false}\right]\\
            &&&}{B^r_\false\left[\left( C''_i[\vec a''_i \diamond \vec b''_i] \right)_{i\in I^r_\false} \diamond \left (D''_j[\vec c''_j \diamond \vec t''_j] \right)_{j\in J^r_\false}\right]}
        \end{alignedat}
        $} node[right,red] at (a.east) {$s_\false$};\\
    }
  \end{alignat*}
  We then have the following property: $J^l = J^l_\true \cup J^l_\false$, and $J^r = J^r_\true \cup J^r_\false$.

  We want to show that for all $j \in J^l \cup J^r,\, \exists t \in \vec t_j.\, t \in \leavest(s \downarrow_R)$. Let $j \in J^l$ (the proof for $j \in J^r$ is similar), then either $j \in J^l_\true$ or $j \in J^l_\false$. In the former case we apply the induction hypothesis to $s_\true$, and in the latter to $s_\false$. Lets check that the premises hold for $s_\true$ (the same proof works for $s_\false$):
  \begin{itemize}
  \item (i) and (ii) trivially hold.

  \item (iii) is simple, as we only removed some nodes from the if-context and (iii) is stable by embedding.

  \item Checking that (iv) holds is straightforward. Assume that there exists $i,j \in I^l_\true \cup I^r_\false \cup \{0\}$ such that $\vec b'_{i} \cap \vec b'_{j} \ne \emptyset$. Since $\vec b'_{i} \subseteq \vec b_i$ and $\vec b'_{j} \subseteq \vec b_j$ we know that $\vec b_{i} \cap \vec b_{j} \ne \emptyset$. Therefore $C_{i}[\vec a_{i} \diamond \vec b_{i}] \equiv  C_{j}[\vec a_{j} \diamond \vec b_{j}]$. Hence w.l.o.g. we can assume that:
    \[
      C'_{i}[\vec a'_{i} \diamond \vec b'_{i}] \equiv  C'_{j}[\vec a'_{j} \diamond \vec b'_{j}] \qquad \text{and} \qquad C''_{i}[\vec a''_{i} \diamond \vec b''_{i}] \equiv  C''_{j}[\vec a''_{j} \diamond \vec b''_{j}]
    \]
  \item Using the inductive property of well-nested couples (item (iv)) we know that the following  couple of sets is well-nested:
    \[
      \left(
        \left\{ C'_{i}[\vec a'_{i} \diamond \vec b'_{i}] \mid i \in I^l \cup I^r \cup \{0\} \right\},
        \left\{D'_{j}[\vec c'_{j} \diamond \vec t'_{j}] \mid j \in J^l \cup J^r \right\}_j
      \right)
    \]
    Since if $\left(\mathcal{C},\mathcal{D}\right)$ is well-nested and $\mathcal{C}' \subseteq \mathcal{C} \wedge \mathcal{D}' \subseteq \mathcal{D}$ then $\left(\mathcal{C}',\mathcal{D}'\right)$ is well-nested, we know that the following  couple of sets is well-nested:
    \[
      \left(
        \left\{ C'_{i}[\vec a'_{i} \diamond \vec b'_{i}] \mid i \in I^l_\true \cup I^r_\true \cup \{0\}\right\},
        \left\{D'_{j}[\vec c'_{j} \diamond \vec t'_{j}] \mid j \in J^l_\true \cup J^r_\true \right\}_j
      \right)
    \]
  \end{itemize}
  Since $\vec a_0' \subset \vec a_0$ (resp. $\vec a_0'' \subset \vec a_0$), we can apply the induction hypothesis to $s_\true$ (resp $s_\false$), which shows that for all $j \in J^l_\true$ (resp. $j \in J^r_\true$), there exists $t \in \vec t'_j$ such that $t \in \leavest(s_\true \downarrow_R)$ (resp. $t \in \leavest(s_\false \downarrow_R)$).

  Let $S = I^l \cup I^r \cup \{0\} \cup J^l \cup J^r$. Let $S_m$ be the subset of $ I^l \cup I^r \cup \{0\} $ such that for all $i \in S_m$, $C_{i}[\vec a_{i} \diamond \vec b_{i}] \equiv C_{m}[\vec a_{m} \diamond \vec b_{m}]$ and $S' = S \backslash S_m$. We now want to apply Proposition~\ref{prop:if-same-st} to show that $t \in  \leavest(s \downarrow_R)$. The only difficulty lies in showing that:
  \[\textstyle
    b_m \cap \left(\bigcup_{i \in S'} \vec a'_i,\vec a''_i,\vec b'_i,\vec b''_i,\vec c'_i,\vec c''_i\right) = \emptyset
  \]
  We know that
  \(
  b_m \cap \left(\bigcup_{i \in S'} \vec a'_i,\vec a''_i,\vec c'_i,\vec c''_i\right) = \emptyset
  \) (since $\vec a'_i \subseteq \vec a_i \backslash \vec b_m,\dots$), so it only remains to show that:
  \begin{equation}\label{eq:abounded-uml-tech-1}
    \vec b_m \cap \bigcup_{i \in S'} \vec b'_i,\vec b''_i = \emptyset
  \end{equation}
  Using hypothesis (iv) we know that for all $i \in S, \vec b_i \cap \vec b_m \ne \emptyset$ implies $i \in S_m$. Therefore since $\vec b'_i \subseteq \vec b_i$ (resp. $\vec b''_i \subseteq \vec b_i$), if $\vec b_m \cap \vec b'_i \ne \emptyset$ (resp. $\vec b_m \cap \vec b''_i \ne \emptyset$) then $i \in S_m$. Since $S' = S \backslash S_m$, we know that \eqref{eq:abounded-uml-tech-1} holds.

  \paragraph{Part 2} The proof of the general case is exactly the same than the one we did for the first inductive case of Part~1.
\end{proof}

%%% Local Variables:
%%% mode: latex
%%% TeX-master: "ms"
%%% End:

\newpage

\section{If-Free Conditionals}
\label{app-section:if-free-conds}

Given an if-free term $s$ in $R$-normal form, $s$ can be rewritten using $R$ into a more complex term:
\[
  u \equiv C\left[\left(D_i\left[\vec a_i \diamond \vec b_i \right]\right)_i \diamond \vec t \right]
\]
that is not if-free. Basically, the following proposition shows that as long as the term $u$ does not contain $\true$ and $\false$ conditionals, the term $s$ will always appear in the right-most and left-most branches of $C$. This is actually an invariant preserved by the term rewriting system $R$ without the rules:
\begin{mathpar}
  \ite{\true}{v}{w} \ra w 

  \ite{\false}{v}{w} \ra w
\end{mathpar}
\begin{proposition}
  \label{prop:left-most-ind}
  For all if-free term $s$ in $R$-normal form, if $s =_R C\left[\left(D_i\left[\vec a_i \diamond \vec b_i \right]\right)_i \diamond \vec t \right]$ where:
  \begin{itemize}
  \item $\vec t \cup \bigcup_i( \vec a_i \cup \vec b_i)$ are if-free and in $R$-normal form.
  \item Let $i$ be such that $D_i\left[\vec a_i \diamond \vec b_i \right]$ is a term appearing on the left-most (resp. right-most) branch of $C$. Then $\false \not \in \vec a_i \cup \vec b_i$ (resp.  $\true \not \in \vec a_i \cup \vec b_i$).
  \end{itemize}
  Then the left-most (resp. right-most) element of $\vec t$ is $s$.
\end{proposition}

\begin{proof}
  If suffices to show that the existence of a decomposition satisfying these two properties is preserved by $\ra_R$, which is simple. We conclude by observing that since $s$ is if-free, the only decomposition of $s \downarrow_R$ into $C\left[\left(D_i\left[\vec a_i \diamond \vec b_i \right]\right)_i \diamond \vec t \right]$ is such that $C \equiv []$. Hence $\vec t$ is a single element $u$, and $u \equiv s \downarrow_R \equiv s$.
\end{proof}

We are now ready to prove Proposition~\ref{prop:iffree-false-body}, which we recall below.
\begin{proposition*}
  % \label{prop:iffree-false}
  Let $b$ an if-free conditional in $R$-normal, with $b \not \equiv \false$ (resp. $b \not \equiv \true$). Then there exists no derivation of $b \sim \false$ (resp. $b \sim \true$) in $\mathcal{A}_\succ$.
\end{proposition*}

\begin{proof}
  We prove only that there is no derivation of $b \sim \false$ in $\mathcal{A}_\succ$ (the proof that there is no derivation of $b \sim \true$ in $\mathcal{A}_\succ$ is exactly the same). We prove this by contradiction. Let $b$ an if-free conditional in $R$-normal form, and let $P$ be such that $P \npfproof b \sim \false$. We choose $b$ such that $P$ is of minimal size. 

  First the minimality of the derivation implies that for all $\sfh \in \setindex(P)$, for all $b_0$ such that $b_0 \lecs^{\sfh} (b,P)$ or $b_0 \lecs^{\sfh} (\false,P)$, $b_0 \not = \false$. Let $H = \cspos(P)$. We now focus on the left-most branch of the proof:\\
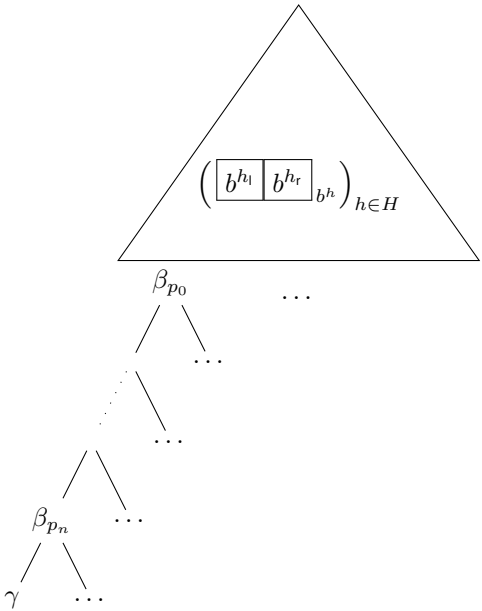
  \begin{minipage}{0.4\linewidth}
    \begin{center}
      \begin{tikzpicture}[sibling distance=3em, level distance=3em]
        \draw (0,0) -- (2.4,-3.4) -- (-2.4,-3.4) node[below=1em,midway] {$\dots$} -- cycle;
        \node at (0,-2.4) {$\left(\splitbox{b^{h_\sfl}}{b^{h_\sfr}}{b^h}\right)_{h\in H}$};
        \path (-1.7,-3.4) node[below] {$\beta_{p_0}$}
        child {
          node {}
          child {
            node {}edge from parent[loosely dotted]
            child {
              node {$\beta_{p_n}$}edge from parent[solid]
              child {
                node {$\gamma$}
              }
              child {node{$\dots$}edge from parent[solid]}
            }
            child {node{$\dots$}edge from parent[solid]}
          }
          child {node{$\dots$}edge from parent[solid]}
        }
        child {node{$\dots$}edge from parent[solid]};
      \end{tikzpicture}
    \end{center}
  \end{minipage}
  \begin{minipage}{0.6\linewidth}
    Let $l \in \prooflabel(P)$. First we show that for all $\beta \lecond^{\epsilon,l} (b,P)$, $\beta \ne_R \false$. Assume that this is not the case, let $\beta =_R \false$ and $\beta'$ be such that $(\beta,\beta') \lesimcond^{\epsilon,l} (b \sim \false,P)$. If $\beta =_R \beta' =_R \false$ then there is an easy proof cut elimination which yields a smaller proof $P'$ of $b \sim \false$. 

    Hence assume $\beta' \not =_R \false$. If $\beta =_R \false$ then $\leavest(\beta\downarrow_R) = \{\false\} = \leavest(\false\downarrow_R)$. As $\beta$ is a normalized basic conditional, using Proposition~\ref{prop:bas-cond-charac} we have $\beta \equiv \false$.

    There exists a derivation of $\beta \sim \beta'$ in $\fas^* \cdot \dup^* \cdot \cca$. Since $\beta \equiv \false$, no rules in $\fas$ are applied. Therefore the derivation is only an application of $\cca$, which is not possible.
    
    Similarly we do not have $\beta \not =_R \false$ and $\beta'  =_R \false$.

    Using Proposition~\ref{prop:bas-cond-charac} we know that $\beta \ne_R \false$ implies that for all $u \in \leavest{(\beta\downarrow_R)}$, $u \not \equiv \false$. Moreover by assumptions, for all $a \in \condst{(\beta\downarrow_R)}$, $a \not \equiv \false$.
    
    We let $(\gamma,\gamma') \leleave^{\epsilon,l} (b \sim \false,P)$ be the left-most elements (as shown in the Figure). For all $a \in \condst(\gamma\downarrow_R)$, $a \not \equiv \false$. Hence if we let $u_0 \in \leavest(\gamma\downarrow_R)$ be the left-most leave element of $\gamma\downarrow_R$, then by Proposition~\ref{prop:left-most-ind} we know that $u_0 \equiv b$.

    Similarly, by applying the exact same reasoning to the other side, we know that the left-most leaf element $u'_0$ of $\gamma'\downarrow_R$ is $\false$, and by Proposition~\ref{prop:bas-cond-charac} we get that $\gamma' \equiv \false$. Since there exists a derivation of $\gamma \sim \gamma'$ in $\fas^* \cdot \dup^* \cdot \cca$, no rule in $\fas$ is applied. Therefore the derivation is only an application of $\cca$. Contradiction.
  \end{minipage}
\end{proof}

We can then ensure that any proof $P$ of $t \sim t'$ is not containing a $\csmb$ or $\obfa$ application on $\true$ or $\false$: if we have a $\csmb$ or $\obfa$ application on $(\true,\true)$ or $(\false,\false)$ then there is a proof cut elimination without it yielding a smaller proof, and the previous proposition ensures that there are no $\csmb$ or $\obfa$ application on $(\true,b),(b,\true), (\false,b)$ or $(b,\false)$ (with $b \not =_R \false,\true$).
\begin{proposition}
  \label{prop:obviouscut}
  For all $P \npfproof t \sim t'$, there exists $P'$ such that $P' \npfproof t \sim t'$ and for all $l \in \prooflabel(P'), h \in \setindex(P'), \sfx \in \{\sfl,\sfr\}$ we have:
  \begin{gather*}
    \forall \beta \in 
    \left((\lecond^{h_\sfx,l}\cup \lecs^{h_\sfx}) (t,P')\right) 
    \cup 
    \left((\lecond^{h_\sfx,l}\cup \lecs^{h_\sfx}) (t',P')\right),\quad
    \{\false,\true\} \cap \leavest(\beta\downarrow_R)  = \emptyset
  \end{gather*}
\end{proposition}

\begin{proof}
  We can construct a proof $P'$ from $P$ through simple proof cut eliminations such that:
  \[
    \{(\true,\true),(\false,\false)\}\cap (\lesimcond^{h_\sfx,l}(t \sim t',P)\cup \lesimcs^{h_\sfx}(t \sim t',P)) = \emptyset
  \]
  Then using Proposition~\ref{prop:iffree-false-body} we know that for all:
  \[
    (\beta,\beta') \in (\lesimcond^{h_\sfx,l}(t \sim t',P)\cup \lesimcs^{h_\sfx}(t \sim t',P))
  \]
  the conditionals $\beta$ and $\beta'$ are such that $\beta \not =_R \false$ and $\beta' \not =_R \false$ (same with $\true$). Finally if $\beta \not =_R \false$ then one can easily check that for all $u \in \leavest(\beta\downarrow_R)$, $u \not \equiv \false$ (idem with $\true$).
\end{proof}

We recall that showed in Lemma~\ref{lem:cond-equiv-body} that if $\vdash_{\mathcal{A}_{\fas}} b,b \sim b',b''$ then $b' \equiv b''$. We are now ready to give the proof of Lemma~\ref{lem:cond-equiv-bis-body}, which generalize this to the case where $\npfproof b,b \sim b',b''$, but only when $b,b',b''$ are if-free.
\begin{lemma*}[\ref{lem:cond-equiv-bis-body}]
  For all $a,a',b,c$ such that their $R$-normal forms are if-free and such that $a =_R a'$, if $P \npfproof a,a' \sim b,c$ then $b =_R c$.
\end{lemma*}

\begin{proof}
  Let $t \equiv \pair{a}{a}$ and  $t' \equiv \pair{b}{c}$, we know that there exists $P'$ such that $P' \npfproof t \sim t'$ since $P \npfproof a,a' \sim b,c$. Moreover using Proposition~\ref{prop:obviouscut} we know that for all $h \in \setindex(P)$, for all $l,\sfx$:
  \begin{gather*}
    \forall \beta \in 
    \left((\lecond^{h_\sfx,l}\cup \lecs^{h_\sfx,l}) (t,P')\right) 
    \cup 
    \left((\lecond^{h_\sfx,l}\cup \lecs^{h_\sfx,l}) (t',P')\right),\quad
    \{\false,\true\} \cap \leavest(\beta\downarrow_R)  = \emptyset
  \end{gather*}
  Let $(\gamma,\gamma') \leleave^{\epsilon,l} (t \sim t',P)$ be the left-most elements of $t$ and $t'$. By Proposition~\ref{prop:left-most-ind} we know that $\pair{a}{a}\downarrow_R \in \leavest(\gamma\downarrow_R)$ and $\pair{b}{c}\downarrow_R \in \leavest(\gamma'\downarrow_R)$. More precisely we know that $\pair{b}{c}$ is the left-most element of $\gamma'\downarrow_R$.

  Since $\gamma \sim \gamma'$ is provable in $\fas^* \cdot \dup^* \cdot \cca$, we know that there exists $\gamma_1,\gamma_2,\gamma_1',\gamma_2'$ such that they are $\ekl^P$-normalized basic terms and $\gamma =_R \pair{\gamma_1}{\gamma_2}$, $\gamma' =_R \pair{\gamma'_1}{\gamma'_2}$, and the formula $\gamma_1,\gamma_2 \sim \gamma_1',\gamma_2'$ is provable in $\fas^* \cdot \dup^* \cdot \cca$. 

  Moreover $a \in \leavest(\gamma_1\downarrow_R)$ and $a \in \leavest(\gamma_2\downarrow_R)$, hence $\leavest(\gamma_1\downarrow_R) \cap \leavest(\gamma_2\downarrow_R) \ne \emptyset$. Using Proposition~\ref{prop:bas-cond-charac} we deduce that $\gamma_1 \equiv \gamma_2$.
  
  Therefore there exists a proof of $\gamma_1,\gamma_1 \sim \gamma_1',\gamma_2'$ in $\fas^* \cdot \dup^* \cdot \cca$, and by Lemma~\ref{lem:cond-equiv-body} we get that $\gamma_1' \equiv \gamma_2'$.

  We conclude by observing that since $\pair{b}{c}$ is the let-most element of $\gamma'\downarrow_R$, $b$ (resp. $c$) is the left-most element of $\gamma_1'$ (resp. $\gamma_2'$). Therefore $b \equiv c$.
\end{proof}

\begin{definition}
  For all term $t$, we let $<^{\ek}_{\textsf{bc}} t$ be the set of $\ek$-normalized basic conditional appearing in $t$, defined inductively by:
  \begin{itemize}
  \item If $t$ is a $\ek$-normalized simple term $C[\vec b \diamond \vec u]$, then:
    \[
      <^{\ek}_{\textsf{bc}} t \;=\; \vec b \;\;\cup\;\; \left(<^{\ek}_{\textsf{bc}} \vec b\right)  \;\;\cup\;\; \left(<^{\ek}_{\textsf{bc}} \vec u\right) 
    \]
  \item If $t$ is a $\ek$-normalized basic term $B[\vec w, (\alpha_i)_i, (\dec_j)_j]$, then:
    \[
      <^{\ek}_{\textsf{bc}} t
      \; = \;
      \bigcup_{i} <^{\ek}_{\textsf{bc}} \alpha_i\;\;\cup\;\; \bigcup_{j} <^{\ek}_{\textsf{bc}} \dec_j
    \]
  \item For all \ek-encryption oracle call $t \equiv \enc{u}{\pk}{r}$, then:
    \[
      <^{\ek}_{\textsf{bc}} t \; = \; <^{\ek}_{\textsf{bc}} u
    \]
  \item For all \ek-decryption oracle call $C[\vec b \diamond \vec u]$, let $s,\sk$ such that terms in $\vec u$ are of the form $\zero(\dec(s[(\alpha_i),(\dec_j)_j],\sk))$ or $\dec(s[(\alpha_i),(\dec_j)_j],\sk)$, and $u$ is if-free. Then:
    \[
      <^{\ek}_{\textsf{bc}} t \;=\; \vec b \;\;\cup\;\; \left(<^{\ek}_{\textsf{bc}} \vec b\right) \;\;\cup\;\; \bigcup_{i} <^{\ek}_{\textsf{bc}} \alpha_i\;\;\cup\;\; \bigcup_{j} <^{\ek}_{\textsf{bc}} \dec_j
    \]
  \end{itemize}
\end{definition}

\begin{proposition}
  \label{prop:acond-charac}
  For all term $\beta$ such that $\beta$ is a $\ek$-normalized basic term, $\ek$-normalized simple term, $\ek$-decryption oracle call or $\ek$-encryption oracle call we have:
  \[
    \acondst(\beta) = \bigcup_{u <^{\ek}_{\textsf{bc}} \beta} \aleavest(u)
  \]
\end{proposition}

\begin{proof}
  We prove this by induction on the order $\lest^{\ek}$.
  \paragraph{Base Case} If $\beta$ is minimal for $\lest^{\ek}$, then we have the following cases:
  \begin{itemize}
  \item $\ek$-decryption oracle call: $\beta$ is of the form $C[\vec b \diamond \vec u]$, and there exists $s,\sk$ such that terms in $\vec u$ are of the form $\zero(\dec(s,\sk))$ or $\dec(s,\sk)$, and $u$ is if-free. Moreover by minimality of $\beta$ the vector of terms $\vec b$ must be empty, since for all $b \in \vec b$, $b$ is a $\ek$-normalized basic term. 
    
    Hence $\acondst(\beta) = \emptyset$. Finally since $\beta$ is minimal there are no $u$ such that $u <^{\ek}_{\textsf{bc}} \beta$.
  \item $\ek$-encryption oracle call case cannot happen, as $\beta$ would not be minimal.
  \item $\ek$-normalized basic term: $\beta$ contains no $\ite{}{}{}$ symbol, hence $\acondst(\beta) = \emptyset$. Moreover since $\beta$ is minimal there are no $u$ such that $u <^{\ek}_{\textsf{bc}} \beta$.
  \item $\ek$-normalized simple term case cannot happen, as $\beta$ would not be minimal.
  \end{itemize}
  \paragraph{Inductive Case} Let $\beta$ be such that for all $\beta' \ne \beta$, if $\beta' \lest^{\ek} \beta$ then the property holds for $\beta'$.
  \begin{itemize}
  \item $\ek$-normalized basic term: $\beta$ is of the form $B[\vec w,(\alpha_i)_i,(\dec_j)_j]$. The result is then immediate by induction hypothesis and using the definition of $\acondst(\cdot)$ and $<^{\ek}_{\textsf{bc}}$:
    \begin{alignat*}{3}
      \acondst(\beta) &\;=\;&& \bigcup_{i}\; \acondst(\alpha_i) &\quad\cup\quad& \bigcup_{j}\; \acondst(\dec_i)
      \tag{By definition of $\acondst(\cdot)$}\\
      &\;=\;&& \bigcup_{i} \bigcup_{u <^{\ek}_{\textsf{bc}} \alpha_i}\; \aleavest(u) &\quad\cup\quad& \bigcup_{j} \bigcup_{u <^{\ek}_{\textsf{bc}} \dec_j}\; \aleavest(u)
      \tag{By induction hypothesis}\\
      &\;=\;&&  \bigcup_{u <^{\ek}_{\textsf{bc}} \beta}\; \aleavest(u)
      \tag{By definition of $<^{\ek}_{\textsf{bc}}$}
    \end{alignat*}
  \item $\ek$-decryption oracle call: $t$ is of the form  $C[\vec b \diamond \vec u]$, where there exists $s,\sk$ such that terms in $\vec u$ are of the form $\zero(\dec(s[(\alpha_i),(\dec_j)_j],\sk))$ or $\dec(s[(\alpha_i),(\dec_j)_j],\sk)$, and $u$ is if-free. Then:
    \begin{alignat*}{5}
      \acondst(\beta) &\;=\;&& \bigcup_{i}\; \acondst(\alpha_i) 
      &\;\cup\;& \bigcup_{j}\; \acondst(\dec_i) 
      &\;\cup\;& \acondst(\vec g) 
      &\;\cup\;& \aleavest(\vec g)\\
      \tag{By definition of $\acondst(\cdot)$}\\
      &\;=\;&& \bigcup_{i} \bigcup_{u <^{\ek}_{\textsf{bc}} \alpha_i}\; \aleavest(u) 
      &\;\cup\;& \bigcup_{j} \bigcup_{u <^{\ek}_{\textsf{bc}} \dec_j}\; \aleavest(u) 
      &\;\cup\;& \bigcup_{u <^{\ek}_{\textsf{bc}} \vec g}\; \aleavest(u)
      &\;\cup\;& \aleavest(\vec g)\\
      \tag{By induction hypothesis: remark that guards in $\vec g$ are $\ek$-normalized basic terms s.t. $\vec g \lebt^{\ek} \beta$}\\
      &\;=\;&&  \bigcup_{u <^{\ek}_{\textsf{bc}} \beta}\; \aleavest(u)
      \tag{By definition of $<^{\ek}_{\textsf{bc}}$}
    \end{alignat*}
  \item $\ek$-encryption oracle call: $t$ is of the form $\enc{s}{\pk}{r}$, then:
    \begin{alignat*}{3}
      \acondst(\beta) &\;=\;&& \acondst(s)      
      \tag{By definition of $\acondst(\cdot)$}\\
      &\;=\;&& \bigcup_{u <^{\ek}_{\textsf{bc}} s}\; \aleavest(u)
      \tag{By induction hypothesis}\\
      &\;=\;&& \bigcup_{u <^{\ek}_{\textsf{bc}} \beta}\; \aleavest(u)
      \tag{By definition of $<^{\ek}_{\textsf{bc}}$}
    \end{alignat*}
  \item $\ek$-normalized simple term: $t$ is of the form  $C[\vec b \diamond \vec v]$. Then:
    \begin{alignat*}{4}
      \acondst(\beta) &\;=\;&& \acondst(\vec b)
      &\;\cup\;& \acondst(\vec v) 
      &\;\cup\;& \aleavest(\vec b)
      \tag{By definition of $\acondst(\cdot)$}\\
      &\;=\;&&  \bigcup_{u <^{\ek}_{\textsf{bc}} \vec b}\; \aleavest(u) 
      &\;\cup\;& \bigcup_{u <^{\ek}_{\textsf{bc}} \vec v}\; \aleavest(u) 
      &\;\cup\;& \aleavest(\vec b)
      \tag{By induction hypothesis}\\
      &\;=\;&&  \bigcup_{u <^{\ek}_{\textsf{bc}} \beta}\; \aleavest(u)
      \tag{By definition of $<^{\ek}_{\textsf{bc}}$}
    \end{alignat*}
  \end{itemize}
\end{proof}

\begin{proposition}
  \label{prop:bil-cap-clio}
  Let $P \npfproof t \sim t'$. Then for all $\sfh,l$ for all $\beta \lebt^{\sfh,l} (t,P)$,  $\acondst(\beta) \cap \aleavest(\beta) = \emptyset$.
\end{proposition}

\begin{proof}
  Let $\sfh,l$ and $\beta \lebt^{\sfh,l} (t,P)$ be such that $\acondst(\beta) \cap \aleavest(\beta) \ne \emptyset$. By Proposition~\ref{prop:acond-charac} this means that there exists a $\ekl$-normalized basic term $u <_{\textsf{bc}}^{\ekl} \beta$ such that $\aleavest(u) \cap \aleavest(\beta) \ne \emptyset$.

  Using Proposition~\ref{prop:bas-cond-charac} we know that $u \equiv \beta$. But $u <_{\textsf{bc}}^{\ekl} \beta$ implies that $u$ is a strict subterm of $\beta$. Absurd.
\end{proof}

\begin{definition}
  Let $P \npfproof t \sim t'$, we know that $t$ is of the form: 
  \[
    t \equiv C\left[\left(\splitbox{b^{h_\sfl}}{b^{h_\sfr}}{b^h}\right)_{h \in H} \diamond \left(
        D_l\left[
          \left(\beta\right)_{\beta \lecond^{\epsilon,l} (t,P)}
          \diamond\left(\gamma\right)_{\gamma \leleave^{\epsilon,l} (t,P)}
        \right]
      \right)_{l \in L}\right]
  \]
  For all $l$, we let:
  \begin{itemize}
  \item $\dcspath^{\epsilon,l}(t,P)$ be the directed path of conditional occurring from the root of $t$ to $D_l[]$ in $P$.
  \item $\dsimcspath^{\epsilon,l}(t \sim t',P)$ be the directed path of pairs of conditionals occurring from the root of $(t,t')$ to $D_l[]$ in $P$.
  \end{itemize}
  We extend this to all $h \in \setindex(P), \sfx \in \{\sfl,\sfr\}$ by having:
  \begin{alignat*}{2}
    &\dcspath^{h_{\sfx},l}(t,P) = \dcspath^{\epsilon,l}(b,\extractx(h,P))\\
    \text{ and }\quad&
    \dsimcspath^{h_{\sfx},l}(t \sim t',P) = \dsimcspath^{\epsilon,l}(b \sim b',\extractx(h,P))
  \end{alignat*}
  where $\extractx(h,P)$ is a proof of $b \sim b'$.
\end{definition}

\toadd{i) Explain why, if there are no FAb, one can show that the leaves of $b^\sfh \sim b'^\sfh$ are abounded. ii) Give the proof cut elimination used in this proof before, with a description and an example. Explain what are the hypothesis necessary for this proof cut elimination. iii) Explain how one can have a proof with where hypothesis hold, e.g. by replacing CSm by CSfull.}

\begin{lemma}
  \label{lem:no-fab-no-cap}
  Let $P \npfproof t \sim t'$. There exists $P'$ such that $P' \npfproof t \sim t'$ and for all $h \in \setindex(P')$ with $h \ne \epsilon$, for all $\sfx \in \{\sfl,\sfr\}$, if we let $\sfh = h_\sfx$ and $P^\sfh  = \extractx(h,P')$ be the proof of $b^{\sfh} \sim b'^{\sfh}$ then for all $l \in \prooflabel(P^\sfh)$:
  \begin{itemize}
  \item[(a)] The proof $P^\sfh$ does not use the $\{\obfa(b,b')\}$ rules.
  \item[(b)] $\cspath^{\sfh,l}(t,P)$ (resp. $\cspath^{\sfh,l}(t',P)$) does not contain two occurrences of the same conditional.
  \item[(c)] For all $\gamma \leleave^{\sfh,l} (t,P')$, $(b^\sfh\downarrow_R) \in \leavest(\gamma\downarrow_R)$ and for all $\gamma' \leleave^{\sfh,l} (t',P')$, $(b'^\sfh\downarrow_R) \in \leavest(\gamma'\downarrow_R)$.
  \item[(d)] For all $\beta \lecond^{\epsilon,l} (t,P')$, $\leavest(\beta\downarrow_R) \cap  \cspath^{\epsilon,l}(t,P) = \emptyset$ (same for $t'$).
  \item[(e)] For all  $\gamma \leleave^{\epsilon,l} (t,P')$,  $\leavest(t\downarrow_R) \cap \leavest(\gamma\downarrow_R) \ne \emptyset$ (same for $t'$).
  \end{itemize}
\end{lemma}

\begin{proof}
  Using Proposition~\ref{prop:obviouscut}, we know that we have $P$ such that $P \npfproof t \sim t'$ and for all $l \in \prooflabel(P), h \in \setindex(P), \sfx \in \{\sfl,\sfr\}$ we have:
  \begin{gather}
    \label{eq:nofalse}
    \forall \beta \in 
    \left((\lecond^{h_\sfx,l}\cup \lecs^{h_\sfx,l}) (t,P)\right) 
    \cup 
    \left((\lecond^{h_\sfx,l}\cup \lecs^{h_\sfx,l}) (t',P)\right),\quad
    \{\false,\true\} \cap \leavest(\beta\downarrow_R)  = \emptyset
  \end{gather}

  First we start by rewriting the proof $P$ so that all $\csm$ application are of the form:
  \begin{equation}
    \infer[\csm]{(\ite{b}{u_i}{v_i})_i \sim (\ite{b'}{u_i'}{v_i'})_i}
    {
      b, (u_i)_i \sim b',(u_i')_i
      &
      b, (v_i)_i \sim b',(v_i')_i
    }\label{eq:mo-fab-no-cap}
  \end{equation}
  We prove by induction on $n$, starting with the inner-most $\cs$ conditionals, that there exists $P$ such that $P \npfproof t \sim t'$, \eqref{eq:nofalse} is true for $P$ and the following properties hold for all $h,h' \in \setindex(P)$:
  \begin{itemize}
  \item[(i)] If $\ifdepth_P(h) \ge n$ then the $\extractl(h,P)$ and $\extractr(h,P)$ do not use the $\{\obfa(b,b')\}$ rules.
  \item[(ii)] If $\ifdepth_P(h) \ge n$ then for all $\sfx,l$, $\cspath^{h_\sfx,l}(t,P)$ and $\cspath^{h_\sfx,l}(t',P)$ do not contain two occurrences of the same conditional.
  \item[(iii)] If $\ifdepth_P(h) \ge n$ then for all $\sfx$, if $\extractx(h,P)$ is the proof of $b \sim b'$ then for all $l$, for all $\gamma \leleave^{h_\sfx,l} (t,P)$, $(b\downarrow_R) \in \leavest(\gamma\downarrow_R)$ and for all $\gamma' \leleave^{h_\sfx,l} (t',P)$, $(b'\downarrow_R) \in \leavest(\gamma'\downarrow_R)$.
  \item[(iv)] If $\ifdepth_P(h) < n$ then for all $h,h' \in \setindex(P)$ such that $h \le h'$, if we let $h''$ be such that $h' = h\cdot h''$ and $\sfx$ be such that $h'' \in \setindex(\extractx(h,P))$, then for all $\sfx'$, for all $l \in \prooflabel(\extract_{\sfx'}(h',P))$, we have 
    \[
      \dcspath^{h_{\sfx},l}(t,P) \supseteq \dcspath^{h'_{\sfx'},l}(t,P)
    \]
  \end{itemize}
  Let $n_{\text{max}}$ be the maximal $\ifdepth$ in the proof of $t \sim t'$:
  \[
    n_{\text{max}} = \max_{h \in \setindex(P)}\ifdepth_P(h)
  \]
  \paragraph{Base Case:} We are going to show that the invariants hold at $n_{\text{max}} + 1$. Invariants (i), (ii) and (iii) are obvious, since there exists no $h$ such that $\ifdepth_P(h) \ge n_{\text{max}} + 1$; and invariant (iv) is a consequence of the rewriting done in \eqref{eq:mo-fab-no-cap}.

  \paragraph{Inductive Case:} Assume that the property holds for $n+1$ and let us show that it holds for $n$.
  
  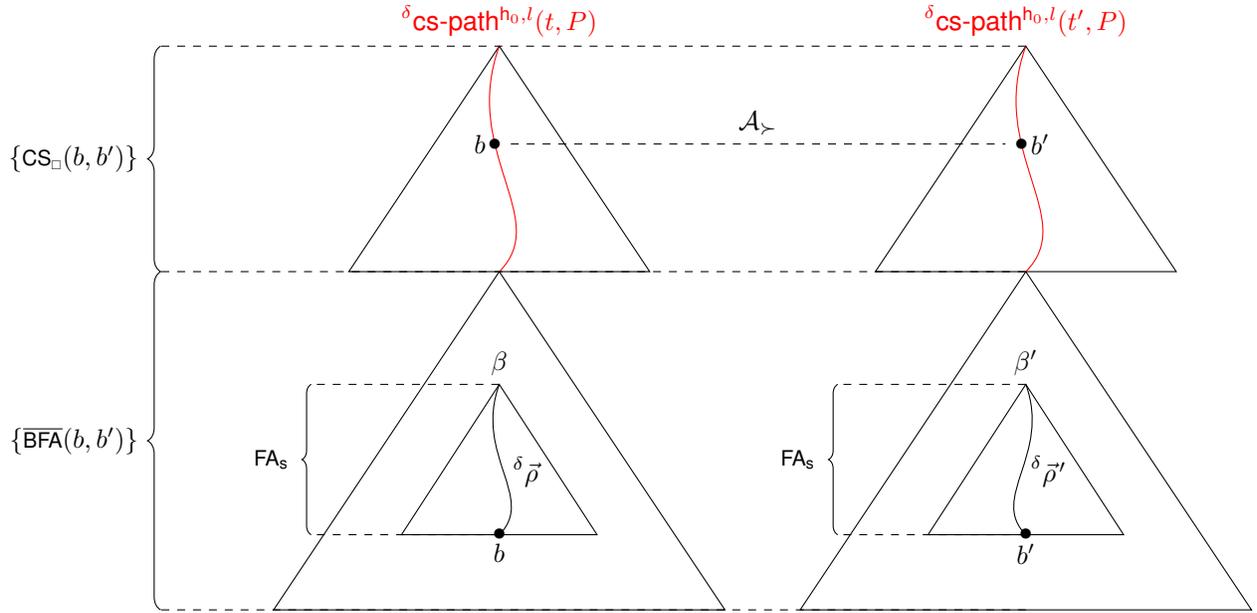
\begin{figure}[tb]
    \begin{center}
      \begin{tikzpicture}
        % Left term
        \node (root) at (0,0) {};
        \path (root.center) ++(0,-3) node (cs){};
        \draw (root.center) -- ++(2,-3) -- ++(-4,0) -- cycle;
        \draw[red] (root.center) 
        node[above]{$\dcspath^{\sfh_0,l}(t,P)$}
        .. controls ++(250:1.5) and ++(45:1) .. (cs.center) 
        node[black,pos=0.35] (a) {$\bullet$} 
        node[black,left] at (a) {$b$};
        % node[pos=0.65,right]{$\dcspath^{\sfh_0,l}(t,P)$};
        
        \draw (cs.center) -- ++(3,-4.5) -- ++(-6,0) -- cycle;
        \draw (root.center) ++(0,-4.5) -- ++(1.3,-2) -- ++(-2.6,0) -- cycle;
        \draw (root.center) ++(0,-4.5)  node[above]{$\beta$} 
        .. controls ++(250:1) and ++(45:0.7) .. ++(0,-2) 
        node[midway,right]{$\dpath{\vec \rho}$} 
        node {$\bullet$} 
        node[below]{$b$};

        \draw[dashed] (root.center) ++(0,-4.5) -- ++(-2.5,0) node (b){};
        \draw[dashed] (root.center) ++(0,-6.5) -- ++(-2.5,0) node (c){};
        \draw [decorate,decoration={brace,amplitude=3.5pt,mirror}] (b.center) -- (c.center) node[midway,left,xshift=-0.5em]{$\fas$};

        % Right term
        \node (root2) at (7,0) {};
        \path (root2.center) ++(0,-3) node (cs2){};
        \draw (root2.center) -- ++(2,-3) -- ++(-4,0) -- cycle;
        \draw[red] (root2.center)
        node[above]{$\dcspath^{\sfh_0,l}(t',P)$}
        .. controls ++(250:1.5) and ++(45:1) .. (cs2.center)
        node[black,pos=0.35] (a2) {$\bullet$} 
        node[black,right] at (a2) {$b'$};
        % node[pos=0.65,right]{$\dcspath^{\sfh_0,l}(t',P)$};
        
        \draw (cs2.center)  -- ++(3,-4.5) -- ++(-6,0) -- cycle;
        \draw (root2.center) ++(0,-4.5) -- ++(1.3,-2) -- ++(-2.6,0) -- cycle;
        \draw (root2.center) ++(0,-4.5)  node[above]{$\beta'$} 
        .. controls ++(-70:1) and ++(135:0.7) .. ++(0,-2) 
        node[midway,right]{$\dpath{\pvec \rho'}$} 
        node {$\bullet$} 
        node[below]{$b'$};

        \draw[dashed] (root2.center) ++(0,-4.5) -- ++(-2.5,0) node (b){};
        \draw[dashed] (root2.center) ++(0,-6.5) -- ++(-2.5,0) node (c){};
        \draw [decorate,decoration={brace,amplitude=3.5pt,mirror}] (b.center) -- (c.center) 
        node[midway,left,xshift=-0.5em]{$\fas$};
        
        % Rules annotations
        \draw[dashed] (root2.center) -- (root.center) -- ++(-4.5,0) node (d){};
        \draw[dashed] (cs2.center) -- (cs.center) -- ++(-4.5,0) node (e){};

        \path (cs.center) -- ++(0,-4.5) node (fa) {};
        \path (cs2.center) -- ++(0,-4.5) node (fa2) {};
        \draw[dashed] (fa2.center) -- (fa.center) -- ++(-4.5,0) node (f){};

        \draw [decorate,decoration={brace,amplitude=5pt,mirror}] (d.center) -- (e.center) 
        node[midway,left,xshift=-0.5em]{$\{\csmb(b,b')\}$};

        \draw [decorate,decoration={brace,amplitude=5pt,mirror}] (e.center) -- (f.center) 
        node[midway,left,xshift=-0.5em]{$\{\obfa(b,b')\}$};

        \draw[dashed] (a) -- (a2) node[midway,above] {$\mathcal{A}_\succ$};
      \end{tikzpicture}
    \end{center}
    \caption{\label{fig:fab-no-cap-proof} Corresponding occurrences of $b$ and $b'$ in the proof of Lemma~\ref{lem:no-fab-no-cap}}
  \end{figure}

  \subparagraph{Step 1} Let $l \in \prooflabel(P)$ and $h_0 \in \hbranch(l)$ such that $\ifdepth_P(h_0) = n$. Let $x_0 \in \{\sfl,\sfr\}$ and $\sfh_0 = {h_0}_{x_0}$. We start by showing that for all $l$, for all $\beta \lecond^{\sfh_0,l} (t,P)$, if there exists $b \in \cspath^{\sfh_0,l}(t,P)$ such that  $b \in \leavest(\beta\downarrow_R)$ then there exists $(b,b') \in \simcspath^{\sfh_0,l}(t,P)$ and $\beta'$ such that $(\beta,\beta') \lesimcond^{\sfh_0,l} (t\sim t',P)$ and:
  \begin{itemize}
  \item $b' \in  \leavest(\beta'\downarrow_R)$.
  \item The directed path $\dpath{\vec \rho}$ (resp. $\dpath{\pvec \rho'}$) of the conditionals occurring from the root of $\beta\downarrow_R$ (resp. $\beta'\downarrow_R$) to the leave $b$  (resp. $b'$) is such that $\dpath{\vec \rho} \subseteq \dcspath^{\sfh_0,l}(t,P)$ (resp. $\dpath{\pvec \rho'} \subseteq \dcspath^{\sfh_0,l}(t,P)$).
  \end{itemize}
  This is described in Fig.~\ref{fig:fab-no-cap-proof}.

  Let $\beta \lecond^{\sfh_0,l} (t,P)$ and $b \in \cspath^{\sfh_0,l}(t,P)$ such that $b \in \leavest(\beta\downarrow_R)$. We know that there exists $b'$ and $\beta'$ such that $(b,b') \in \simcspath^{\sfh_0,l}(t,P)$ and $(\beta,\beta') \lesimcond^{\sfh_0,l} (t\sim t',P)$.

  Let $h \in \cspos(\extract_{x_0}(h_0,P))$ and $\sfx$ be the direction taken in $l$ at $h$ be such that $\extract(h,P)$ is the rule $\csmb(b,b')$. We know that $\extractx(h,P)$ is a proof of $a \sim a'$, where $a =_R b$ and $a' =_R b'$. As $\ifdepth(h) = n+1$ we know by induction hypothesis (i) that $\extractx(h,P)$ does not uses $\{\obfa(b,b')\}$. Hence the set $\leleave^{\epsilon,l}(a,\extractx(h,P))$ is the singleton $\{\gamma_l\}$ and the set $\leleave^{\epsilon,l}(a',\extractx(h,P))$ is the singleton $\{\gamma'_l\}$. Let $H = \setindex(\extractx(h,P))$, we have:
  \[
    a \equiv C\left[(b^g)_{g \in H} \diamond \left(\gamma_{l_a}\right)_{l_a}\right] \qquad
    a' \equiv C\left[(b'^g)_{g \in H} \diamond \left(\gamma'_{l_a}\right)_{l_a}\right]
  \]

  By induction hypothesis (iii) we know that $b \in \leavest(\gamma_{l}\downarrow_R)$ and $b' \in \leavest(\gamma'_{l}\downarrow_R)$. $\gamma_{l}$ and $\beta$ are $\ek_{l}$-normalized basic terms, hence using Proposition~\ref{prop:bas-cond-charac} we know that $\beta \equiv \gamma_{l}$. We can extract from the branch $l$ of $P$ a proof of $\gamma_{l},\beta \sim \gamma'_{l},\beta'$ in $\fas^* \cdot \dup^* \cdot \bcca$. Therefore, using Lemma~\ref{lem:cond-equiv-body}, we get that $\beta' \equiv \gamma'_{l}$. Since $b' \in \leavest(\gamma'_{l} \downarrow_R)$, we deduce that $b' \in \leavest(\beta' \downarrow_R)$. This concludes the proof of the first bullet point.
  
  By induction hypothesis (iv) we know that 
  \[
    \dcspath^{(h_0)_{\sfx_0},l}(t,P) \supseteq \dcspath^{h_{\sfx},l}(t,P)
    \quad \wedge \quad
    \dcspath^{(h_0)_{\sfx_0},l}(t',P) \supseteq \dcspath^{h_{\sfx},l}(t',P)
  \]
  By definition of $\vec \rho$, $\condst(\gamma_{l}\downarrow_R) \supseteq \vec \rho$. More precisely, using the facts that $a \equiv C\left[(b^g)_{g \in H} \diamond \left(\gamma_l\right)_l\right]$ and since $\condst(a\downarrow_R) = \{b\}$, and invariant (ii), we can show that $\dpath{\vec \rho} \subseteq \dcspath^{h_{\sfx},l}(t,P)$. By consequence, $\dpath{\vec \rho} \subseteq \dcspath^{(h_0)_{\sfx_0},l}(t,P)$. Similarly we show that $\dpath{\pvec{\rho}'} \subseteq \dcspath^{(h_0)_{\sfx_0},l}(t',P)$.

  \subparagraph{Step 2} 
  By doing some proof cut elimination, we can guarantee that for all $l$, for all $\beta \lecond^{\sfh_0,l} (t,P)$:
  \[
    \leavest(\beta\downarrow_R) \cap \cspath^{\sfh_0,l}(t,P) = \emptyset
  \]
  Assume this is not the case: using \textbf{Step 1} we have:
  \[
    \dpath{\vec \rho} \subseteq \dcspath^{(h_0)_{\sfx_0},l}(t,P)
    \quad \wedge \quad
    \dpath{\pvec{\rho}'} \subseteq \dcspath^{(h_0)_{\sfx_0},l}(t',P)
  \]
  Therefore we can rewrite $\beta$ and $\beta'$ into, respectively, $b$ and $b'$ (this is possible because we have an inclusion between the \emph{directed paths}, not just the paths). We can then rewrite $b$ and $b'$ into $\true$ if we are on the \textsf{then} branch of $b$ and $b'$ (i.e. $\sfx = \sfl$), and $\false$ if we are on the \textsf{else} branch (i.e. $\sfx = \sfr$). Finally we get rid of $\true$ and $\false$ using $R$, and check that the resulting proof verifies \eqref{eq:nofalse} and the induction invariants.

  \subparagraph{Step 2 b.} Then we show that we can assume that (ii) holds through some proof rewriting, while maintaining invariant (iv). 

  Let $(a,a'),(b,b') \lesimcs^{\sfh_0}(t,P)$ such that $a \equiv b$ and they are on the same branch $l$. Since they are on the same branch, we can extract a proof $Q \npfproof a,a \sim a',b'$. Moreover $a\downarrow_R,a'\downarrow_R,b'\downarrow_R$ are if-free, therefore by Lemma~\ref{lem:cond-equiv-bis-body} we have $a' \equiv b'$. We then do our standard proof cut elimination to get rid of the duplicate. We need to make sure that this ensure that invariant (iv) at rank $n$ holds for $Q$: this follows from the fact that invariant (iv) holds for $P$ at rank $n+1$ and that the cut takes place at depth $n$.

  \subparagraph{Step 3} We then show that (iii) holds. Let $b^{\sfh_0}, b'^{\sfh_0}$ be such that $\extract_{\sfx_0}(h,P) \npfproof b^{\sfh_0} \sim b'^{\sfh_0}$. We know that:
  \[
    b^{\sfh_0} \equiv
    C\left[\left(\splitbox{b^{h_\sfl}}{b^{h_\sfr}}{b^h})_{x}\right)_{h \in H^{\sfh_0}} 
      \diamond
      \left(D_l^{\sfh_0}
        \left[ 
          (\beta)_{\beta \lecond^{\sfh_0,l} (t,P)} 
          \diamond
          (\gamma)_{\gamma \leleave^{\sfh_0,l} (t,P)}
        \right]\right)_{l \in L^{\sfh_0}}
    \right]
  \]
  where $H^{\sfh_0} = \cspos(\extract_{\sfx_0}(h_0,P))$ and $L^{\sfh_0} = \prooflabel(\extract_{\sfx_0}(h_0,P))$.

  To prove that for all $l$, for all $\gamma \leleave^{\sfh_0,l} (t,P)$, we have $b^{\sfh_0}\downarrow_R \in  \leavest(\gamma\downarrow_R)$, we only need to show that the hypotheses of Proposition~\ref{prop:uml-abounded-tech} hold for $b^{\sfh_0}$ (then we do the same thing with $b'^{\sfh_0}$ to show that for all $\gamma' \leleave^{\sfh_0,l} (t',P)$ we have $b'^{\sfh_0}\downarrow_R \in  \leavest(\gamma'\downarrow_R)$):
  \begin{itemize}
  \item (\ref{prop:uml-abounded-tech}.i): the only difficulty lies in proving that for all $\beta \lecond^{\sfh_0,l}(t,P)$, $\condst(\beta\downarrow_R) \cap \leavest(\beta\downarrow_R) = \emptyset$, which is shown in Proposition~\ref{prop:bil-cap-clio}.
  \item (\ref{prop:uml-abounded-tech}.ii): this is a consequence of the fact that \eqref{eq:nofalse} holds for $P$.
  \item (\ref{prop:uml-abounded-tech}.iii): for pairs in $(\cspath^{\sfh_0,l}(t,P))^2$ this was shown in \textbf{Step 2 b}. For couples of positions in $D^{\sfh_0}_l\times D^{\sfh_0}_l$ we have a proof cut elimination: let $p < p'$ be the positions in $b^{\sfh_0}$ of $\beta_0,\beta_1 \lecond^{\sfh_0,l}(t,P)$ on the same branch such that $\leavest(\beta_0) \cap \leavest(\beta_1) \ne \emptyset$. By Proposition~\ref{prop:bas-cond-charac} we know that $\beta_0 \equiv \beta_1$. Let $\beta_0',\beta_1'$ be the conditionals at positions, respectively, $p$ and $p'$ in $b'^{\sfh_0}$. We know that $(\beta_0,\beta_0'),(\beta_1,\beta_1') \lecond^{\sfh_0,l}(t \sim t',P)$. We can extract from $P$ a proof of:
    \[
      \beta_0,\beta_0 \sim \beta_0',\beta_1'
    \]
    in $\fas^* \cdot \dup^* \cdot \bcca$, hence using Lemma~\ref{lem:cond-equiv-body} we get that $\beta_0' \equiv \beta_1'$. Therefore we can do the following proof cut elimination: if $p'$ is on the $\textsf{then}$ branch of $p$ then we can rewrite $\beta_1$ and $\beta_1'$ into $\true$ in, respectively, $b^{\sfh_0}$ and $b'^{\sfh_0}$. We then rewrite the two terms using $R$ to remove the $\true$ conditionals. This yields a new proof $Q$ in proof normal form, such that \eqref{eq:nofalse} and the induction invariants hold. We do a similar cut elimination with $\false$ if $p'$ is in the $\textsf{else}$ of $p$.

    Finally the result proven at \textbf{Step 2} shows that we do not have cross cases $\cspath^{\sfh_0,l}(t,P) \times D^{\sfh_0}_l$.
  \item (\ref{prop:uml-abounded-tech}.iv): this is a consequence of Corollary~\ref{cor:bas-cond-pull}.\eqref{item:bas-cond-charac-beta}.
  \item (\ref{prop:uml-abounded-tech}.v): this is a consequence of Lemma~\ref{lem:well-nested}.
  \end{itemize}

  \subparagraph{Step 4} We conclude by showing that we can get rid of the $\{\obfa(b,b')\}$ applications.

  Using Corollary~\ref{cor:bas-cond-pull}.\eqref{item:bas-cond-charac-gamma} and the proof $Q$ constructed at \textbf{Step 3}, we know that for all $\gamma,\gamma' \leleave^{\sfh_0,l}(t,Q)$, $\gamma \equiv \gamma'$ (and the same holds for $(t',Q)$). Therefore there is a proof cut elimination that allows us to remove all $\{\obfa(b,b')\}$ applications, by rewriting:
  \[
    D_l\left[\_ \diamond \left(\gamma\right)_{\gamma \leleave^{\sfh_0,l}(t,Q)}\right] 
    \quad\text{ and }\quad
    D_l\left[\_ \diamond \left(\gamma'\right)_{\gamma \leleave^{\sfh_0,l}(t',Q)}\right]
  \]
  into, respectively, $\gamma_0$ and $\gamma'_0$ (where $\gamma_0 \leleave^{\sfh_0,l}(t,Q)$ and $\gamma_0' \leleave^{\sfh_0,l}(t',Q)$).

  \paragraph{Conclusion} To conclude, we can first observe that the properties (a),(b) and (c) are implied by, respectively, (i), (ii) and (iii) for $n = 0$. The proof that (d) (resp. (e)) holds is exactly the same than the one we did at \textbf{Step 2} (resp. \textbf{Step 3}).
\end{proof}

%%% Local Variables:
%%% mode: latex
%%% TeX-master: "ms"
%%% End:

\newpage

\section{Bounding the Basic Terms}
\label{app-section:bounding}

\subsection{\texorpdfstring{$\alpha$}{alpha}-Bounded Conditionals}
\label{subsection:abounded-cond}

We are ready to do the final proof cut eliminations, which will yield derivation of bounded size w.r.t. $|t\downarrow_R| + |t'\downarrow_R|$. To bound the size of cut-free derivations, we are going to bound the size of all normalized basic terms and case-study conditionals appearing in such derivations. To do this, we first introduce the notion of $(t,P)$-\abounded terms, where $P \npfproof t\sim t'$, and then prove that $(t,P)$-\abounded terms are of bounded size w.r.t. $|t\downarrow_R| + |t'\downarrow_R|$. Basically, a term $\beta$ in $\lebt^{\sfh,l} (t,P)$ or $ \cspath^{\sfh,l} (t,P)$ is $(t,P)$-\abounded if we are in one of the following case:
\begin{itemize}
\item $\beta$ is a normalized basic term, and $\beta$ has a leaf term appearing in $\st(t\downarrow_R)$. Since $\beta$ is uniquely characterized by its leaf terms, this bound $\beta$.
\item Let $\beta'$ be the term matching $\beta$ on the \emph{right}. If $\beta'$ shares a leaf term with $\st(t'\downarrow_R)$, then, by the previous observation, $\beta'$ is bounded. Since $\beta$ and $\beta'$ differ only by the content of their encryptions, this also bound $\beta$.
\item If $\beta$ is a case-study conditional (i.e. in $\cspath^{\sfh,l} (t,P)$), and if there exists a $(t,P)$-\abounded normalized basic term $\varepsilon$ such that $\beta$ appears in $\varepsilon$'s leaf terms. Indeed, since $\varepsilon$ is bounded, it has finitely many leaf terms, which are of bounded size. Hence $\beta$ is also of bounded size.
\item If $\beta$ is a normalized basic terms used in the sub-proof of $b\sim b'$, where $b$ and $b'$ are $(t,P)$-\abounded case-study conditionals, and if $b$ appears in $\beta$'s leaf terms. Again, since $\beta$ is uniquely characterized by any of its leaf terms, and since $b$ is bounded, we know that $\beta$ is bounded.
\item Finally, if $\beta$ is a decryption guard of some decryption oracle call $d$, where $d$ appears in a $(t,P)$-\abounded normalized basic term $\zeta$. Since $\zeta$ is bounded, and since $\beta$ is a sub-term of $\zeta$, the term $\beta$ is also bounded.
\end{itemize}
We formally define what is a $(t,P)$-\abounded terms.
\begin{definition}
  For all $P \npfproof t \sim t'$, the set of $(t,P)$-\abounded terms is the smallest subset of:
  \[
    \big\{
    \beta \mid \exists \sfh,l.\, \beta \lebt^{\sfh,l} (t,P)
    \big\}
    \cup
    \big\{
    b \mid \exists \sfh.\, b \in \cspath^{\sfh,l} (t,P)
    \big\}
  \]
  such that for all $\sfh,l$, for all $\beta \mathbin{(\lebt^{\sfh,l}  \cup \cspath^{\sfh,l})} (t,P)$,  $\beta$ is $(t,P)$-\abounded if:
  \begin{itemize}
  \item \textbf{Base case:} $\sfh = \epsilon$ and $\leavest(\beta\downarrow_R) \cap \st(t\downarrow_R)\ne \emptyset$.

  \item \textbf{Base case:} $\sfh = \epsilon$ and there exists $\beta'$ such that:
    \[
      (\beta,\beta')
      \mathbin{(\lesimleave^{\epsilon,l} \cup
        \lesimcond^{\epsilon,l}  \cup
        \cspath^{\epsilon,l})} (t \sim t',P)
    \]
    and $\leavest(\beta'\downarrow_R) \cap \st(t'\downarrow_R)\ne \emptyset$.

  \item \textbf{Inductive case, same label:} $\beta \in \cspath^{\sfh,l}(t,P)$ and there exists $\varepsilon \lebt^{\sfh,l} (t,P)$ such that $\varepsilon$ is $(t,P)$-\abounded and $\beta \in {\leavest(\varepsilon \downarrow_R)}$.

  \item \textbf{Inductive case, different labels:} $\beta \lebt^{\sfh,l}(t,P)$, there exists $\sfh'$ such that $\sfh \in \cspos(\sfh')$ and $b \in \cspath^{\sfh',l} (t,P)$ such that $b$ is $(t,P)$-\abounded and $b \in \leavest(\beta \downarrow_R)$.

  \item \textbf{Inductive case, guard:} $\beta \lebt^{\sfh,l}(t,P)$, there exists $\varepsilon \lebt^{\sfh,l} (t,P)$ such that:
    \begin{itemize}
    \item $\varepsilon \equiv B[\pvec{w},(\alpha_i)_i,(\dec_j)_j]$ is $(t,P)$-\abounded.
    \item $\beta$ is a guard of a $\ekl^P$-decryption oracle call $d \in (\dec_j)_j$.
    \end{itemize}
  \end{itemize}
\end{definition}

We continue our proof cut eliminations, starting from the derivations constructed in Lemma~\ref{lem:no-fab-no-cap}. We let $P \aproof t \sim t'$ be the restriction of $\npfproof$ to derivations satisfying the properties guaranteed by Lemma~\ref{lem:no-fab-no-cap} which use only $(t,P)$-\abounded terms. Moreover, we require that no basic conditionals appears twice on the same branch.
\begin{definition}
  \label{def:abounded}
  For all proof $P$, term $t,t'$, we write $P \aproof t \sim t'$ if:
  \begin{itemize}
  \item[(I)] $P \npfproof t \sim t'$ and the properties (a) to (e) of Lemma~\ref{lem:no-fab-no-cap} hold.
  \item[(II)] The following sets are sets of, respectively, $(t,P)$-\abounded and $(t',P)$-\abounded terms:
    \begin{alignat*}{2}
      \big\{\beta \mid \exists \sfh,l.\, \beta \lebt^{\sfh,l} (t,P') \big\}
      &\;\cup\;&&
      \big\{b \mid \exists \sfh.\, b \lecs^{\sfh} (t,P') \big\}\\
      \big\{\beta' \mid \exists \sfh,l.\, \beta' \lebt^{\sfh,l} (t',P') \big\}
      &\;\cup\;&&
      \big\{b' \mid \exists \sfh.\, b' \lecs^{\sfh} (t',P') \big\}
    \end{alignat*}
  \item[(III)] For every $l \in \prooflabel(\epsilon)$, for every path $\vec \rho$ of $\ekl^P$-normalized basic conditional from the root of $t$ to some leave, $\vec \rho$ does not contain any duplicates. The same property must hold for $t'$.
  \end{itemize}
\end{definition}
We can now give the proof of Lemma~\ref{lem:bil-abounded-body}, which we recall below.
\begin{lemma*}[\ref{lem:bil-abounded-body}]
  $\aproof$ is complete with respect to $\npfproof$. 
\end{lemma*}

\begin{proof}
  Let $P$ be such that $P \npfproof t \sim t'$, where $P$ is obtained using Lemma~\ref{lem:no-fab-no-cap}. Therefore $P$ satisfies the item (I) of Definition~\ref{def:abounded}. Now, we are going to build from $P$ a proof $P'$ of $t\sim t'$ that satisfies the item (II) and (III) of Definition~\ref{def:abounded}.

  We are going to show that, if there exists $\beta$ in:
  \[
    \big\{\beta \mid \exists \sfh,l.\, \beta \lebt^{\sfh,l} (t,P') \big\}
    \cup
    \big\{b \mid \exists \sfh.\, b \lecs^{\sfh} (t,P') \big\}
  \]
  such that $\beta$ is not $(t,P)$-\abounded, then there is a cut elimination removing $\beta$ (we describe the cut elimination used later in the proof). Moreover, the resulting proof will have a smaller number of basic terms which are not $(t,P)$-\abounded, hence we will conclude by induction. First, we want to pick a term $\beta$ maximal for a carefully chosen relation.

  \paragraph{Order $<_g$}
  Let $<_g$ be the transitive closure of the relation $\ll_g$ on:
  \[
    \bigcup_{\sfh\in\setindex(P)}
    \big\{
    (\beta,\sfh) \mid \exists l. \beta  \lebt^{\sfh,l} (t,P)
    \big\}
    \cup
    \bigcup_{\sfh\in\setindex(P)}
    \big\{
    (b,\sfh) \mid \exists l. b \in \cspath^{\sfh,l} (t,P)
    \big\}
  \]
  defined by:
  \[
    (\zeta,\sfh) \ll_g (\zeta',\sfh') \text{ iff }
    \begin{cases}
      \sfh = \sfh' \wedge \zeta,\zeta' \lebt^{\sfh,l} (t,P)
      \wedge
      \text{ $\zeta$ is a guard of some decryption oracle call
        $d \in \st(\zeta')$}\\
      \sfh = \sfh' \wedge \zeta \in \cspath^{\sfh,l} (t,P)
      \wedge \zeta' \lebt^{\sfh,l} (t,P)
      \wedge \zeta \in \leavest(\zeta'\downarrow_R)\\
      \sfh > \sfh' \wedge \zeta \lebt^{\sfh,l} (t,P)
      \wedge \zeta' \in \cspath^{\sfh',l} (t,P)
      \wedge \zeta' \in \leavest(\zeta\downarrow_R)
    \end{cases}
  \]
  First we show that $<_g$ is a strict order. As it is transitive, we just need to show that it is an antisymmetric relation. For all $\sfh$, the restriction $<_g^\sfh$ of $<_g$ to:
  \[
    \big\{(\beta,\sfh) \mid \exists l. \beta  \lebt^{\sfh,l} (t,P)\big\}
    \cup
    \big\{(b,\sfh) \mid \exists l. b \in \cspath^{\sfh,l} (t,P)\big\}
  \]
  is a strict order, as it is included in the embedding relation. To show that $<_g$ is a strict order on its full domain, we simply use the facts that for all $\sfh$, $<_g^\sfh$ is a strict order and that when we go from the domain of $<_g^\sfh$ to the domain of $<_g^{\sfh'}$, we have $\sfh' > \sfh$.

  W.l.o.g. we assume that $(\beta,\sfh)$ is maximal for $<_g$ among the set of terms that are not $(t,P)$-\abounded. Consider an arbitrary $l$ such that $\sfh \in \hbranch(l)$. Since $\beta$ is not $(t,P)$-\abounded, we know that if $\beta$ is a guard of some decryption oracle call $d \in \st(\zeta)$ with $\zeta \lebt^{\sfh,l} (t,P)$, then $\zeta$ is not $(t,P)$-\abounded. By maximality of $\beta$, it follows that if $\beta \lebt^{\sfh,l}(t,P)$ then $\beta$ is not a decryption guard of any $\zeta \lebt^{\sfh,l}(t,P)$.

  \paragraph{Case $\sfh = \epsilon$}
  First we are going to describe what to do for $\sfh = \epsilon$. From Lemma~\ref{lem:no-fab-no-cap}.(e), we know that for every $l \in \prooflabel(P)$, for all $\gamma \leleave^{\epsilon,l} (t,P)$, the basic term $\gamma$ is $(t,P)$-\abounded. Therefore $\beta \not \leleave^{\epsilon,l} (t,P)$. Moreover, from Lemma~\ref{lem:no-fab-no-cap}.(d) we get that $\beta \lecond^{\epsilon,l} (t,P)$ and $\beta \in \cspath^{\epsilon,l} (t,P)$ are mutually exclusive. Putting everything together, we have three cases:
  \begin{enumerate}[(i)]
  \item\label{item:ab1} either $\beta  \mathbin{(\not \leleave^{\epsilon,l} \cup \lecond^{\epsilon,l})} (t,P)$ and $\beta \not \in \cspath^{\epsilon,l} (t,P)$.
  \item\label{item:ab2} or $\beta  \mathbin{(\not \leleave^{\epsilon,l} \cup \not \lecond^{\epsilon,l})} (t,P)$ and $\beta \in \cspath^{\epsilon,l} (t,P)$.
  \item\label{item:ab3}  $\beta  \mathbin{(\not \leleave^{\epsilon,l} \cup \not \lecond^{\epsilon,l})} (t,P)$ and $\beta \not \in \cspath^{\epsilon,l} (t,P)$.
  \end{enumerate}
  We first focus on case \ref{item:ab1}. We explain how to deal with \ref{item:ab2} and \ref{item:ab3} later.
  \begin{itemize}
  \item \textbf{\ref{item:ab1}, Part 1} Assume that we are in case i). Let $\beta'$ be such that $(\beta,\beta') \mathbin{(\lesimcond^{\epsilon,l})} (t \sim t',P)$. Since $\beta$ is not $(t,P)$-\abounded we know that for all $u \in \leavest(\beta\downarrow_R)$, for all $u' \in \leavest(\beta'\downarrow_R)$,  $u$ and $u'$ are spurious in, respectively, $t$ and $t'$. We let:
    \begin{gather*}
      t \equiv C\big[\pvec{b}_{cs} \diamond
      D_l\left[
        \left(\beta_i\right)_{i \in J}
        \diamond
        \left(\gamma_m\right)_{m \in M}
      \right]
      ,\Delta\big]\\
      t' \equiv C\big[\pvec{b}'_{cs} \diamond
      D_l
      \left[
        \left(\beta'_i\right)_{i \in J}
        \diamond
        \left(\gamma'_m\right)_{m \in M}
      \right]
      ,\Delta'
      \big]
    \end{gather*}
    where, for every $i \in J$, $(\beta_i,\beta_i') \lesimcond^{\epsilon,l} (t\sim t',P)$, and for every $m \in M$, $(\gamma_m,\gamma_m') \lesimleave^{\epsilon,l} (t \sim t',P)$.  Moreover, we assume that for every $i \in J$, the hole $[]_i$ (which is mapped to $\beta_i$) appears exactly once in $D_l$. We define the set of indices $I = \{i \in J \mid  \beta \equiv \beta_{i} \}$. Using Corollary~\ref{cor:bas-cond-pull}.(i), we know that:
    \[
      I = \left\{
        i \in J \mid  \leavest(\beta\downarrow_R)
        \cap \leavest(\beta_i\downarrow_R)
        \ne \emptyset
      \right\}
    \]
    We know that we have a proof of \((\beta_{i})_{i \in I} \sim (\beta'_{i})_{i \in I}\) in the fragment $\frakf(\fas^* \cdot \dup^* \cdot \bcca)$. Therefore:
    \begin{equation}
      \label{eq:abounded-lem-0}
      \forall b,b' \in   \{\beta'_{i} \mid i \in I\}, b \equiv b' \equiv \beta'
    \end{equation}
    Indeed, if $|I| = 1$ then this is obvious, and if $|I|>1$ we use Lemma~\ref{lem:cond-equiv-body} (since all the terms on the left are the same). We let $I'= \{i \in J \mid  \beta' \equiv \beta'_{i} \}$. Using the same proof than for $I$, we know that $I' = \{i \in J\mid  \leavest(\beta'\downarrow_R) \cap \leavest{(\beta'_i\downarrow_R)} \ne \emptyset \}$. We deduce from this that:
    \begin{equation}
      \label{eq:abounded-lem-1}
      \forall b,b' \in   \{\beta_{i} \mid i \in I'\}, b \equiv b' \equiv \beta
    \end{equation}
    From \eqref{eq:abounded-lem-0} we get that \( I \subseteq I' \) and conversely from \eqref{eq:abounded-lem-1} we get that \( I' \subseteq I \). Therefore we have the equality $I = I'$.

  \item \textbf{\ref{item:ab1}, Part 2}
    For every $i \not \in I$, using Lemma~\ref{lem:bas-cond-restr-spurious2} on $\beta$ we know that there exists $\tilde \beta_{i}[]$ such that:
    \[
      \tilde \beta_{i}[\beta] \equiv \beta_{i}
      \qquad \tand \qquad
      \leavest(\beta\downarrow_R)
      \cap
      \condst(\tilde \beta_{i}[]\downarrow_R) = \emptyset
    \]
    Similarly, for every $m \in M$, there exists $\tilde \gamma_{m}[]$ such that:
    \[
      \tilde \gamma_{m}[\beta] \equiv \gamma_{m}
      \qquad \tand \qquad
      \leavest(\beta\downarrow_R)
      \cap
      \condst(\tilde \gamma_{m}[]\downarrow_R) = \emptyset
    \]
    Then we have:
    \begin{alignat*}{2}
      t \;&\equiv&&\;
      C\left[\pvec{b}_{cs} \diamond
        \left(D_l
          \left[\left(\beta_i\right)_{i\in J}
            \diamond\left(\gamma_{m}\right)_{m \in M}\right],
          \Delta\right)\right] \\
      \;&\equiv&&\;
      C\left[\pvec{b}_{cs} \diamond
        \left(D_l
          \left[
            \left((\beta)_{i \in I},(\tilde \beta_{i}[\beta])_{i \not \in I}\right)
            \diamond\left(\tilde \gamma_{m}[\beta]\right)_{m \in M}
          \right],
          \Delta\right)\right]
    \end{alignat*}
    Let $C_\beta[\pvec{b}_\beta \diamond \pvec{u}_\beta] \equiv \beta\downarrow_R$. We have:
    \begin{alignat*}{2}
      &D_l\left[
        \left((\beta)_{i \in I},(\tilde \beta_{i}[\beta])_{i \not \in I}\right)
        \diamond\left(\tilde \gamma_{m}[\beta]\right)_{m \in M}
      \right] \\
      =_R\;\;&
      \begin{alignedat}[t]{2}
        \ite{C_\beta[\pvec{b}_\beta \diamond \pvec{u}_\beta]}
        {&
          D_l\left[
            \left(
              (\true)_{i \in I},(\tilde \beta_{i}[\true])_{i \not \in I}
            \right)
            \diamond\left(\tilde \gamma_{m}[\true]\right)_{m \in M}
          \right]
          \\}{&
          D_l\left[
            \left(
              (\false)_{i \in I},(\tilde \beta_{i}[\false])_{i \not \in I}
            \right)
            \diamond\left(\tilde \gamma_{m}[\false]\right)_{m \in M}
          \right] }
      \end{alignedat}
    \end{alignat*}
    Since $\pvec{u}_\beta = \leavest(\beta\downarrow_R)$, for every $u \in \pvec{u}_\beta$, $i \in J$ and $m \in M$, we know that $u \not \in \condst(\tilde \beta_{i}[]\downarrow_R)$ and $u \not \in \condst(\tilde \gamma_{m}[]\downarrow_R)$. Let $\vec{\rho}$ be the conditionals appearing on the path from the root of $t$ to $D_l[\_]$. Using Lemma~\ref{lem:no-fab-no-cap}.(d), we know that $\pvec{u}_\beta \cap \vec \rho = \emptyset$. Let $(u_o)_{o\in O}$ be such that $\pvec{u} \equiv (u_o)_{o \in O}$. By applying Proposition~\ref{prop:spurious-replace} to all $u$ we know that:
    \begin{alignat*}{2}
      &C\left[\pvec{b}_{cs} \diamond \left(\begin{alignedat}[c]{2}
            \ite{C_\beta\left[\pvec{b}_\beta \diamond \pvec{u}_\beta\right]}{&
              D_l\left[\left(
                  (\true)_{i \in I},(\tilde \beta_{i}[\true])_{i \not \in I}
                \right)
                \diamond\left(\tilde \gamma_{i}[\true]\right)_m\right]
              \\}{&
              D_l\left[
                \left(
                  (\false)_{i \in I},(\tilde \beta_{i}[\false])_{i \not \in I}
                \right)
                \diamond\left(\tilde \gamma_{i}[\false]\right)_m\right]
            }
          \end{alignedat},\Delta\right)\right]\displaybreak[0]\\
      =_R \;\;&
      C\left[\pvec{b}_{cs} \diamond \left(\begin{alignedat}[c]{2}
            \ite{C_\beta\left[\pvec{b}_\beta \diamond (\true)_o\right]}{&
              D_l\left[\left(
                  (\true)_{i \in I},(\tilde \beta_{i}[\true])_{i \not \in I}
                \right)
                \diamond\left(\tilde \gamma_{i}[\true]\right)_m\right]
              \\} {&
              D_l\left[\left(
                  (\false)_{i \in I},(\tilde \beta_{i}[\false])_{i \not \in I}
                \right)
                \diamond\left(\tilde \gamma_{i}[\false]\right)_m\right]
            }
          \end{alignedat},\Delta\right)\right]\\
      =_R \;\;&
      C\left[
        \pvec{b}_{cs}
        \diamond
        \left(
          D_l\left[
            \left(
              (\true)_{i \in I},(\tilde \beta_{i}[\true])_{i \not \in I}
            \right)
            \diamond
            \left(
              \tilde \gamma_{i}[\true]\right)_m\right],\Delta
        \right)
      \right]
      \numberthis\label{eq:GClFapjQyAQgRBw}
    \end{alignat*}

  \item \textbf{\ref{item:ab1}, Part 2.b} We do exactly the same thing on the other side: for all $i \not \in I$ we know that there exists $\tilde \beta'_{i}[]$ such that:
    \[
      \tilde \beta'_{i}[\beta'] \equiv \beta'_{i}
      \qquad \tand \qquad
      \leavest(\beta'\downarrow_R)
      \cap
      \condst(\tilde \beta'_{i}[]\downarrow_R) = \emptyset
    \]
    And, for every $m \in M$, there exists $\tilde \gamma'_{m}[]$ such that:
    \[
      \tilde \gamma'_{m}[\beta'] \equiv \gamma'_{m}
      \qquad \tand \qquad
      \leavest(\beta'\downarrow_R)
      \cap
      \condst(\tilde \gamma'_{m}[]\downarrow_R) = \emptyset
    \]
    Then by the same reasoning we have:
    \begin{alignat*}{2}
      t' \;&\equiv&&\;
      C\left[\pvec{b}'_{cs} \diamond
        \left(D_l
          \left[\left(\beta'_i\right)_i
            \diamond\left(\gamma'_{m}\right)_{m \in M}\right],
          \Delta'\right)\right] \\
      \;&\equiv&&\;
      C\left[\pvec{b}'_{cs} \diamond
        \left(D_l
          \left[
            \left(
              (\beta')_{i \in I},(\tilde \beta'_{i}[\beta'])_{i \not \in I}
            \right)
            \diamond\left(\tilde \gamma'_{m}[\beta']\right)_{m \in M}\right],
          \Delta'\right)\right]\displaybreak[0]\\
      &\;=_R&&\;
      C\left[\pvec{b}'_{cs} \diamond
        \left(
          D_l\left[
            \left(
              (\true)_{i \in I},(\tilde \beta'_{i}[\true])_{i \not \in I}
            \right)
            \diamond\left(\tilde \gamma'_{m}[\true]\right)_{m \in M}
          \right],
          \Delta'\right)\right]
      \numberthis\label{eq:OmKSKAEHPeuoZtP}
    \end{alignat*}
    Observe that corresponding sub-terms of \eqref{eq:GClFapjQyAQgRBw} and \eqref{eq:OmKSKAEHPeuoZtP} can be matched to corresponding sub-terms of $t$ and $t'$. It is straightforward to build a proof of the equivalence of \eqref{eq:GClFapjQyAQgRBw} and \eqref{eq:OmKSKAEHPeuoZtP} using $P$, except for the $\cca$ applications side-conditions. We argue why the side-conditions carry over from the derivation $P$ later in the proof.

  \item \textbf{\ref{item:ab2} and \ref{item:ab3}} The case \ref{item:ab2} works similarly to the case \ref{item:ab1}, except that we use Lemma~\ref{lem:cond-equiv-bis-body} instead of Lemma~\ref{lem:cond-equiv-body}. The case \ref{item:ab3} is exactly like the case \ref{item:ab1} when taking $I = \emptyset$.
  \end{itemize}

  \paragraph{Case $\;\sfh \ne \epsilon$}
  In that case, thanks to Lemma~\ref{lem:no-fab-no-cap}.(a), we know that $\beta \not \lecond^{\sfh,l} (t,P)$. We have three cases:
  \begin{enumerate}[(a)]
  \item either $\beta \leleave^{\sfh,l} (t,P)$: using Lemma~\ref{lem:no-fab-no-cap}.(c), there exists $\sfh_0,b^h$ such that $\sfh \in \cspos(\sfh_0)$, $b^h \in \cspath^{\sfh_0,l}(t,P)$ and $(b^h\downarrow_R) \in \leavest(\beta\downarrow_R)$. Since $\sfh \in \cspos(\sfh_0)$ implies that $\sfh_0 < \sfh$, we know that $\beta <_g b^h$. We then have two cases. Either $b^h$ is $(t,P)$-\abounded, and then using the inductive case for different labels of the definition of $(t,P)$-\abounded terms, we know that $\beta$ is $(t,P)$-abounded. Absurd. Or $b^h$ is not $(t,P)$-\abounded, which contradicts the maximality of $\beta$ among the set of terms which are not $(t,P)$-abounded. Absurd.
  \item either $\beta  \not \leleave^{\sfh,l} (t,P)$ and $\beta \in \cspath^{\sfh,l} (t,P)$: this case is done exactly like case (ii).
  \item either $\beta  \not \leleave^{\sfh,l} (t,P)$ and $\beta \not \in \cspath^{\sfh,l} (t,P)$: this case is done exactly like case (iii).
  \end{enumerate}

  \paragraph{Valid Proof Rewriting}
  We do the rewritings described above for every $\sfh$ such that $(\beta,\sfh)$ is maximal for $<_g$, and for every $l$ such that $\beta \lebt^{\sfh,l} (t,P)$ or $\beta \in \cspath^{\sfh,l} (t,P)$, \emph{simultaneously}. It remains to check that this is a valid cut elimination. The only difficulty lies in checking that all the side-conditions of the $\cca$ axiom hold. This is tedious, but here are the key ingredients:
  \begin{itemize}
  \item $\beta$ is not a guard, and the encryptions that need to be guarded in a decryption are invariant by our proof cut elimination. Therefore decryptions that were well-guarded before are still well-guarded after the cut.
  \item We did the proof rewriting simultaneously for all $\sfh$ such that $(\beta,\sfh)$ is maximal for $<_g$. Consider $\sfh'$ such that $(\beta,\sfh')$ is not maximal for $<_g$: then there exists $\sfh$ such that $(\beta,\sfh)$ is maximal for $<_g$ and $\sfh < \sfh'$. Therefore, the sub-proof at index $\sfh'$ is removed by the proof rewriting. This ensure that, for all branch $l$ where a rewriting occurred, we removed all occurrences of $\beta$. Therefore, if an encryption used to contain $\beta$ then all occurrences of this encryption have been rewritten in the same way. This guarantees that the freshness condition on encryption randomness still holds.
  \item The length constraints on encryption oracle calls still holds thanks to the branch invariance property of the length predicate $\eql{\_}{\_}$.
  \end{itemize}

  \paragraph{Conclusion}
  This concludes the proof of the second bullet point of the definition $\aproof$. The third bullet point is much simpler. We want to show that for all $l \in \prooflabel(\epsilon)$, for every path $\vec \rho$ of $\ekl^P$-normalized basic conditional from the root of $t$ to some leave, $\vec \rho$ does not contain any duplicates. We show this by proof cut elimination as follows: let $(\beta,\beta'_0) \lesimcond^{\epsilon,l}(t,P)$ and $(\beta,\beta'_1) \lesimcond^{\epsilon,l}(t,P)$, using Lemma~\ref{lem:cond-equiv-body} we have $\beta_0' \equiv \beta_1'$. Since they are on the same branch, one may rewrite the lowest occurrence of $\beta$ and $\beta'_0$ into their $\textsf{then}$ branch (we could also use the $\textsf{else}$ branch). This yield a smaller proof, and one can check that all the other properties are invariant of this proof cut elimination. We directly concludes by induction.
\end{proof}

\subsection{Bounding the Number of Nested Basic Conditionals}
We use the previous lemma to bound the number of basic conditionals appearing in a proof $P \aproof t \sim t'$. Looking at the definition of $(t,P)$-\abounded terms, one may try to show that for every $\beta \in (\lebt^{\sfh,l} (t,P) \cup \cspath^{\sfh,l}(t,P))$, if $\beta$ is $(t,P)$-\abounded then there exists $u \in \leavest(\beta\downarrow_R)$ such that $u \in \st(t\downarrow_R) \cup \st(t'\downarrow_R)$. Since $\st(t\downarrow_R) \cup \st(t'\downarrow_R)$ is finite, and since a basic term is uniquely characterized by any of its leaves, this would allow us to bound the number of basic terms appearing in $P \aproof t \sim t'$.

Unfortunately, this is not always the case. Indeed, consider $(\beta,\beta') \lecond^{\sfh,l}(t \sim t',P)$ such that $\beta'$ has a leaf term appearing in $t'$, but $\beta$ shares no leaf term with $\beta'$ nor $t$:
\begin{mathpar}
  \leavest(\beta\downarrow_R)\cap\leavest(\beta'\downarrow_R) = \emptyset

  \leavest(\beta\downarrow_R) \cap \st(t\downarrow_R) = \emptyset

  \leavest(\beta'\downarrow_R) \cap \st(t'\downarrow_R) \ne \emptyset
\end{mathpar}
$\beta'$ is \abounded since it shares a leaf term with $t'$, and using the second case, $\beta$ is \abounded too. But $\beta$ shares no leaf term with $t$ and $t'$.

Still, we can bound $\beta$. Since $(\beta,\beta') \lecond^{\sfh,l}(t \sim t',P)$, we observe that $\beta \equiv B[\pvec{w},(\alpha_i)_i,(\dec_j)_j]$ and $\beta' \equiv B[\pvec{w},(\alpha'_i)_i,(\dec'_j)_j]$. Using the fact that $\leavest(\beta'\downarrow_R) \cap \st(t'\downarrow_R)$ and that $\beta$ is a $\ekl$-normalized basic term, we know that every leaf $u \in \leavest(\beta\downarrow_R)$ is in $\st(t'\downarrow_R)$, \emph{modulo the content of the $\ekl$-encryption oracle calls}. This motivate the introduction of the notion of \emph{leaf frame}.

\paragraph{Leaf frame}
Let $\beta$ be a $\ekl$-normalized basic term, and $u,v \in \leavest(\beta\downarrow_R)$ be leaf terms of $\beta$. Then $u$ and $v$ only differ by their encryptions. That is, if one replace all the zero decryptions $\zero(\dec(\_,\sk))$ by $\dec(\_,\sk)$, and all the leaves of encryptions $\enc{m}{\pk}{\nonce}$ by $\enc{[]_\alpha}{\pk}{\nonce}$ (where $\alpha$ is the unique term of $\encs_l$ such that $\alpha \equiv \enc{\_}{\pk}{\nonce}$) in $u$ and in $v$ then you get the same context. We formalize this below, and use it to generalize Proposition~\ref{prop:bas-cond-charac}.
\begin{definition}
  Let $P \aproof t \sim t'$ and $l$ be a branch label in $\prooflabel(P)$. We define the left \emph{leaf frame} $\lframe_l^P$ of $\beta \in (\lebt^{\sfh,l} (t,P) \cup \cspath^{\sfh,l}(t,P))$ inductively as follows:
  \[
    \lframe_l^P(s) \;\equiv\;
    \begin{cases}
      \enc{[]_\alpha}{\pk}{\nonce} &
      \text{ if }
      \exists \alpha \equiv \enc{m}{\pk}{\nonce}
      \in \encs^P_l \wedge s \equiv \enc{\_}{\pk}{\nonce}\\
      \dec(\lframe_l^P(s),\sk) &
      \text{ if } \sk \in \keys_l^P \wedge s \equiv \zero(\dec(s,\sk))\\
      \lframe_l^P(v) &
      \text{ if } s \equiv \ite{b}{u}{v}\\
      f((\lframe_l^P(u_i))_i) &
      \text{ otherwise}
    \end{cases}
  \]
  We also let $\ulframe_l^P(\beta)$ be $\lframe_l^P(\beta)$ where we make every hole variable appear at most once, by replacing a hole variable $[]_\alpha$ occurring at position $p$ in $\beta$ by $[]_{\alpha,p}$.
  
  We define the right \emph{leaf frame} $\rframe_l^P$ (and its underlined version $\urframe_l^P$) of $\beta \in (\lebt^{\sfh,l} (t',P) \cup \cspath^{\sfh,l}(t',P))$, using $\encs_l'^P$ instead of $\encs_l^P$.
\end{definition}

\begin{remark}
  We have two remarks:
  \begin{itemize}
  \item We state some results only for $\lframe$. The corresponding results for $\rframe$ are obtained by symmetry.
  \item The hole variables in $\ulframe_l^P(\beta)$ are annotated by both the position $p$ of the hole \emph{and} the encryption $\alpha$ that appears at $p$ in $\beta$. By consequence, if two normalized basic terms $\beta$ and $\beta'$ are such that $\ulframe_l^P(\beta)$ and $\ulframe_l^P(\beta')$ share a hole variable $[]_{\alpha,p}$, it means that $\beta$ and $\beta'$ contain the \emph{same encryption $\alpha$ at the same position $p$}. This is crucial, as we want $\ulframe_l^P$ to uniquely characterize normalized basic terms.
    \qedhere
  \end{itemize}

\end{remark}

\begin{example}
  For all $\ekl^P$-decryption oracle call $\dec$ guarding $\dec(s[(\alpha_i)_i,(\dec_j)_j],\sk)$, if for all $i$, $\alpha_i \equiv \enc{\_}{\pk_i}{\nonce_i}$~then:
  \[
    \lframe_l^P(\dec) \equiv
    \dec\Big(
    s\Big[
    \big( \enc{[]_{\alpha_i}}{\pk_i}{\nonce_i} \big)_i,
    \big( \lframe_l^P(\dec_j)\big)_j
    \Big],\sk
    \Big)
  \]
  We also give an example of $\ulframe_l^P$. Assuming that $\alpha_0 \equiv \enc{A}{\pk}{\nonce_0}$ and $\alpha_1\equiv\enc{B}{\pk}{\nonce_1}$ are encryptions in $\encs_l^P$:
  \[
    \ulframe_l^P\left(
      \pair{\alpha_0}
      {\pair{\alpha_1}
        {\alpha_0}}
    \right) \equiv
    \spair{\enc{[]_{\alpha_0,00}}{\pk}{\nonce_0}}
    {\spair{\enc{[]_{\alpha_1,100}}{\pk}{\nonce_1}}
      {\enc{[]_{\alpha_0,110}}{\pk}{\nonce_0}}}
    \qedhere
  \]
\end{example}

\begin{proposition}
  \label{prop:f0}
  Let $P \aproof t \sim t'$ and $l \in \prooflabel(P)$. Let $b$ be an if-free term in $R$-normal form. For every $\ekl$-normalized basic terms $\gamma$, if $b \in \leavest(\gamma\downarrow_R)$ then $\lframe_l^P(b) \equiv \lframe_l^P(\gamma)$.
\end{proposition}

\begin{proof}
  This is by induction on the size of $\gamma$.
\end{proof}

\begin{proposition}
  \label{prop:base-term-frame-charac}
  Let $P \aproof t \sim t'$ and $l \in \prooflabel(P)$. For every $\ekl$-normalized basic terms $\beta,\beta'$, if $\lframe_l^P(\beta) \equiv \lframe_l^P(\beta')$ then $\beta \equiv \beta'$.
\end{proposition}

\begin{proof}
  The proof is exactly the same than for Proposition~\ref{prop:bas-cond-charac}.
\end{proof}

\begin{proposition}
  \label{prop:cs-implies-bt}
  Let $P \aproof t \sim t'$ and $l \in \prooflabel(P)$. For all $\sfh$, if $(b,b') \lesimcs^{\sfh,l} (t\sim t',P)$ then there exists $\sfh'$ and $(\gamma,\gamma') \mathbin{(\lesimcond^{\sfh',l} \cup \lesimleave^{\sfh',l})} (t \sim t',P)$ such that $b \in \leavest(\gamma\downarrow_R)$ and $b' \in \leavest(\gamma'\downarrow_R)$.
\end{proposition}

\begin{proof}
  Let $h,\sfx$ be such that $\sfh = h_\sfx$. Let $h_0 \in \cspos(\extractx(h,P))$ and $\sfx_0$ be such that $\sfx_0$ is the direction taken in $l$ at position $h_0$, and such that $Q = \extract_{\sfx_0}(h_0,P)$ is a proof of $b \sim b'$.

  Using the fact that the sub-proofs of $\csmb$ conditionals of $P$ do not use the $\obfa$ rule, we know that $Q$ lies in the fragment:
  \[
    \frakf(\csmb \cdot \fas^* \cdot \dup^* \cdot \bcca)
  \]
  Let $(\gamma,\gamma') \lesimleave^{\epsilon,l} (b \sim b',Q)$. Using the property (c) of Lemma~\ref{lem:no-fab-no-cap} (which holds thanks to $\aproof$), we know that $b \in \leavest(\gamma\downarrow_R)$ and $b \in \leavest(\gamma'\downarrow_R)$.
\end{proof}

\begin{proposition}
  \label{prop:sim-equ-frame}
  Let $P \aproof t \sim t'$ and $l \in \prooflabel(P)$. For all $\sfh$, if $(\beta,\beta') \mathbin{(\lesimcond^{\sfh,l} \cup \lesimleave^{\sfh,l} \cup \simcspath^{\sfh,l})} (t\sim t',P)$ then $\lframe_l^P(\beta) \equiv \rframe_l^P(\beta')$.
\end{proposition}

\begin{proof}
  First we deal with the case $(\beta,\beta') \mathbin{(\lesimcond^{\sfh,l} \cup \lesimleave^{\sfh,l})} (t\sim t',P)$. We know that we can extract a proof $Q$ (from $P$) such that $Q \aproof \beta \sim \beta'$ and $Q$ is in the fragment $\frakf(\fas^*\cdot \dup^* \cdot \bcca)$. The result follows from the definitions of $\lframe_l^P$ and $\rframe_l^P$.

  Now we deal with the case $(\beta,\beta') \mathbin{(\simcspath^{\sfh,l})} (t\sim t',P)$. Using Proposition~\ref{prop:cs-implies-bt} we know that there exists $\sfh'$ and $(\gamma,\gamma') \mathbin{(\lesimcond^{\sfh',l} \cup \lesimleave^{\sfh',l})} (t\sim t',P)$ such that $\beta \in \leavest(\gamma\downarrow_R)$ and $\beta' \in \leavest(\gamma'\downarrow_R)$. Since $\beta$ is if-free and in $R$-normal form, we obtain that $\lframe_l^P(\beta) \equiv \lframe_l^P(\gamma)$ by applying Proposition~\ref{prop:f0}. Similarly $\rframe_l^P(\beta') \equiv \rframe_l^P(\gamma')$. Moreover, from the previous case, we get that $\lframe_l^P(\gamma) \equiv \rframe_l^P(\gamma')$. Hence $\lframe_l^P(\beta) \equiv \rframe_l^P(\beta')$.
\end{proof}

\begin{proposition}
  \label{prop:lf-imp-ulf}
  Let $P \aproof t \sim t'$ and $l \in \prooflabel(P)$. For every $\ekl$-normalized basic terms $\beta,\beta'$, $\lframe_l^P(\beta) \equiv \lframe_l^P(\beta')$ if and only if $\ulframe_l^P(\beta) \equiv \ulframe_l^P(\beta')$.
\end{proposition}

\begin{proof}
  This is obvious, since the hole variable annotations in $\ulframe_l^P$ uniquely characterize both the position of the hole and the encryption appearing at this position.
\end{proof}

\begin{proposition}
  \label{prop:remove-subst-lframe}
  Let $P \aproof t \sim t'$ and $l \in \prooflabel(P)$. For every $\ekl$-normalized basic terms $\beta,\beta'$ and substitutions $\theta,\theta'$, if $\ulframe_l^P(\beta)\theta \equiv \ulframe_l^P(\beta')\theta'$ then $\ulframe_l^P(\beta) \equiv \ulframe_l^P(\beta')$.
\end{proposition}

\begin{proof}
  We prove this by induction on the size of $\beta$. The base case is trivial, lets deal with the inductive case. Let $\beta$ and $\beta'$ be $\ekl^P$-normalized basic terms, we know that $\beta \equiv B[\pvec{w}, (\alpha_i)_i, (\dec_j)_j]$ where:
  \begin{itemize}
  \item for every $i$, $\alpha_i \equiv \enc{m_i}{\pk_i}{\nonce_i} \in \encs_l^P$.
  \item for every $j$, $\dec_j$ is a decryption oracle call for $\dec(s_j,\sk_j)$ in $\decs_l^P$.
  % \item for every $i$, the hole corresponding to $\alpha_i$ (resp. $\alpha_i'$) appears exactly once in $B$ (resp. $B'$), at position $p_i$ (resp. $p_i'$).
  \end{itemize}
  Similarly, we have a decomposition of $\beta'$ into $B'[\pvec{w}', (\alpha'_i)_i, (\dec'_j)_j]$. By definition of $\lframe_l^P$, and using the fact that $\fresh{\rands_l^P}{\pvec{w}}$, we have:
  \[
    \lframe_l^P(\beta) \equiv
    B[
    \pvec{w},
    (\enc{[]_{\alpha_i}}{\pk_i}{\nonce_i})_i,
    \dec(\lframe_l^P(s_j),\sk_j)]
  \]
  Similarly:
  \[
    \lframe_l^P(\beta') \equiv
    B'[
    \pvec{w}',
    (\enc{[]_{\alpha'_i}}{\pk'_i}{\nonce'_i})_i,
    \dec(\lframe_l^P(s'_j),\sk'_j)]
  \]
  We have three cases:
  \begin{itemize}
  \item Either $\beta \equiv \enc{m}{\pk}{\nonce} \in \encs_l^P$. Then $\ulframe_l^P(\beta) \equiv \enc{[]_{\beta,0}}{\pk}{\nonce}$. By definition of $\lframe$, and using the fact that $\ulframe_l^P(\beta)\theta \equiv \ulframe_l^P(\beta')\theta'$, we get that $\beta'$ is of the form $\enc{m'}{\pk}{\nonce}$. We deduce from the freshness side condition of $\nonce$ that $m' \equiv m$.
  \item Or $\beta \equiv \dec$ where $\dec$ is a $\ekl^P$-decryption oracle call guarding $\dec(s,\sk)$. Then $\ulframe_l^P(\beta) \equiv \dec(\ulframe_l^P(s),\sk)\mu$, where $\mu$ is the substitution that lifts positions of $s$ into positions of $\dec(s,\sk)$, i.e. for every $\alpha \in \encs_l^P$ and position $p \in \pos(s)$:
    \[
      \mu([]_{\alpha,p}) \equiv []_{\alpha,0\cdot p}
    \]
    By definition of $\lframe$, and using the fact that $\ulframe_l^P(\beta)\theta \equiv \ulframe_l^P(\beta')\theta'$ and that $\beta'$ is a $\ekl^P$-normalized basic term, we get that $\beta'$ is also some $\dec'$ where $\dec'$ is a $\ekl^P$-decryption oracle call guarding $\dec(s',\sk)$.

    Moreover we have $\ulframe_l^P(s)\mu\theta \equiv \lframe_l^P(s')\mu\theta$, and $s,s'$ are $\ekl^P$-normalized basic terms. Hence by induction hypothesis $\ulframe_l^P(s) \equiv \ulframe_l^P(s')$, which concludes this case.
  \item Or we are not in one of the two cases above. Then, there exists $f \in \nizsig$ s.t. $\beta \equiv f(u_1,\dots,u_n)$ and $\beta' \equiv f(u'_1,\dots,u'_n)$, where $u_1,\dots,u_n$ and $u_1',\dots,u_n'$ are $\ekl^P$-normalized basic term. Hence $\ulframe_l^P(\beta)$ and $\ulframe_l^P(\beta')$ both starts with the function symbol $f$.

    Moreover, if we let, for very $1 \le i \le n$, $\mu_i$ be the lifting substitution such that, for every $\alpha \in \encs_l^P$ and position $p$, $\mu_i([]_{\alpha,p}) \equiv []_{\alpha,i\cdot p}$, then:
    \begin{mathpar}
      \ulframe_l^P(\beta)
      \equiv
      f(\ulframe_l^P(u_1)\mu_1,\dots,\ulframe_l^P(u_n)\mu_n)
      
      \ulframe_l^P(\beta')
      \equiv
      f(\ulframe_l^P(u'_1)\mu_1,\dots,\ulframe_l^P(u'_n)\mu_n)
    \end{mathpar}
    We apply $\theta$ to the equations above, and use the fact that $\ulframe_l^P(\beta)\theta \equiv \ulframe_l^P(\beta')\theta$: 
    \begin{alignat*}{2}
      f(\ulframe_l^P(u_1)\mu_1\theta,\dots,\ulframe_l^P(u_n)\mu_n\theta)
      &\;\;\equiv\;\;&&
      \ulframe_l^P(\beta)\theta\\
      &\;\;\equiv\;\;&&
      \ulframe_l^P(\beta')\theta\\
      &\;\;\equiv\;\;&&
      f(\ulframe_l^P(u'_1)\mu_1\theta,\dots,\ulframe_l^P(u'_n)\mu_n\theta)
    \end{alignat*}
    Hence, for every $1 \le i \le n$, $\ulframe_l^P(u_i)\mu_i\theta \equiv \ulframe_l^P(u'_i)\mu_i\theta$. By induction hypothesis, we deduce that $\ulframe_l^P(u_i) \equiv \ulframe_l^P(u'_i)$. Therefore $\ulframe_l^P(\beta) \equiv \ulframe_l^P(\beta')$.\qedhere
  \end{itemize}
\end{proof}

\begin{definition}
  We let $<_\st$ be the strict, well-founded, subterm ordering.
\end{definition}

\paragraph{Nested Sequences of Basic Conditionals}
We want to bound the number of nested basic conditional appearing in $P \aproof t \sim t'$. Using the contrapositive of Proposition~\ref{prop:base-term-frame-charac}, we know that when $\beta <_{\st} \beta'$ we have $\lframe_l^P(\beta) \not \equiv \lframe_l^P(\beta')$. Moreover, using Proposition~\ref{prop:lf-imp-ulf} and Proposition~\ref{prop:remove-subst-lframe}, we know that $\lframe_l^P(\beta) \not \equiv \lframe_l^P(\beta')$ implies that $\ulframe_l^P(\beta)\theta \not \equiv \ulframe_l^P(\beta')\theta'$ (for every substitutions $\theta,\theta'$).

Therefore, for any sequence of nested $\ekl^P$-normalized basic conditionals:
\[
  \beta_1 <_{\st} \dots <_{\st} \beta_n
\]
and for any substitutions $\theta_1,\dots,\theta_n$, we know that $(\ulframe_l^P(\beta_i)\theta_i)_{1 \le i \le n}$ is a sequence of pair-wise distinct terms. Tu use this, we prove that there there exists a sequence of substitutions $\theta_1,\dots,\theta_n$ such that:
\[
  \big\{
    \ulframe_l^P(\beta_1)\theta_1, \dots, \ulframe_l^P(\beta_n)\theta_n
  \big\}
  \subseteq \mathcal{B}(t,t')
\]
where $\mathcal{B}(t,t')$ is a set of bounded size w.r.t. $|t| + |t'|$. Since the $(\lframe_l^P(\beta_i)\theta_i)_{1 \le i \le n}$ are pair-wise distinct, using a pigeon-hole argument we get that $n \le |\mathcal{B}(t,t')|$.

We outline the end of this sub-section. First, we define the set of terms $\mathcal{B}(t,t')$, and show the existence of the substitutions $(\theta_i)_i$. Then, we bound the size of $\mathcal{B}(t,t')$. Finally, we bound the number of nested basic conditional $n$ using a pigeon-hole argument.

\begin{definition}
  Let $u$ be an if-free term. We let $\zeta_\keys(u)$ be the set of terms obtained from $u$ by replacing some occurrences of $\zero(\dec(w,\sk))$ by $\dec(w,\sk)$ (where $\sk \in \keys$), non-deterministically stopping at some encryptions. Formally:
  \[
    \zeta_\keys(u) =
    \begin{cases}
      \{\dec(v,\sk)\mid w \in v \in \zeta_\keys(w)\}&
      \text{ if } u \equiv \zero(\dec(w,\sk)) \tand \sk \in \keys\\
      \{u\} \cup \{ \enc{v}{\pk(\nonce)}{\nonce_r} \mid v \in \zeta_\keys(m)\}
      & \text{ if } u \equiv \enc{m}{\pk(\nonce)}{\nonce_r}
      \tand \sk(\nonce)\in \keys\\
      \{f(v_1,\dots,v_n) \mid \forall i, v_i \in \zeta_\keys(u_i)\} &
      \text{ otherwise, where } u \equiv f(u_1,\dots,u_n)
    \end{cases}
  \]
  Moreover, given a set of ground terms $\cals$, we let $\guards_\keys(\cals)$ be an over-approximation of the set of guards of terms in $\cals$:
  \[
    \guards_\keys(\cals) =
    \big\{
    \eq{s}{\alpha}
    \mid
    \dec(s,\sk(\nonce))\in\cals
    \wedge \alpha \equiv \enc{\_}{\pk(\nonce)}{\_} \in \st(s)
    \wedge \sk(\nonce) \in \keys
    \big\}
  \]
\end{definition}

\begin{definition}
  Let $\cals_{\key}(t)$ be the set of private keys appearing in $t\downarrow_R$, i.e. $\cals_\key(t) = \{\sk(\nonce) \mid \sk(\nonce) \in \st(t\downarrow_R) \}$.
  For every term $t$, we let $\mathcal{B}(t)$ be the set:
  \[
    \calb(t) =
    \bigcup_{\keys \subseteq \cals_\key(t)}
    \bigcup_{
      \scriptsize
      \begin{alignedat}{2}
        &&&u \in \st(\leavest(t\downarrow_R))\\
        &\vee&& u \in \st(\condst(t\downarrow_R))
      \end{alignedat}
    }
    \zeta_\keys(u)
    \cup
    \guards_\keys(\zeta_\keys(u))
  \]
  Moreover, we let $\mathcal{B}(t,t') = \mathcal{B}(t) \cup \mathcal{B}(t')$.
\end{definition}

\begin{proposition}
  \label{prop:lframe-theta-exists-zeta}
  Let $P\aproof t\sim t'$ and $l \in \prooflabel(P)$. Let $\beta$ be a $\ek_l^P$-normalized basic conditional. Then, for every $u \in \leavest(\beta\downarrow_R)$, there exists $\theta$ such that $\ulframe_l^P(\beta)\theta \in\zeta_\keys(u)$.
\end{proposition}

\begin{proof}
  We show this by induction on $|\beta|$.
  \begin{itemize}
  \item If $\beta$ is an encryption $\enc{m}{\pk}{\nonce} \in \encs_l^P$, then $\ulframe_l^P(\beta) \equiv \enc{[]_{\beta,0}}{\pk}{\nonce}$ and:
    \[
      \leavest(\beta\downarrow_R) =
      \left\{
        \enc{v}{\pk}{\nonce}\mid v \in \leavest(m\downarrow_R)
      \right\}
    \]
    Let $u \in \leavest(\beta\downarrow_R)$, there exists $u_m \in \leavest(m\downarrow_R)$ such that $u \equiv \enc{u_m}{\pk}{\nonce}$. Let $\theta$ be the substitution mapping $[]_{\beta,0}$ to $u_m$. Then:
    \[
      \ulframe_l^P(\beta)\theta \equiv \enc{u_m}{\pk}{\nonce}
      \equiv u \in \zeta_{\keys_l^P}(u)
    \]
  \item If $\beta$ is a decryption oracle call in $\decs_l^P$ for $\dec(s,\sk)$, the:
    \[
      \leavest(\beta\downarrow_R)\subseteq
      \left\{
        \dec(u_s,\sk) \mid u_s \in \leavest(s\downarrow_R)
      \right\}
      \cup
      \left\{
        \zero(\dec(u_s,\sk)) \mid u_s \in \leavest(s\downarrow_R)
      \right\}
    \]
    Let $u \in \leavest(\beta\downarrow_R)$, there exists $u_s \in \leavest(s\downarrow_R)$ such that $u \equiv \dec(u_s,\sk)$ or $u \equiv \zero(\dec(u_s,\sk))$. Since $s$ is a $\ekl^P$-normalized basic term, by induction hypothesis we have $\theta$ such that $\ulframe_l^P(s)\theta\in \zeta_{\keys_l^P}(u_s)$. Moreover:
    \[
      \ulframe_l^P(\beta) \equiv \dec(\ulframe_l^P(s)\mu,\sk)
    \]
    where $\mu$ is a renaming of hole variables. Let $\theta' = \mu^{-1}\theta$, then:
    \[
      \ulframe_l^P(\beta)\theta'
      \equiv \dec(\ulframe_l^P(s)\mu\mu^{-1}\theta,\sk)
      \equiv \dec(\ulframe_l^P(s)\theta,\sk)
      \in \zeta_{\keys_l^P}(u)
    \]
  \item Otherwise, $\beta \equiv f(\beta_1,\dots,\beta_n)$ where, for every $1\le i \le n$, $\beta_i$ is a $\ekl^P$-normalized basic term. Then, using the fact that $\beta$ is a $\ekl^P$-normalized basic term, we check that:
    \[
      \leavest(\beta\downarrow_R)\subseteq
      \left\{
        f(v_1,\dots,v_n) \mid
        \forall i, v_i \in \leavest(\beta_i\downarrow_R)
      \right\}
    \]
    Let $u \in \leavest(\beta\downarrow_R)$, there exists $v_1,\dots,v_n$ such that for every $1\le i\le n$ $v_i \in \leavest(\beta_i\downarrow_R)$ and $u \equiv f(v_1,\dots,v_n)$. By induction hypothesis, there exists $\theta_1,\dots,\theta_n$ such that for every $1\le i \le n$:
    \[
      \ulframe_l^P(\beta_i)\theta_i\in \zeta_{\keys_l^P}(v_i)
    \]
    For very $1 \le i \le n$, let $\mu_i$ be the lifting substitution such that, for every $\alpha \in \encs_l^P$ and position $p$, $\mu_i([]_{\alpha,p}) \equiv []_{\alpha,i\cdot p}$. Then:
    \[
      \ulframe_l^P(\beta) \equiv
      f(\ulframe_l^P(\beta_1)\mu_1,\dots,\ulframe_l^P(\beta_n)\mu_n)
    \]
    Observe that the substitutions $(\mu_i\theta_i)_{1\le i \le n}$ have disjoint domains. Let $\theta = \mu_1\theta_1\dots\mu_n\theta_n$. Then:
    \[
      \ulframe_l^P(\beta)\theta \equiv
      f(\ulframe_l^P(\beta_1)\mu_1\theta_1,\dots,\ulframe_l^P(\beta_n)\mu_n\theta_n)
    \]
    We know that $f$ cannot be the function symbol $\zero(\_)$ (since $\fa$ cannot be applied on $\zero(\_)$. It follows that:
    \[
      f(\ulframe_l^P(\beta_1)\mu_1\theta_1,
      \dots,
      \ulframe_l^P(\beta_n)\mu_n\theta_n)
      \in \zeta_{\keys_l^P}(u)
      \qedhere
    \]
  \end{itemize}
\end{proof}
We lift the previous result to \abounded conditionals.
\begin{lemma}
  \label{lem:abounded-st}
  Let $P \aproof t \sim t'$, $l$ a branch label in $\prooflabel(P)$, $\sfh$ a proof index and $\beta \in (\lebt^{\sfh,l} (t,P) \cup \cspath^{\sfh,l}(t,P))$.  If $\beta$ is $(t,P)$-\abounded then there exists a substitution $\theta$ s.t. $\ulframe_l^P(\beta)\theta \in \mathcal{B}(t,t')$.
\end{lemma}

\begin{proof}
  We prove this by induction on the well-founded order underlying the inductive definition of $(t,P)$-\abounded terms.
  \begin{itemize}
  \item \textbf{Base case:} Assume $\sfh = \epsilon$ and $\leavest(\beta\downarrow_R) \cap \st(t\downarrow_R)\ne \emptyset$. Let $u \in \leavest(\beta\downarrow_R) \cap \st(t\downarrow_R)$, we have $u$ in $R$-normal form and if-free, therefore $u \in \st(\leavest(t\downarrow_R) \cup \condst(t\downarrow_R))$. Moreover, by Proposition~\ref{prop:lframe-theta-exists-zeta}, there exists $\theta$ such that $\ulframe_l^P(\beta)\theta \in \zeta_{\keys_l^P}(u)$. Hence $\ulframe_l^P(\beta)\theta \in \mathcal{B}(t,t')$.

  \item \textbf{Base case:}  Assume $\sfh = \epsilon$ and there exists $\beta'$ such that:
    \[
      (\beta,\beta')
      \mathbin{
        (\lesimleave^{\epsilon,l} \cup
        \lesimcond^{\epsilon,l}  \cup
        \lesimcs^{\epsilon})}
      (t \sim t',P)
      \quad\text{ and }\quad
      \leavest(\beta'\downarrow_R) \cap \st(t'\downarrow_R)\ne \emptyset
    \]
    By Proposition~\ref{prop:sim-equ-frame} we know that $\lframe_l^P(\beta) \equiv \rframe_l^P(\beta')$. By Proposition~\ref{prop:lf-imp-ulf}, we deduce that $\ulframe_l^P(\beta) \equiv \urframe_l^P(\beta')$. From the previous case we know that there exists $\theta$ such that $\urframe_l^P(\beta')\theta \in \mathcal{B}(t')$. Therefore $\ulframe_l^P(\beta)\theta \in \mathcal{B}(t')$.

  \item \textbf{Inductive case, same label:} Assume $\beta \in \cspath^{\sfh,l}(t,P)$ and that there exists $\varepsilon \lebt^{\sfh,l} (t,P)$ such that $\varepsilon$ is $(t,P)$-\abounded and $\beta \in \leavest(\varepsilon \downarrow_R)$. By induction hypothesis we have $\theta$ such that $\ulframe_l^P(\varepsilon)\theta \in \mathcal{B}(t,t')$. We know that $\beta$ is if-free and in $R$-normal form and that $\varepsilon$ is a $\ekl^P$-normalized basic term. Therefore, by Proposition~\ref{prop:f0}, we have $\lframe_l^P(\beta) \equiv \lframe_l^P(\varepsilon)$. Hence, using Proposition~\ref{prop:lf-imp-ulf}, $\ulframe_l^P(\beta)\theta \in \mathcal{B}(t,t')$.

  \item \textbf{Inductive case, different labels:} Similar to the previous case.

  \item \textbf{Inductive case, guard:} If there exists $\varepsilon \lebt^{\sfh,l} (t,P)$ such that:
    \begin{itemize}
    \item $\varepsilon \equiv B[\pvec{w},(\alpha_i)_i,(\dec_j)_j]$ is $(t,P)$-\abounded.
    \item $\beta$ is a guard of a $\ekl^P$-decryption oracle call $d \in (\dec_j)_j$.
    \end{itemize}
    By induction hypothesis there exists $\theta$ such that $\ulframe_l^P(\varepsilon)\theta \in \mathcal{B}(t,t')$. Moreover let $(\pk_i)_i$ and $(\nonce_i)_i$ be such that $\forall i, \alpha_i \equiv \enc{\_}{\pk_i}{\nonce_i}$. Then:
    \[
      \lframe_l^P(\varepsilon)
      \equiv
      B\left[
        \pvec{w},
        \big(\enc{[]_{\alpha_i}}{\pk_i}{\nonce_i}\big)_i,
        \big(\lframe_l^P(\dec_j) \big)_j
      \right]
    \]
    Therefore there exists a renaming of hole variables $\mu$ such that $\ulframe_l^P(d)\mu\theta \in \st(\ulframe_l^P(\varepsilon)\theta)$. Since $\mathcal{B}(t,t')$ is closed under $\st$, this implies that:
    \[
      \ulframe_l^P(d)\mu\theta \in \mathcal{B}(t,t')
      % \numberthis\label{eq:qnmPzCyaatluCtO}
    \]
    $d$ is of the form $\dec(s,\sk)$ where $\sk \in \keys$. Since members of $\guards_\keys(\_)$ are of the form $\eq{\_}{\_}$, we know that there exists some $u \in \st(\leavest(t\downarrow_R) \cup \condst(t\downarrow_R))$ such that $\ulframe_l^P(d)\mu\theta \in \zeta_\keys(u)$. Since $\beta$ is a guard of $d$, $\beta$ is of the form $\eq{s}{\alpha}$ where $\alpha$ is an encryption under key $\pk$ (corresponding to $\sk$) and randomness $\nonce$ appearing directly in $s$. It follows that:
    \begin{mathpar}
      \lframe_l^P(d) \equiv \dec(\lframe_l^P(s),\sk)

      \lframe_l^P(\beta) \equiv \eq{\lframe_l^P(s)}{\enc{[]_{\alpha}}{\pk}{\nonce}}
    \end{mathpar}
    Since $\alpha$ appears directly in $s$, and since $\ulframe_l^P(d)\mu\theta\in \zeta_\keys(u)$, there exists $\theta'$ such that:
    \[
      \ulframe_l^P(\beta)\theta'\in
      \guards_\keys(\zeta_\keys(u))
      \subseteq \calb(t,t')
      \qedhere
    \]
  \end{itemize}
\end{proof}

We now bound the size of $\mathcal{B}(t)$.
\begin{proposition}
  \label{prop:bound-b-t-tp}
  For every term $t$, for every $u \in \mathcal{B}(t)$, we have $|u| \le |t\downarrow_R|$. Moreover:
  \[
    |\mathcal{B}(t)| \le
    |t\downarrow_R|^2.2^{|t\downarrow_R|}
  \]
\end{proposition}

\begin{proof}
  An over-approximation of the set of terms $\zeta_\keys(u)$ is obtained from $u$ by choosing a subset of positions of $u$ where decryptions over keys in $\keys$ occur, and removing $\zero$ before the subterms at these positions (if there is one). Hence each element of $\zeta_\keys(u)$ is of size at most $|u|$.  Moreover, for every $u \in \st(\leavest(t\downarrow_R) \cup \condst(t\downarrow_R))$, we have $u \in \st(t\downarrow_R)$, and therefore ${|u|} \le {|t\downarrow_R|}$.  Therefore the set $\zeta_\keys(u)$ contains terms of size at most $|t\downarrow_R|$.

  Let $\dec(s,\sk) \in \zeta_\keys(u)$, then $|\dec(s,\sk)| = |s| + 3$ and for every $\alpha$ appearing in $s$:
  \[
    |\eq{s}{\alpha}| = |s| + |\alpha| + 1 \le 2|s| + 1
    \le 2|\dec(s,\sk)| \le 2|t\downarrow_R|
  \]
  Hence the set $\guards_\keys(\zeta_\keys(u))$ contains terms of size at most $2|t\downarrow_R|$. We deduce that for every $v \in \mathcal{B}(t)$, $|v| \le 2|t\downarrow_R|$. Moreover, by upper-bounding the positions of $\dec(s,\sk)$ where an encryption might be, there are at most $|s| - 1\le |t\downarrow_R| - 1$ such $\alpha$, independently of the set of keys $\keys$. It follows that:
  \[
    \Big|
    \bigcup_{\calk\subseteq\cals_\key(t)}\guards_\keys(\zeta_\keys(u))
    \Big|
    \le|\zeta_\keys(u)| . (|t\downarrow_R| - 1)
  \]
  Independently of the set of keys $\keys$ chosen, we have at most $|\st(t\downarrow_R)| \le |t\downarrow_R|$ choices for $u$, and the set $\bigcup_{\calk\subseteq\cals_\key(t)}\zeta_\keys(u)$ contains at most $2^{|u|}\le 2^{|t\downarrow_R|}$ elements (we choose the positions where we remove $\zero$s). Hence:
  \begin{alignat*}{2}
    \Big|
    \bigcup_{\calk\subseteq\cals_\key(t)}\zeta_\keys(u)\cup \guards_\keys(\zeta_\keys(u))
    \Big|
    &\;\;\le\;\;&&
    \Big|\,
    \bigcup_{\calk\subseteq\cals_\key(t)}\zeta_\keys(u)
    \,\Big|
    +
    \Big|\,
    \bigcup_{\calk\subseteq\cals_\key(t)}\guards_\keys(\zeta_\keys(u))
    \,\Big|\\
    &\;\;\le\;\;&&
    |\zeta_\keys(u)| + (|t\downarrow_R| - 1).|\zeta_\keys(u)|
    \le
    |t\downarrow_R|.2^{|t\downarrow_R|}
  \end{alignat*}
  By consequence:
  \[
    |\calb(t)| \le
    |t\downarrow_R| .
    |t\downarrow_R|.2^{|t\downarrow_R|}
    \le
    |t\downarrow_R|^2 . 
    2^{|t\downarrow_R|}
    \qedhere
  \]
\end{proof}

Finally, we apply a pigeon-hole argument to bound the number of nested basic terms.
\begin{lemma}
  \label{lem:bound-depth}
  Let $P \aproof t \sim t'$. Let $l$ be a branch label in $\prooflabel(P)$, $\sfh$ a proof index. Let $(\beta_i)_{i \le n}$ such that for all $i$, $\beta_i\lebt^{\sfh,l} (t,P)$. If $\beta_1 <_\st \dots <_\st \beta_n$ then $n \le |\mathcal{B}(t,t')|$.
\end{lemma}

\begin{proof}
  For every $i \ne j$, we know, using Proposition~\ref{prop:base-term-frame-charac}, that $\lframe_l^P(\beta_i) \not \equiv \lframe_l^P(\beta_j)$. By Proposition~\ref{prop:lf-imp-ulf}, we deduce that $\ulframe_l^P(\beta_i) \not \equiv \ulframe_l^P(\beta_j)$. Since $P \aproof t \sim t'$, we know that for every $i$, $\beta_i$ is $(t,P)$-\abounded. Using Lemma~\ref{lem:abounded-st}, we deduce that for every $i$, there exists a substitution $\theta_i$ such that:
  \[
    \ulframe_l^P(\beta_i)\theta_i \in \mathcal{B}(t,t')
  \]
  Using the contrapositive of Proposition~\ref{prop:remove-subst-lframe}, we have that for every $i \ne j$:
  \[
    \ulframe_l^P(\beta_i)\theta_i \not \equiv \ulframe_l^P(\beta_j)\theta_j
  \]
  Therefore, by a pigeon-hole argument, $n \le |\mathcal{B}(t,t')|$.
\end{proof}

\subsection{Candidate Sequences}
Let $P \aproof t\sim t'$. For all $n \le |\mathcal{B}(t,t')|$, we are going to define the set $\pterms{n}$ of normalized basic terms that may appear in $P$ using $n$ nested basic terms. We then show that these sets are finite and recursive, and give an upper-bound on their size which does not depend on $n$. This allows us to conclude by showing that the existence of a proof using our (complete) strategy is decidable.

\begin{definition}
  An $\alpha$-context $C$ is a context such that all holes appear below the encryption function symbol, with proper randomness and encryption key. More precisely, for every position $p \in \pos(C)$, if $C_{|p} \equiv []$ then $p = p'\cdot 0$ and there exist two nonces $\nonce,\nonce_r$ such that $C_{|p'} \equiv \enc{[]}{\pk(\nonce)}{\nonce_r}$.

  Moreover, we require that every hole appears at most once.
\end{definition}

\begin{remark}
  For every $\beta \lebt^{\sfh,l}(t,P)$, the context $\lframe_l^P(\beta)$ is an $\alpha$-context.
\end{remark}

Let $t$ and $t'$ be two ground terms. We now define what is a \emph{valid candidate sequence} $(\calu_n,\cala_n)_{n \in \mathbb{N}}$ for $t,t'$. Basically, $\calu_n$ corresponds to basic terms at nested depth $n$ that could appear, on the left, in a proof of $\aproof t \sim t'$, while $\cala_n$ is the set of left encryptions oracle calls built using basic terms in $\calu_{n-1}$.
\begin{definition}
  Let $t,t'$ be two terms. A sequence of pairs of sets of ground terms $(\mathcal{U}_n,\mathcal{A}_n)_{n \in \mathbb{N}}$ is a \emph{valid candidate sequence} for $t,t'$ if:
  \begin{itemize}
  \item $\pterms{0} = \mathcal{B}(t,t')$ and $\aterms{0} = \emptyset$.
  \item For $n \ge 0$, $\aterms{n+1}$ can be any set of terms that satisfies the following constraints (with the convention that $\aterms{-1} = \emptyset$): $\aterms{n+1}$ contains $\aterms{n}$, and for all $\alpha \in \aterms{n+1}\backslash\aterms{n}$, $\alpha\equiv \enc{D[\pvec{b} \diamond \pvec{u}]}{\pk(\nonce_p)}{\nonce_r}$ where:
    \begin{itemize}
    \item $\pvec{b} \cup \pvec{u}$ are in $\pterms{n-1}$ and there exists $\enc{\_}{\_}{\nonce_r} \in \st(t\downarrow_R) \cup \st(t'\downarrow_R)$.
    \item for every branch $\vec \rho \subseteq \pvec{b}$ of $D[\pvec{b} \diamond \pvec{u}]$, $\vec \rho$ does not contain duplicates.
    \item $\aterms{n}$ does not contain any terms of the form $\enc{\_}{\_}{\nonce_r}$.
    \end{itemize}
  \item For $n>0$, we let $\pterms{n+1}$ is the set of term defined from $\pterms{n}$ and $\aterms{n}$ as follows: $\pterms{n+1}$ contains $\pterms{n}$, plus any element that can be obtained through the following construction:
    \begin{itemize}
    \item Take a $\alpha$-context $C$ such that there exists $\theta$ with $C\theta \in \mathcal{B}(t,t')$.
    \item Let $[]_1,\dots,[]_a$ be the variables of $C$, and let $\alpha_1,\dots,\alpha_a$ be encryptions in $\cala_n$. For all $1 \le k \le a$, let $s_i$ be such that $\enc{s_i}{\_}{\_} \equiv \alpha_i \in \aterms{n}$.
    \item Let $v_0 \equiv C[(s_i)_{1 \le i \le a}]$. Then let $v$ be the term obtained from $v_0$ as follows: take positions $p_1,\cdots,p_o \in \pos(C)$ such that for all $1 \le i \le o$, $C_{|p_i} \equiv \dec(\_,\sk_i)$ (where $\sk_i$ is a valid private key, i.e. of the form $\sk(\nonce_i)$); for every $1 \le i \le o$, replace in $v_0$ the subterm $\dec(s,\sk)$ at position $p$ by $D[\pvec{g} \diamond \pvec{w}]$, where $\pvec{g}$ are terms in $\pterms{n}$ of the form $\eq{s}{\alpha}$ (with $\alpha \equiv \enc{\_}{\_}{\nonce_\alpha} \in \aterms{n}$ and $\alpha$ directly appears in $s$) and $\forall w \in \pvec{w}$, $w \equiv \dec(s,\sk)$ or $w \equiv \zero(\dec(s,\sk))$.
    \end{itemize}
  \end{itemize}
\end{definition}

\begin{proposition}
  \label{prop:candidatesequence}
  Let $P \aproof t \sim t'$. For $l \in \prooflabel(P)$, there exists a valid candidate sequence $(\mathcal{U}_n,\mathcal{A}_n)_{n \in \mathbb{N}}$ for $t,t'$ such that:
  \[
    \bigcup_{\sfh} \lebt^{\sfh,l} (t,P)
    \subseteq
    \bigcup_{n < |\mathcal{B}(t,t')|} \pterms{n}
    \qquad \tand \qquad
    \bigcup_{\sfh}
    \cspath^{\sfh,l}(t,P)
    \subseteq
    \bigcup_{n < |\mathcal{B}(t,t')|} \leavest\left(\pterms{n}\downarrow_R\right)
  \]
\end{proposition}

\begin{proof}
  First, we show that there exists a valid candidate sequence such that the inclusion holds when taking the union over $\mathbb{N}$ on the right, and s.t. for every $n$, $\cala_n$ contains only valid encryptions in $\encs_l^P$, i.e.:
  \begin{equation}
    \label{eq:prop-aterms}
    \cals =
    \bigcup_{\sfh} \lebt^{\sfh,l} (t,P)
    \subseteq
    \bigcup_{n < +\infty} \pterms{n}
    \qquad\tand\qquad
    \bigcup_{n\in\mathbb{N}} \cala_n
    \subseteq
    \encs_l^P
  \end{equation}
  Before starting the construction of the valid candidate sequence, we make some observations: if one fixes $(\aterms{n})_{n \in\mathbb{N}}$, there is at most one sequence $(\pterms{n})_{n \in \mathbb{N}}$ such that $(\pterms{n},\aterms{n})_{n \in\mathbb{N}}$ is a valid candidate sequence.

  Moreover this sequence is non-decreasing in $(\aterms{n})_{n \in\mathbb{N}}$. More precisely, if $(\pterms{n},\aterms{n})_{n \in\mathbb{N}}$  and $(\ptermsp{n},\atermsp{n})_{n \in\mathbb{N}}$ are valid candidate sequences such that for every $n$, $\aterms{n} \subseteq \atermsp{n}$, then for every $n$, $\pterms{n} \subseteq \ptermsp{n}$.

  We now describe a procedure that recursively construct $\mathcal{S}' \subseteq \mathcal{S}$ and a valid candidate sequence $(\pterms{n},\aterms{n})_{n \in\mathbb{N}}$ such that $\mathcal{S}'$ is a subset of $\bigcup_{n \le +\infty} \pterms{n}$ (eventually, we will show that $\cals'=\cals$). Moreover we require $(\aterms{n})_{n \in\mathbb{N}}$ to be minimal in the following sense: if $\alpha \equiv C[\pvec{b} \diamond \pvec{u}]$ is in $\aterms{n+1}\backslash\aterms{n}$ then there exists $v \in \pvec{b} \cup \pvec{u}$ such that $v \in \pterms{n}\backslash \pterms{n-1}$ (in other words, we add new encryptions in $\cala_n$ as soon as we can).

  Initially we take $\aterms{n} = \emptyset$ for every $n$, $(\pterms{n})_{n \in \mathbb{N}}$ such that $(\pterms{n},\aterms{n})_{n \in\mathbb{N}}$ is a valid candidate sequence and $\mathcal{S}' = \emptyset$. While $\mathcal{S}' \ne \mathcal{S}$, we pick an element $\beta$ in $\mathcal{S} \backslash \mathcal{S}'$ such that $\beta$ is minimal for $<_{\st}$ in $\cals\backslash\cals'$. Then we add $\beta$ to $\mathcal{S}'$ and update $(\aterms{n})_{n \in \mathbb{N}}$ as follows:
  \paragraph{Case 1}
  If $\beta$ is minimal for $<_{\st}$ in $\mathcal{S}$, we have $\beta$ of the form $B[\pvec{w},(\alpha_i)_{i \in I},(\dec_j)_{j \in J}]$. By minimality of $\beta$, we have $I = \emptyset$ and for all $j \in J$, $\dec_j$ has no encryptions in $\encs_l^P$, and by consequence no guards. It follows that $\beta$ is if-free and in $R$-normal form, hence $\lframe_l^P(\beta) \equiv \beta$. By consequence, using Lemma~\ref{lem:abounded-st}, we get that $\beta \in \mathcal{B}(t,t') = \pterms{0}$ (since $\pterms{0}$ does not depends on the sets $(\aterms{n})_{n \in \mathbb{N}}$).

  \paragraph{Case 2}
  Let $\beta$ such that for all $\beta' <_{\st} \beta$, $\beta' \in \mathcal{S}'$. Since $\cals'\subseteq \cup_{n\in\mathbb{N}}\pterms{n}$, and since $\left\{\beta'\mid \beta' <_{\st} \beta\right\}$ is finite, there exists $n_m$ such that:
  \[
    \left\{\beta'\mid \beta' <_{\st} \beta\right\}
    \cap
    \left(\lebt^{\sfh,l} (t,P) \cup \cspath^{\sfh,l}(t,P)\right)
    \subseteq \bigcup_{0 \le n \le n_m} \pterms{n}
  \]
  From Lemma~\ref{lem:abounded-st} we have a substitution $\theta$ such that:
  \[
    \ulframe_l^P(\beta)\theta \in \mathcal{B}(t,t')
  \]
  We then just need to show that we can obtain $\beta$ from $\ulframe_l^P(\beta)$ using the procedure defining $\pterms{n_m+1}$:
  \begin{itemize}
  \item For all encryption $\alpha \equiv \enc{m}{\pk}{\nonce} \in \st(\beta) \cap \encs^P_l$, we know that $m \equiv C[\pvec{b} \diamond \pvec{u}]$ where $\pvec{b},\pvec{u} <_{\st} \beta$. Hence $\pvec{b},\pvec{u}$ are in $\cup_{0 \le n \le n_m} \pterms{n}$. We then have two cases:
    \begin{itemize}
    \item either $\cup_{n \in \mathbb{N}}\aterms{n}$ already contains an encryption $\alpha'$ with randomness $\nonce$. Since $\cup_{n\in\mathbb{N}} \cala_n \subseteq \encs_l^P$, and using the side-condition of the $\CCA$ application, we know that $\alpha \equiv \alpha' \in \cup_{n \in \mathbb{N}}\aterms{n}$. By minimality of the $(\aterms{n})_{n \in \mathbb{N}}$ we know that $\alpha \in \aterms{n_m+1}$.
    \item or $\cup_{n \in \bbN}\aterms{n}$ does not contain an encryption with randomness $\nonce$. Then we simply add $\alpha$ to $\aterms{n'}$, where $n' \le n_m + 1$ is the smallest possible: we know that there exists such a $n'$ since adding $\alpha$ to $\aterms{n}$ yields, after completion of the $(\pterms{n})_{n \in \mathbb{N}}$, a valid candidate sequence (one can check that for all branch $\vec \rho$ of $C[\pvec{b} \diamond \pvec{u}]$, $\vec \rho$ does not contain duplicates, using the third bullet point of the definition of $\aproof$).
    \end{itemize}
    Then we replace in $\ulframe_l^P(\beta)$ the holes $[]_\alpha,\_$ by $\enc{C[\pvec{b} \diamond \pvec{u}]}{\pk}{\nonce}$. This produce a term $v_0$.
  \item Finally we also replace in $v_0$ every occurrence of $\dec(\_,\sk)$ or $\zero(\dec(\_,\sk))$ in $\st(\lframe_l^P(\beta))$ by the corresponding $\ekl^P$-decryption oracle call, which is possible since the guards $\pvec{g}$ of this decryption oracle calls are such that $\pvec{g} <_{\st} \beta$, hence are in $\cup_{0 \le n \le n_m} \pterms{n}$.
  \end{itemize}

  \paragraph{Conclusion}
  We show that when $\mathcal{S} = \mathcal{S}'$ we have:
  \begin{equation}
    \label{eq:prop-dont-remember}
    \mathcal{S} \cap \bigcup_{n < +\infty} \pterms{n}
    \;=\;
    \mathcal{S} \cap \bigcup_{n < |\mathcal{B}(t,t')|} \pterms{n}
  \end{equation}
  Assume that $\mathcal{S} \cap \pterms{|\mathcal{B}(t,t')| - 1} \subsetneq \mathcal{S} \cap \pterms{|\mathcal{B}(t,t')|} $, take $\beta \in \mathcal{S} \cap (\pterms{|\mathcal{B}(t,t')|} \backslash \pterms{|\mathcal{B}(t,t')| - 1})$. We know that $\beta \equiv B[\pvec{w}, (\alpha_i)_i,(\dec_j)_j]$ and that there is an encryption $\alpha$ in $(\alpha_i)_i$ or in the encryptions of the $(\dec_j)_j$ such that $\alpha \in \aterms{|\mathcal{B}(t,t')|-1} \backslash \aterms{|\mathcal{B}(t,t')|-2}$ (otherwise $\beta$ would be in $\mathcal{S} \cap \pterms{|\mathcal{B}(t,t')| - 1}$). Let $\alpha \equiv \enc{C[\pvec{b} \diamond \pvec{u}]}{\pk}{\nonce}$, by minimality of the $(\aterms{n})_{n \in \mathbb{N}}$ we know that there is some $v \in \pvec{b} \cup \pvec{u}$ such that $v \in \pterms{|\mathcal{B}(t,t')|-1} \backslash \pterms{|\mathcal{B}(t,t')|-2}$. Since $\beta$ is in $\mathcal{S}$ and since $v$ is a $\ekl^P$-normalized basic term appearing in $\beta$ we know that $v \in \mathcal{S}$. Let $\beta_0 \equiv \beta$, $\beta_1 \equiv v$, we have $v \in \mathcal{S}\cap(\pterms{|\mathcal{B}(t,t')|-1} \backslash \pterms{|\mathcal{B}(t,t')|-2})$. By induction we can build a sequence of terms $\beta_n$, for $n \in \{0,\dots,|\mathcal{B}(t,t')|\}$ such that for all $0 \le n \le |\mathcal{B}(t,t')| $, $\beta_n \in \mathcal{S}\cap(\pterms{|\mathcal{B}(t,t')|-i}\backslash \pterms{|\mathcal{B}(t,t')|-(i+1)})$ and $\beta_{n+1} <_{\st} \beta_{n}$ (with the convention $\pterms{-1} = \emptyset$). We built a sequence of terms in $\mathcal{S}$, strictly ordered by $<_{\st}$ and of length $|\mathcal{B}(t,t')|+1$. This contradicts Lemma~\ref{lem:bound-depth}. Absurd.

  To finish, it remains to show that:
  \[
    \bigcup_{\sfh}
    \cspath^{\sfh,l}(t,P)
    \subseteq
    \bigcup_{n < |\mathcal{B}(t,t')|} \leavest\left(\pterms{n}\downarrow_R\right)
  \]
  Let $b$ in $\bigcup_{\sfh} \cspath^{\sfh,l}(t,P)$. Using Proposition~\ref{prop:cs-implies-bt} we know that there exists $\gamma \lebt^{\sfh',l}(t,P)$ such that $b \in \leavest(\gamma\downarrow_R)$. Since $\gamma \in \bigcup_{n < |\mathcal{B}(t,t')|} \pterms{n}\downarrow_R$, we have $ b \in \bigcup_{n < |\mathcal{B}(t,t')|} \leavest\left(\pterms{n}\downarrow_R\right)$.
\end{proof}

\begin{proposition}
  \label{prop:bound-acontext}
  For all terms $ u$, let $\mathcal{C}_{u}$ be the set of $\alpha$-contexts:
  \[
    \mathcal{C}_u =
    \left\{
      C\mid \exists \theta.\,
      C\theta \equiv u \wedge \text{every hole appears at most once}
    \right\}
  \]
  and $\mathcal{C}^\alpha_u$ be $\mathcal{C}_u$ quotiented by the $\alpha$-renaming of holes relation. Then $|\mathcal{C}^\alpha_u| \le 2^{|u|}$.
\end{proposition}

\begin{proof}
  The set of contexts $\mathcal{C}^\alpha_u$ can be injected in the subsets of positions of $u$ as follows: for every context $C$, associate to $C$ the set of positions of $u$ such that $C_{|p}$ is a hole. This is invariant by $\alpha$-renaming and uniquely characterizes $C$ modulo hole renaming. It follows that there are less element of $\mathcal{C}^\alpha_u$ than subsets of $\pos(u)$, i.e. $2^{|\pos(u)|} = 2^{|u|}$.
\end{proof}

\begin{proposition}
  \label{prop:boundpterms}
  Let $t$ and $t'$ be two ground terms, $N = |t\downarrow_R| + |t'\downarrow_R|$. For every valid candidate sequence $(\mathcal{U}_n,\mathcal{A}_n)_{n \in \mathbb{N}}$ and $n \in \mathbb{N}$:
  \begin{mathpar}
    \left|\aterms{n} \right| \le N

    \left| \pterms{n} \right| \le N^2.2^{3.N}
  \end{mathpar}
\end{proposition}

\begin{proof}
  For every $n$,  $\aterms{n}$ contains only terms of the form $\alpha \equiv \enc{m}{\pk}{\nonce_r}$, where $\enc{\_}{\_}{\nonce_r} \in \st(t\downarrow_R) \cup \st(t'\downarrow_R)$. Moreover, $\aterms{n}$ cannot contain two encryptions using the same randomness. Therefore $\left|\aterms{n} \right| \le N$.

  For every $n$, the only leeway we have while constructing the terms in $\pterms{n}$ is in the choice of the $\alpha$-context $C$, as the content of the encryptions is determined by $\aterms{n-1}$, and the guards that are added are determined by $\pterms{n-1}$. The $\alpha$-context $C$ is picked in the following set:
  \[
    \bigcup_{u \in \mathcal{B}(t,t')}\mathcal{C}_u^\alpha
  \]
  which, using Proposition~\ref{prop:bound-b-t-tp} and Proposition~\ref{prop:bound-acontext}, we can bound by:
  \[
    \Big|\bigcup_{u \in \mathcal{B}(t,t')}\mathcal{C}_u^\alpha \Big|
    \quad\le \quad
    \sum_{u \in \mathcal{B}(t,t')}\left|\mathcal{C}_u^\alpha \right|
    \quad\le \quad
    \sum_{u \in \mathcal{B}(t,t')} 2^{2.N}
    \quad\le \quad
    N^2.2^N.2^{2.N}
    \quad= \quad
    N^2.2^{3.N}
    \qedhere
  \]
\end{proof}

\begin{proposition}
  \label{prop:cand-seq-term-size}
  Let $t,t'$ be two ground terms and $N = |t\downarrow_R| + |t'\downarrow_R|$. For every valid candidate sequence $(\mathcal{U}_n,\mathcal{A}_n)_{n \in \mathbb{N}}$ and $n \in \mathbb{N}$:
  \[
    \forall u \in
    \bigcup_{n < |\mathcal{B}(t,t')|} \pterms{n},\, |u| \le 2^{Q(N)\,.\,2^{4.N}}
  \]
  Where $Q(X)$ is a polynomial of degree $4$.
\end{proposition}

\begin{proof}
  Even though there are at most $|\mathcal{B}(t,t')|.N^2.2^{3.N}$ distinct basic terms appearing in branch $l$ at proof index $\sfh$, these terms may be much larger. Let $U_n$ (resp. $A_n$) be an upper bound on the size of a term in $\pterms{n}$ (resp. $\aterms{n}$). Then for every $0 \le n < |\mathcal{B}(t,t')|$ and $\alpha \in \aterms{n+1}\backslash \aterms{n}$, $\alpha$ is of the form $\enc{C[\pvec{b} \diamond \pvec{u}]}{\pk}{\nonce}$, where $\pvec{b},\pvec{u}$ are in $\pterms{n}$ and $C$ is such that no term appears twice on the same branch. Recall that we call branch the ordered list of \emph{inner conditionals}, which does not include the final leaf. If follows that $C$ is of depth at most $|\pterms{n}| + 1$, and therefore has at most $2^{|\pterms{n}| + 2} - 1$ conditional and leaf terms. To bound $|C[\pvec{b} \diamond \pvec{u}]|$, we need to bound the size of each of its internal and leaf terms, which we do using $U_{n}$. We~get:
  \[
    \big|C[\pvec{b} \diamond \pvec{u}]\big|
    \le |C| + |C| \,.\, U_n \le
    2.|C| \,.\, U_n \le
    2^{|\pterms{n}| + 3}\,.\, U_n
  \]
  since $U_n$ is greater than $1$ (terms can not be of size $0$). Therefore $|\alpha| \le 4 + 2^{|\pterms{n}| + 3}\,.\, U_n$. Using the bound from Proposition~\ref{prop:boundpterms}, we can take:
  \[
    A_{n} = 4 + 2^{N^2.2^{3.N} + 3}\,.\, U_n
  \]
  Now let $u \equiv C[(\alpha_i)_{i \in I},(\dec_j)_{j \in J}]$ in $\pterms{n+1}\backslash\pterms{n}$. We know that $\forall i \in I, |\alpha_i| \le A_n$. There are at most $|C|$ hole occurrences in $C$, hence $|I| \le |C|$ and $|J|\le |C|$. To bound $|u|$, we also need to bound the size of the decryption guards. There are at most $N$ guards for each decryption (since only element of $\aterms{n}$ may be guarded, and $|\aterms{n}| \le N$), and each guard is in $\pterms{n}$, so of size bounded by $U_n$. Moreover, guarded decryptions have at most $N+1$ leaf, where each life is of size at most $|C[(\alpha_i)_{i \in I},([])_{j \in J}]| + 1\le |C| + |I|.A_n + 1$. Hence every decryption's size is upper-bounded by:
  \[
    N + N.U_n + (N+1).(|C| + |I|.A_n + 1)
  \]
  Finally $|C|$ is such that there there exists $\theta$ such that $C\theta \in \mathcal{B}(t,t')$, hence $|C| \le 2.N$ using Proposition~\ref{prop:bound-b-t-tp}. Hence, assuming $U_n \ge N$ (which will be the case):
  \begin{alignat*}{2}
    \left|C[(\alpha_i)_{i \in I},(\dec_j)_{j \in J}]\right|
    &\;\;\le\;\;&&
    |C| + |I|.A_n + |J|.(N + N.U_n + (N+1).(|C| + |I|.A_n + 1))\\
    &\;\;\le\;\;&&
    2N + 2N.A_n + 2N.(N + N.U_n + (N+1).(2N + 2N.A_n + 1))
  \end{alignat*}
  Seen as a multi-variate polynomial in $N$, $A_n$ and $U_n$, we have only monomials $N$, $N.A_n$, $N^2$, $N^2.U_n$, $N^3$ and $N^3.A_n$. Hence there exists a constant $L$ such that:
  \[
    u \le L.N^3(A_n+U_n)
    \le
    L.N^3(4 + 2^{N^2.2^{3.N} + 3}.U_n + U_n)
  \]
  Hence there exists some polynomial $Q_0$ of degree two such that $u \le 2^{Q_0(N).2^{3N}}.U_n$. We let $U_0 = N$, and $U_{n+1} = 2^{Q_0(N).2^{3N}}.U_n$. Then:
  \[
    U_{|\mathcal{B}(t,t')| - 1}
    \le 2^{|\mathcal{B}(t,t')|.Q_0(N).2^{3N}}.U_n
    \le 2^{N^2.2^N.Q_0(N).2^{3N}}.U_n
    \le 2^{N^2.Q_0(N).2^{4N}}.U_n
  \]
  Hence we have a polynomial $Q(N) = N^2.Q_0(N)$, which is of degree four.
\end{proof}

\begin{corollary}
  \label{cor:bound-size-bt}
  Let $P \aproof t \sim t'$ and $N = |\mathcal{B}(t,t')|$. For $l \in \prooflabel(P)$ and for all proof index $\sfh$:
  \[
    \forall u \in
    \left(\lebt^{\sfh,l} (t,P) \cup \cspath^{\sfh,l}(t,P) \right),\,
    |u| \le 2^{Q(N)\,.\,2^{4.N}}
  \]
\end{corollary}

\begin{proof}
  Direct consequence of Proposition~\ref{prop:candidatesequence} and Proposition~\ref{prop:cand-seq-term-size}.
\end{proof}

To conclude, we only need to bound the number of nested $\csmb$ conditionals.
\begin{proposition}
  \label{prop:bound-csmb-nested}
  Let $P \aproof t \sim t'$ and $(\sfh_i)_{1 \le i \le n}$ be a sequence of indices of $P$ such that for every $1 \le i < n$, $\sfh_{i+1} \in \cspos_P(\sfh_{i})$ and $\sfh_1 = \epsilon$. Then $n \le |\mathcal{B}(t,t')| + 1$. Moreover $|\prooflabel(P)| \le 2^{|\mathcal{B}(t,t')|}$.
\end{proposition}

\begin{proof}
  Let $l \in \prooflabel(P)$ be such that $\sfh_n \in \hbranch(l)$. The proof consists in building an increasing sequence of $\ekl^P$-normalized basic terms $\beta_1 <_{\st} \dots <_{\st} \beta_m$ from $(\sfh_i)_{1 \le i \le n}$ of length $m \ge n$. We then concludes using Lemma~\ref{lem:bound-depth}.

  If $\sfh_n \ne \epsilon$, then $\sfh_n$ is of the form $h^n_{\sfx_n}$. We know that $\extract_{\sfx_n}(h^n,P)$ is a proof of $b^n \sim b'^n$ in $\mathcal{A}_{\csmb}$. Moreover $b^n\downarrow_R$ is in $\cspath^{\sfh_{n-1},l}(t,P)$ and is $(t,P)$-\abounded. Be definition of $(t,P)$-\abounded terms, we know that there exists $(\beta_{n,j})_{1 \le j \le k_n}$ (with $k_n \ge 1$) such that:
  \begin{itemize}
  \item for all $1 \le j \le k_n$, $\beta_{n,j} \lebt^{\sfh_{n-1},l}(t,P)$.
  \item $b^n\downarrow_R \in \leavest(\beta_{n,1}\downarrow_R)$.
  \item $\beta_{n,k_n} \leleave^{\sfh_{n-1},l}(t,P)$.
  \item for all $1 \le j < k_n$, $\beta_{n,j}$ is a guard of a decryption in $\beta_{n,j+1}$, and therefore $\beta_{n,j} <_{\st} \beta_{n,j+1}$.
  \end{itemize}
  If $\sfh_{n-1} \ne \epsilon$, then since $\beta_{n,k_n} \leleave^{\sfh_{n-1},l}(t,P)$ is $(t,P)$-\abounded, and since for any $\beta \lebt^{\sfh_{n-1},l}(t,P)$, $\beta_{n,j}$ is not a guard of $\beta$, we know that we are in the inductive case with different labels of the definition of  $(t,P)$-\abounded terms. Therefore there exists  $b^{n-1} \in \cspath^{\sfh_{n-2},l}(t,P)$ such that $b^{n-1} \in \leavest(\beta_{n,k_n})$.

  We then iterate this process until we reach $\epsilon$, building sequences $(\beta_{i,j})_{1 < i \le n, 1 \le j \le k_i}$ and $(b^i)_{1 < i \le n}$. Since for all $i$, $b^{i-1} \in \leavest(\beta_{i,k_i}\downarrow_R)$ and $b^{i-1} \in \leavest(\beta_{i-1,1} \downarrow_R)$ we know, using Proposition~\ref{prop:bas-cond-charac}, that $\beta_{i,k_i} \equiv \beta_{i-1,1}$. Therefore we have:
  \[
    \beta_{n,1} <_{\st} \dots <_{\st} \beta_{n,k_n}
    \equiv
    \beta_{n-1,1} <_{\st} \dots <_{\st} \beta_{n-1,k_{n-1}}
    \dots<_{\st}
    \beta_{3,k_{3}} \equiv \beta_{2,1} <_{\st} \dots <_{\st} \beta_{2,k_2}
  \]
  Moreover, for all $i$ we have $k_i \ge 1$, therefore we built an increasing sequence of $\ekl^P$-normalized basic terms of length at least $n-1$. It follows, using Lemma~\ref{lem:bound-depth}, that $n-1 \le |\mathcal{B}(t,t')|$.

  To upper-bound $|\prooflabel(P)|$, we only need to observe that we cannot have two $\csmb$ applications on the same conditional in a given branch. Consider the binary tree associated to the $\csmb$ applications in $P$, labelled by the corresponding $\csmb$ conditionals (say, on the left). Then this tree is of depth at most $|\mathcal{B}(t,t')|$, and therefore has at most $2^{|\mathcal{B}(t,t')|}$ leaves.
\end{proof}

\begin{theorem*}[Main Result]
  The following problem is decidable:\\
  \textbf{Input:} A ground formula $\vec u \sim \vec v$.\\
  \textbf{Question:} Is $\textsf{Ax} \wedge \vec u \not \sim \vec v$ unsatisfiable?
\end{theorem*}

\begin{proof}
  Let $\pvec{u} = u_1,\dots,u_n$, $\pvec{v} = v_1,\dots,v_n$ and:
  \begin{mathpar}
    t \equiv \pair{u_1}{\pair{\dots}{\pair{u_{n-1}}{u_n}}}

    t' \equiv \pair{v_1}{\pair{\dots}{\pair{v_{n-1}}{v_n}}}
  \end{mathpar}
  Using the $\fa_{\pair{\_}{\_}}$ axiom, we know that if $\pvec{u} \sim \pvec{v}$ is derivable then $t \sim t'$ is derivable. Conversely, we show that $t \sim t'$ is derivable then  $\pvec{u} \sim \pvec{v}$ is derivable. For every $3\le i\le n$, let $\rho_i[]$ be the $i$-th projection defined using $\pi_1$ and $\pi_2$ by:
  \begin{mathpar}
    \forall n>i \ge 1,\rho_{i} \equiv \pi_1(\pi_2^{i-1}([]))

    \rho_{n}[] \equiv \pi_2^{n-1}([])
  \end{mathpar}
  Then:
  \[
    \infer[R]{
     \pvec{u} \sim \pvec{v}
   }{
     \infer[\fa^*]{
       (\rho_i[t])_{1\le i \le n} \sim (\rho_i[t'])_{1\le i \le n}
     }{
       t \sim t'
     }
    }
  \]
  Hence $t \sim t'$ is derivable iff $\pvec{u} \sim \pvec{v}$ is derivable. Moreover, the corresponding proof of $\pvec{u} \sim \pvec{v}$ is of polynomial size in the size of the proof of $t \sim t'$. Therefore w.l.o.g. we can focus on the case $|\pvec{u}| = |\pvec{v}| = 1$.

  Let $N = |\st(t\downarrow_R)| + |\st(t'\downarrow_R)|$. Using Proposition~\ref{prop:bound-csmb-nested}, we have bounded the number of branches of the proof tree (by $2^{N^2.2^{N}}$), and the number of nested $\csmb$ conditionals. For every branch, we non-deterministically guesses a set of \abounded basic terms that can appear in a proof $P$ of $P \aproof t \sim t'$ using the valid candidate sequence algorithm (in polynomial time in $\mathcal{O}(N.2^{3.N}.2^{Q(N).2^{4.N}})$, using Proposition~\ref{prop:boundpterms} and Proposition~\ref{prop:cand-seq-term-size}). Then the procedure guesses the rule applications, and checks that the candidate derivation is a valid proof. This is done in polynomial time in the size of the candidate derivation. Remark that to check whether the leaves are valid $\CCA$ instances we use the polynomial-time algorithm describe in Proposition~\ref{prop:cca-small-restr}. Finally, since $|t\downarrow_R|$ is at most exponential with respect to $|t|$, this yields a $3$-\textsc{NExpTime} decision procedure that shows the decidability of our problem.
\end{proof}

%%% Local Variables:
%%% mode: latex
%%% TeX-master: "ms"
%%% End:

\fi
\end{document}

%%% Local Variables:
%%% mode: latex
%%% TeX-master: t
%%% End: